\documentclass[a4paper,UKenglish,cleveref, autoref, thm-restate]{lipics-v2021}
\usepackage{cite}
\usepackage{amsmath,amssymb,amsfonts}
\usepackage{stmaryrd }
\usepackage{textcomp}
\usepackage{proof}
\def\BibTeX{{\rm B\kern-.05em{\sc i\kern-.025em b}\kern-.08em
    T\kern-.1667em\lower.7ex\hbox{E}\kern-.125emX}}

\usepackage{etoolbox}

\usepackage{amsmath,mathtools}
\usepackage{booktabs}

\usepackage[dvipsnames,prologue]{xcolor}
\definecolor{mygreen}{RGB}{23,109,22}
\definecolor{mygray}{RGB}{121,121,121}
\definecolor{mypurple}{RGB}{141,75,141}
\definecolor{mytomato}{RGB}{171,67,62}
 \NewDocumentCommand{\TODO}{m m}%
   {{\bfseries\color{#1}[#2]}}%

 \RenewDocumentCommand{\TODO}{m m}{}

\newcommand{\lang}{\m{SINTEGRITY}}

\newcommand{\maxsec}[1]{\textcolor{BrickRed}{#1}}
\newcommand{\runsec}[1]{\textcolor{OliveGreen}{#1}}

\usepackage{graphicx}

\usepackage{mathpartir}

\usepackage{hyperref}

\Crefname{section}{\S\!}{\S}
\Crefname{figure}{Fig.}{Figs.}
\Crefname{tabular}{Tab.}{Tabs.}
\Crefname{theorem}{Thm.}{Thms.}
\Crefname{lemma}{Lem.}{Lems.}
\Crefname{definition}{Def.}{Defs.}

\newcommand{\eg}{e.g.,~}
\newcommand{\ie}{i.e.,~}

\newcommand{\resp}{resp.,~}
\newcommand{\respb}{resp}

\newcommand{\mc}[1]{\mathcal{#1}}
\newcommand{\mi}[1]{\mathit{#1}}
\newcommand{\m}[1]{\mathsf{#1}}

\newcommand{\mb}[1]{\mathbf{#1}}

\newcommand{\alt}{\,\mid\,}
\newcommand{\meq}{=_\xi}
\newcommand{\defined}{\,\triangleq\,}
\newcommand{\dom}[1]{\m{dom}(#1)}

\newcommand{\D}{\Delta}
\newcommand{\Sig}{\Sigma}

\newcommand{\intchoicesymb}{\oplus}
\newcommand{\extchoicesymb}{\&}
\newcommand{\chanoutsymb}{\otimes}
\newcommand{\chaninsymb}{\multimap}
\newcommand{\intchoice}[1]{\oplus\{\ell{:}#1_{\ell}\}_{\ell \in L}}
\newcommand{\extchoice}[1]{\&\{\ell{:}#1_{\ell}\}_{\ell \in L}}
\newcommand{\chanout}[2]{#1 \otimes #2}
\newcommand{\chanin}[2]{#1 \multimap #2}
\newcommand{\one}{\mb{1}}

\newcommand{\dlabel}[2]{\langle #1, #2 \rangle}
\newcommand{\dlabelr}[2]{{\color{red}\langle #1, #2 \rangle}}
\newcommand{\dlabelb}[2]{\langle #1, #2 \rangle}
\newcommand{\ptyp}[8]{{\color{red}\Psi}; #1 \vdash_{\Sig} (#2) {@ \dlabel{#3}{#4}} :: #5 {:} #6 \dlabel{#7}{#8} }
\newcommand{\ptypr}[8]{{\color{red}\Psi}; #1 \vdash_{\Sig} (#2) {{\color{red}@} \dlabelr{#3}{#4}} :: #5 {:} #6 \dlabelr{#7}{#8} }
\newcommand{\ptypb}[8]{\Psi; #1 \vdash_{\Sig} #2 @ \dlabel{#3}{#4} :: #5 {:} #6 \dlabel{#7}{#8} }
\newcommand{\ptypbr}[8]{{\color{red}\Psi}; #1 \vdash_{\Sig} #2 {\color{red}@} \dlabelr{#3}{#4} :: #5 {:} #6 \dlabelr{#7}{#8} }
\newcommand{\latcstr}[1]{\Psi \Vdash #1}
\newcommand{\latcstrnot}[1]{\color{mytomato}{\Psi \not\Vdash #1}}
\newcommand{\barbeq}[4]{\Psi \vDash #1 \sim_{\dlabelb{#2}{#3}} #4}
\newcommand{\outact}[4]{\uparrow_{#1^{\dlabelb{#2}{#3}}}.#4}
\newcommand{\inact}[4]{\downarrow_{#1^{\dlabel{#2}{#3}}}.#4}
\newcommand{\subst}[3]{[#1 / #2] #3}
\newcommand{\substmap}[3]{#1 \Vdash #2 : #3}
\newcommand{\substmapap}[2]{\hat{#1}(#2)}

\newcommand{\typdef}[2]{#1 = #2}
\newcommand{\typwfmd}[1]{\Vdash_{\Sig} #1\; \mb{wfmd}}
\newcommand{\procdef}[9]{{\it \color{red}\Psi'}; #1 \vdash_{\Sig} #2 = #3 {\it\color{red}@} \dlabel{#4}{#5} :: #6 {:} #7 \dlabel{#8}{#9}}
\newcommand{\procdefr}[9]{{\it \color{red}\Psi'}; #1 \vdash_{\Sig} #2 = #3 {\it\color{red}@} \dlabelr{#4}{#5} :: #6 {:} #7 \dlabelr{#8}{#9}}
\newcommand{\procdefb}[9]{{\it \color{red}\Psi}; #1 \vdash_{\Sig} #2 = #3 @ \dlabel{#4}{#5} :: #6 {:} #7 \dlabel{#8}{#9}}

\newcommand{\procdefbsr}[9]{{\it \color{red}\Psi}; #1 \vdash #2 = #3 {\it \color{red}@} \dlabelr{#4}{#5} :: #6 {:} #7 \dlabelr{#8}{#9}}
\newcommand{\sigwfmd}[1]{\Vdash_{\Sigma; {\it \color{red}\Psi_0}} #1 \; \mathbf{sig}}
\newcommand{\labsnd}[3]{#1.#2; #3}
\newcommand{\labrcv}[2]{\mb{case}\, #1 (\ell {\Rightarrow} #2_{\ell})_{\ell \in L}}
\newcommand{\labrcvb}[2]{\mb{case}\, #1 (\ell {\Rightarrow} #2_{\ell})_{\ell \in I}}
\newcommand{\chnsnd}[3]{\mb{send}\, #1\, #2; #3}
\newcommand{\chnrcv}[3]{#1 {\leftarrow} \mb{recv}\,#2; #3_{#1}}
\newcommand{\cls}[1]{\mb{close}\, #1}
\newcommand{\wait}[2]{\mb{wait}\,#1; #2}
\newcommand{\spwn}[9][\gamma]{(#2^{\dlabel{#3}{#4}} \leftarrow #5 [#1] \leftarrow #6)@\dlabel{#7}{#8}; #9_{#2}}
\newcommand{\spwnb}[9][\gamma]{(#2^{\dlabelb{#3}{#4}} \leftarrow #5 [#1] \leftarrow #6)@\dlabelb{#7}{#8}; #9_{#2}}
\newcommand{\fwd}[2]{#1 \leftarrow #2}
\newcommand{\cproj}[2]{#1{\Downarrow}#2}

\newcommand{\mnode}{\color{black}{\mathcal{D}}}

\newcommand{\msc}[1]{\mbox{\sc #1}}

\newcommand{\ctype}[4]{#1; #2 \Vdash #3 :: #4}
\newcommand{\proc}[3]{\mathbf{proc}(#1, #2\, @ #3)}
\newcommand{\procb}[6]{\proc{#1\dlabel{#2}{#3}}{#4}{\dlabel{#5}{#6}}}
\newcommand{\msg}[1]{\mathbf{msg}(#1)}

\newcommand{\ostep}{\textcolor{OliveGreen}{\;\; \mapsto \;\;}} % open step
\newcommand{\ostepb}{\textcolor{OliveGreen}{\mapsto \;\;}} % open step
\newcommand{\fwder}[1]{\mathtt{Fwd}_{#1}}

\nolinenumbers
\usepackage{listings}
\usepackage{mathtools}
\makeatletter
\lst@Key{matchrangestart}{f}{\lstKV@SetIf{#1}\lst@ifmatchrangestart}
\def\lst@SkipToFirst{%
    \lst@ifmatchrangestart\c@lstnumber=\numexpr-1+\lst@firstline\fi
    \ifnum \lst@lineno<\lst@firstline
        \def\lst@next{\lst@BeginDropInput\lst@Pmode
        \lst@Let{13}\lst@MSkipToFirst
        \lst@Let{10}\lst@MSkipToFirst}%
        \expandafter\lst@next
    \else
        \expandafter\lst@BOLGobble
    \fi}
\makeatother

\lstdefinelanguage{sintegrity}{
    language=C++,
    morestring=[b]",
    morecomment=[s]{(*}{*)},
    morekeywords=[1]{
        secrecy, theory, signature, exec, end,
        type, stype, session,
        send, forward, select, case, receive, instantiate, spawn,
        wait, close,
        proc, provide, using, at, to,
        unit, &, *, -o, +, |_|
    },
    morekeywords=[2]{
        \#bank, \#alice, \#bob, \#guest
    },
    otherkeywords={
        \#bank, \#alice, \#bob, \#guest,
        &, *, -o, +, |_|
    },
}

\definecolor{background_color}{RGB}{240, 240, 240}
\definecolor{string_color}    {RGB}{180, 156,   0}
\definecolor{keyword_color}   {RGB}{ 64, 100, 255}
\definecolor{lattice_color}   {RGB}{180, 156,   0}
\definecolor{comment_color}   {RGB}{  0, 117, 110}
\definecolor{number_color}    {RGB}{ 84,  84,  84}
\DeclareCaptionFormat{listing}{\hfill#1#2#3\vskip1pt}
\newrobustcmd{\code}[2][]{{\sloppy
\ifmmode
    \text{\colorbox{background_color}{\lstinline[language=sintegrity,#1]`#2`}}
\else
    {\colorbox{background_color}{\lstinline[language=sintegrity,#1]`#2`}}%
\fi}}
\lstnewenvironment{codeblock}[1][]{\lstset{language=sintegrity,numbers=none,#1}}{}

\newcommand{\bnfdef}{\mathrel{::=}}
\newcommand{\bnfalt}{\mathrel{\mid}}

\newenvironment{synchart}[1]%
{\begin{equation*}
      \label{eqn:#1}
      \renewcommand{\bnfalt}{\mathrel{\phantom{\mid}}}
      \begin{array}{llc@{\quad\extracolsep{\fill}}lll}
}
{\end{array}\end{equation*}\ignorespacesafterend}

\newenvironment{synchart*}[1]%
{\begin{equation*}
      \label{syn:#1}
      \renewcommand{\bnfalt}{\mathrel{\phantom{\mid}}}
      \begin{array}{llc@{\quad\extracolsep{\fill}}lll}
        \textit{Sort} & & & \textit{Abstract} & \textit{Concrete} & \textit{Description} \\
}
{\end{array}\end{equation*}\ignorespacesafterend}

\usepackage{pdfpages}

\title{Regrading Policies for Flexible Information Flow Control in Session-Typed Concurrency}

\author{Farzaneh Derakhshan}{
  Illinois Institutie of Technology, USA
}{
  fderakhshan@iit.edu
}{
  https://orcid.org/0000-0002-2156-2606
}{% Funding ack
}
\author{Stephanie Balzer}{
  Carnegie Mellon University, USA
}{
  balzers@cs.cmu.edu
}{
  https://orcid.org/0000-0002-8347-3529
}{Supported in part by the Air Force Office of Scientific Research under award number FA9550-21-1-0385 (Tristan Nguyen, program manager).
Any opinions, findings and conclusions or recommendations expressed here are those of the author(s) and do not necessarily reflect the views of the U.S. Department
of Defense.
}
\author{Yue Yao}{Carnegie Mellon University, USA}{yueyao@cs.cmu.edu}{https://orcid.org/0000-0001-8523-5156}{}

\authorrunning{Derakhshan, Balzer, and Yao} %TODO mandatory. First: Use abbreviated first/middle names. Second (only in severe cases): Use first author plus 'et al.'

\Copyright{Farzaneh Derakhshan, Stephanie Balzer, and Yue Yao}

\begin{document}

\maketitle

% force page numbers, remove for final
\thispagestyle{plain}
\pagestyle{plain}

\begin{abstract}
Noninterference guarantees that an attacker cannot infer secrets by interacting with a program.
Information flow control (IFC) type systems assert noninterference by tracking the level of information learned (\textsf{pc}) and
disallowing communication to entities of lesser or unrelated level than the \textsf{pc}.
Control flow constructs such as loops are at odds with this pattern because they necessitate downgrading
the \textsf{pc} upon recursion to be practical.
In a concurrent setting, however, downgrading is not generally safe.
This paper utilizes \emph{session types} to track the flow of information
and contributes an IFC type system for message-passing concurrent processes
that allows downgrading the pc upon recursion.
To make downgrading safe, the paper introduces \emph{regrading policies}.
Regrading policies are expressed in terms of integrity labels,
which are also key to safe composition of entities with different regrading policies.
The paper develops the type system and proves \emph{progress-sensitive noninterference} for well-typed processes,
ruling out timing attacks that exploit the relative order of messages.
The type system has been implemented in a type checker,
which supports security-polymorphic processes.

\end{abstract}

\section{Introduction}
\label{sec:intro}
With the emergence of new applications, such as Internet of Things and cloud computing,
today's software landscape has become increasingly \emph{concurrent}.
A dominant computation model adopted by such applications is \emph{message passing},
where several concurrently running processes connected by channels exchange messages.
A further common aspect is the need for security,
ensuring that confidential information is not leaked to a (malevolent) observer.

\emph{Information flow control} (IFC) type systems
\cite{VolpanoARTICLE1996,SmithVolpanoPOPL1998,SabelfeldIEE2003}
rule out information leakage by type checking.
These systems statically track the level of information learned by an entity and
disallow propagation to parties of lesser or unrelated levels, given a security lattice.
The ultimate property to be asserted by an IFC type system is noninterference,
a program equivalence statement up to the confidentiality level of an observer.
The gold standard is \emph{progress-sensitive noninterference} (PSNI) \cite{HeidinSabelfeldMartkoberdorf2011},
which treats divergence as an observable outcome.
PSNI thus only equates a divergent program with another diverging one,
whereas \emph{progress-insensitive noninterference} (PINI) regards divergence to be equal to any outcome.
Especially in a concurrent setting, PSNI is a sine qua non
because the \emph{termination channel} \cite{SabelfeldIEE2003} can be scaled to many parallel computations,
each leaking ``just'' one bit~\cite{StefanICFP2012,AlpernasOOPSLA2018}.

Guaranteeing PSNI, or even PINI for that matter,
can become both a blessing and a curse in a concurrent setting.
To ensure such a strong property, IFC type systems have to be very restrictive.
The troublemakers, in particular, are control flow constructs, such as loops and if statements.
Whereas IFC type systems for sequential languages allow the \textsf{pc} label\footnote{{The pc (program counter) label approximates the level of confidential information learned up to the current execution point.}}
to be lowered to its previous level for the continuation of a control flow construct,
even if the construct itself runs at high,
this treatment is no longer safe in a concurrent setting \cite{SmithVolpanoPOPL1998}.
To uphold noninterference, IFC type systems for concurrent languages typically forbid high-security loop guards
and may even put restrictions on if statements, depending on thread scheduling and attacker model
\cite{SmithVolpanoPOPL1998,SabelfeldSandsCSFW2000,SabelfeldMantelSAS2002}.

The use of linearity provides some relief
\cite{ZdancewicMyersESOP2001,ZdancewicMyersARTICLE2002,ZdancewicMyersCSFW2003,DerakhshanLICS2021,BalzerARXIV2023},
allowing high-security loop guards.
Linearity also facilitates race freedom, key to guaranteeing observational determinism
and, thus, the absence of certain timing attacks
\cite{ZdancewicMyersCSFW2003,DerakhshanLICS2021,BalzerARXIV2023}.
A family of concurrent languages that employ linearity are \emph{session types}
\cite{HondaCONCUR1993,HondaESOP1998,HondaPOPL2008,CairesCONCUR2010,WadlerICFP2012}.
Session types are used for message-passing concurrency,
typically in the context of process calculi, where concurrently running processes communicate along channels.
A distinguishing characteristic of session types is their ability to assert \emph{protocol adherence}.
A session-typed channel prescribes not only the types of values that can be transported over the channel
but also their relative sequencing.

In this paper, we develop a flow-sensitive IFC session type system that not only supports {recursive processes with arbitrary recursion guards, including high-security ones,}
but also identifies synchronization patterns that make it safe for the {process body} to
downgrade to the initial \textsf{pc} level upon recursion.
We refer to this adjustment of confidentiality level as \emph{regrading}.
To enforce the safety of regrading, we complement confidentiality with \emph{integrity} \cite{BibaTR1977}.
Integrity allows prescribing a process a \emph{regrading policy},
ensuring that any confidential information learned during the high-security parts of the loop
cannot be rolled forward to the next iteration.
Processes are \emph{polymorphic} in the confidentiality and integrity labels,
ensuring maximal flexibility of the IFC type system.

We contribute a type checker for our IFC type system, yielding the language $\lang$.
The type checker supports security-polymorphic processes using local \emph{security theories}.
Well-typed processes in $\lang$ enjoy PSNI.
To prove this result, we develop a \emph{logical relation} for integrity,
showing that well-typed processes are self-related (fundamental theorem, \Cref{lem:reflexivity}).
We then prove that the logical relation is closed under parallel composition and
that related processes are bisimilar (adequacy theorem, \Cref{thm:adeq}).

Regrading is related to robust declassification~\cite{ZdancewicMyersCSFW2001,ZdancewicARTICLE2002,ZdancewicMFPS2003,MyersARTICLE2006,ChongMyersCSFW2006,AskarovMyersARTICLE2011}, as both allow downgrading the \textsf{pc} using integrity.
In contrast to declassification, which deliberately releases information and thus intentionally weakens noninterference,
regrading preserves noninterference.
The distinction also manifests in how integrity is used.
Whereas integrity is used in robust declassification to convey how trustworthy the information is on which a regrading decision is based,
integrity in our work is used to impose extra synchronization policies on regrading processes to prevent leakage by downgrading the \textsf{pc} upon recursion.
As such, regrading constitutes a more permissive IFC mechanism.

\textbf{Contributions:}
\begin{itemize}
\item The notion of a regrading policy to downgrade a process' confidentiality, retaining PSNI.
\item The language $\lang$, a flow-sensitive IFC session type system for asynchronous message-passing with confidentiality and integrity to support regrading policies.
\item A logical relation for integrity to prove that $\lang$ processes satisfy PSNI.
\item A type checker for $\lang$, available as an artifact.
\end{itemize}

The complete formalization with proofs is available in the appendix.

\section{Motivating example and background}
\label{sec:motivating-examples}
This section provides an introduction to session-typed programming and IFC control
based on a running example.
Our language $\lang$ is an \emph{intuitionistic linear session type} language
\cite{CairesCONCUR2010,ToninhoESOP2013};
thus, our presentation is specific to that family of session types.

We use a simple bank survey as an example.
The survey is carried out by an analyzer at a bank to decide whether to buy or sell a share of stock.
The analyzer's decision depends on the opinion of two groups of participants, queried by two surveyors,
and a strategy provided by a tactician.
For simplicity, we assume that each group of participants only consists of one participant,
and the surveyors simply pass the opinion of their participant to the analyzer.

\begin{figure*}[ht]
\centering
\includegraphics[width=1\textwidth]{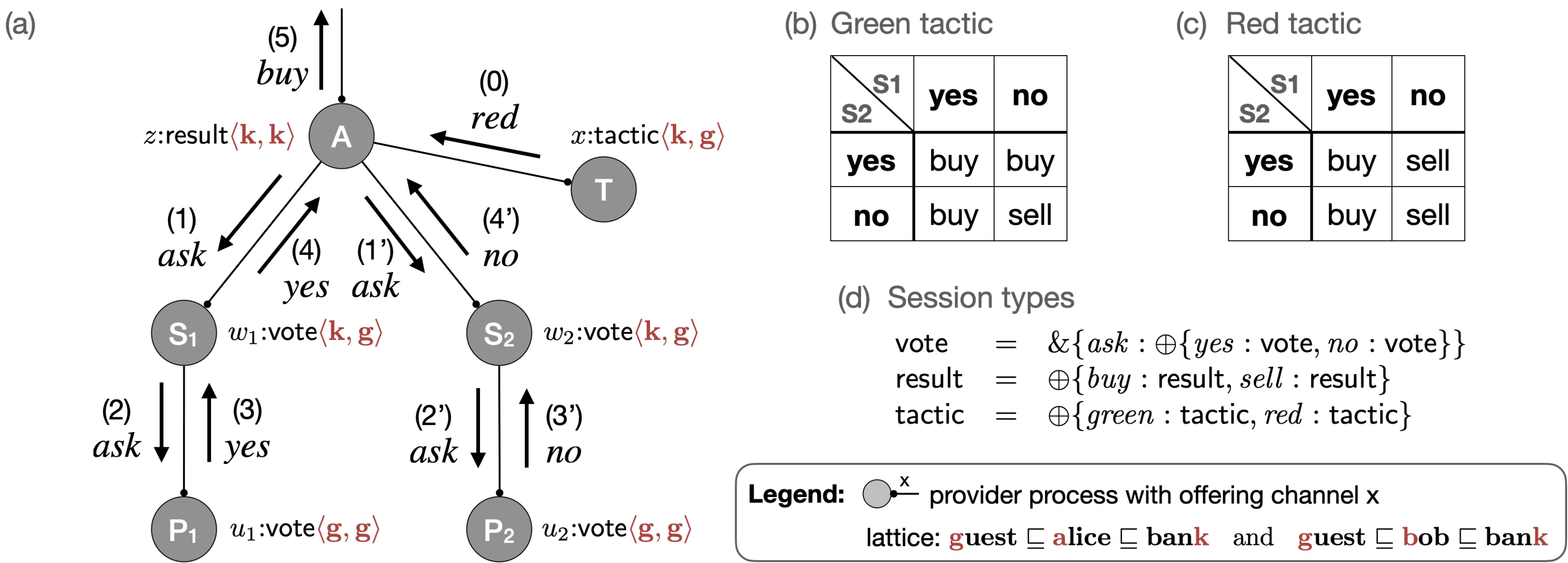}
\vspace{-5pt}
\caption{Bank survey: (a) process configuration, (b)/(c) red/green tactic, (d) session types.}
\label{fig:analyzer}
\vspace{-10pt}
\end{figure*}

A runtime configuration of processes for this example is shown in \Cref{fig:analyzer}(a):
the analyzer process {\small$\m{A}$}, the tactician process $\m{T}$, the surveyor processes {\small$\m{S_1}$} and {\small$\m{S_2}$},
along with their participant processes {\small$\m{P_1}$} and {\small$\m{P_2}$}, \respb.
The processes are connected by the channels {\small$u_1$}, {\small$u_2$}, {\small$w_1$}, {\small$w_2$}, {\small$x$}, and {\small$z$}.
The figure shows the communications between the analyzer, surveyors, and participants along these channels, with arrows indicating the message being exchanged.
The analyzer sends the message {\small$\mi{ask}$} to surveyor {\small$\m{S_1}$} to request a poll (1).
Surveyor {\small$\m{S_1}$} then sends the message {\small$\mi{ask}$} to participant {\small$\m{P_1}$} to get their opinion about buying a share (2).
Once the surveyor receives {\small$\m{P_1}$}'s vote (\ie either {\small$yes$} or {\small$no$}) (3),
it relays the vote back to the analyzer (4).
The analogous communication pattern is repeated between the analyzer and surveyor {\small$\m{S_2}$} and participant {\small$\m{P_2}$} (1'--4').
The final decision whether to {\small$\mi{buy}$} or {\small$\mi{sell}$} (5) of the analyzer is based on the tactic provided by the tactician.
For simplicity, we assume that the tactician chooses either a {\small$\mi{green}$} or {\small$\mi{red}$} tactic (0). 
In the green tactic, the analyzer decides to buy the share if at least one of the surveyors votes yes. 
In the red tactic, the analyzer buys the stock if the first surveyor votes to buy,
regardless of the opinion of the second one (see \Cref{fig:analyzer}(b-c)). 

The protocols for these communications can be specified by the session types shown in \Cref{fig:analyzer}(d),
using the connectives of \Cref{tab:session_types}.
The connectives are drawn from intuitionistic linear logic and obey the following grammar:

{\small\(A,B::= {\oplus}\{\ell : A_\ell\}_{\ell \in L} \mid {\&}\{\ell : A_\ell\}_{\ell \in L}  \mid A \otimes B \mid A \multimap B \mid 1 \mid Y,\)}

\noindent 
where {\small$L$} ranges over finite sets of labels denoted by {\small$\ell$} and {\small$k$}, {amounting to primitive values in our system.}
Type variable {\small$Y$} is a fixed point whose definition {\small$Y=A$} is given in a global signature {\small$\Sigma$}.
The latter is used to define general recursive types.
Recursive types must be \emph{contractive}~\cite{GayHoleARTICLE2005}, demanding a message exchange before recurring,
and \emph{equi-recursive}~\cite{CraryPLDI1999},
avoiding explicit (un)fold messages and relating types up to their unfolding.
All three types {\small$\m{vote}$}, {\small$\m{result}$}, and {\small$\m{tactic}$} are recursive.

\begin{table*}
\centering
\caption{$\lang$ constructs. Upper half: types and terms (\textbf{b}efore and \textbf{a}fter exchange), operational meaning, and \textbf{p}olarity.
Lower half: spawn and forward terms and operational meaning.}
\label{tab:session_types}
\vspace{-3pt}
\begin{small}
    \begin{tabular}{@{}lllllr@{}}
    \toprule
    \multicolumn{2}{@{}l}{\textbf{Session type (b/a)}} &
    \multicolumn{2}{l}{\textbf{Process term (b/a)}} &
    \textbf{Description}  \\[3pt]
    %%%
    $x {:} \intchoice{A}$ & $x {:} A_k$ & $\labsnd{x}{k}{P}$ &
    $P$ & provider sends label $k$ along $x$, continues with $P$  \\
    %%%
     & & $\labrcv{x}{Q}$ & $Q_k$ &
    \multicolumn{2}{l}{client receives label $k$ along $x$, continues with $Q_k$} \\[2pt]
    %%%
    $x {:} \extchoice{A}$ & $x {:} A_k$ & $\labrcv{x}{P}$ &
    $P_k$ & provider receives label $k$ along $x$, continues with $P_k$  \\
    %%%
     & & $\labsnd{x}{k}{Q}$ & $Q$ & \multicolumn{2}{l}{client sends label $k$ along $x$, continues with $Q$}  \\[2pt]
    %%%
    $x {:} \chanout{A}{B}$ & $x {:} B$ & $\chnsnd{y}{x}{P}$ &
    $P$ & provider sends channel $y {:} A$ along $x$, continues with $P$   \\
    %%%
     & & $\chnrcv{z}{x}{Q}$ & ${Q_y}$ &
    \multicolumn{2}{l}{client receives channel $y {:} A$ along $x$, continues with $Q_y$}  \\[2pt]
    %%%
    $x {:} \chanin{A}{B}$ & $x {:} B$ & $\chnrcv{z}{x}{P}$ & ${P_y}$ &
    provider receives channel $y {:} A$ along $x$, continues with $P_y$  \\
    %%%
     & & $\chnsnd{y}{x}{Q}$ & $Q$ &
    \multicolumn{2}{l}{client sends channel $y {:} A$ along $x$, continues with $Q$}  \\[2pt]
    %%%
    $x {:} \one$ & - & $\cls{x}$ &
    - & provider sends ``$\m{close}$'' along $x$ and terminates\\
    %%%
     & & $\wait{x}{Q}$ & $Q$ & \multicolumn{2}{l}{client receives ``$\m{close}$'' along $x$, continues with Q} \\[2pt]
    $x:Y$ & $x:A$ & - & - & recursive type definition $Y = A$ ($Y$ occurs in $A$) \\[3pt]
    %%%
    %%%
    \multicolumn{6}{@{}l}{\textbf{Judgmental rules}}\\[3pt]
    %%%
    \multicolumn{4}{@{}l}{$\spwn{x}{c}{e}{X}{\D}{c_0}{e_0}{Q}$} & \multicolumn{2}{l}{spawn $X$ along $x^{\dlabel{c}{e}}$ with arguments $\D$, substitution $\gamma$,} \\
    & & & & \multicolumn{2}{l}{and running security $\dlabel{c_0}{e_0}$, then continue with $Q_x$}  \\[2pt]
    %%%
    \multicolumn{4}{@{}l}{$\fwd{x}{y}$} & \multicolumn{2}{l}{forward $x {:} A$ to $y {:} A$}\\
    % %%%
    \bottomrule

\end{tabular}
\vspace{-15pt}
\end{small}
\end{table*}

\Cref{tab:session_types} provides an overview of $\lang$ types and terms.
A crucial characteristic of session-typed processes is that
a process \emph{changes} its type along with the messages it exchanges.
A process' type therefore always reflects the current protocol state.
\Cref{tab:session_types} lists state transitions caused by a message exchange in columns 1 and 2 with
corresponding process terms in columns 3 and 4.
Column 5 describes the computational behavior of a type.

Linearity ensures that every channel connects exactly two processes, thus imposing a \emph{tree} structure on a configuration of processes, as witnessed by \Cref{fig:analyzer}(a). 
We adopt a form of session types corresponding with intuitionistic linear logic, which moreover
introduces a distinction between the two processes connected by a channel, identifying one as the \emph{parent} and the other as the \emph{child}, turning the configuration into a \emph{rooted tree}.
The parent and child processes have mutually dual perspectives on the protocol of their connecting channel:
The child has the perspective of the \emph{provider} and the parent that of a \emph{client}.
Column 5 of \Cref{tab:session_types} describes the perspective of the client and provider for each type.
We assign a polarity to each session type which determines whether the type has a sending semantics or a receiving semantics.
% 
% %
For positive types, the provider sends, and the client receives; for negative types, the provider receives, and the client sends.
The types with positive polarity are $\intchoice{A}$, $\chanout{A}{B}$, and $1$, and the types with negative polarity are $\extchoice{A}$ and $\chanin{A}{B}$.

Each process in a configuration is uniquely identified by the channel that connects it to its parent, which we also refer to as its \emph{offering} (or \emph{providing}) channel. We consider the session type of a process to be the protocol of its offering channel.
For example, the participant process {\small$\m{P_1}$} in \Cref{fig:analyzer} has type {\small$\m{vote}$},
which is also the type of the process' offering channel {\small$u_1$} that connects {\small$\m{P_1}$} to its client {\small$\m{S_1}$}.
We say that the client {\small$\m{S_1}$} \emph{uses} the channel {\small$u_1$}.

The connectives {\small$\chanoutsymb$} and {\small$\chaninsymb$}, not used in the example,
allow sending channels along channels.
Such higher-order channels change the connectivity structure of a configuration: from the perspective of the provider, 
{\small$\chanoutsymb$} turns a child into a sibling and {\small$\chaninsymb$} a sibling into a child. The former is achieved by sending a subtree to the parent and the latter by receiving a subtree from the parent.
\Cref{sec:bank-examples} showcases an example that uses higher-order channels.

It is now time to explore how to implement the processes of our bank survey example.
\Cref{fig:survey-impl} gives the process definitions of the analyzer and surveyor.
A process definition consists of the process signature (first two lines) and body (after {\small$=$}).
The first line indicates the typing of channel variables used by the process (left of {\small$\vdash$}) and the type of the providing channel variable (right of {\small$\vdash$}).
The former are going to be child nodes of the process.
The second line binds the channel variables.
In {\small$\lang$}, {\small$\leftarrow$} generally denotes variable bindings. 
The channels and the process definitions  are annotated with confidentiality and integrity levels (\eg {\small$\maxsec{\dlabel{\mb{bank}}{\mb{guest}}}$} and {\small$\runsec{@\dlabel{\mb{guest}}{\mb{guest}}}$}).
We will later describe the meaning of these annotations; the reader can safely ignore them for now.

\begin{figure*}
\centering
\begin{small}
\input{figs/survey-impl.tex}
\end{small}
\caption{Secure process implementations of analyzer and surveyor (see \Cref{fig:analyzer}), accepted by $\lang$,
but rejected by existing IFC session type systems.}
\vspace{-5pt}
\label{fig:survey-impl}
\vspace{-12pt}
\end{figure*}

The analyzer first waits to receive a tactic from the tactician along channel {\small$x$}.
In either branch (\ie {\small$\mi{green}$} or {\small$\mi{red}$}), the analyzer proceeds by requesting a vote from surveyors {\small$\m{S_1}$} and {\small$\m{S_2}$},
after which it communicates its decision along its offering channel {\small$z$} before recurring.
We remark that the notation {\small$z\leftarrow \m{A}\leftarrow w_1, w_2, x$} used for a tail call does not precisely match up with \Cref{tab:session_types}
because we are deferring a discussion of security annotations
and substitutions for security-polymorphic processes to \Cref{sec:term-typing}.
Moreover, a tail call is syntactic sugar for a spawn combined with a forward;
\ie {\small$z\leftarrow \m{A}\leftarrow w_1, w_2, x$} desugars to {\small$z'\leftarrow \m{A}\leftarrow w_1, w_2, x; \fwd{z}{z'}$}.

We implement a surveyor by two processes {\small$\m{S}$} and {\small$\m{S}'$} to take advantage of $\lang$'s support for regrading,
as we will detail in \Cref{sec:motivating-examples-down}.
The surveyor starts out as process {\small$\m{S}$} and calls process {\small$\m{S}'$} right after having received the request from its parent, the analyzer.

Suppose that the tactic is a secret that a participant shall not deduce. 
The implementations in \Cref{fig:survey-impl} respect this security condition:
the analyzer interacts with the participants via the surveyors the same
regardless of the tactic it received.
Existing IFC session type systems \cite{DerakhshanLICS2021,BalzerARXIV2023}, however,
reject these implementations,
because they view the analyzer as tainted as soon as it learns the secret tactic,
and disallow further communication with the participants via the surveyor.
This paper relaxes this restriction---while preserving PSNI---and allows
the tainted surveyor to interact with the participants while putting safeguards in place (synchronization patterns, \Cref{sec:downgradingpolicy}--\Cref{sec:reclassification-nutshell} and \Cref{sec:sync-patterns}) that prevent the surveyor from leaking the tactic.

\section{Key ideas}
\label{sec:key-ideas}
This section develops the main ideas underlying our flexible IFC session type system;
the type system and dynamics is given in \Cref{sec:type-system}.
The latter is asynchronous, \ie non-blocking sends and blocking receives (see \Cref{sec:vanilla-dynamics} and \Cref{sec:def-dynamics}). 
An asynchronous semantics allows for a more permissive noninterference statement since message receipt is not observable.

It may be helpful to foreshadow our attacker model (detailed in \Cref{subsubsec:attacker-model}).
We assume that an attacker knows the implementation of all processes and can observe messages sent over channels with lower or equal confidentiality level than the attacker. 
The attacker cannot measure time but can observe the relative order in which messages are sent along different observable channels.
As we aim for PSNI, we need to ensure that an attacker is unable to deduce any information from non-reactiveness either.

\subsection{Regrading confidentiality}\label{sec:motivating-examples-down}

It is now time to consider the red annotations {\small$\maxsec{\dlabel{\mb{c}}{\mb{e}}}$} on channels
and the green annotations {\small$\runsec{@\dlabel{\mb{c_0}}{\mb{e_0}}}$} on process terms in \Cref{fig:survey-impl},
where {\small$\mb{c}$}, {\small$\mb{d}$}, {\small$\mb{c_0}$}, and {\small$\mb{e_0}$} range over levels in the security lattice
{\small$\mb{guest} \sqsubseteq \mb{alice} \sqsubseteq \mb{bank}$} and {\small${\mb{guest} \sqsubseteq \mb{bob} \sqsubseteq \mb{bank}}$}.
We focus on the first components {\small$\maxsec{\mb{c}}$} and {\small$\runsec{\mb{c_0}}$} for now,
which denote confidentiality labels.
They are 
% in line with
{ adopted from}
existing IFC session type systems \cite{DerakhshanLICS2021,BalzerARXIV2023},
which are based solely on confidentiality.

The first component {\small$\maxsec{\mb{c}}$}  of the pair {\small$\maxsec{\dlabel{\mb{c}}{\mb{e}}}$}
indicates the \emph{maximal confidentiality} of a process, \ie the maximal level of secret information the process may ever obtain.  
As to be expected, the analyzer ({\small$\m{A}$}), the tactician ({\small$\m{T}$}), and both surveyors ({\small$\m{S}_1$} and {\small$\m{S}_2$}) have maximal confidentiality {\small$\maxsec{\mb{bank}}$},
as they are affiliated with the bank and have the clearance of knowing the secret tactic.
The processes associated with the participants have the lowest maximal confidentiality {\small$\maxsec{\mb{guest}}$},
as they must not gain any information about the bank's secrets. 

The first component {\small$\runsec{\mb{c_0}}$} of the pair {\small$\runsec{@\dlabel{\mb{c_0}}{\mb{e_0}}}$}
denotes a process' \emph{running confidentiality}.
It denotes the highest level of secret information a process has obtained so far
and thus is analogous to the \textsf{pc} label in imperative languages,
making the type system flow-sensitive.
The running confidentiality is capped by the maximal confidentiality, \ie {\small$\runsec{\mb{c_0}} \sqsubseteq \maxsec{\mb{c}}$}.
When defining a process, a programmer must indicate the process' maximal confidentiality as well as the \emph{initial} running confidentiality at which the process starts out when spawned.

An IFC type system increases
the running confidentiality accordingly, whenever information of higher confidentiality is received, and
disallow sends from senders with a higher or incomparable running confidentiality than the recipient.
For example, the analyzer starts with the running confidentiality {\small$\runsec{\mb{guest}}$}.
When it receives the secret from the tactician, its running confidentiality increases to {\small$\runsec{\mb{bank}}$}. 
After the receive, the analyzer can still send the message {\small$\mi{ask}$} to a surveyor as the maximal confidentiality of the surveyor is {\small$\maxsec{\mb{bank}}$}. 
However, as soon as the surveyor receives this message from the analyzer, its running confidentiality increases to {\small$\runsec{\mb{bank}}$},
which prevents it from sending messages to participants, whose maximal confidentiality is {\small$\maxsec{\mb{guest}}$},
because {\small$\runsec{\mb{bank}} \not\sqsubseteq \maxsec{\mb{guest}}$}.

To address this limitation of existing IFC session type systems,
we develop \emph{regrading policies}.
A regrading policy is polymorphic in a level {\small$f$} of the security lattice and certifies that, when regrading the running confidentiality to {\small$f$},
any secrets of confidentiality {\small$d_s \not\sqsubseteq f$} learned so far
will not affect future communications of confidentiality at most {\small$f$} after regrading.

To convey the regrading policy that a process must obey,
we 
%use
{ introduce}
\emph{integrity} annotations,
amounting to the second components in the pairs {\small$\runsec{@\dlabel{\mb{c_0}}{\mb{e_0}}}$} and {\small$\maxsec{\dlabel{\mb{c}}{\mb{e}}}$}.
We refer to {\small$\runsec{\mb{e_0}}$} as the \emph{running integrity} of the process
and to {\small$\maxsec{\mb{e}}$} as the \emph{minimal integrity} of the process.
The running integrity specifies what level a process is allowed to regrade to
and is capped by the minimal integrity, \ie {\small$\runsec{\mb{e_0}} \sqsubseteq \maxsec{\mb{e}}$}.
For example, the surveyor process {\small$\m{S}$} runs at {\small$\runsec{@\dlabel{\mb{bank}}{\mb{guest}}}$} after having received the request from the analyzer,
where the running integrity {\small$\runsec{\mb{guest}}$} licenses it to drop its running confidentiality as low as {\small$\runsec{\mb{guest}}$} upon tail-calling,
but forces it to obey that policy too.
The minimal integrity {\small$\maxsec{\mb{e}}$} of a process is naturally capped by its maximal confidentiality {\small$\maxsec{\mb{c}}$},
\ie {\small$\maxsec{\mb{e}} \sqsubseteq \maxsec{\mb{c}}$}, because a process cannot learn (and thus drop) more secrets than it is licensed to.
As a result, a process with maximal confidentiality and minimal integrity {\small$\maxsec{\dlabel{\mb{c}}{\mb{c}}}$} effectively amounts to a non-regrading process.

{We draw both integrity and confidentiality levels from the same security lattice,
but interpret integrity levels \emph{dually}, as usual: the lower a level in the lattice, the higher its integrity\footnote{{
We adopt the following convention to avoid any confusion: we use ``running integrity'', ``minimal integrity'', and ``integrity level'' for elements in the security lattice, and otherwise just ``integrity''.  Thus, when the integrity level in the lattice increases, the integrity decreases.}}.
For regrading this means that the lower the level a process regrades to, the stricter the process' policy becomes.
The $\lang$ type system thus increases the running integrity of a process upon receiving from a process with a higher minimal integrity and
disallows sends from a process of a higher or incomparable running integrity than the minimal integrity of the recipient (see \Cref{sec:type-system}).}
 
The process definitions in \Cref{fig:survey-impl} only use concrete levels from the security lattice for confidentiality and integrity annotations.
To increase code reusability, $\lang$ supports \emph{security-polymorphic} process definitions.
Such definitions range over security variables for confidentiality and integrity {levels}
and may state constraints on these variables.
The constraints must be satisfied upon spawning,
which is checked by the $\lang$ type checker using a security theory.
\Cref{sec:type-system} expands on security-polymorphic process definitions.

\subsection{The need for regrading policies}\label{sec:downgradingpolicy}

While a regrading policy licenses regrading, it also imposes restrictions on a process' communication patterns to guarantee noninterference.
To distill these restrictions,
we next explore insecure implementations of the analyzer-surveyor example from \Cref{sec:motivating-examples} that do not satisfy PSNI.

\begin{figure*}
\centering
\begin{small}
\input{figs/insecure-impl.tex}
\end{small}
 \vspace{-5pt}
\caption{Insecure hasty analyzer $\m{A_H}$ and reckless analyzer $\m{A_R}$, rejected by $\lang$.}
 \vspace{-20pt}
\label{fig:insecure-impl}
\end{figure*}

\subsubsection{Hasty analyzer - optimization may introduce a timing attack}

In the red tactic, the decision of the analyzer does not depend on the {result} provided by the second surveyor.
Hence, one may be tempted to optimize the analyzer implementation
by refraining from asking the opinion of the second surveyor in the branch corresponding to the red tactic
(see {\small$\m{A_H}$} in \Cref{fig:insecure-impl}).  
As appealing as this optimization seems, it leads to a leak.
An attacker of confidentiality level {\small$\maxsec{\mb{guest}}$} can simultaneously observe the sequence of messages transmitted along channels {\small$u_1$} and {\small$u_2$} of confidentiality {\small$\maxsec{\mb{guest}}$},
which connect the participants to the surveyors, 
and thus, can deduce which secret tactic was chosen:
 in case of the green tactic, the sequence of messages along {\small$u_1$} and {\small$u_2$} has the recurrence
 {\small$u_1.\mathit{ask}; u_1. (\mathit{yes/no}); u_2.\mathit{ask}; u_2. (\mathit{yes/no})$},
whereas it has the recurrence
{\small$u_1.\mathit{ask}; u_1. (\mathit{yes/no})$} for the red tactic.
Observing, for example, the sequence
{\small$u_1.\mathit{ask}; u_1.(\mathit{yes/no}); u_2.\mathit{ask}; u_2.(\mathit{yes/no}); u_1.\mathit{ask}; u_1.(\mathit{yes/no}),$}
the attacker can deduce that the first tactic used was green and the second one was red.
These leaks constitute \emph{timing attacks} because the attacker cannot deduce the secret by only looking at a single channel, but needs to observe the relative timing of messages passed along two or more channels.

\subsubsection{Reckless analyzer - be careful with synchronization}

The previous example shows that a send along a channel, present in one branch, but omitted from another, may lead to a leak.
One may naively suspect that these leaks only involve sends.
The analyzer version {\small$\m{A_R}$} in \Cref{fig:insecure-impl} showcases the opposite:  
mismatched receives are at least as dangerous as mismatched sends.
In the original implementation (\Cref{fig:survey-impl}), the analyzer synchronizes the communications of surveyors and participants across branches, ensuring, in particular, that the second participant always casts their vote after the first.
The reckless analyzer {\small$\m{A_R}$} breaks this synchronization in the red branch by swapping the order of {\small$\mathbf{case}\, w_1$ and $w_2.\mathit{ask}$}.
This minimal change allows the two surveyors to run concurrently when the tactic is red and produce the sequence of messages
{\small $u_2.\mathit{ask}; u_2. (\mathit{yes/no}); u_1.\mathit{ask}; u_1. (\mathit{yes/no})$}
along channels $u_1$ and $u_2$,
a sequence that is impossible to produce in the green tactic {(recall that receives are blocking, but sends are not)}.
Both {\small$\m{A_H}$} and {\small$\m{A_R}$} leak the secret with a timing attack,
\ie the simultaneous observation of the relative order of sends along several channels.

There is a subtle connection between timing attacks and leaks due to the non-reactivity of a process.
For instance, let us assume that the second participant loops internally and never casts its vote.
The attacker can then deduce the secret tactic in the hasty implementation of the analyzer by only observing the communications of the first participant along $u_1$: 
the sequence {\small$u_1.\mathit{ask}; u_1. (\mathit{yes/no}); u_1.\mathit{ask}; u_1. (\mathit{yes/no})$} indicates that the prior tactic was red.  
A similar scenario holds for the reckless analyzer when the first participant is non-reactive.

\subsection{Regrading policies in a nutshell}\label{sec:reclassification-nutshell}

Our model allows the running confidentiality of a process to be dropped as low as its running integrity. 
Performing such a venturous act, needs a corresponding safety net in place: a regrading policy that is polymorphic in the running integrity to preserve noninterference.
The examples in \Cref{sec:downgradingpolicy} suggest that a regrading policy must enforce the following properties:
\begin{enumerate}
\item The continuation of a process after regrading must not depend on any secret higher than or incomparable to its running integrity.
That is, when branching on a secret {\small$\runsec{\mb{d}_s}$}, the same process must be spawned for the recursive call in every branch,
if that process regrades to a level $\runsec{\mb{e}_0}$ such that {\small$ \runsec{\mb{d}_s} \not \sqsubset \runsec{\mb{e}_0}$}.

\item Whether a process reaches its regrading point or not must not depend on any secret higher than or incomparable to its running integrity.
\end{enumerate}
The latter property
 is violated in both analyzer implementations of \Cref{fig:insecure-impl}, amounting to a leak.
In the hasty implementation {\small$\m{A_H}$}, the second surveyor only gets to the regrading point if the secret tactic is green.
In the reckless implementation {\small$\m{A_R}$},
if the secret tactic is green, the second surveyor gets to the regrading point only if the first participant casts their vote,
whereas if the secret is red, there is no such chaining. 

The above properties
capture semantically what conditions secure processes that employ regrading must meet to observe PSNI.
In \Cref{sec:sync-patterns} we develop static checks that, when satisfied by a process,
ensure that the process also meets these semantic conditions.
We refer to those checks as \emph{synchronization pattern} checks,
and they are enforced by the $\lang$ type checker.
The pattern checks are of the form {\small$\barbeq{P}{d}{f}{Q}$} and
synchronize {\small$P$} and {\small$Q$} in terms of their communication actions:
if {\small$P$} outputs along channel {\small$x$}, so must {\small$Q$}, and if {\small$P$} inputs along channel {\small$x$}, so must {\small$Q$}, and vice versa.
The pattern checks are invoked pairwise for every two branches, {\small$P_i$} and {\small$P_j$}, in a {\small$\mb{case}$} statement,
requiring that {\small$\barbeq{P_i}{d}{f}{P_j}$}.
The check is conditioned on the running confidentiality {\small$d$} and running integrity {\small$f$} at the branching point.

An important feature of our regrading policies is that they are \emph{compositional}.
We take advantage of the fact that intuitionism imposes a rooted tree structure on process configurations
and require that a configuration aligns with the security lattice:
for every child process and parent process with maximal confidentiality and minimal integrity
$\maxsec{\dlabel{\mb{c}}{\mb{e}}}$ and $\maxsec{\dlabel{\mb{c}'}{\mb{e}'}}$, \resp
it must hold that $\maxsec{\dlabel{\mb{c}}{\mb{e}}} \sqsubseteq \maxsec{\dlabel{\mb{c}'}{\mb{e}'}}$,
ensuring that a child process can learn at most as much as its parent and has at least an as stringent regrading policy as its parent.
We design our type system to preserve this property as an invariant.

\section{Blueprint for Formal Development}
\label{sec:roadmap}
Before delving into the formal development, we review the statics and dynamics of a vanilla intuitionistic session type system and give a roadmap for the upcoming technical sections.
We use the intuitionistic session type system introduced by Toninho et al.~\cite{ToninhoESOP2013,CairesCONCUR2010} as our vanilla intuitionistic session type system.
$\lang$ enhances such a vanilla session type system with confidentiality and integrity annotations to establish PSNI.
{$\lang$ adopts the former from existing intuitionistic IFC session type systems~\cite{BalzerARXIV2023,DerakhshanLICS2021}.
The integrity annotations as well as the synchronization patterns are contributions unique to $\lang$.
The addition requires us to define the relationships between all these levels, expressed as invariants, and the development of synchronization patterns. 
Similar to the system in~\cite{BalzerARXIV2023} our language supports general recursion and allows processes to be polymorphic in confidentiality levels. 
$\lang$ extends label polymorphism to also accommodate integrity levels.
}

\subsection{Vanilla intuitionistic session types - statics}\label{sec:vanilla-statics}
Process terms and session types are built by the grammar in~\Cref{sec:motivating-examples} and~\Cref{tab:session_types}.
The process typing judgment is of the form
{\small\(\Delta \vdash_{\Sigma} P :: x{:}A\),}
\noindent to be read as: \textit{``Process {\small$P$}
provides a session of type {\small$A$} along channel {\small$x$},
given the typing of sessions offered along channels in {\small$\Delta$}}.
{\small$\Delta$} is a \emph{linear} typing context consisting of the channels connecting {\small$P$} to its children, and {\small$x$} connects {\small$P$} to its parent. The global signature {\small$\Sigma$} includes recursive type definitions and process definitions.

\subsubsection{Process term typing}\label{sec:term-typing-v}
\Cref{fig:typing-rules} lists the process term typing rules.
The parts in red are specific to $\lang$ and can be ignored for now; we discuss them in~\Cref{sec:type-system}.
As is usual in intuitionistic linear session type languages, the rules are given in a sequent calculus. When read from bottom to top, the rules closely follow the behavior described in Table~\ref{tab:session_types}: right rules describe a type from the point of view of a provider, and left rules from the point of view of a client.
For example, rule {\small$\intchoicesymb_{R_1}$} describes the behavior of the process that provides a channel with the protocol {\small$\oplus\{\ell{:}A_\ell\}_{\ell \in L}$}: it chooses a label {\small $k{\in} L$} and sends it to the client along channel {\small$x$}, and then continues by checking process {\small$P$} providing {\small$A_k$} in the premise. Note that the typing rules {\small$\intchoicesymb_{R_1}$} and {\small$\intchoicesymb_{R_2}$} are identical, ignoring the security annotations.
Rule {\small$\msc{Fwd}$} ensures that the type of the two channels involved in forwarding is the same.
Rule {\small$\msc{Spawn}$} spawns a new child process {\small$X$} along the fresh channel {\small$x$}; it first checks that {\small$X$} is defined in the signature (first premise) and thus is well-typed and then continues with type-checking the continuation {\small$Q$} (last premise).

\subsubsection{Signature checking}

To support general recursive types, we employ a global signature $\Sigma$
comprised of all process definitions.
Each process definition is typed individually,
assuming that the other processes in the signature are well-typed.
The signature also comprises recursive type definitions.
When typing a process with a recursive protocol, the signature is consulted to unfold the definition.

For example, the signature for the bank survey example in~\Cref{sec:motivating-examples} consists of the definitions for processes $\m{A}$, $\m{S}$, and $\m{S'}$ as shown in~\Cref{fig:analyzer} and the definition of recursive types as shown in~\Cref{fig:analyzer}(d). 
In our formal development, we use a more concise syntax for process definitions than what is shown in~\Cref{fig:analyzer}. In particular, we write them in the form of {\small$\Delta \vdash X=P{::} (z{:}A)$}. For instance, the concise version of the process definition for process $\m{S}$ in~\Cref{fig:analyzer}, ignoring its security annotations, is {\small $u{:}\m{vote} \vdash \m{S} = {case} \, w \, \,(\mathit{ask} \Rightarrow (w\leftarrow \m{S}'\leftarrow u)){::} w{:}\m{vote}
$}.

Type checking starts with typing the signature by the rules listed in~\Cref{fig:signature}; again, ignore the parts in red for now, as they will be discussed later in~\Cref{sec:type-system}. 
The rules are in a sequent calculus and should be read from bottom to top.
Rule {$\Sig_3$} ensures that each process definition in the signature is well-typed. It invokes the process term typing judgment for a process definition relative to the entire global signature {\small$\Sig$} (fifth premise) and continues with checking the rest of the signature (sixth premise).
Rule { $\Sig_2$} ensures that all recursive types in the signature are well-formed via its first premise, the judgment {\small$\typwfmd{A}$}. This judgment denotes a well-formed session type definition, which, if recursive,
must be \emph{equi-recursive}~\cite{CraryPLDI1999} and \emph{contractive}.
Equi-recursiveness ensures that types are related up to their unfolding without requiring explicit (un)fold messages (see rules $\msc{TVar}_R$ and $\msc{TVar}_L$).
Contractiveness demands an exchange before recurring.

\subsection{Vanilla intuitionistic session types - dynamics}\label{sec:vanilla-dynamics} 
At runtime, process definitions result in a configuration of processes structured as a forest of rooted trees. The nodes in the forest represent runtime processes and messages, denoted as {\small$\mathbf{proc}(y_\alpha; P)$} and {\small$\mathbf{msg}(M)$}, \respb. 
We use metavariables {\small$\mathcal{C}$} and {\small$\mathcal{D}$} to refer to a configuration and formally define it as a set of runtime processes and messages (the nodes in the tree). The connection between the nodes will be inferred through configuration typing.
In {\small$\mathbf{proc}(y_\alpha, P)$}, the metavariable {\small$y_\alpha$} represents the process' offering channel, and $P$ represents the process' source code (where free variables have been substituted by channels). Runtime messages {\small$\mathbf{msg}(M)$} are a special form of processes created to model asynchronous communication: we implement asynchronous sends by spawning off the message {\small$\mathbf{msg}(M)$} that carries the sent message {\small$M$}.
A sent message {\small$M$} can be of the form 
{\small $x.k$}, {\small $\mb{send}\, y\,x$}, or {\small$\cls{x}$}, corresponding to label output, channel output, and a termination message, \respb.

Runtime channels {\small$y_\alpha$} are annotated with a generation subscript {\small$\alpha$}, which distinguishes them from channel variables {\small$y$} used in the statics. 
Using generation subscripts, we can ensure that both the sender and receiver agree on a new name for the continuation channel without explicitly passing the name in a message. 
We will see an example of using generation subscripts in the next paragraph.

\subsubsection{Asynchronous dynamics}\label{sec:dynamicsv}
{We chose an asynchronous semantics for $\lang$ because it weakens the attacker model, allowing a more permissive IFC enforcement, and is also a sensible model for practical purposes.}
The dynamics is given in \Cref{fig:dynamics} in terms of multiset rewriting rules \cite{CervesatoARTICLE2009} (again, for now the parts in red can be ignored).
Multiset rewriting rules express the dynamics as state transitions between configurations and are \emph{local} in that they only mention the parts of a configuration they rewrite.

For example, in case of {\small$\otimes_{\m{snd}}$},  the provider {\small$\mathbf{proc}(y_\alpha, \mathbf{send}\,x_\beta\,y_\alpha;P)$} spawns off the message process {\small$\mathbf{msg}(\mathbf{send}\,x_\beta\,y_\alpha)$}, indicating that the channel {\small$x_\beta$} is sent over channel {\small$y_\alpha$}.
Since sends are non-blocking, the provider steps to its continuation {\small$\mathbf{proc}(y_{\alpha+1}, ( [y_{\alpha+1}/y_\alpha]P))$}, allocating a new generation {\small$\alpha{+}1$} of the carrier channel {\small$y_\alpha$}. 
In {\small$\otimes_{\m{recv}}$}, upon receipt of the message, the receiving client process {\small $\mathbf{proc}(y_\alpha, w \leftarrow \mathbf{recv}\, y_\alpha; P)$} will increment the generation of the carrier channel in its continuation.
The scenario is similar for {\small$\oplus_{\m{snd}}$} and {\small$\oplus_{\m{rcv}}$}, but the sent message is a label in this case, and 
similar for {\small$\multimap_{\m{snd}}$},\,{\small$\multimap_{\m{rcv}}$} and {\small$\&_{\m{snd}}$},\,{\small$\&_{\m{rcv}}$}, except that in these cases the sender is the client and the receiver the provider. 
In the rules for the termination protocol, i.e., {\small$1_{\m{snd}}$} and {\small$1_{\m{rcv}}$}, there is no continuation channel.
Rule $\msc{Spawn}$ creates a process offering along a fresh runtime channel $x_0$ by looking up the definition of the spawnee in the signature.

The dynamics for the forwarding process $\mathbf{proc}(y_\alpha, y_\alpha \leftarrow x_\beta)$ is often described as fusing the two channels, $y_\alpha$ and $x_\beta$. 
We, however, represent forward as syntactic sugar by including forwarder processes defined by structural induction on the type of the channels involved in the forward, amounting to an identity expansion.
The reader may see the appendix for the details.

\subsubsection{Configuration typing}
The configuration typing judgment is of the form {\small$\Delta\Vdash_\Sigma \mc{C}:: \Delta'$} indicating that the configuration {\small$\mc{C}$} provides sessions along the channels in {\small$\Delta'$}, using sessions provided along channels in {\small$\Delta$}.
{\small$\Delta$} and {\small$\Delta'$} are both linear contexts, consisting of actual runtime channels of the form {\small$y_{\alpha}{:}B$}. 
We often use the term \emph{open configurations} to emphasize that our configurations may have external free channels in both $\Delta$ and $\Delta'$ to communicate with the environment. This is in contrast to restricting $\Delta$ to be an empty context, which means the configuration only has external free channels to communicate with a client.

\Cref{fig:config-typing} shows the typing rules, enforcing that the configuration is structured as a forest and the source code of each node is well-typed.
For brevity, \Cref{fig:config-typing} omits a channel's generation as well as $\Sigma$, which is fixed.
The $\mathbf{emp}$ rule types an empty forest.
The $\mathbf{comp}$ rule types each tree in the forest. The $\mathbf{proc}$ rule and the $\mathbf{msg}$ rule check the well-typedness of the root node of a tree when it is a process or message, resp., using the last premises. Well-typedness of the remaining forest is checked by the eighth and fourth premise of the latter two rules, resp.
The last premise of the $\mathbf{msg}$ rule calls message typing rules, which we provide in the appendix~\Cref{apx:sec:async_dynamics}.

The typing rules ensure progress and preservation, i.e.,
the dynamics can always step an open configuration {\small$\Delta\Vdash \mc{C}:: \Delta'$} to {\small$\Delta\Vdash \mc{C}':: \Delta'$}.

\subsection{Roadmap for $\lang$}\label{sec:roadmap-2}
To develop the ideas discussed in~\Cref{sec:key-ideas} and establish PSNI, we supplement the vanilla type system with a security layer. Here, we provide a roadmap to the key parts of our development.
\subsubsection{Regrading policy type system}
The first step in our formal development is to enrich the process term typing judgment with security levels as
{\small\(\ptypb{\D}{P}{c_0}{e_0}{x}{A}{c}{e}.\)}
\noindent Here, $\Psi$ denotes a security theory which includes the security lattice and polymorphic confidentiality and integrity variables. The pair $\langle c_0, e_0 \rangle$ denotes the running confidentiality and integrity of the process, aka its taint level. The pair $\langle c,e \rangle$ denotes the max confidentiality and min integrity of the process. Similarly, each channel in $\Delta$ is annotated with a pair of confidentiality and integrity levels denoting its provider's max confidentiality and min integrity.
%

% We use security labels to annotate configurations and configuration typing judgments, similar to how we do it for process judgments. 
%
Similarly, we use security labels to annotate configurations and configuration typing judgments.
In particular, runtime processes in  configuration {\small $\mathcal{C}$} now have the form {\small $\mathbf{proc}(y_\alpha\langle c,e \rangle, P@ \langle c',e' \rangle)$}, where {\small$\langle c,e \rangle$} is the pair of max confidentiality and min integrity of the process, and {\small$\langle c',e' \rangle$} is the pair of its running confidentiality and integrity.

The typing rules include security constraints highlighted in red---the ones we have been ignoring in~\Cref{sec:vanilla-statics}.  The purpose of these security annotations is to (i) ensure that the taint levels are propagated correctly, (ii) prevent a tainted process from sending information to a process with a lower max confidentiality/higher min integrity, (iii) ensure that a process regrades its running confidentiality only as low as its running integrity, and (iv) verify that the process indeed adheres to the policy enforced by its running integrity. The first three conditions are enforced by imposing the security constraints on the process term typing rules in~\Cref{fig:typing-rules}.
 The last check is enforced by the synchronization pattern checks in~\Cref{fig:sync-patterns}.

\subsubsection{PSNI via a logical relation}
Our ultimate goal is to prove that well-typed $\lang$ processes enjoy PSNI.
We prove PSNI as an equivalence up to an attacker's confidentiality level $\xi$ using a logical relation, which then delivers a process bisimulation.

To define PSNI for an open configuration in the shape of a tree {\small$\Psi_0; \Delta \Vdash \mathcal{D}:: u_\alpha {:}T\langle c,e \rangle$},
given a global security lattice {\small $\Psi_0$} fixed for an application,
we consider the external free channels {\small $y_\beta{:}B\langle c',e' \rangle \in \Delta, u_\alpha {:}T\langle c,e \rangle$} with max confidentiality $c' \sqsubseteq \xi$.
We call the set of these channels that connect a configuration to its environment and that are observable to an attacker, the \emph{confidentiality interface}.

Such an open configuration satisfies noninterference if, when composed with different high-confidentiality processes, behaves the same along the confidentiality interface. 
We prove that all well-typed open configurations enjoy PSNI by designing a logical relation and showing that (i) all well-typed configurations are self-related (fundamental theorem, \Cref{lem:reflexivity}) and (ii) any two related configurations are bisimilar (adequacy theorem, \Cref{thm:adeq}).

To prove these results, our logical relation needs to consider some free channels in {\small $\Delta, u_\alpha {:}T\langle c,e \rangle$} that are not directly observable in terms of their confidentiality but can have an observable effect due to their integrity.  We thus define a superset of the confidentiality interface that additionally contains channels {\small $y_\beta{:}B\langle c',e' \rangle \in \Delta, u_\alpha {:}T\langle c,e \rangle$} with min integrity $e' \sqsubseteq \xi$.
We call this interface the \emph{integrity interface}.

\section{Regrading policy type system}
\label{sec:type-system}
This section formalizes $\lang$'s type system with synchronization patterns and asynchronous dynamics.
$\lang$ supports security-polymorphic process definitions,
an example of which is discussed in \Cref{sec:bank-examples}.

\subsection{Process term typing}\label{sec:term-typing}

Let us recall the process term typing judgment from~\Cref{sec:roadmap-2}:
$$\ptypb{\D}{P}{c_0}{e_0}{x}{A}{c}{e}.$$

\noindent We read it as: \emph{``Process {\small$P$},
with maximal confidentiality and minimal integrity {\small$\dlabel{c}{e}$}
and running confidentiality and integrity {\small$\dlabel{c_0}{e_0}$},
provides a session of type {\small$A$} along channel {\small$x$},
given the typing of sessions offered along channels in {\small$\Delta$} and
given a security theory {\small$\Psi$}''}.
{\small $\Delta$} is a \emph{linear} typing context with the grammar {\small$\Delta ::= \cdot \mid {x{:}A}\langle{c},{e}\rangle, \Delta$}.
A security theory {\small$\Psi$} is used for type checking security-polymorphic process definitions.
It consists of the global security lattice {\small$\Psi_0$} which is fixed for an application, security variables {\small$\psi$}, and constraints on the variables (see \Cref{sec:bank-examples} and~\Cref{apx:sec:lattice} in the appendix).

We impose the following properties on the typing judgment, as discussed in detail in~\Cref{sec:key-ideas}. These properties are maintained by typing as invariants. When reading them, note that ``high integrity'' and ``low confidentiality'' both mean a ``lower level'' in the security lattice.
\begin{enumerate}
\item[(a)] {\small$\forall y{:}B\dlabel{d}{f} \in \D . \latcstr{d \sqsubseteq c}, \latcstr{f \sqsubseteq e}$}: ensuring that a child process can learn at most as much as its parent and has at least an as stringent regrading policy as its parent.
\item[(b)] {\small$\latcstr{c_0 \sqsubseteq c}$ and $\latcstr{e_0 \sqsubseteq e}$}: ensuring that a process knows at most as much as it is licensed to and adheres to at least an as stringent regrading policy as it promises.
\item[(c)] {\small$\latcstr{e_0 \sqsubseteq c_0}$ and $\latcstr{e \sqsubseteq c}$}: ensuring that a process cannot drop more secrets than it knows and is licensed to learn, \respb.
\end{enumerate}
Moreover, the typing rules for input and output have to conform to the following schema to make sure that the running confidentiality and running integrity correctly reflect the taint level and that a tainted process does not leak information via a send:
\begin{enumerate}
  \item[(1)] \textbf{after} receiving a message, the running confidentiality and running integrity of the receiving process must be
  increased to \textbf{at least} the maximal confidentiality and minimal integrity of the sending process, and
\item[(2)]\label{lab:before-snd} \textbf{Before} sending a message, the running confidentiality and running integrity of the sending process must be
  \textbf{at most} the maximal confidentiality and minimal integrity of the receiving process.
\end{enumerate}
Conforming to this schema leads to the premises of the form
{\small $\latcstr{\dlabel{d_1}{f_1} = \dlabel{c}{e} \sqcup \dlabel{d_0}{f_0}}$} and {\small $\latcstr{\dlabel{d_0}{f_0} \sqsubseteq \dlabel{c}{e}}$}
to meet condition (1) and (2), \resp above.
The judgments are defined formally in~\Cref{apx:sec:lattice} in the appendix.

It is time to consider the red security annotations of the typing rules in~\Cref{fig:typing-rules}. We explain how the rules satisfy conditions (1) and (2) above:

\begin{figure*}
\begin{center}
\begin{small}
    \def \MathparLineskip {\lineskip=0.2cm}
\begin{mathpar}
\inferrule*[right=$\intchoicesymb_{R_1}$]
{{\it \color{red}\latcstr{c=e}} \\
k \in L \\
\ptypbr{\D}{P}{c_0}{e_0}{x}{A_k}{c}{e}}
{\ptypr{\D}{\labsnd{x^{\dlabelr{c}{e}}}{k}{P}}{c_0}{e_0}{x}{\intchoice{A}}{c}{e}}

\inferrule*[right=$\intchoicesymb_{R_2}$]
{{\it\latcstrnot{c=e}} \\
{\it \color{red}\forall i,j \in L. \,A_i=A_j} \\
k \in L \\
\ptypbr{\D}{P}{c_0}{e_0}{x}{A_k}{c}{e}}
{\ptypr{\D}{\labsnd{x^{\dlabelr{c}{e}}}{k}{P}}{c_0}{e_0}{x}{\intchoice{A}}{c}{e}}

\inferrule*[right=$\intchoicesymb_L$]
{{\it \color{red}\latcstr{\dlabelr{d_1}{f_1} = \dlabelr{c}{e} \sqcup \dlabelr{d_0}{f_0}}} \\
\forall k \in L \\
\ptypbr{\D, x{:}A_k\dlabelr{c}{e}}{Q_k}{d_1}{f_1}{z}{C}{d}{f} \\
{\it \color{red}\forall i,j \in L. \barbeq{Q_i}{d_1}{f_1}{Q_j}}}
{\ptypr{\D, x{:}\intchoice{A}\dlabelr{c}{e}}{\labrcv{x^{\dlabelr{c}{e}}}{Q}}{d_0}{f_0}{z}{C}{d}{f}}

\inferrule*[right=$\extchoicesymb_R$]
{\forall k \in L \\
\ptypbr{\D}{P_k}{c}{e}{x}{A_k}{c}{e} \\
{\it \color{red}\forall i,j \in L. \barbeq{P_i}{c}{e}{P_j}}}
{\ptypr{\D}{\labrcv{x^{\dlabelr{c}{e}}}{P}}{c_0}{e_0}{x}{\extchoice{A}}{c}{e}}

\inferrule*[right=$\extchoicesymb_{L_1}$]
{{\it \color{red}\latcstr{c=e}} \\
{\it \color{red}\latcstr{\dlabelr{d_0}{f_0} \sqsubseteq \dlabelr{c}{e}}} \\
k \in L \\
\ptypbr{\D, x {:} A_k \dlabelr{c}{e}}{Q}{d_0}{f_0}{z}{C}{d}{f}}
{\ptypr{\D, x {:} \extchoice{A} \dlabelr{c}{e}}{\labsnd{x^{\dlabelr{c}{e}}}{k}{Q}}{d_0}{f_0}{z}{C}{d}{f}}

\inferrule*[right=$\extchoicesymb_{L_2}$]
{\latcstrnot{c=e} \\
{\it \color{red}\forall i,j \in L. \,A_i=A_j}\\
{\it \color{red}\latcstr{\dlabelr{d_0}{f_0} \sqsubseteq \dlabelr{c}{e}}} \\
k \in L \\
\ptypbr{\D, x {:} A_k \dlabelr{c}{e}}{Q}{d_0}{f_0}{z}{C}{d}{f}}
{\ptypr{\D, x {:} \extchoice{A} \dlabelr{c}{e}}{\labsnd{x^{\dlabelr{c}{e}}}{k}{Q}}{d_0}{f_0}{z}{C}{d}{f}}

\inferrule*[right=$\chanoutsymb_R$]
{\ptypbr{\D}{P}{c_0}{e_0}{x}{B}{c}{e}}
{\ptypr{\D, y{:}A\dlabelr{d}{f}}{\chnsnd{y}{x^{\dlabelr{c}{e}}}{P}}{c_0}{e_0}{x}{\chanout{A}{B}}{c}{e}}

\inferrule*[right=$\chanoutsymb_L$]
{{\color{red}\latcstr{\dlabelr{d_1}{f_1} = \dlabelr{c}{e} \sqcup \dlabelr{d_0}{f_0}}} \\
\ptypbr{\D, x{:}B\dlabelr{c}{e}, y{:}A\dlabelr{c}{e}}{Q}{d_1}{f_1}{z}{C}{d}{f}}
{\ptypr{\D, x{:}\chanout{A}{B}\dlabelr{c}{e}}{\chnrcv{y^{\dlabelr{c}{e}}}{x^{\dlabelr{c}{e}}}{Q}}{d_0}{f_0}{z}{C}{d}{f}}

\inferrule*[right=$\chaninsymb_R$]
{\ptypbr{\D, y{:}A\dlabelr{c}{e}}{P}{c}{e}{x}{B}{c}{e}}
{\ptypr{\D}{\chnrcv{y^{\dlabelr{c}{e}}}{x^{\dlabelr{c}{e}}}{P}}{c_0}{e_0}{x}{\chanin{A}{B}}{c}{e}}

\inferrule*[right=$\chaninsymb_L$]
{{\it \color{red}\latcstr{\dlabelr{d_0}{f_0} \sqsubseteq \dlabelr{c}{e}}} \\
\ptypbr{\D, x{:} B\dlabelr{c}{e}}{Q}{d_0}{f_0}{z}{C}{d}{f}}
{\ptypr{\D, x{:} \chanin{A}{B}\dlabelr{c}{e}, y{:}A\dlabelr{c}{e}}{\chnsnd{y}{x^{\dlabelr{c}{e}}}{Q}}{d_0}{f_0}{z}{C}{d}{f}}

\inferrule*[right=$\msc{Fwd}$]
{{\color{red}\Psi \Vdash \langle c_1, e_1 \rangle = \langle c_2,e_2 \rangle}}
{\ptypr{y{:}A\dlabelr{c_1}{e_1}}{\fwd{x^{\dlabelr{c_2}{e_2}}}{y^{\dlabelr{c_1}{e_1}}}}{c_0}{e_0}{x}{A}{c_2}{e_2}}

\inferrule*[right=$\msc{Spawn}$]
{\procdefr{\D_1'}{X}{P}{\psi_0}{\omega_0}{x}{A}{\psi}{\omega} \in \Sig \\
{\it \color{red}\substmap{\Psi}{\gamma}{\Psi'}} \\
{\it \color{red}\substmapap{\gamma}{\D_1'} = \D_1} \\
{\it \color{red}\latcstr{\dlabelr{\substmapap{\gamma}{\psi}}{\substmapap{\gamma}{\omega}} \sqsubseteq \dlabelr{d}{f}}} \\
{\it \color{red}\latcstr{f_0 \sqsubseteq \substmapap{\gamma}{\psi_0}}} \\
{\it \color{red}\latcstr{f_0 \sqsubseteq \substmapap{\gamma}{\omega_0}}} \\
\ptypbr{\D_2, x{:}A\dlabelr{\substmapap{\gamma}{\psi}}{\substmapap{\gamma}{\omega}}}{Q}{d_0}{f_0}{z}{C}{d}{f}}
{\ptypr{\D_1, \D_2}{\spwn{x}{\substmapap{\gamma}{\psi}}{\substmapap{\gamma}{\omega}}{X}{\D_1}{\substmapap{\gamma}{\psi_0}}{\substmapap{\gamma}{\omega_0}}{Q}}{d_0}{f_0}{z}{C}{d}{f}}

\inferrule*[right=$\one_R$]
{\strut}
{\ptypr{\cdot}{\cls{x^{\dlabelr{c}{e}}}}{c_0}{e_0}{x}{\one}{c}{e}}

\inferrule*[right=$\one_L$]
{{\it \color{red}\latcstr{\dlabelr{d_1}{f_1} = \dlabelr{c}{e} \sqcup \dlabelr{d_0}{f_0}}} \\
\ptypbr{\D}{Q}{d_1}{f_1}{z}{C}{d}{f}}
{\ptypr{\D, x{:}\one\dlabelr{c}{e}}{\wait{x^{\dlabelr{c}{e}}}{Q}}{d_0}{f_0}{z}{C}{d}{f}}
\end{mathpar}
\end{small}
\end{center}
\vspace{-10pt}
\caption{Process term typing rules of $\lang$.}
\label{fig:typing-rules}
\vspace{-15pt}
\end{figure*}

\begin{itemize}
  \item{$\oplus$:} There are two versions of the right rule for {\small$\oplus$}.
Both versions establish condition~(2) on sends without extra premises by the invariant~(b).
  The difference between the two versions lies in whether {\small$\latcstr{c=e}$} is derivable or not derivable ({\small$\latcstrnot{c=e}$}). If {\small$\latcstr{c=e}$} is derivable, then rule {\small$\oplus_{R_1}$} applies; if it is not, rule {\small$\oplus_{R_2}$} applies.
  In the former case, the client of {\small $x$}, on the receiving side, adjusts its running integrity to at least {\small $e{=}c$} upon receiving the sent message, and thus, it cannot regrade to a lower (or unrelated) level than {\small $c$}. 
  In the latter case, the min integrity {\small $e$} of the process is strictly lower than its max confidentiality {\small $c$}.
  This means that the client of {\small $x$} might, in fact, continue to have its running integrity as low as {\small $e\sqsubset c$} and, at some point in the future, drop its running confidentiality to {\small $e$} and start sending to channels with lower (or unrelated) confidentiality than {\small $c$}.
  The additional premise {\small $\forall i,j \in L. \,A_i=A_j$}  in  {\small$\oplus_{R_2}$} prevents potential leaks through different continuation protocols at that future point, i.e., it ensures that the client's future communications with channels of lower confidentiality level than {\small $c$} do not depend on the continuation protocol chosen now.

  The first premise of rule {\small $\oplus L$} updates the running integrity and confidentiality based on {\small $x$}'s security levels to enforce condition~(1) for receives.
  Moreover, as explained in~\Cref{sec:reclassification-nutshell}, the third premise invokes the pattern check pairwise for every two branches conditioned on the running confidentiality {\small$d_1$} and running integrity {\small$f_1$} after the receive.
  We detail the synchronization pattern check rules later in \Cref{sec:sync-patterns}.

  \item{$\&$:} The left and right rules for $\&$ are dual to $\oplus$, except that the sends in $\&_{L_1}$ and $\&_{L_2}$ have to be guarded by their second and third premises, \resp to ensure condition~(2) on sends.
  In $\&_{R}$, the updated running confidentiality and running integrity is equal to the max confidentiality and max integrity by invariant~(b).
  \item{$\otimes$, $\multimap$, $1$:} 
  The rules for the rest of the connectives use the same set of premises to ensure conditions~(1) and~(2).
  Rules $\otimes_{R}$ and $\multimap_{L}$, moreover, ensure that a channel can be sent over another channel only if they have the same security levels.
  \item{$\msc{fwd}$:} The forward rule requires that the security levels of the involved channels match.
  \item{$\msc{Spawn}$:} The rule relies on an order-preserving substitution $\substmap{\Psi}{\gamma}{\Psi'}$,
  guaranteeing that the security terms provided by the spawner comply with the order expected among those terms by the spawnee. 
  The substitution maps the security terms in the context in the signature to the one provided by the spawner, i.e., $\substmapap{\gamma}{\D_1'} = \D_1$.
The rule also establishes invariants~(a)-(c) for the newly spawned process via the premise
  $\latcstr{\dlabel{\substmapap{\gamma}{\psi}}{\substmapap{\gamma}{\omega}} \sqsubseteq \dlabel{d}{f}}$.
 The running confidentiality and the running integrity of the spawned process will result from applying the substitution to the corresponding levels in the signature, i.e., $\substmapap{\gamma}{\psi_0}$ and $\substmapap{\gamma}{\omega_0}$, \respb.
The premises $\latcstr{f_0 \sqsubseteq \substmapap{\gamma}{\psi_0}}$ and $\latcstr{f_0 \sqsubseteq \substmapap{\gamma}{\omega_0}}$ allow the newly spawned process to start its running confidentiality and integrity at least at the spawner's running integrity $f_0$.
  % for running confidentiality and integrity, 
This facilitates regrading to $f_0$ in case of a tail call.
  Note that $\latcstr{f_0 \sqsubseteq \substmapap{\gamma}{\omega_0}}$ prevents the spawnee from employing more pattern checks
  than the spawner because the spawnee would otherwise be affected by the spawners negligence.

\end{itemize}

{\bf Signature checking.} 
The syntax of process definitions in the signature is also enhanced with the security levels and is of the form
{\small$\Psi; \Delta \vdash X=P @ \langle \psi_0, \omega_0 \rangle{::} (z{:}A \langle \psi, \omega \rangle)$.}
\Cref{fig:signature} lists the signature checking rules.
Signature checking happens relative to a globally fixed security lattice $\Psi_0$ of concrete security levels.
Rule $\Sigma_3$ initiates type-checking of a process definition via its fifth premise and enforces invariants (a)-(c) on the process via the first four premises.

\begin{figure*}
\begin{center}
\begin{small}
\def \MathparLineskip {\lineskip=0.13cm}
\begin{mathpar}
\inferrule[$\msc{TVar}_R$]
{\typdef{Y}{A} \in \Sig \\
\ptypbr{\D}{P}{c_0}{e_0}{x}{A}{c}{e}}
{\ptypbr{\D}{P}{c_0}{e_0}{x}{Y}{c}{e}}
\and
\inferrule[$\msc{TVar}_L$]
{\typdef{Y}{A} \in \Sig \\
\ptypbr{\D, x{:}A\dlabelr{c}{e}}{Q}{d_0}{f_0}{z}{C}{d}{f}}
{\ptypbr{\D, x{:}Y\dlabelr{c}{e}}{Q}{d_0}{f_0}{z}{C}{d}{f}}
\and
\inferrule*[right=$\Sig_1$]
{\strut}
{\sigwfmd{(\cdot)}}
\and
\inferrule*[right=$\Sig_2$]
{\typwfmd{A} \\
\sigwfmd{\Sig'}}
{\sigwfmd{\typdef{Y}{A}, \Sig'}}
\and
\inferrule*[right=$\Sig_3$]
{\mathit{\color{red}\forall i \in \{1\dots n\}. \latcstr{\dlabelr{\psi_i}{\omega_i} \sqsubseteq \dlabelr{\psi}{\omega}}, \latcstr{\omega_i \sqsubseteq \psi_i}}\\
\mathit{\color{red}\latcstr{\dlabelr{\psi_0}{\omega_0} \sqsubseteq \dlabelr{\psi}{\omega}}}\\
\mathit{\color{red}\latcstr{\omega_0 \sqsubseteq \psi_0}}\\
\mathit{\color{red}\latcstr{\omega \sqsubseteq \psi}}\\
\ptypbr{y_1 {:} B_1\dlabelr{\psi_1}{\omega_1}, \dots, y_n {:} B_n\dlabelr{\psi_n}{\omega_n}}{P}{\psi_0}{\omega_0}{x}{A}{\psi}{\omega} \\
\sigwfmd{\Sig'}}
{\sigwfmd{\procdefbsr{y_1 {:} B_1\dlabelr{\psi_1}{\omega_1}, \dots, y_n {:} B_n\dlabelr{\psi_n}{\omega_n}}{X}{P}{\psi_0}{\omega_0}{x}{A}{\psi}{\omega}, \Sig'}}
\end{mathpar}
\end{small}
\end{center}
\vspace{-5pt}
\caption{Signature checking rules of $\lang$.}
\label{fig:signature}
\end{figure*}

\subsection{Synchronization patterns}\label{sec:sync-patterns}

To check synchronization patterns, we use the judgment {\small$\barbeq{P}{d}{f}{Q}$}, defined inductively in \Cref{fig:sync-patterns}.
The judgment states that process terms {\small$P$} and {\small$Q$} are \emph{synchronized} in terms of their communication pattern,
meaning that if {\small$P$} outputs along channel {\small$x$}, so must {\small$Q$}, and that if {\small$P$} inputs along channel {\small$x$}, so must {\small$Q$}, and vice versa.
The check is \emph{conditioned} on
the running confidentiality {\small$d$} and running integrity {\small$f$} of the recipient after branching, and is pairwise called for all branches of a {\small$\mb{case}$} statement.
\begin{figure*}
\begin{center}
\begin{small}
    \def \MathparLineskip {\lineskip=0.15cm}  
\begin{mathpar}
\inferrule[$\msc{Unsync}_1$]
{\latcstrnot{d \sqsubseteq f} \\
\latcstr{d \sqsubseteq e} \\
\barbeq{P}{d}{f}{Q}}
{\barbeq{\outact{x}{c}{e}{P}}{d}{f}{Q}}
\and
\inferrule[$\msc{Unsync}_2$]
{\latcstrnot{d \sqsubseteq f} \\
\latcstr{d \sqsubseteq e} \\
\barbeq{P}{d}{f}{Q}}
{\barbeq{P}{d}{f}{\outact{x}{c}{e}{Q}}}
\and
\inferrule*[right=$\msc{Unsync}_3$]
{\latcstr{d \sqsubseteq f}}
{\barbeq{P}{d}{f}{Q}}
\and
\inferrule*[right=$\msc{Unsync-Spawn}_1$]
{\latcstrnot{d \sqsubseteq f} \\
\latcstr{d \sqsubseteq e_0} \\
\forall y{:}B\dlabelb{c'}{e'} \in \D . \latcstr{d \sqsubseteq e'}\\
\barbeq{P}{d}{f}{Q}}
{\barbeq{\spwnb{x}{c}{e}{X}{\D}{c_0}{e_0}{P}}{d}{f}{Q}}
\and
\inferrule*[right=$\msc{Unsync-Spawn}_2$]
{\latcstrnot{d \sqsubseteq f} \\
\latcstr{d \sqsubseteq e_0} \\
\forall y{:}B\dlabelb{c'}{e'} \in \D . \latcstr{d \sqsubseteq e'}\\
\barbeq{P}{d}{f}{Q}}
{\barbeq{P}{d}{f}{\spwnb{x}{c}{e}{X}{\D}{c_0}{e_0}{Q}}}
\and
\inferrule[$\msc{SndLab}$]
{\latcstrnot{d \sqsubseteq f} \\
\latcstrnot{d \sqsubseteq e} \\
\barbeq{P}{d}{f}{Q}}
{\barbeq{\labsnd{x^{\dlabelb{c}{e}}}{k}{P}}{d}{f}{\labsnd{x^{\dlabelb{c}{e}}}{\ell}{Q}}}
\and
\inferrule[$\msc{RcvLab}$]
{\latcstrnot{d \sqsubseteq f} \\
\forall j \in I, k \in L. \barbeq{P_j}{d}{f \sqcup e}{Q_k}}
{\barbeq{\labrcvb{x^{\dlabelb{c}{e}}}{P}}{d}{f}{\labrcv{x^{\dlabelb{c}{e}}}{Q}}}
\and
\inferrule[$\msc{SndChn}$]
{\latcstrnot{d \sqsubseteq f} \\
\latcstrnot{d \sqsubseteq e} \\
\barbeq{P}{d}{f}{Q}}
{\barbeq{\chnsnd{y}{x^{\dlabelb{c}{e}}}{P}}{d}{f}{\chnsnd{y}{x^{\dlabelb{c}{e}}}{Q}}}
\and
\inferrule[$\msc{RcvChn}$]
{\latcstrnot{d \sqsubseteq f} \\
\barbeq{\subst{y}{y_1}{P}}{d}{f \sqcup e}{\subst{y}{y_2}{Q}}}
{\barbeq{\chnrcv{y_1}{x^{\dlabelb{c}{e}}}{P}}{d}{f}{\chnrcv{y_2}{x^{\dlabelb{c}{e}}}{Q}}}
\and
\inferrule*[right=$\msc{Sync-Spawn}$]
{\latcstrnot{d \sqsubseteq f}\\
\color{mytomato}{(\latcstrnot{d \sqsubseteq e_0} \; \m{or} \; \exists y{:}B\dlabelb{c'}{e'} \in \D . \latcstrnot{d \sqsubseteq e'})} \\
\barbeq{\subst{x}{x_1}{P}}{d}{f}{\subst{x}{x_2}{Q}}}
{\barbeq{\spwnb{x_1}{c}{e}{X}{\D}{c_0}{e_0}{P}}{d}{f}{\spwnb{x_2}{c}{e}{X}{\D}{c_0}{e_0}{Q}}}
\and
\inferrule*[right=$\msc{Fwd}$]
{\latcstrnot{d \sqsubseteq f}}
{\barbeq{\fwd{x^{\dlabelb{c_1}{e_1}}}{y^{\dlabelb{c_2}{e_2}}}}{d}{f}{\fwd{x^{\dlabelb{c_1}{e_1}}}{y^{\dlabelb{c_2}{e_2}}}}}
\\
\inferrule*[right=$\msc{Close}$]
{\latcstrnot{d \sqsubseteq f}}
{\barbeq{\cls{x^{\dlabelb{c}{e}}}}{d}{f}{\cls{x^{\dlabelb{c}{e}}}}}
\and
\inferrule*[right=$\msc{Wait}$]
{\latcstrnot{d \sqsubseteq f} \\
\barbeq{P}{d}{f \sqcup e}{Q}}
{\barbeq{\wait{x^{\dlabelb{c}{e}}}{P}}{d}{f}{\wait{x^{\dlabelb{c}{e}}}{Q}}}
\end{mathpar}
\end{small}
\end{center}
\vspace{-10pt}
\caption{Synchronization pattern checking rules of $\lang$.}
\vspace{-15pt}
\label{fig:sync-patterns}
\end{figure*}

Let us assume that right after branching, the known secret of a process (its running confidentiality) is of level {\small$d$}.
The goal of the synchronization pattern checks is to rule out any leakage of this secret of level $d$ via regrading.
Such leakage is only possible if the process (or any process that receives this secret from it) regrades to a lower or unrelated level than the secret $d$. However, if $d \sqsubseteq f$, we know that this can never happen.
Therefore, if {\small$\latcstr{d \sqsubseteq f}$}, the judgment {\small$\barbeq{P}{d}{f}{Q}$} trivially holds.
This case is handled by Rule {\small$\msc{Unsync}_3$} and is a base case of the inductive definition.

The interesting case is when $\latcstrnot{d \sqsubseteq f}$, meaning that  the process can potentially regrade to a lower (or unrelated) level than $d$. In this case, the rules have to ensure that the secret $d$ does not affect the ability of the process itself or the processes communicating with it to reach a regrading point. 
Furthermore, the secret $d$ cannot affect the continuation of the process after regrading.
In this case, the rules consider whether the next action in $P$ and $Q$ is a receive, send (except close), spawn, close, or forward:
\begin{itemize}
  \item The receives are checked to be synchronized in {\small$P$} and {\small$Q$} by the rules {\small$\msc{RcvLab}$} and {\small$\msc{RcvChan}$}. The pattern check is invoked inductively on the continuation, with updated running integrity ({\small$f \sqcup e$}) to take into account the receive.
   The confidentiality of the learned secret $d$, however, remains constant under inductive invocations as it has to continue preventing the leak of the original secret.
  The receives have to be synchronized as long as {\small$\latcstrnot{d \sqsubseteq f}$} holds, since different receives in {\small$P$} and {\small$Q$} might result in one branch reaching the regrading point and the other one not (related to non-reactiveness). 

  \item Different sends in two branches of a process does not impact whether or not the process itself reaches a regrading point (sends are non-blocking). But, it may impact whether or not the other process, on the receiving side, reaches the regrading point based on the secret.
If the carrier channel's min integrity $e$ is high enough, the receiving process cannot regrade to a level lower (or unrelated) than $d$, and we do not need to synchronize the sends. 
The sends must only be synchronized if the carrier channel's min integrity $e$ is not greater than or equal to the level $d$ of the secret ({\small$d \not\sqsubseteq e$}).
  Rules $\msc{Unsync}_1$ and $\msc{Unsync}_2$ correspond to the former case where   {\small$d \sqsubseteq e$}; for brevity, in these rules, we use process terms with any output prefix defined as
  {\small\(\outact{x}{c}{e}{P} \defined
  \labsnd{x^{\dlabel{c}{e}}}{k}{P} \alt
  \chnsnd{y}{x^{\dlabel{c}{e}}}{P}. \)}
  And rules $\msc{SndLab}$ and $\msc{SndChan}$ correspond to the latter where {\small$d \not\sqsubseteq e$}.
  In either case, the pattern check is invoked inductively on the continuation, with unchanged running integrity.
  \item Similar to the reasoning in the case of sends, if the running integrity of the spawned process and the min integrity of all its channels are high enough, there is no need to synchronize the spawns ($\msc{Unsync-Spawn}$ rules). Otherwise, the two branches must spawn the same processes with the same arguments ($\msc{Sync-Spawn}$).
  
  \item Rules $\msc{Close}$ and $\msc{Fwd}$ are the other base cases of the inductive definition. They insist that the two branches {\small$P$} and {\small$Q$} can synchronize their termination behavior.
  \end{itemize}

\subsection{Configuration typing and asynchronous dynamics}\label{sec:def-dynamics}
The configuration typing judgment is of the form 
{\small$\Psi_0; \Delta \Vdash \mathcal{C}:: \Delta'$},
where {\small$\Psi_0$} is the security lattice and {\small$\mathcal{C}$} is a set of runtime processes {\small $\mathbf{proc}(y_\alpha\langle c,e \rangle, P@ \langle c',e' \rangle)$} and messages {\small$\mathbf{msg}(M)$}.
\Cref{fig:config-typing} shows the configuration typing. The security premises in the {\small$\mathbf{proc}$} and {\small$\mathbf{msg}$} rules enforce the invariants (a)-(c) on the process term judgment before invoking process typing.

\begin{figure*}
\begin{center}
\begin{small}
    \def \MathparLineskip {\lineskip=0.15cm}
\begin{mathpar}
\inferrule*[right=$\mathbf{emp}$]
{\strut}
{{\it\color{red}\Psi_0}; x{:}A{\it\color{red}[\langle d,e \rangle]} \Vdash \cdot :: (x{:}A{\it\color{red}[\langle d,e \rangle]})}
\and
\inferrule*[right=$\mathbf{comp}$]
{{\it\color{red}\Psi_0}; \Delta_0 \Vdash \mathcal{C}:: \Delta \\
{\it\color{red}\Psi_0}; \Delta'_0 \Vdash \mc{C}_1:: x{:}A{\it\color{red}[\langle d, e\rangle]}}
{{\it\color{red}\Psi_0}; \Delta_0, \Delta'_0 \Vdash \mathcal{C}\, \mc{C}_1 :: \Delta, x{:} A{\it\color{red}[\langle d, e\rangle]}}
%\inferrule*[right=$\mathbf{emp}_2$]
%{\strut}
%{\Psi_0; \cdot \Vdash \cdot :: (\cdot)}
%
\and
\inferrule*[right=$\mathbf{proc}$]
{{\it\color{red}\Psi_0 \Vdash d_1 \sqsubseteq d} \\
{\it\color{red}\Psi_0 \Vdash e_1 \sqsubseteq e }\\
{\it\color{red}\forall y{:}B[\langle d',e' \rangle] \in \Delta'_0, \Delta \, (\Psi_0 \Vdash d' \sqsubseteq d)} \\
{\it\color{red}\forall y{:}B[\langle d',e' \rangle] \in \Delta'_0, \Delta \, (\Psi_0 \Vdash e' \sqsubseteq e)} \\
{\it\color{red}\Psi_0 \Vdash e_1 \sqsubseteq d_1} \\
{\it\color{red}\Psi_0 \Vdash e \sqsubseteq d} \\
{\it\color{red}\forall y{:}B[\langle d',e' \rangle] \in \Delta'_0, \Delta \, (\Psi_0 \Vdash e' \sqsubseteq d')} \\
{\it\color{red}\Psi_0}; \Delta_0 \Vdash \mathcal{C}:: \Delta \\
{\it\color{red}\Psi_0}; \Delta'_0, \Delta \vdash P{\it\color{red}@\langle d_1,e_1 \rangle}:: (x{:}A{\it\color{red}[\langle d,e \rangle]})}
{{\color{red}\Psi_0}; \Delta_0, \Delta'_0 \Vdash \mathcal{C}\, \mathbf{proc}(x{\color{red}[\langle d,e \rangle]}, P@{\color{red}\langle d_1,e_1 \rangle}):: (x{:}A{\color{red}[\langle d,e \rangle]})}
\and
\inferrule*[right=$\mathbf{msg}$]
{{\it\color{red}\forall y{:}B[\langle d',e' \rangle] \in \Delta'_0, \Delta \, (\Psi_0 \Vdash d' \sqsubseteq d)} \\
{\it\color{red}\Psi_0 \Vdash e \sqsubseteq d}\\
{\it\color{red}\forall y{:}B[\langle d',e' \rangle] \in \Delta'_0, \Delta \, (\Psi_0 \Vdash e' \sqsubseteq d')} \\
{\it\color{red}\Psi_0}; \Delta_0 \Vdash \mathcal{C}:: \Delta \\
{\it\color{red}\Psi_0}; \Delta'_0, \Delta \vdash M{\it\color{red}@\langle d,e \rangle}:: (x{:}A{\it\color{red}[\langle d,e \rangle]})}
{{\it\color{red}\Psi_0}; \Delta_0, \Delta'_0 \Vdash \mathcal{C}, \mathbf{msg}(M):: (x{:}A{\it\color{red}[\langle d,e \rangle]})}
\end{mathpar}
\end{small}
\end{center}
\vspace{-10pt}
\caption{Configuration typing rules of $\lang$.}

\label{fig:config-typing}
\end{figure*}

The dynamic rules in \Cref{fig:dynamics} take care of updating the running confidentiality and integrity of each process after a receive.  For brevity, we write {\small$\langle p \rangle$} to refer to a pair of confidentiality and integrity labels {\small$\langle c,e \rangle$}.
Rule {\small$\msc{Spawn}$} relies on the substitution mapping {\small$\gamma$} given by the programmer and its lifting {\small$\hat{\gamma}$} to the process term level.
It looks up the definition of process {\small$X$} in the signature and instantiates the security variables occurring in the process body using {\small$\gamma$}.
The condition {\small$\hat{\gamma}(\Psi')=\Psi_0$} ensures that all security variables are instantiated with a concrete security level.
For brevity, we omit a channel generations as well as {\small$\Sigma$}, which is fixed.

\begin{figure*}
\begin{center}
\begin{small}
\input{figs/dynamics}
 \vspace{-5pt}
\caption{Asynchronous dynamics of $\lang$.}
\label{fig:dynamics}
 \vspace{-15pt}
\end{small}
\end{center}
\end{figure*}

\subsection{Banking example}\label{sec:bank-examples}

The following example implements a bank that authorizes transactions made by its customers and sends a copy to their bank accounts. 
In line with our security lattice, we assume that the bank has two customers, Alice and Bob.
To authenticate themselves, a customer sends their token to the bank.
The bank then verifies the token and, if the verification is successful, sends the message $\mathit{succ}$ to the customer, otherwise the message $\mathit{fail}$.
Moreover, if the verification is successful, the bank creates a transaction statement and sends it to another process that represents the account of the customer in the bank.
Once done, the bank continues to serve the next customer by making a recursive call.
We assume that the bank alternates between its two customers, Alice and Bob, by making a mutually recursive call from $\m{Bank_A}$, which serves Alice, to $\m{Bank_B}$, which serves Bob, and vice versa. 
At each recursive call, a bank process regrades its running confidentiality to interact with the next customer.
The example showcases a characteristic feature of our type system: it accepts an implementation for a bank that interactively communicates with Alice and Bob without jeopardizing noninterference.

The following session types dictate the above protocol:

\begin{small}
\begin{minipage}[t]{0.5\textwidth}
\begin{tabbing}
$\mathsf{customer}= \oplus \{$\=$\mathit{tok_{black}}: \& \{\mathit{succ}: \mathsf{customer}, \mathit{fail}: \mathsf{customer}\},$\\
\>$\mathit{tok_{white}}: \& \{\mathit{succ}: \mathsf{customer}, \mathit{fail}: \mathsf{customer}\}\}$\\
\end{tabbing}
\end{minipage}
\hspace{8mm}
\begin{minipage}[t]{0.5\textwidth}
\begin{tabbing}
$\mathsf{account}= \mathsf{transfer} \multimap \mathsf{account}$\\
$\mathsf{transfer} = \oplus\{\mathit{transaction}: 1\}$
\end{tabbing}
\end{minipage}
\end{small}
%
% In this simple implementation we assume that Alice and Bob both have the same token (white), it is easy to see how we can change it (provide in the appendix?)
% If the token is true, the bank sends done to Alice, and spawns a bill for alice and sends it to its account in the bank usin $\multimap$. Then it will switch and cal Bank to communicate with bob.
%
\noindent \Cref{fig:poly-bank-impl} shows the process implementations {\small$\m{Bank_A}$, $\m{Bank_B}$}, {\small$\m{Customer_A}$}, and {\small$\m{Statement_A}$}.
The latter two are the implementation of Alice's customer and statement process, \respb.
The implementation of {\small$\m{Customer_A}$} is as expected.
{\small$\m{Statement_A}$} signals a single transfer by sending the label transaction and terminates.
The implementation of corresponding processes for Bob, i.e., {\small$\m{Customer_B}$}, and {\small$\m{Statement_B}$}, would be similar.
%
% The example is typed using the security theory {\small$\Psi:= \langle \mc{V}, \mc{E}\rangle$ with $\mc{V}:= \{\psi,\psi'\}$} and
% {\small$\mc{E}:= \{\mb{guest} \sqsubseteq \psi \sqsubseteq \mb{bank}, \mb{guest} \sqsubseteq \psi' \sqsubseteq \mb{bank}\},$} and substitution
% {\small$\gamma:=\{ \psi \mapsto \psi, \psi' \mapsto \psi'\}$}.
%
{The example is typed using the security theory {\small $\Psi$}, consisting of the concrete security lattice {\small$\Psi_0$}, the security variables {\small$\psi$} and {\small$\psi'$}, and the set of constraints {\small $\{\mb{guest} \sqsubseteq \psi \sqsubseteq \mb{bank}, \mb{guest} \sqsubseteq \psi' \sqsubseteq \mb{bank}\}$}. 
% The constraints enforce a relation between security variables and concrete security levels.
%  in the concrete lattice {\small$\Psi_0$}. 
%
(See the appendix for the formal definition of a security theory.)}
To execute this program using the dynamics in \Cref{fig:dynamics}, we provide the order-preserving substitution {\small$\Psi_0 \Vdash \gamma' :: \Psi$},  defined as {\small$\gamma':=\{ \psi \mapsto \mb{alice}, \psi' \mapsto \mb{bob}\}$}.

\begin{figure*}
\centering
\begin{small}
\input{figs/poly-bank-impl.tex}
\end{small}
 \vspace{-3pt}
\caption{Security-polymorphic process definitions.}
\label{fig:poly-bank-impl}
 \vspace{-10pt}
\end{figure*}

Let us examine the pattern checks {\small$\sim_{\langle \psi,\mb{guest}\rangle}$} invoked by {\small$\mb{case}\, y_1(\dots)$ in $\mb{Bank}_A$},
relating the branches corresponding to black and white tokens.
The sends along {\small$y_1$} match in both branches, as demanded by \msc{SndLab} (since {\small$\Psi \not \Vdash \psi\, {\sqsubseteq}\, \mb{guest}$}),
even though the sent labels are not the same.
The unsynchronized spawn and send along {\small$w_1$} is verified by \msc{Unsync-Spawn$_1$} and \msc{Unsync$_1$}, resp., since {\small$\Psi \Vdash \psi \sqsubseteq \psi$}.
The matching tail calls are verified with \msc{Sync-Spawn}.

\section{Progress-sensitive noninterference}
\label{sec:psnix}
This section presents our main result, PSNI, which we prove using a logical relation.
\subsection{Attacker model}
\label{subsubsec:attacker-model}
The attacker model assumes a configuration \mit{$\mathcal{D}$} with prior annotation of its free channels with security levels, the attacker's confidentiality level \mit{$\xi$}, and a nondeterministic scheduler. The attacker knows the source code of \mit{$\mathcal{D}$},  can only observe the messages sent along the free channels of \mit{$\mathcal{D}$} with confidentiality level {\small$c \sqsubseteq \xi$}, and cannot measure the passing of time.

\subsection{Noninterference via an integrity logical relation}
Noninterference amounts to a process
equivalence up to the confidentiality level $\xi$ of an observer.
In a message-passing system, it boils down to an equivalence of a configuration with interacting processes.
This section focuses on noninterference for tree-shaped configurations. The definition can be extended to forests by enforcing pairwise relation between their trees.

An open configuration {\small$\Psi_0; \Delta \Vdash \mc{D}:: x_\alpha{:}A \dlabel{c}{e}$} has the free channels {\small$\Delta$} and {\small$x_\alpha$} to communicate with its external environment; it sends outgoing messages to and receives incoming messages from the environment along these free channels. 
Two observationally equivalent configurations may only differ in outgoing messages of confidentiality level {\small$c_o \not\sqsupseteq \xi$}, assuming that the incoming messages of confidentiality level {\small$c_i \sqsubseteq \xi$} are the same.
We introduce a \emph{logical relation} that captures this idea and accounts for integrity and regrading policies.

The logical relation relates two open configurations {\small$\mnode_1$} and {\small$\mnode_2$}---the two runs of the program under consideration---and asserts that {\small$\mnode_1$} and {\small$\mnode_2$} send related messages to the environment, if they receive related messages from the environment.
The term interpretation of the logical relation, defined in \Cref{fig:rel_term}, allows the first configuration {\small$\mnode_1$} to step internally until the configuration is ready to send or receive a message across at least one external channel. Then, it requires the second configuration {\small$\mnode_2$} to step internally so that the resulting configurations are in the value interpretation of the logical relation, defined in \Cref{fig:rel_value_left} and \Cref{fig:rel_value_right}. We call the external channels, e.g., {\small${\color{mygreen}\Delta} \Vdash {\color{mygreen}K}$} in \Cref{fig:rel_term}, the interface of {\small$\mnode_1$} and {\small$\mnode_2$}. The metavariable $K$ stands for either $x_{\alpha}{:}A \langle c,e \rangle$ or simply $\_{:}1\langle \top, \top \rangle$ which refers to an arbitrary unobservable channel.

\begin{figure*}
    {\footnotesize
    \[\begin{array}{lcl}
    ({\mc{B}_1}, {\mc{B}_2}) \in \mc{E}^\xi_{ \Psi_0}\llbracket \Delta \Vdash K \rrbracket^{m+1} &
    \m{iff} &
     (\mc{D}_1;\mc{D}_2)\in \m{Tree}_{\Psi_0}(\Delta \Vdash K)\,\m{and} \forall \,\Upsilon_1,\, \Theta_1, \, \mc{D}'_1.\,\m{if}\; \mc{D}_1\mapsto^{{*}_{\Upsilon_1; \Theta_1}} \; \mc{D}'_1\,  \m{then}\\
    &&\exists\, \Upsilon_2\,\mc{D}'_2.\, \m{\,such \, that\,}\,\mc{D}_2\mapsto^{{*}_{ \Upsilon_2}}\mc{D}'_2 \; \m{and}\; \Upsilon_1 \subseteq \Upsilon_2\, \m{and}\\
    && \forall\, y_\alpha \in \mb{Out}(\Delta \Vdash K).\, \mb{if}\, y_\alpha \in \Upsilon_1.\, \mb{then}\, ({\mc{D}'_1}; {\mc{D}'_2})\in \mc{V}^\xi_{\Psi_0}\llbracket  \Delta \Vdash K \rrbracket_{\cdot;y_\alpha}^{m+1}\,\m{and}\,\\[2pt]
    && \forall\, y_\alpha \in \mb{In}(\Delta \Vdash K). \mb{if}\, y_\alpha \in \Theta_1.\, \mb{then}\, ({\mc{D}'_1}; {\mc{D}'_2})\in \mc{V}^\xi_{\Psi_0}\llbracket  \Delta \Vdash K \rrbracket_{y_\alpha;\cdot}^{m+1}\\[4pt]
      
    ({\mc{B}_1}, {\mc{B}_2}) \in \mc{E}^\xi_{ \Psi_0}\llbracket \Delta \Vdash K \rrbracket^{0} &
    \m{iff}&
    (\mc{D}_1; \mc{D}_2)\in \m{Tree}_{\Psi_0}(\Delta \Vdash K)
        \end{array}\]
    }
    \vspace{-5pt}
    \caption{Term interpretation of logical relation.}
    \label{fig:rel_term}
    \vspace{-10pt}
    \end{figure*}
    
\begin{figure*}
   % \begin{small}
    %
    \footnotesize
    \begin{mathpar}
    \begin{array}{ll}
      (l_1)&  (\mc{D}_1; \mc{D}_2)\in \mc{V}^\xi_{ \Psi_0}\llbracket \Delta, y_\alpha{:}1\langle c, e \rangle \Vdash K\rrbracket_{y^\alpha; \cdot}^{m+1}\\
   & \m{iff} \quad 
   \quad
    (\mc{D}_1; \mc{D}_2) \in \m{Tree}_{\Psi_0}(\Delta, y_\alpha{:}1\langle c, e \rangle \Vdash K)\, \m{then} \\
    & \;\quad \qquad (\mathbf{msg}( \mathbf{close}\,y_\alpha^{\langle c, e \rangle})\mc{D}_1; \mathbf{msg}( \mathbf{close}\,y_\alpha^{\langle c, e \rangle} )\mc{D}_2)\in \mathcal{E}^\xi_{\Psi_0}\llbracket \Delta \Vdash K\rrbracket^{m}    \\[4pt]
    %
   %  \hline
    %
    (l_2) & (\mc{D}_1;\mc{D}_2) \in \mathcal{V}^\xi_{ \Psi_0}\llbracket \Delta, y_\alpha:\oplus\{\ell{:}A_\ell\}_{\ell \in I}\langle c, e \rangle \Vdash K\rrbracket_{y^\alpha; \cdot}^{m+1}\\
    & \m{iff} \quad 
    \quad
    %c \sqsubseteq e \, \m{and}\,
    (\mc{D}_1;\mc{D}_2) \in \m{Tree}_{\Psi_0}(\Delta, y_\alpha{:}\oplus\{\ell{:}A_\ell\}_{\ell \in I}\langle c, e \rangle \Vdash K)\, \m{and}\, \forall k_1,k_2 \in I. \m{if} (c \sqsubseteq \xi \rightarrow k_1=k_2)\, \m{then}\\
   &\;\quad \qquad (\mathbf{msg}(y_\alpha^{\langle c, e \rangle}.k_1)\mc{D}_1; \mathbf{msg}(y_\alpha^{\langle c, e \rangle}.k_2)\mc{D}_2) \in  \mathcal{E}^\xi_{ \Psi_0}\llbracket \Delta, y_{\alpha+1}{:}A_{k_1}{\langle c, e \rangle} \Vdash K \rrbracket^{m}    \\[4pt]
      %
   %   \hline
      %
   %   (l_2') & (\mc{D}_1;\mc{D}_2) \in \mathcal{V}^\xi_{ \Psi_0}\llbracket \Delta, y_\alpha:\oplus\{\ell{:}A'\}_{\ell \in I}\langle c, e \rangle \Vdash K\rrbracket_{y^\alpha; \cdot}^{m+1}\\
    %  & \m{iff} \quad 
    %  \quad
    %  c \not \sqsubseteq e\, \m{and}\, (\mc{D}_1;\mc{D}_2) \in \m{Tree}_{\Psi_0}(\Delta, y_\alpha{:}\oplus\{\ell{:}A_\ell\}_{\ell \in I}\langle c, e \rangle \Vdash K)\, \m{and}\\
     % &\;\quad \qquad \mathcal{C}_1=\mathcal{C}'_1\mathbf{msg}( y_\alpha^{\langle c, e \rangle}.k; y_\alpha^{\langle c, e \rangle} \leftarrow u^{\langle c, e \rangle}_\delta) \, \m{and}\, \m{if}\, \mathcal{C}_2=\mathcal{C}'_2 \mathbf{msg}(y_\alpha^{\langle c, e \rangle}.j;y_\alpha^{\langle c, e \rangle} \leftarrow u^{\langle c, e \rangle}_\delta )\,\m{then}\\
    % &\;\quad \qquad (\lre{\mc{C}'_1} {\mathbf{msg}(y_\alpha^{\langle c, e \rangle}.k; y_\alpha^{\langle c, e \rangle} \leftarrow u^{\langle c, e \rangle}_\delta)\mc{D}_1} {\mc{F}_1},\lre{\mc{C}'_2} {\mathbf{msg}(y_\alpha^{\langle c, e \rangle}.j; y_\alpha^{\langle c, e \rangle} \leftarrow u^{\langle c, e \rangle}\delta)\mc{D}_2} {\mc{F}_2}) \in  \mathcal{E}^\xi_{ \Psi_0}\llbracket \Delta, y_{\alpha+1}{:}A'{\langle c, e \rangle}] \Vdash K \rrbracket^{m}    \\[4pt]
        %
      %   \hline
        %
      (l_3) &   (\mc{D}_1;\mc{D}_2) \in \mathcal{V}^\xi_{ \Psi_0}\llbracket \Delta, y_\alpha{:}\&\{\ell{:}A_\ell\}_{\ell \in I}\langle c, e \rangle \Vdash K \rrbracket_{\cdot; y^\alpha}^{m+1}\\
      & \m{iff} \quad 
      \quad
       (\mc{D}_1;\mc{D}_2) \in \m{Tree}_{\Psi_0}(\Delta, y_\alpha{:}\&\{\ell{:}A_\ell\}_{\ell \in I}\langle c, e \rangle \Vdash K)\, \m{and}\, \exists k_1,k_2 \in I.  (c \sqsubseteq \xi \rightarrow k_1=k_2)\, \m{and}\,\\
      &\;\quad \qquad \mathcal{D}_1=\mathbf{msg}( y_\alpha^{\langle c, e \rangle}.k_1) \mathcal{D}'_1 \, \m{and}\ \mathcal{D}_2=\mathbf{msg}( y_\alpha^{\langle c, e \rangle}.k_2)\mathcal{D}'_2\,  \,\m{and} \\
     &\;\quad \qquad ( \mc{D}'_1; \mc{D}'_2)
      \in  \mathcal{E}^\xi_{ \Psi_0}\llbracket \Delta, y_{\alpha+1}{:}A_{k_1}{\langle c, e \rangle} \Vdash K \rrbracket^{m}     \\[4pt]
    %
      %   \hline
        %
    (l_4) & (\mc{D}_1;\mc{D}_2)  \in  \mathcal{V}^\xi_{ \Psi_0}\llbracket \Delta, y_\alpha{:}A\otimes B\langle c, e \rangle \Vdash K \rrbracket_{y^\alpha; \cdot}^{m+1} \\
    & \m{iff} \quad 
    \quad
    (\mc{D}_1;\mc{D}_2) \in \m{Tree}_{\Psi_0}(\Delta, y_\alpha{:}A\otimes B\langle c, e \rangle \Vdash K)\,\m{and} \forall x_\beta {\not \in} \mathit{dom}(\Delta,y_\alpha{:} A \otimes B \langle c,e \rangle, K).\\
     &\;\quad \qquad (\mathbf{msg}( \mb{send}x_\beta^{\langle c, e \rangle},y_\alpha^{\langle c, e \rangle})\mc{D}_1; \mathbf{msg}( \mb{send}x_\beta^{\langle c, e \rangle},y_\alpha^{\langle c, e \rangle})\mc{D}_2)  \in  \mathcal{E}^\xi_{ \Psi_0}\llbracket \Delta,x_\beta{:}A\langle c, e \rangle, y_{\alpha+1}{:}B\langle c, e \rangle \Vdash K \rrbracket^{m}  \\[4pt]
     %
   %   \hline
     % 
    (l_5) & (\mc{D}_1;\mc{D}_2)  \in \mathcal{V}^\xi_{\Psi_0}\llbracket \Delta', \Delta'', y_\alpha{:}A\multimap B\langle c, e \rangle \Vdash K \rrbracket_{\cdot;y^\alpha}^{m+1}\\
    & \m{iff} \quad 
    \quad
    (\mc{D}_1;\mc{D}_2) \in \m{Tree}_{\Psi_0}(\Delta', \Delta'', y_\alpha{:}A\multimap B\langle c, e \rangle \Vdash K)\,\m{and}\,\\
    & \;\quad \qquad \mathcal{D}_1=\mathcal{T}_1\mathbf{msg}( \mb{send}x_\beta^{\langle c, e \rangle}\,y_\alpha^{\langle c, e \rangle})\,\mathcal{D}''_1\, \m{and} \,\m{for}\, \mc{T}_1 \in \m{Tree}_{\Psi_0}(\Delta' \Vdash x_\beta{:}A\langle c, e \rangle)\\
    & \;\quad \qquad \mathcal{D}_2=\mathcal{T}_2\mathbf{msg}( \mb{send}x_\beta^{\langle c, e \rangle}\,y_\alpha^{\langle c, e \rangle})\,\mc{D}''_2\, \m{and}  \,\m{for}\, \mc{T}_2 \in \m{Tree}_{\Psi_0}(\Delta' \Vdash x_\beta{:}A\langle c, e \rangle)\,\m{and} \\
    & \;\quad \qquad (\mc{T}_1; \mc{T}_2) \in  \mathcal{E}^\xi_{ \Psi_0}\llbracket \Delta' \Vdash x_\beta{:}A\langle c, e \rangle \rrbracket^{m} \,\m{and}\\
    &\;\quad \qquad(\mathcal{D}''_1;\mathcal{D}''_2) \in  \mathcal{E}^\xi_{ \Psi_0}\llbracket \Delta'', y_{\alpha+1}{:}B\langle c, e \rangle \Vdash K \rrbracket^{m}  
    \end{array}
    \end{mathpar}
    %
   % \end{small}
   \vspace{-5pt} 
   \caption{Value interpretation of logical relation for \emph{left} communications.}
    \label{fig:rel_value_left}
    \vspace{-10pt}
    \end{figure*}

\begin{figure*}
{\footnotesize%
\[\begin{array}{ll}
(r_1)& ({\mc{D}_1 ;\mc{D}_2})\in \mc{V}^\xi_{ \Psi_0}\llbracket \cdot \Vdash y_\alpha{:}1\langle c,e \rangle\rrbracket_{\cdot;y^\alpha}^{m+1}\\
&\;\m{iff}\quad 
\quad (\mc{D}_1; \mc{D}_2 ) \in \m{Tree}_{\Psi_0}(\cdot \Vdash y_\alpha)\, \m{and} \,
%\\
%& \; \quad \qquad 
\mc{D}_1=\mathbf{msg}(\mathbf{close}\,y_\alpha^{\langle c,e \rangle} ) \, \m{and}\, \mc{D}_2=\mathbf{msg}( \mathbf{close}\,y_\alpha^{\langle c,e \rangle})\\[2pt]
%  \hline
(r_2)&(\mc{D}_1;\mc{D}_2)\in \mathcal{V}^\xi_{\Psi_0}\llbracket (\Delta \Vdash y_\alpha{:}\oplus\{\ell{:}A_\ell\}_{\ell \in I}\langle c,e \rangle)\rrbracket_{\cdot;y^\alpha}^{m+1} \\
&\;\m{iff} \quad \quad
%c \sqsubseteq e \, \m{and}\, 
(\mc{D}_1;\mc{D}_2) \in \m{Tree}_{\Psi_0}(\Delta \Vdash y_\alpha:\oplus\{\ell{:}A_\ell\}_{\ell \in I}\langle c,e \rangle)\, \m{and} \exists k_1, k_2 \in I.\,\, (c \sqsubseteq \xi \rightarrow k_1=k_2) \\
&\;\quad \qquad\, \mathcal{D}_1=\mathcal{D}'_1\mathbf{msg}( y_\alpha^{\langle c,e \rangle}.k_1) \, \m{and}\, \mathcal{D}_2=\mathcal{D}'_2 \mathbf{msg}(y_\alpha^{\langle c,e \rangle}.k_2)\,\\
&\;\quad \qquad\, \m{and}\, ( {\mc{D}'_1}; \mc{D}'_2) \in  \mathcal{E}^\xi_{ \Psi_0}\llbracket \Delta \Vdash y_{\alpha+1}{:}A_{k_1}\langle c,e \rangle\rrbracket^{m}\\[2pt]
% \hline
(r_3) &(\mc{D}_1;\mc{D}_2) \in \mathcal{V}^\xi_{ \Psi_0}\llbracket \Delta \Vdash y_\alpha{:}\&\{\ell{:}A_\ell\}_{\ell \in I}\langle c,e \rangle\rrbracket_{y^\alpha; \cdot}^{m+1} \\
&\;\m{iff} \quad \quad
%c \sqsubseteq e \,\m{and}\,
(\mc{D}_1;\mc{D}_2) \in \m{Tree}_{\Psi_0}(\Delta \Vdash y_\alpha{:}\&\{\ell{:}A_\ell\}_{\ell \in I}\langle c,e \rangle)\,\m{then}\forall k_1, k_2 \in I. \m{if}\, (c \sqsubseteq \xi \rightarrow k_1=k_2)\,  \m{then}\,\\
&\; \quad \qquad  (\mc{D}_1 \mathbf{msg}( y_\alpha^{\langle c,e \rangle}.k_1) , \mc{D}_2\mathbf{msg}( y_\alpha^{\langle c,e \rangle}.k_2))\in  \mathcal{E}^\xi_{\Psi_0}\llbracket \Delta \Vdash y_{\alpha+1}{:}A_{k_1}\langle c,e \rangle\rrbracket^{m}\\[2pt]
   % \hline
(r_4) & (\mc{D}_1;\mc{D}_2)  \in \mathcal{V}^\xi_{ \Psi_0}\llbracket \Delta',\Delta'' \Vdash y_\alpha{:}A\otimes B\langle c,e \rangle\rrbracket_{\cdot; y^\alpha}^{m+1} \\
&\;\m{iff} (\mc{D}_1;\mc{D}_2) \in \m{Tree}_{\Psi_0}(\Delta', \Delta'' \Vdash y_\alpha{:}A\otimes B\langle c,e \rangle)\,\m{and}\, \exists x_\beta.\\
&\; \quad \qquad
 \mathcal{D}_1=\mathcal{D}'_1\mathcal{T}_1\mathbf{msg}( \mb{send}\,x_\beta^{\langle c,e \rangle}\,y_\alpha^{\langle c,e \rangle} )\, \m{for}\, \mc{T}_1 \in \m{Tree}_{\Psi_0}(\Delta'' \Vdash x_\beta{:}A\langle c,e \rangle)\, \m{and}\\
&\; \quad \qquad \mathcal{D}_2=\mathcal{D}'_2\mc{T}_2 \mathbf{msg}( \mb{send}\,x_\beta^{\langle c,e \rangle}\,y_\alpha^{\langle c,e \rangle})\,\m{for}\, \mc{T}_2 \in \m{Tree}_{\Psi_0}(\Delta'' \Vdash x_\beta{:}A\langle c,e \rangle)\, \m{and}\\
&\; \quad \qquad
 (\mc{T}_1;\mc{T}_2)  \in  \mathcal{E}^\xi_{ \Psi_0}\llbracket \Delta'' \Vdash x_{\beta}{:}A\langle c,e \rangle\rrbracket^{m}  \,\m{and}\\
 &\; \quad \qquad
 (\mc{D}'_1;\mc{D}'_2)  \in  \mathcal{E}^\xi_{ \Psi_0}\llbracket \Delta' \Vdash y_{\alpha+1}{:}B\langle c,e \rangle\rrbracket^{m}  \\[2pt]
%   \hline
(r_5) & (\mc{D}_1;\mc{D}_2)  \in \mathcal{V}^\xi_{ \Psi_0}\llbracket \Delta \Vdash y_\alpha{:}A\multimap B\langle c,e \rangle\rrbracket_{y^\alpha; \cdot}^{m+1}\\
& \; \m{iff} \quad \quad
(\mc{D}_1;\mc{D}_2) \in \m{Tree}_{\Psi_0}(\Delta \Vdash y_\alpha{:}A\multimap B\langle c,e \rangle)\, \m{and}\,  \forall x_\beta {\not\in} \mathit{dom}(\Delta,y_\alpha{:}A \multimap B\langle c,e \rangle).\\
& \;\quad \qquad (\mc{D}_1\mathbf{msg}( \mb{send}\,x_\beta^{\langle c,e \rangle}\,y_\alpha^{\langle c,e \rangle}); \mc{D}_2\mathbf{msg}(\mb{send}\,x_\beta^{\langle c,e \rangle}\,y_\alpha^{\langle c,e \rangle}))  \in  \mathcal{E}^\xi_{ \Psi_0}\llbracket \Delta, x_\beta{:}A\langle c,e \rangle\Vdash y_{\alpha+1}{:}B\langle c,e \rangle\rrbracket^{m} 
   \end{array}\]%
}
\vspace{-3pt}
\caption{Value interpretation of logical relation for \emph{right} communications.}
\label{fig:rel_value_right}
\vspace{-10pt}
\end{figure*}

The idea is to build an interface consisting of those external channels of the configurations that may impact the attacker's observations. As such, not only do we need to include the observable channels, i.e., with confidentiality level {\small$c \sqsubseteq \xi$}, in the interface, but also those with higher integrity than the observer, i.e., with integrity level {\small$e \sqsubseteq \xi$}. After all, if a channel's integrity is high enough (and thus its level is low), the messages along it may affect an observable outcome via synchronization patterns.
We call such an interface \emph{integrity interface} since low-confidentiality channels are all high-integrity by typing.

To build an \emph{integrity interface} for {\small$\mnode_1$} and {\small$\mnode_2$}, we close off their external low-integrity ({\small$e\not \sqsubseteq \xi$}) channels on the left by composing the channels with any well-typed provider and on the right with any well-typed client. We may use different low-integrity clients and providers to compose with each program run. These clients/providers can send different and unsynchronized messages along their high-confidentiality and low-integrity connections to {\small$\mnode_1$} and {\small$\mnode_2$}. The term interpretation is designed to ensure that well-typed configurations do not leak these different messages to the attacker. \Cref{fig:rel_eq} defines an equivalence relation between two configurations based on this idea: it composes them with low-integrity providers/clients and calls the term interpretation symmetrically on the compositions.
\noindent  In the definition, we use the projection function to build the integrity interface, e.g., {\small$\Delta {\Downarrow^{\m{ig}}} \xi$} projects out the channels {\small $y_\beta{:}A\langle c,e \rangle \in \Delta$} with $\xi \not \sqsubseteq e$. 
The predicate {\small$\mc{D}_1\in \m{Tree}_{\Psi_0}({\color{mygreen}\Delta} \Vdash {\color{mygreen}K})$} indicates that the configuration {\small$\mnode_1$} is well-typed.
 In the term and value interpretations, we generalize this predicate to the binary case, {\small$(\mc{D}_1; \mc{D}_2)\in \m{Tree}_{\Psi_0}({\color{mygreen}\Delta} \Vdash {\color{mygreen}K})$} indicating that both {\small$\mnode_1$} and {\small$\mnode_2$} are of the same type.

\begin{figure*}
{\small    \[
    \begin{array}{lll}
        (\Delta_1 \Vdash \mc{D}_1:: x_\alpha {:}A_1\langle c_1, e_1 \rangle) \equiv^{\Psi_0}_{\xi} (\Delta_2 \Vdash \mc{D}_2:: y_\beta {:}A_2\langle c_2,e_2 \rangle)  \;\;\m{iff}\\ 
       \;\; 
       {\mc{D}_1} \in \m{Tree}({\Delta_1} \Vdash x_\alpha {:}A_1\langle c_1, e_1 \rangle)\,\, \m{and}\,\,{\mc{D}_2} \in \m{Tree}({\Delta_2} \Vdash y_\beta {:}A_2\langle c_2, e_2 \rangle)\,\, \m{and}\,\,
      \Delta_1 {\Downarrow^{\m{ig}}} \xi= \Delta_2 {\Downarrow^{\m{ig}}} \xi=\Delta\,\, \m{and} \\ 
      \;\; x_\alpha {:}A_1\langle c_1, e_1 \rangle \Downarrow^{\m{ig}} \xi= y_\beta {:}A_2\langle c_2, e_2 \rangle \Downarrow^{\m{ig}} \xi= K\,\,\m{and}\,\,
      \forall \, \mc{B}_1 \in \mb{L\text{-}IProvider}^\xi(\Delta_1).\, \forall \, \mc{B}_2 \in \mb{L\text{-}IProvider}^\xi(\Delta_2).\\ 
       \;\;\forall \mc{T}_1 \in \mb{L\text{-}IClient}^\xi(x_\alpha {:}A_1\langle c_1, e_1 \rangle). \,\forall \mc{T}_2 \in \mb{L\text{-}IClient}^\xi(y_\beta {:}A_2\langle c_2, e_2 \rangle).\\
     \;\; \forall\,m.\,(\mc{B}_1\mc{D}_1\mc{T}_1, \mc{B}_2\mc{D}_2\mc{T}_2) \in \mc{E}^\xi_{\Psi_0}\llbracket{\Delta} \Vdash {K} \rrbracket^{m},\,\,
      \m{and}\,\,
      \forall\,m.\,(\mc{B}_2\mc{D}_2\mc{T}_2, \mc{B}_1\mc{D}_1\mc{T}_1) \in \mc{E}^\xi_{\Psi_0}\llbracket {\Delta} \Vdash {K} \rrbracket^{m}.
    \end{array}  
    \]
% \vspace{-3pt}
    \[\begin{array}{lll}
        \cdot \,\,\in \mb{L\text{-}IProvider}^\xi(\cdot)\\
        \mc{B} \in \mb{L\text{-}IProvider}^{\xi}(\Delta, x_\alpha{:}A\langle c, e \rangle )\, \,\; \m{iff}\\
        \quad e \not \sqsubseteq \xi  \, \m{and}\,
        \mc{B} = \mc{B}' \mc{T} \,\, \m{and}\,\, \mc{B}' \in \mb{L\text{-}IProvider}^\xi(\Delta)\,\,\m{and}\,\, \mc{T} \in \m{Tree}(\cdot \Vdash x_\alpha{:}A\langle c, e \rangle), \mb{or}\\
         \quad e\sqsubseteq \xi  \, \m{and}\,
       \mc{B} \in \mb{L\text{-}IProvider}^\xi(\Delta)\\[5pt]
        \mc{T} \in \mb{L\text{-}IClient}^\xi(x_\alpha{:}A\langle c, e \rangle)\, \,\;  \m{iff} \,\;\\
        \quad  e \not \sqsubseteq \xi  \, \m{and}\, \mc{T} \in \m{Tree}(x_\alpha{:}A\langle c, e \rangle \Vdash \_:1\langle \top, \top\rangle),\,\, \mb{or}\,\, e \sqsubseteq \xi  \, \m{and}\, \mc{T}= \cdot
      \end{array}\]
   }
   \vspace{-5pt}
   \caption{Logical equivalence.}
   \vspace{-10pt}
    \label{fig:rel_eq}
    \end{figure*}

The term interpretation allows stepping configuration  {\small$\mc{D}_1\mapsto^{*_{ \Upsilon_1; \Theta_1}}{\mc{D}'_1}$}
    \noindent  by iterated application of the rewriting rules defined in \Cref{fig:dynamics}.
    The star expresses that zero to multiple internal steps can be taken.
    The superscripts {\small$\Upsilon_1; \Theta_1$} denote two sets of channels occurring in the interface {\small${\color{mygreen}\Delta} \Vdash {\color{mygreen}K}$}.
    The set {\small$\Theta_1$} collects the \emph{incoming} channels, \ie channels that a process in $\mnode_1$ is ready to receive from, and the set {\small$\Upsilon_1$} collects the \emph{outgoing} channels, \ie channels with a message in {\small$\mnode_1$} ready to be sent.
    Assuming that {\small$\mnode_1$} steps to  {\small$\mc{D}'_1$}, generating the outgoing channels {\small$\Upsilon_1$}, {\small$\mnode_2$} must be stepped {\small$\mc{D}_2\mapsto^{*_{ \Upsilon_2}}{\mc{D}'_2}$} to produce at least the same set of outgoing channels, i.e., the set  {\small$\Upsilon_2$} such that {\small$\Upsilon_1 \subseteq \Upsilon_2$}.
The term interpretation then calls the value interpretation on the resulting configurations {\small$\mc{D}_1',\mc{D}_2'$} \emph{for every} channel that has a message ready for transmission in $\mc{D}'_1$, and thus $\mc{D}'_2$, and \emph{for every} channel that has a process waiting for a message in $\mc{D}'_1$.
Insisting on $\Upsilon_2$ being a superset of $\Upsilon_1$ ensures progress-sensitive noninterference without timing attacks:
if a configuration produces observable messages along a set of channels, the other configuration has to be able to produce the equivalent set of messages with zero or some internal steps.
The term interpretation uses focus channels as a subscript to the value interpretation to support simultaneous communications --- when there are multiple messages ready to be sent or received along channels in the interface.
The subscript \mit{$\cdot; y_\alpha$} indicates that \mit{$y_\alpha \in \Upsilon_1$} and \mit{$y_\alpha; \cdot$} that \mit{$y_\alpha \in \Theta_1$}.

The value interpretation accounts for every message sent from or received by {\small$\mnode_1$} and {\small$\mnode_2$},
amounting to two cases per connective: one for a message exchanged along a channel in {\small${\color{mygreen}K}$} and one for a message exchanged along a channel in {\small${\color{mygreen}\Delta}$}.
We refer to the former as communications to the \emph{right} (\Cref{fig:rel_value_right}) and the latter as communications to the \emph{left} (\Cref{fig:rel_value_left}).
The value interpretation generally establishes the following pattern: it asserts relatedness of outgoing messages, but assumes relatedness of incoming messages.
For example, $\&$ on the left ($l_3$ in \Cref{fig:rel_value_left}) \emph{asserts} the sending of related messages and pushes the messages into the environment, yielding {\small$\mnode_1'$}, {\small$\mnode_2'$}.
Now,
{\small$\mnode_1'$}, {\small$\mnode_2'$},  can each step internally, \eg to consume the incoming messages, requiring them to be in the term interpretation. 
On the other hand, $\&$
on the right ($r_{3}$ in \Cref{fig:rel_value_right}) \emph{assumes} receipt of related messages and pushes the messages into the configurations {\small$\mnode_1$} and {\small$\mnode_2$}.
\emph{Relatedness} for messages is determined by how they can impact the attacker's observations. If their carrier channel is observable to the attacker, i.e., has confidentiality level {\small$c \sqsubseteq \xi$}, then related messages must have the same labels. But if the channel only affects the attacker's observations via synchronization patterns, related messages may have different labels. The clause $c \sqsubseteq \xi \rightarrow k_1=k_2$ in the value interpretation conveys this, enforcing equality of the communicated labels only if the channel is observable.

Relatedness for higher-order types (\mit{$\chanoutsymb$} and \mit{$\chaninsymb$}) is a bit more subtle.
In particular, it requires  future observations along the exchanged channels to be related.
For example, $l_5$ in \Cref{fig:rel_value_left} for \mit{$\multimap$}-left asserts
existence of a message \mit{$\mathbf{msg}( \mb{send}x_\beta^{\langle c, e \rangle}\,y_\alpha^{\langle c, e \rangle})$}
and of subtrees \mit{$\mathcal{T}_1$} and \mit{$\mathcal{T}_2$}
in \mit{$\mc{D}_1$} and \mit{$\mc{D}_2$}.
The clause comprises two invocations of the term relation,
\mit{$ (\mc{T}_1; \mc{T}_2) \in  \mathcal{E}^\xi_{ \Psi_0}\llbracket \Delta' \Vdash x_\beta{:}A\langle c, e \rangle \rrbracket^{m}$},
asserting that future observations to be made along the sent channel \mit{$x_\beta$} are related, and
$(\mathcal{D}''_1;\mathcal{D}''_2) \in  \mathcal{E}^\xi_{ \Psi_0}\llbracket \Delta'', y_{\alpha+1}{:}B\langle c, e \rangle \Vdash K \rrbracket^{m}$,
asserting that the continuations \mit{$\mc{D}''_1$} and \mit{${\mc{D}''_2}$} are related.
Conversely, $r_5$ in \Cref{fig:rel_value_right} for \mit{$\multimap$}-right assumes
receipt of a message $\mathbf{msg}( \mb{send}x_\beta^{\langle c, e \rangle}\,y_\alpha^{\langle c, e \rangle})$ and invokes the term relation
{\small$(\mc{D}_1\mathbf{msg}( \mb{send}\,x_\beta^{\langle c,e \rangle}\,y_\alpha^{\langle c,e \rangle}); \mc{D}_2\mathbf{msg}(\mb{send}\,x_\beta^{\langle c,e \rangle}\,y_\alpha^{\langle c,e \rangle}))  \in  \mathcal{E}^\xi_{ \Psi_0}\llbracket \Delta, x_\beta{:}A\langle c,e \rangle\Vdash y_{\alpha+1}{:}B\langle c,e \rangle\rrbracket^{m}$}.

As we support general recursive types, we need an index to stratify our logical relation~\cite{AhmedESOP2006,AppelMcAllesterTOPLAS2001}. 
We tie our index to the number of \emph{observations} that can be made along the interface {\small${\color{mygreen}\Delta} \Vdash {\color{mygreen}K}$}, as suggested in~\cite{BalzerARXIV2023}.
We thus bound the value and term interpretation of our logical relation by the number of observations {\small$m$}, for {\small$m \geq 0$}, and attach them as superscripts to the relation's interface {\small${\color{mygreen}\Delta} \Vdash {\color{mygreen}K}$}.
The base case of the term interpretation, i.e., {\small$m=0$}, is a trivial relation.

The fundamental theorem states that any well-typed $\lang$ configuration is equivalent to itself up to the level of an arbitrary observer.
    \begin{theorem}[Fundamental theorem]\label{lem:reflexivity}
        For all security levels {\small$\xi$}, and a well-typed configuration {\small\({\Psi_0; \Delta \Vdash \mathcal{D}:: u_\alpha {:}T\langle c,e \rangle}\)}  we have {\small\((\Delta \Vdash \mathcal{D}:: u_\alpha {:}T\langle c, e \rangle)  \equiv^{\Psi_0}_{\xi} (\Delta \Vdash \mathcal{D}:: u_\alpha {:}T\langle c, e \rangle).\)} 
    \end{theorem}
\begin{proof}
   {We present a proof sketch; see the appendix for details.
    By the definition of logical equivalence (\Cref{fig:rel_eq}), we first need to close off the external low integrity ({\small$e\not \sqsubseteq \xi$}) channels in {\small$\Delta$} and {\small$u_\alpha {:}T\langle c, e \rangle$} by composing {\small$\mathcal{D}$} with arbitrary well-typed providers and clients, \respb.  We use  low integrity clients, {\small$\mathcal{T}_1$} and {\small$\mathcal{T}_2$}, and low-integrity providers, {\small$\mathcal{B}_1$} and {\small$\mathcal{B}_2$}, to compose with each run, resulting in two configurations, {\small$\mathcal{D}_1=\mathcal{B}_1\mathcal{D}\mathcal{T}_1$ and $\mathcal{D}_2=\mathcal{B}_2\mathcal{D}\mathcal{T}_2$}.
    Configurations {\small$\mathcal{D}_1$} and {\small$\mathcal{D}_2$} are both well-typed for the integrity interface: {\small\({\Psi_0; \Delta' \Vdash \mathcal{D}_i::K}\)}
    where {\small$\Delta'=\Delta {\Downarrow^{\m{ig}}} \xi$} and  {\small $K= u_\alpha {:}T\langle c,e \rangle {\Downarrow^{\m{ig}}} \xi$}. By the definition in~\Cref{fig:rel_eq}, it is enough to show that {\small$\mathcal{D}_1$} and {\small$\mathcal{D}_2$} are in the term interpretation with the integrity interface, i.e., {\small$\forall\,m.\,(\mc{D}_1, \mc{D}_2) \in \mc{E}^\xi_{\Psi_0}\llbracket {\Delta'} \Vdash {K} \rrbracket^{m}$} and {\small $\forall\,m.\,(\mc{D}_2, \mc{D}_1) \in \mc{E}^\xi_{\Psi_0}\llbracket {\Delta'} \Vdash {K} \rrbracket^{m}$}.
    We prove the former by induction on $m$; the proof of the latter is symmetric. 
Specifically, we prove a more general theorem (\Cref{thm:ap_main} in the appendix) for any {\small$\mathcal{D}_1$} and {\small$\mathcal{D}_2$} with the same observable outcome, using the notion of relevant nodes (\Cref{def:relevant-node} in the appendix).
    By the definition of the term interpretation in~\Cref{fig:rel_term}, the base case ({\small$m=0$}) is straightforward. 
    For the inductive case, following the first row of~\Cref{fig:rel_term}, 
    we assume arbitrary {\small$\Upsilon_1$}, {\small$\Theta_1$}, and {\small$\mc{D}'_1$} such that
    {\small$\mc{D}_1\mapsto^{{*}_{\Upsilon_1; \Theta_1}} \; \mc{D}'_1$}. We apply a lemma (\Cref{lem:indinvariant1} in the appendix) stating that {\small$\mc{D}_2$} can simulate the internal steps taken by {\small$\mc{D}_1$}, producing at least the same set of outgoing channels {\small$\Upsilon_2$}, i.e., {\small$\mc{D}_2\mapsto^{{*}_{\Upsilon_2}} \; \mc{D}'_2$}, such that {\small$\mc{D}'_1$} and {\small$\mc{D}'_2$} continue to have the same observable outcomes. 
    Finally, for every channel $x_\alpha$ in {\small$\Upsilon_1$} and {\small$\Theta_1$}, we case analyze on the type of $x_\alpha$, showing that {\small$\mc{D}'_1$} and {\small$\mc{D}'_2$} are in the value interpretation, with $x_\alpha$ being the focus channel. To do so, we use the induction hypothesis to establish that after the corresponding communication with the environment, the continuations of {\small$\mc{D}'_1$} and {\small$\mc{D}'_2$} are related by the term interpretation for a smaller index.
      }
\end{proof}
 \subsection{Adequacy}
 Next, we prove an adequacy theorem showing that two logically equivalent configurations are \emph{bisimilar} up to observations of confidentiality $\xi$. 
 
 For adequacy, we are interested in \emph{a confidentiality interface}, i.e., channels with observable max confidentiality {\small$c \sqsubseteq \xi$}; after all, our goal is to prove that the configurations are equivalent up to the confidentiality of an observer. 
Because the integrity interface of our logical relation is a \emph{superset} of the confidentiality interface, we need to close off those channels in the integrity interface that are of high-confidentiality ({\small$c \not \sqsubseteq \xi$}).
Note that these high-confidentiality channels are of high-integrity ({\small$e \sqsubseteq \xi$}).
To close off these channels, we compose the open configurations with high-confidentiality clients and providers, possibly different ones for each program run.
These high-integrity clients and providers are connected to the open configurations via high-integrity channels and,
as a result, may affect the observable outcome of the two runs via synchronization patterns.
We therefore require them to be logically equivalent.

Based on this idea, \Cref{fig:bisim} defines the \emph{bisimulation up to} confidentiality $\xi$ denoted as {\small$\approx_a^{\xi}$}, for two well-typed configurations: it first composes the two configurations with high-confidentiality providers ({\small$\mb{Hrel\text{-}IProvider}$}) and clients ({\small$\mb{H\text{-}CClient}$}), while insisting that the high-integrity parts of the providers and clients are logically equivalent (using the relations {\small$\mb{Hrel\text{-}IProvider}$} and {\small$\mb{Hrel\text{-}IClient}$}). Then it invokes an asynchronous bisimulation {\small$\approx_a$} on the compositions.
The definition uses a projection function {\small${\Downarrow^{\m{cf}}} \xi$} to build the confidentiality interface, e.g., {\small$\Delta {\Downarrow^{\m{cf}}} \xi$} projects out the channels in {\small$\Delta$} with confidentiality {\small$c \not \sqsubseteq \xi$}. 

The \emph{asynchronous bisimulation} {\small$\approx_a$} invoked by the definition in~\Cref{fig:bisim} uses a labeled transition system (LTS) following the standard definition of asynchronous bisimulation~\cite{SangiorgiWalkerBook2001}. The relation {\small$\mc{D}_1 \approx_a \mc{D}_2$}  states that every internal step or external action of {\small$\mc{D}_1$} can be (weakly) simulated by {\small$\mc{D}_2$} and vice-versa.  For example, when {\small$\mc{D}_1$} takes an action by sending output {\small$q$}  via an external channel {\small$x_\alpha$}, i.e., {\small$\mc{D}_1 \xrightarrow{\overline{x_\alpha}\,q} \mc{D}'_1$}, the bisimulation ensures that for some  {\small$\mc{D}'_2$}, we have {\small$\mc{D}_2 \xRightarrow{\overline{x_\alpha}\,q} \mc{D}'_2$} and {\small$\mc{D}'_1 \approx_a \mc{D}'_2$}.
Here, {\small$\xRightarrow{\overline{x_\alpha}\,q}$} stands for taking zero or more internal steps before outputting {\small$q$} along {\small$x_\alpha$}. The full definition of bisimulation is in the appendix.

 \begin{figure*}
{\small  
\[\begin{array}{ll}
  \Delta_1 \Vdash \mc{D}_1 :: x_\alpha {:}A_1\langle c_1 , e_1 \rangle \approx^\xi_a \Delta_2 \Vdash \mc{D}_2:: y_\beta {:}A_2\langle c_2 , e_2 \rangle\,\,\m{iff}\\
 \,\, \mc{D}_1\in \m{Tree}({\Delta_1} \Vdash x_\alpha {:}A_1\langle c_1 , e_1 \rangle) \,\m{and}\, \mc{D}_2\in \m{Tree}({\Delta_2} \Vdash y_\beta {:}A_2\langle c_2 , e_2 \rangle)\, \m{and}
  \,\,\\\,\, \Delta= \Delta_1 \Downarrow^{\m{cf}} \xi= \Delta_2  \Downarrow^{\m{cf}} \xi\,\,\,\m{and}\,\,\, K=y_\beta {:}A_2\langle c_2 , e_2 \rangle \Downarrow^{\m{cf}} \xi =  x_\alpha {:}A_1\langle c_1 , e_1 \rangle \Downarrow^{\m{cf}} \xi\,\, \,\m{and}\,\\ 
  \,\, \Delta'= \Delta_1 \Downarrow^{\m{ig}} \xi= \Delta_2 \Downarrow^{\m{ig}} \xi\,\,\,\m{and}\,\,\, K'=y_\beta {:}A_2\langle c_2 , e_2 \rangle \Downarrow^{\m{ig}} \xi =  x_\alpha {:}A_1\langle c_1 , e_1 \rangle \Downarrow^{\m{ig}} \xi\,\, \,\m{and}\,\\
  \, \, \forall \mc{B}_1 \in \mb{H\text{-}CProvider}^\xi(\Delta_1), \mc{B}_2\in \mb{H\text{-}CProvider}^\xi(\Delta_2).\\
\,\,\forall \mc{T}_1 \in \mb{H\text{-}CClient}^\xi(x_\alpha{:}A_1\langle c_1 , e_1 \rangle),\, \mc{T}_2 \in \mb{H\text{-}CClient}^\xi(y_\beta{:}A_2\langle c_2 , e_2 \rangle).\\
\,\,\m{if} (\mc{B}_1, \mc{B}_2)\in \mb{Hrel\text{-}IProvider}^\xi(\Delta'\backslash \Delta)\, \m{and}\,(\mc{T}_1, \mc{T}_2)\in \mb{Hrel\text{-}IClient}^\xi(K' \backslash K)\, \m{then}\,\mc{B}_1\mc{D}_1\mc{T}_1 \approx_a \mc{B}_2\mc{D}_2\mc{T}_2.
\end{array}
\]
\[\begin{array}{lll}
      \cdot \,\,\in \mb{H\text{-}CProvider}^\xi(\cdot) \\
      \mc{B} \in \mb{H\text{-}CProvider}^{\xi}(\Delta, x_\alpha{:}A\langle c, e \rangle )   \, \,\;\m{iff}\\ 
      \quad c \not \sqsubseteq \xi  \, \m{and}\,
      \mc{B} = \mc{B}' \mc{T} \,\, \m{and}\,\, \mc{B}' \in \mb{H\text{-}CProvider}^\xi(\Delta)\,\,\m{and}\,\, \mc{T} \in \m{Tree}(\cdot \Vdash x_\alpha{:}A\langle c, e \rangle), \mb{or}\\
      \quad c \sqsubseteq \xi  \, \m{and}\,
     \mc{B} \in \mb{H\text{-}CProvider}^\xi(\Delta)\\[5pt]
      \mc{T} \in \mb{H\text{-}CClient}^\xi( x_\alpha{:}A\langle c, e \rangle )\,  \, \,\; \m{iff} \\
      \quad c \not \sqsubseteq \xi  \, \m{and}\, \mc{T} \in \m{Tree}(x_\alpha{:}A\langle c, e \rangle \Vdash \_:1\langle \top, \top\rangle), \mb{or}\\
      \quad c \sqsubseteq \xi  \, \m{and}\, \mc{T}= \cdot\\
    \end{array}\]
    \[\begin{array}{lll}
      (\cdot, \cdot) \,\,\in \mb{Hrel\text{-}IProvider}^\xi(\cdot) \\
      (\mc{B}_1, \mc{B}_2) \in \mb{Hrel\text{-}IProvider}^{\xi}(\Delta, x_\alpha{:}A\langle c, e \rangle )  \;\;\, \m{iff} \\
      \quad e \not \sqsubseteq \xi  \, \m{and}\,
      \mc{B}_i = \mc{B}'_i \mc{T}_i \,\, \m{and}\,\, (\mc{B}'_1,\mc{B}'_2) \in \mb{Hrel\text{-}IProvider}^\xi(\Delta), \mb{or}\\
      \quad e \sqsubseteq \xi  \, \m{and} \mc{B}_i = \mc{B}'_i \mc{T}_i \,\, \m{and}\,\, (\mc{B}'_1,\mc{B}'_2) \in \mb{Hrel\text{-}IProvider}^\xi(\Delta)\,\,\m{and}\,\, \cdot \Vdash\mc{T}_1 \equiv_\xi^{\Psi_0} \mc{T}_2:: x_\alpha{:}A\langle c, e \rangle.\\[5pt]
      (\cdot, \cdot) \in \mb{Hrel\text{-}IClient}^\xi( \_\langle \top, \top \rangle )\\
      (\mc{T}_1, \mc{T}_2) \in \mb{Hrel\text{-}IClient}^\xi( x_\alpha{:}A\langle c, e \rangle )\;\;\, \m{iff} \\
      \quad e \not \sqsubseteq \xi\,\, \mb{or}\,\,  e \sqsubseteq \xi  \, \m{and}\, x_\alpha{:}A\langle c, e \rangle  \Vdash \mc{T}_1 \equiv_{\xi}^{\Psi_0} \mc{T}_2 :: \_:1\langle \top, \top\rangle
    \end{array}\]

}
   \vspace{-10pt}
   \caption{Asynchronous bisimulation up to observations of confidentiality $\xi$.}
   \vspace{-12pt}
    \label{fig:bisim}
    \end{figure*}

Now we are ready to present our adequacy theorem stating that, given an observer level $\xi$, logically equivalent configurations are bisimilar up to observations of confidentiality $\xi$.
The proof of the theorem relies on a compositionality lemma, which ensures a harmony between asserts and assumes in the value-interpretation of the logical relation.
\begin{lemma}[Compositionality]\label{lem:compose}
 {\small \(
  \forall m.\, (\mc{D}_1; \mc{D}_2) \in \mc{E}_{\Psi_0}^\xi\llbracket \Delta, u_\alpha^{\langle c,e \rangle}{:}T \Vdash K \rrbracket^m\, \m{and}\,
    \forall m.\, (\mc{T}_1; \mc{T}_2) \in \mc{E}_{\Psi_0}^\xi\llbracket \Delta'\Vdash  u_\alpha^{\langle c,e \rangle}{:}T \rrbracket^m\,\,\)}  if and only if
  {\small\, \(\forall k.\, (\mc{T}_1 \mc{D}_1; \mc{T}_2\mc{D}_2) \in \mc{E}_{\Psi_0}^\xi\llbracket \Delta', \Delta\Vdash K \rrbracket^k.
  \)}
  
  \end{lemma}

\begin{theorem}[Adequacy]\label{thm:adeq}
    If 
    {\small$(\Delta_1 \Vdash \mc{D}_1:: x_\alpha {:}A_1\langle c_1, e_1 \rangle) \equiv^{\Psi_0}_{\xi} (\Delta_2 \Vdash \mc{D}_2:: y_\beta {:}A_2\langle c_2, e_2 \rangle)$} then {\small$(\Delta_1 \Vdash \mc{D}_1:: x_\alpha {:}A_1\langle c_1, e_1 \rangle) \approx_a^{\xi} (\Delta_2 \Vdash \mc{D}_2:: y_\beta {:}A_2\langle c_2, e_2 \rangle)$}.
    
\end{theorem}
\begin{proof}
Recall from~\Cref{fig:bisim} that {\small$\approx_a^{\xi}$} composes {\small$\mc{D}_1$} and {\small$\mc{D}_2$} with arbitrary high-confidentiality ({\small$c \not\sqsubseteq \xi$}) clients and providers, building a confidentiality interface {\small$\llbracket\Delta^c \Vdash K^c\rrbracket$}. 
Let us call the high-confidentiality providers {\small$\mc{B}_1$} and {\small$\mc{B}_2$} and the high-confidentiality clients {\small$\mc{T}_1$} and {\small$\mc{T}_2$}.
We can partition the providers into high-integrity and low-integrity parts to get {\small$\mc{B}_1=\mc{B}^{\mi{HI}}_i\mc{B}^{\mi{LI}}_i$} (similarly for the clients {\small$\mc{T}_1=\mc{T}^{\mi{HI}}_i\mc{T}^{\mi{LI}}_i$}), where superscripts {\small$\mi{HI}$} and {\small$\mi{LI}$} correspond to high-integrity and low-integrtiy parts, resp.
Our goal is to prove {\small$\mc{B}^{HI}_1\mc{B}^{LI}_1\mc{D}_1 \mc{T}^{HI}_1\mc{T}^{LI}_1 \approx_a \mc{B}^{HI}_2\mc{B}^{LI}_2\mc{D}_2 \mc{T}^{HI}_2\mc{T}^{LI}_2$}.

{\bf Step 1.}
The first step is to show that the two compositions are related by the term interpretation as well, i.e., {\small$ \forall m. (\mc{B}^{\mi{HI}}_1\mc{B}^{\mi{LI}}_1\mc{D}_1 \mc{T}^{\mi{HI}}_1\mc{T}^{\mi{LI}}_1; \mc{B}^{\mi{HI}}_2\mc{B}^{\mi{LI}}_2\mc{D}_2 \mc{T}^{\mi{HI}}_2\mc{T}^{\mi{LI}}_2) \in \mathcal{E}\llbracket\Delta^c \Vdash K^c\rrbracket$}.
To do so, we can use the definition from~\Cref{fig:rel_eq} for $\equiv^{\Psi_0}_{\xi}$ to compose $\mc{D}_1$ and $\mc{D}_2$ with given low integrity clients and providers {\small$\mc{T}^{\mi{LI}}_1$, $\mc{T}^{\mi{LI}}_2$, $\mc{B}^{\mi{LI}}_1$}, and {\small$\mc{B}^{\mi{LI}}_2$} to build the integrity interface {\small$\llbracket\Delta^i \Vdash K^i\rrbracket$} and get {\small$ \forall m. (\mc{B}^{\mi{LI}}_1\mc{D}_1 \mc{T}^{\mi{LI}}_1; \mc{B}^{\mi{LI}}_2\mc{D}_2 \mc{T}^{\mi{LI}}_2) \in \mathcal{E}\llbracket\Delta^i \Vdash K^i\rrbracket$}.
However, this is not enough to achieve our goal as the relation pertains to the integrity interface, and thus, the composition only includes the low integrity providers/clients.
To build the confidentiality interface and include the high integrity parts, we use the fact that the high integrity providers  {\small$\mc{B}^{\mi{HI}}_1$ and $\mc{B}^{\mi{HI}}_2$} (and clients {\small$\mc{T}^{\mi{HI}}_1$} and {\small$\mc{T}^{\mi{HI}}_2$}) are themselves logically equivalent. 
We use our compositionality lemma (\Cref{lem:compose}) to compose the high-integrity channels in the integrity interface with these providers/clients and show that the composition results in two logically equivalent configurations, i.e., {\small$ \forall m. (\mc{B}^{\mi{HI}}_1\mc{B}^{\mi{LI}}_1\mc{D}_1 \mc{T}^{\mi{HI}}_1\mc{T}^{\mi{LI}}_1; \mc{B}^{\mi{HI}}_2\mc{B}^{\mi{LI}}_2\mc{D}_2 \mc{T}^{\mi{HI}}_2\mc{T}^{\mi{LI}}_2) \in \mathcal{E}\llbracket\Delta^c \Vdash K^c\rrbracket$}.

{\bf Step 2.} We complete the proof by connecting our logically related configurations to an observational equivalence relation for session types~\cite{BalzerARXIV2023}, which is proved sound and complete for asynchronous bisimulation. 
We first show that our integrity term interpretation implies the observational equivalence relation  in~\cite{BalzerARXIV2023} when we consider a confidentiality interface, and then use their soundness result to show that the integrity term interpretation {\small$ \forall m. (\mc{B}^{\mi{HI}}_1\mc{B}^{\mi{LI}}_1\mc{D}_1 \mc{T}^{\mi{HI}}_1\mc{T}^{\mi{LI}}_1; \mc{B}^{\mi{HI}}_2\mc{B}^{\mi{LI}}_2\mc{D}_2 \mc{T}^{\mi{HI}}_2\mc{T}^{\mi{LI}}_2) \in \mathcal{E}\llbracket\Delta^c \Vdash K^c\rrbracket$} implies bisimilarity {\small$\mc{B}^{HI}_1\mc{B}^{LI}_1\mc{D}_1 \mc{T}^{HI}_1\mc{T}^{LI}_1 \approx_a \mc{B}^{HI}_2\mc{B}^{LI}_2\mc{D}_2 \mc{T}^{HI}_2\mc{T}^{LI}_2$}.

Next, we briefly explain how our integrity logical relation coincides with the observational equivalence relation (for well-typed configurations) in~\cite{BalzerARXIV2023} when considering a confidentiality interface.
The observational equivalence in~\cite{BalzerARXIV2023} is defined via a logical relation similar to the one developed in this paper, but only considering the confidentiality interface. Let us call the term and value interpretations of our logical relation {\small$\mathcal{E}^i$ and $\mathcal{V}^i$} (defined in~\Cref{fig:rel_term,fig:rel_value_left,fig:rel_value_right}) and the ones defined in~\cite{BalzerARXIV2023} {\small$\mathcal{E}^c$} and {\small$\mathcal{V}^c$}, resp.
The relation {\small $\mathcal{E}^i$} is invoked for the integrity interface {\small$\llbracket\Delta^i \Vdash K^i\rrbracket$} in~\Cref{fig:rel_eq}, and similarly {\small$\mathcal{E}^c$} is invoked for the confidentiality interface {\small$\llbracket\Delta^c \Vdash K^c\rrbracket$}, where, by definition, {\small$\Delta^c$} is a subset of (or equal to) {\small$\Delta^i$} and {\small$K^c$} is a subset of (or equal to) {\small$K^i$}.
As the integrity logical relation may contain non-observable channels ({\small$c \not\sqsubseteq \xi$}), it only insists that the same labels are sent when communication is along observable channels.
Concretely, {\small$\mathcal{V}^i$} in lines {\small$(l_2)$}, {\small$(l_3)$}, {\small$(r_2)$}, and {\small$(r_3)$} only insists that the labels {\small$k_1$} and {\small$k_2$} sent/received along a channel are the same if {\small$c \sqsubseteq \xi$}.
However, {\small$\mathcal{V}^c$} always enforces sending the same labels, since a priori the condition {\small$c \sqsubseteq \xi$} holds for all the channels in its interface. 
In all other regards, {\small$\mathcal{E}^i$} and {\small$\mathcal{V}^i$} have the same definition as {\small$\mathcal{E}^c$} and {\small$\mathcal{V}^c$}.
As a result, it is straightforward to observe that given an interface {\small$\Delta^c \Vdash K^c$} with only observable channels (channels with {\small$c\sqsubseteq \xi$}), we have {\small$ \forall m.\,(\mathcal{D}_1;\mathcal{D}_2) \in \mathcal{E}^i\llbracket\Delta^c \Vdash K^c\rrbracket^m$} iff {\small$\forall m.\,(\mathcal{D}_1;\mathcal{D}_2) \in \mathcal{E}^c\llbracket\Delta^c \Vdash K^c\rrbracket^m$}.
\end{proof}

\section{Related work}
\label{sec:related}
{\bf IFC type systems using linearity.}
Conceptually most closely related to our work is the work by Zdancewic and Myers
on \emph{ordered linear continuations} \cite{ZdancewicMyersESOP2001,ZdancewicMyersARTICLE2002}.
The authors consider continuation-passing style (CPS) security-typed languages
to verify noninterference not only for source-level programs but also compiled code.
The authors observe that the possibility to lower the \textsf{pc} label upon exiting control flow constructs,
present in imperative source-level languages,
is no longer available in a CPS language.
To rectify the loss of flexibility they introduce ordered linear continuations.
Similar to our pattern checks, ordered linear continuations allow downgrading of the \textsf{pc} label after branching on high,
because linearity enforces the continuations to be used in every branch, in the order prescribed.
In contrast to our work, the authors only consider a sequential language and only prove PINI.
Our work moreover establishes the connection to integrity,
facilitating regrading policies that are polymorphic in the security lattice for ultimate flexibility.

In another line of work Zdancewic and Myers again employ linearity
for increased flexibility and a stronger noninterference statement \cite{ZdancewicMyersCSFW2003}.
The authors consider a concurrent language with a store and first-class channels.
Their main focus is observational determinism, ensuring immunity to internal timing attacks and
attacks that exploit information about thread scheduling.
To this end the authors introduce linear channels and a race freedom analysis.
Given that $\lang$ enjoys confluence, like other linear session type languages,
it rules out timing attacks that exploit the relative order of messages,
which seems to be a stronger property than immunity to internal timing attacks considered by the authors.
Moreover, we establish PSNI for $\lang$, whereas the authors only prove PINI.

{\bf IFC session type systems.}
In terms of underlying language, the work most closely related to ours is the one by Derakhshan et al. \cite{DerakhshanLICS2021,BalzerARXIV2023}.
The authors develop an IFC type system for the same family of linear session types but only consider confidentiality.
Their system annotates the process term judgments with running and max confidentiality. 
Their typing rules only ensure that the running confidentiality (aka taint level) is updated correctly after each receive and that a tainted process does not leak information via send.
In particular, the rules do not allow decreasing the taint level at any point.
As a result the authors' type system suffers from the same restrictiveness as other IFC type systems for concurrent languages,
requiring each loop iteration to run at the maximal confidentiality reached throughout an arbitrary iteration.
For example, the authors' IFC type system rejects the banking example in~\Cref{sec:bank-examples}: as soon as the bank receives a message from one customer, say Alice, it will be tainted and cannot send a message to any other customer, say Bob.
In fact, the authors' IFC type system rejects all well-typed examples presented in this paper even though they enjoy PSNI. 
We make our IFC type system more 
flexible 
 by designing synchronization policies to enable regrading of the taint level and using integrity labels to make the policies composable, both of which are novel to our system. 
Designing these composable policies was an intricate task, particularly due to dealing with both concurrency and general recursion.

Our logical relation for integrity is inspired by Balzer et al.'s~\cite{BalzerARXIV2023} logical relation for equivalence.
The logical relation for equivalence is defined based on the confidentiality interface. Our logical relation, however, is based on the larger integrity interface to enable the proof of the fundamental theorem.
We prove our adequacy theorem by proving compositionality for our logical relation,
which then allows us to recast our logical relation in terms of the logical relation for equivalence by the authors,
delivering adequacy as a corollary.

{\bf IFC type systems for multiparty session types and process calculi.}
IFC type systems have also been explored for multiparty session types
\cite{CapecchiCONCUR2010,CapecchiARTICLE2014,CastellaniARTICLE2016,Ciancaglini2016}.
These works explore declassification \cite{CapecchiCONCUR2010,CapecchiARTICLE2014} and
flexible runtime monitoring techniques \cite{CastellaniARTICLE2016,Ciancaglini2016}.
Our work not only differs in use of session type paradigm (\ie binary vs. multiparty)
but also in use of a logical relation for showing noninterference.
Our work is more distantly related with IFC type systems for process calculi
\cite{HondaESOP2000,HondaYoshidaPOPL2002,CrafaARTICLE2002,CrafaTGC2005,CrafaARTICLE2007,HennessyARTICLE2005,HennessyARTICLE2005,KOBAYASHI2005,ZdancewicMyersCSFW2003,PottierCSFW2002}.
These works prevent information leakage by associating a security label with channels/types/actions \cite{HondaYoshidaPOPL2002}, read/write policies with channels \cite{HennessyARTICLE2005,HennessyARTICLE2005}, or capabilities with expressions \cite{CrafaARTICLE2002}.
Honda et al. \cite{HondaYoshidaPOPL2002} also use a substructural type system and
prove a sound embedding of Dependency Core Calculus (DCC) \cite{AbadiPOPL1999} into their calculus.
Our work sets itself apart in its use of session types and meta theoretic developments based on logical relations.
{Moreover, our IFC type system 
is more permissive as it allows for regrading of the taint level, while preserving noninterference.}

{\bf Declassification.}
Our notion of regrading may seem related to declassification,
which has extensively been studied for IFC type systems for functional and imperative languages
\cite{FerrariSP1997,AbadiTACS1997,MyersLiskovTOSEM2000,ZdancewicMyersCSFW2001,ZdancewicARTICLE2002,ZdancewicMFPS2003,MyersARTICLE2006,ChongMyersCSFW2006,AskarovMyersARTICLE2011}
and allows an entity to downgrade its level of confidentiality.
However, our work significantly differs from declassification as it preserves PSNI,
whereas declassification systems \emph{deliberately} release information and thus intentionally \emph{weaken} noninterference.

In particular, robust declassification \cite{ZdancewicMyersCSFW2001,ZdancewicARTICLE2002,ZdancewicMFPS2003,MyersARTICLE2006,ChongMyersCSFW2006,AskarovMyersARTICLE2011} prevents adversaries from exploiting downgrading of confidentiality, by complementing confidentiality with integrity.
It uses integrity to ensure that downgrading decisions can be trusted, \ie cannot be influenced by an attacker.
As such, only high-integrity data can influence the taint level to be lowered. 
This is similar to our system, where the higher the integrity of a process, the lower level it can regrade the taint level.
The difference, however, is that we enforce extra synchronization policies on our high-integrity processes to ensure that they cannot induce information leaks by lowering the taint level. 
This contrasts with work on robust declassification, which introduces leakage intentionally and thus compromises noninterference.

\bibliography{main-arxiv}

\section*{Appendix}
\section{Abstract syntax}\label{apx:abstract-syntax}

\Cref{apx:fig:abstract-syntax} defines the abstract syntax of $\lang$.
Lines without a left-hand side are separated by $\alt$ from their preceding line.

\begin{figure*}
\begin{center}
\renewcommand{\tabcolsep}{2mm}
\begin{tabular}{@{}llll@{}}
\\
\textbf{Sort} &
\phantom{$\defined$} &
\textbf{Abstract Form} &
\textbf{Remarks} \\
%%%
Metavariable & $\defined$
& $\Psi_0$ & concrete security lattice $\langle \mathcal{L}, \mc{E}_0, \sqcup, \sqcap \rangle$ \\
%%%
& & $\iota, \eta \in \mathcal{L}$ & set of concrete security levels \\
%%%
& & $\xi \in \mathcal{L}$ & concrete security level of observer \\
%%%
& & $E_0 \in \mathcal{E}_0$ & set of relations $E_0$ of form $\iota \sqsubseteq \iota'$ \\
%%%
& & $\Psi$ & security theory $\langle \mathcal{V}, \mathcal{E}, \sqcup, \sqcap \rangle$ \\
%%%
& & $\psi, \omega \in \mathcal{V}$ & set of security variables \\ 
%%%
& & $c, d, e, f \in \mathcal{S}$ & security terms of $\Psi$ \\
%%%
& & $p \in \mathcal{S} \times \mathcal{S}$ & pair of security terms of $\Psi$ \\
%%%
& & $E \in \mathcal{E}$ & set of relations $E$ of form $c \sqsubseteq d$ \\
%%%
& & $x_\alpha, x_\beta, x_\gamma, x_\delta$ & channel \\
%%%
& & $x, y, z, u, v, w$ & channel varianle\\
%%%
& & $j, k, \ell \in I, L$ & set of labels \\
%%%
& & $\Delta, \Lambda$ & linear typing contexts(channels) \\
%%%
& & $\Omega$ & linear typing contexts (variables)\\
% %%%
% & & $\Gamma$ & linear security typing contexts (channels)\\
% %%%
% & & $\Xi$ & linear security typing contexts(variables) \\
%%%
& & $K$ & linear typing context singleton(channel) \\
% %%%
% & & $K^s$ & linear security typing context singleton(channel) \\
%%%
% & & $\gamma_{\m{sec}}, \hat{\gamma_{\m{sec}}}, \delta_{\m{sec}}, \hat{\delta_{\m{sec}}}$ & order-preserving substitution of security elements \\
%%%
%& & $\mc{B}$ & closed process configuration \\
%%%
& & $\gamma, \hat{\gamma}, \delta, \hat{\delta}$ & channel var./ (order-pres.) label subst. \\
%%%
& & $ \mc{A}, \mc{B}, \mc{C}, \mc{D}, \mc{T}$ &  process configuration in $\lang$\\
% %%%
% % & & $ \mathbb{A}, \mathbb{B}, \mathbb{C}, \mathbb{D},  \mathbb{T}$ &  process configuration in $$\\
% & & $\gamma, \hat{\gamma}, \delta, \hat{\delta}$ & order-preserving substitution \\
% %%%
% & & $\mc{B}$ & closed process configuration \\
% %%%
% & & $\mc{C}, \mc{A}$ & open process configuration \\
%%%
% & & $\mc{P}(\mc{C})$ & a valid permutation of $\mc{C}$ \\
% %%%
% & & $\mnode, \mtree$ & program under consideration \\
% %%%
% & & $\substc, \substf$ & bottom and top closing contexts \\
%%%
Type $A, B, C, T$ & $\defined$
& $\intchoice{A}$ & internal choice, at least one label \\
%%%
& & $\extchoice{A}$ & external choice, at least one label \\
%%%
& & $\chanout{A}{B}$ & channel output \\
%%%
& & $\chanin{A}{B}$ & channel input \\
%%%
& & $\one$ & termination \\
%%%
& & $Y$ & type variable \\
%%%
Definition $X, Y$ & $\defined$
& $\procdefb{\D}{X}{P}{\psi_0}{\omega_0}{x}{A}{\psi}{\omega}$ & process definition \\
%%%
& & $Y = A$ & type definition \\
%%%
Process $P, Q$ & $\defined$
& $\labsnd{x}{k}{P}$ & label output \\
%%%
& & $\labrcv{x}{P}$ & label input \\
%%%
& & $\chnsnd{y}{x}{P}$ & channel output \\
%%%
& & $\chnrcv{y}{x}{P}$ & channel input \\
%%%
& & $\cls{x}$ & terminate process \\
%%%
& & $\wait{x}{Q}$ & wait for process to terminate \\
%%%
& & $\spwn{x}{c}{e}{X}{\D}{c_0}{e_0}{Q}$ & spawn \\
%%%
& & $\fwd{x}{y}$ & forward $x$ to $y$ \\
%%%
%%%
Messages $M,N$ & $\defined$
& $x.k$ & label output \\
%%%
& & $\mb{send}\, y\, x$ & channel output \\
%%%
& & $\cls{x}$ & terminate process 
\end{tabular}
\caption{Abstract syntax of $\lang$.}
\label{apx:fig:abstract-syntax}
\end{center}
\end{figure*}

\section{Concrete security lattice and security theories}\label{apx:sec:lattice}

Process configurations and terms are typed relative to a concrete security lattice and a security theory.
A \emph{concrete} lattice $\Psi_0$ is defined globally for an application and consists of concrete security levels $\iota$.
Our running example lattice
\[\mb{guest} \sqsubseteq \mb{alice} \sqsubseteq \mb{bank} \qquad \mb{guest} \sqsubseteq \mb{bob} \sqsubseteq \mb{bank}\]
\noindent is an example of a concrete security lattice.
Polymorphic process definitions make use of a \emph{security theory} $\Psi$, ranging over security variables $\psi$ and
concrete security levels $\iota$ from the given concrete security lattice $\Psi_0$.
At run-time, all security variables occurring in polymorphic processes will be replaced with concrete security levels.

\begin{definition}[Concrete security lattice and security theory]\label{apx:def:lattice}

Let $\Psi_0 \defined \langle \mathcal{L}, \mc{E}_0, \sqcup \rangle$ be a concrete join semi-lattice with a partial order $\mc{E}_0$
over concrete security levels $\iota, \eta \in \mathcal{L}$ such that $E_0 \in \mc{E}_0$ is of the form $\iota \sqsubseteq \iota'$.
\noindent We define a security theory $\Psi \defined \langle \mathcal{V}, \mathcal{E}\rangle$
that augments $\Psi_0$ with security variables $\psi$ and relations $\mc{E}$ over them such that
\begin{itemize}
\item $\psi \in \mathcal{V}$ is a set of security variables

\item $c, d \in \mc{S}$ is a set of security terms of the security theory $\Psi$, defined by the grammar
\[c, d:= c \sqcup d \mid \iota \mid \psi \]
\item $E \in \mathcal{E}$ is a set of relations over the elements $c, d$ of the security theory where $E = c \sqsubseteq d$.

\end{itemize}
For convenience, we define the projection $\mc{E}(\Psi)$ to extract the relations $\mc{E}$ of $\Psi$.
We write
\[\latcstr{E}\]
if $E$ is consequence of the lattice theory based on the relations.
The judgment is defined in \Cref{apx:fig:latcstr-def}.
\end{definition}

\begin{figure}
\centering
\begin{mathpar}

\inferrule
{c \sqsubseteq d \in \mc{E}_0}
{\latcstr{c \sqsubseteq d}}

\inferrule
{c \sqsubseteq d \in \mc{E}(\Psi)}
{\latcstr{c \sqsubseteq d}}

\inferrule
{\strut}
{\latcstr{c \sqsubseteq c}}

\inferrule
{\latcstr{c_1 \sqsubseteq c_2} \\
\latcstr{c_2 \sqsubseteq c_3}}
{\latcstr{c_1 \sqsubseteq c_3}}

\inferrule
{\latcstr{c \sqsubseteq d'} \\
\latcstr{d \sqsubseteq d'}}
{\latcstr{c \sqcup d \sqsubseteq d'}}

\inferrule
{\strut}
{\latcstr{c \sqsubseteq c \sqcup d}}

\inferrule
{\strut}
{\latcstr{d \sqsubseteq c \sqcup d}}

\end{mathpar}
\caption{Inductive definition of $\latcstr{E}$.}
\label{apx:fig:latcstr-def}
\end{figure}

Process definitions and process spawning rely on an order-preserving substitution, defined as follows:

\begin{definition}[Order-preserving substitution]\label{apx:def:subst}

Let $\gamma \in \mc{V} \rightarrow \mc{S}$ be a total function that maps security variables to security terms.
Its lifting $\hat{\gamma}$ to other syntactic objects, such as security terms $c$, process terms $P$, typing contexts $\Delta$, and security lattices $\Psi$,
is defined by structural induction over the syntactic object, replacing simultaneously all variable occurrences $\psi_i$ in the object with $\gamma(\psi_i)$.
\noindent The function $\gamma$ and its lifting $\hat{\gamma}$ must be order-preserving, ensuring that if $\Psi' \Vdash E$, then $\hat{\gamma}(\Psi') \Vdash \hat{\gamma}(E)$.
To link a spawner and a spawnee, we define an order-preserving substitution $\Psi \Vdash \gamma : \Psi'$
\[\Psi \Vdash \gamma : \Psi' \defined \text{if}\; \Psi' \Vdash E, \text{then}\; \Psi \Vdash \hat{\gamma}(E)\]
\noindent Composition $\gamma \circ \gamma'$ of two substitutions $\gamma$ and $\gamma'$ is defined as usual.
We observe that if both $\gamma$ and $\gamma'$ are order-preserving so is their composition.

\end{definition}

The type system requires every spawner to provide an order-preserving substitution $\Psi \Vdash \gamma : \Psi'$ for the security variables of the spawnee (rule \msc{Spawn}).
As a result, the security theory $\Psi'$ of the spawnee must be satisfied by the security theory $\Psi$ of the spawner,
\ie $\text{if}\; \Psi' \Vdash E, \text{then}\; \Psi \Vdash \hat{\gamma}(E)$.
Moreover, if the spawner provides a substitution $\delta$ for $\Psi_0$, \ie $\Psi_0 \Vdash \delta : \Psi$,
so does the spawnee, \ie $\Psi_0 \Vdash \gamma \circ \delta : \Psi'$.

\section{Regrading policy type system}\label{apx:sec:type-system}

\subsection{Process term typing}\label{apx:sec:term-typing}

The process term typing judgment is of the following form:
\[\ptypb{\D}{P}{c_0}{e_0}{x}{A}{c}{e}\]
where $c_0$ denotes the \emph{running confidentiality}, $e_0$ the \emph{running integrity},
$c$ the \emph{maximal confidentiality}, and $e$ the \emph{minimal integrity}.
The running integrity indicates the level of secrets the process is allowed to regrade to
and thus becomes the process' running confidentiality upon a tail call.
The following presuppositions are imposed on the typing rules:

\begin{enumerate}
\item $\latcstr{e_0 \sqsubseteq c_0 \sqsubseteq c}$ and $\latcstr{e_0 \sqsubseteq e \sqsubseteq c}$,
\item $\forall y{:}B\dlabel{d}{f} \in \D . \latcstr{d \sqsubseteq c}, \latcstr{f \sqsubseteq e}$.
\end{enumerate}

\noindent Moreover, the typing rules for output and input have to conform to the following schema:

\begin{enumerate}

\item \textbf{Before} sending a message, the running confidentiality and running integrity of the sending process must be
  \textbf{at most} the maximal confidentiality and minimal integrity of the receiving process, and

\item \textbf{after} receiving a message, the running confidentiality and running integrity of the receiving process must be
  increased to \textbf{at least} the maximal confidentiality and minimal integrity of the sending process.

\end{enumerate}

\subsubsection{Internal and external choice}

\begin{mathpar}
\inferrule*[right=$\intchoicesymb_{R_1}$]
{\latcstr{c=e} \\
k \in L \\
\ptypb{\D}{P}{c_0}{e_0}{x}{A_k}{c}{e}}
{\ptyp{\D}{\labsnd{x^{\dlabel{c}{e}}}{k}{P}}{c_0}{e_0}{x}{\intchoice{A}}{c}{e}}

\inferrule*[right=$\intchoicesymb_{R_2}$]
{\latcstrnot{c=e} \\
\forall i,j \in L. \,A_i=A_j \\
k \in L \\
\ptypb{\D}{P}{c_0}{e_0}{x}{A_k}{c}{e}}
{\ptyp{\D}{\labsnd{x^{\dlabel{c}{e}}}{k}{P}}{c_0}{e_0}{x}{\intchoice{A}}{c}{e}}

\inferrule*[right=$\intchoicesymb_L$]
{\latcstr{\dlabel{d_1}{f_1} = \dlabel{c}{e} \sqcup \dlabel{d_0}{f_0}} \\
\forall k \in L \\
\ptypb{\D, x{:}A_k\dlabel{c}{e}}{Q_k}{d_1}{f_1}{z}{C}{d}{f} \\
\forall i,j \in L. \barbeq{Q_i}{d_1}{f_1}{Q_j}}
{\ptyp{\D, x{:}\intchoice{A}\dlabel{c}{e}}{\labrcv{x^{\dlabel{c}{e}}}{Q}}{d_0}{f_0}{z}{C}{d}{f}}
\end{mathpar}

\begin{mathpar}
\inferrule*[right=$\extchoicesymb_R$]
{\forall k \in L \\
\ptypb{\D}{P_k}{c}{e}{x}{A_k}{c}{e} \\
\forall i,j \in L. \barbeq{P_i}{c}{e}{P_j}}
{\ptyp{\D}{\labrcv{x^{\dlabel{c}{e}}}{P}}{c_0}{e_0}{x}{\extchoice{A}}{c}{e}}

\inferrule*[right=$\extchoicesymb_{L_1}$]
{\latcstr{c=e} \\
\latcstr{\dlabel{d_0}{f_0} \sqsubseteq \dlabel{c}{e}} \\
k \in L \\
\ptypb{\D, x {:} A_k \dlabel{c}{e}}{Q}{d_0}{f_0}{z}{C}{d}{f}}
{\ptyp{\D, x {:} \extchoice{A} \dlabel{c}{e}}{\labsnd{x^{\dlabel{c}{e}}}{k}{Q}}{d_0}{f_0}{z}{C}{d}{f}}

\inferrule*[right=$\extchoicesymb_{L_2}$]
{\latcstrnot{c=e} \\
\forall i,j \in L. \,A_i=A_j\\
\latcstr{\dlabel{d_0}{f_0} \sqsubseteq \dlabel{c}{e}} \\
k \in L \\
\ptypb{\D, x {:} A_k \dlabel{c}{e}}{Q}{d_0}{f_0}{z}{C}{d}{f}}
{\ptyp{\D, x {:} \extchoice{A} \dlabel{c}{e}}{\labsnd{x^{\dlabel{c}{e}}}{k}{Q}}{d_0}{f_0}{z}{C}{d}{f}}
\end{mathpar}

\subsubsection{Channel output and input}

\begin{mathpar}
\inferrule*[right=$\chanoutsymb_R$]
{\ptypb{\D}{P}{c_0}{e_0}{x}{B}{c}{e}}
{\ptyp{\D, y{:}A\dlabel{d}{f}}{\chnsnd{y}{x^{\dlabel{c}{e}}}{P}}{c_0}{e_0}{x}{\chanout{A}{B}}{c}{e}}

\inferrule*[right=$\chanoutsymb_L$]
{\latcstr{\dlabel{d_1}{f_1} = \dlabel{c}{e} \sqcup \dlabel{d_0}{f_0}} \\
\ptypb{\D, x{:}B\dlabel{c}{e}, y{:}A\dlabel{c}{e}}{Q}{d_1}{f_1}{z}{C}{d}{f}}
{\ptyp{\D, x{:}\chanout{A}{B}\dlabel{c}{e}}{\chnrcv{y^{\dlabel{c}{e}}}{x^{\dlabel{c}{e}}}{Q}}{d_0}{f_0}{z}{C}{d}{f}}

\inferrule*[right=$\chaninsymb_R$]
{\ptypb{\D, y{:}A\dlabel{c}{e}}{P}{c}{e}{x}{B}{c}{e}}
{\ptyp{\D}{\chnrcv{y^{\dlabel{c}{e}}}{x^{\dlabel{c}{e}}}{P}}{c_0}{e_0}{x}{\chanin{A}{B}}{c}{e}}

\inferrule*[right=$\chaninsymb_L$]
{\latcstr{\dlabel{d_0}{f_0} \sqsubseteq \dlabel{c}{e}} \\
\ptypb{\D, x{:} B\dlabel{c}{e}}{Q}{d_0}{f_0}{z}{C}{d}{f}}
{\ptyp{\D, x{:} \chanin{A}{B}\dlabel{c}{e}, y{:}A\dlabel{c}{e}}{\chnsnd{y}{x^{\dlabel{c}{e}}}{Q}}{d_0}{f_0}{z}{C}{d}{f}}
\end{mathpar}

\subsubsection{Identity}

\begin{mathpar}
\inferrule*[right=$\msc{Fwd}$]
{\Psi \Vdash \langle c_1, e_1 \rangle = \langle c_2,e_2 \rangle}
{\ptyp{y{:}A\dlabel{c_1}{e_1}}{\fwd{x^{\dlabel{c_2}{e_2}}}{y^{\dlabel{c_1}{e_1}}}}{c_0}{e_0}{x}{A}{c_2}{e_2}}
\end{mathpar}

\subsubsection{Forwarder process}
\begin{definition}\label{apx:def:fwder}
  For all $Y = A \in \Sigma$, we extend $\Sigma$ by adding the following definition to the signature $\Sigma$:
\[\begin{array}{ll}
  \{\psi_i= \psi_i, \psi_c=\psi_c\}; x':A\langle \psi_c, \psi_i\rangle \vdash_\Sigma F_Y = \fwder{A,{x'}^{\langle \psi_c, \psi_i\rangle}  \leftarrow {y'}^{\langle \psi_c, \psi_i\rangle }} @\langle \psi_c, \psi_i\rangle  :: y'{:}A\langle \psi_c, \psi_i\rangle   & \quad
  Y = A \in \Sigma
\end{array}\]
Here $F_Y$ is a specific process variable assigned to the forwarder process for type variable $Y$, and $\fwder{A, y^{\langle c, e \rangle}\leftarrow x^{\langle c, e \rangle}}$ is defined by induction on the structure of $A$ as a function from type $A$ to process terms as follows:
\begin{mathpar}
\begin{array}{lll}
\fwder{\oplus\{\ell{:}A_\ell\}_{\ell \in L}, y^{\langle c, e \rangle} \leftarrow x^{\langle c, e \rangle}}&:=& \mathbf{case} \, x( \ell \Rightarrow y.\ell; \fwder{A_\ell, y^{\langle c, e \rangle} \leftarrow x^{\langle c, e \rangle}})_{\ell \in L}\\[6pt]
\fwder{ \&\{\ell{:}A_\ell\}_{\ell \in L},  y^{\langle c, e \rangle} \leftarrow x^{\langle c, e \rangle}}&:=&\mathbf{case} \, y( \ell \Rightarrow x.\ell; \fwder{A_\ell, y^{\langle c, e \rangle}\leftarrow x^{\langle c, e \rangle}})_{\ell \in L}\\[6pt]
\fwder{A \otimes B,y^{\langle c, e \rangle}\leftarrow x^{\langle c, e \rangle}}&:=&w \leftarrow \mathbf{recv} x; \mathbf{send} \, w\,y; \fwder{B,y^{\langle c, e \rangle} \leftarrow x^{\langle c, e \rangle}}\\[6pt]
\fwder{A \multimap B,y^{\langle c, e \rangle} \leftarrow x^{\langle c, e \rangle}}&:=& w \leftarrow \mathbf{recv} y; \mathbf{send} \, w\,x; \fwder{B,y^{\langle c, e \rangle} \leftarrow x^{\langle c, e \rangle}}\\[6pt]
\fwder{1,y^{\langle c, e \rangle} \leftarrow x^{\langle c, e \rangle}}&:=& \mathbf{wait}\,x^{\langle c, e \rangle}; \mathbf{close}\, y^{\langle c, e \rangle}\\[6pt]
\fwder{Y; y^{\langle c, e \rangle} \leftarrow x^{\langle c, e \rangle}}&:=&F_Y[c\mapsto \psi_c, e \mapsto \psi_i][x^{\dlabel{c}{e}},y^{\dlabel{c}{e}}] \qquad Y =A \in \Sigma\\[6pt]
\end{array}
\end{mathpar}
\end{definition}

\begin{mathpar}

  \inferrule*[right=$\msc{D-Fwd}$]
  { \{\psi_i= \psi_i, \psi_c=\psi_c\}; z:C\langle \psi_c, \psi_i\rangle \vdash_\Sigma F_Y = \fwder{C,w^{\langle \psi_c, \psi_i\rangle}  \leftarrow z^{\langle \psi_c, \psi_i\rangle }} @\langle \psi_c, \psi_i\rangle  :: w{:}C\langle \psi_c, \psi_i\rangle \in \Sig\\ 
  \gamma(z:C\langle \psi_i, \psi_c\rangle, w:C\langle \psi_i, \psi_c\rangle)= z:C\langle c, e\rangle, y:C\langle c, e\rangle\\
  Y=A \in \Sig\\
  \substmap{\Psi}{\gamma}{\psi_i= \psi_i, \psi_c=\psi_c}}
  {\Psi; z{:}C\langle{c,e}\rangle \vdash_{\Sigma} F_Y[\gamma][x^{\dlabel{c_1}{e_1}},y^{\dlabel{c_1}{e_1}}] @ \langle c, e \rangle :: y{:}C\langle{c,e}\rangle}
  \end{mathpar}

\begin{lemma}
 Given the extended signature $\Sigma$, for all type $A$, there is a derivation for $ x:A\langle{c,e}\rangle \vdash_\Sigma  \fwder{A, y^{\langle{c,e}\rangle} \leftarrow x^{\langle{c,e}\rangle}}:: y: A\langle{c,e}\rangle$.
\end{lemma}
\begin{proof}
The proof is by induction on the structure of type $A$. In a base case, where $A$ is a type variable $Y$, i.e., $A=Y$ for $Y = C \in \Sigma$, the proof is straightforward by the $\msc{D-Fwd}$ rule since  $\{\psi_i= \psi_i, \psi_c=\psi_c\}; z_1:C\langle \psi_c, \psi_i\rangle \vdash_\Sigma F_Y = \fwder{C,x^{\langle \psi_c, \psi_i\rangle}  \leftarrow z^{\langle \psi_c, \psi_i\rangle }_1} @\langle \psi_c, \psi_i\rangle  :: x_1{:}C\langle \psi_c, \psi_i\rangle \in \Sig$.
The proof of other cases is straightforward.

\end{proof}

\subsubsection{Spawn}
\begin{mathpar}
\inferrule*[right=$\msc{Spawn}$]
{\procdef{\D_1'}{X}{P}{\psi_0}{\omega_0}{x}{A}{\psi}{\omega} \in \Sig \\
\substmap{\Psi}{\gamma}{\Psi'} \\
\substmapap{\gamma}{\D_1'} = \D_1 \\
\latcstr{\dlabel{\substmapap{\gamma}{\psi}}{\substmapap{\gamma}{\omega}} \sqsubseteq \dlabel{d}{f}} \\
\latcstr{f_0 \sqsubseteq \substmapap{\gamma}{\psi_0}} \\
\latcstr{f_0 \sqsubseteq \substmapap{\gamma}{\omega_0}} \\
\ptypb{\D_2, x{:}A\dlabel{\substmapap{\gamma}{\psi}}{\substmapap{\gamma}{\omega}}}{Q}{d_0}{f_0}{z}{C}{d}{f}}
{\ptyp{\D_1, \D_2}{\spwn{x}{\substmapap{\gamma}{\psi}}{\substmapap{\gamma}{\omega}}{X}{\D_1}{\substmapap{\gamma}{\psi_0}}{\substmapap{\gamma}{\omega_0}}{Q}}{d_0}{f_0}{z}{C}{d}{f}}

% \inferrule*[right=$\msc{Spawn}$]
% {\procdef{\D_1'}{X}{P}{\psi_0}{\omega_0}{x}{A}{\psi}{\omega} \in \Sig \\
% \substmap{\Psi}{\gamma}{\Psi'} \\
% \substmapap{\gamma}{\D_1'; \psi_0; \omega_0; \psi; \omega} = \D_1; c_0; e_0; c; e \\
% \latcstr{\dlabel{c}{e} \sqsubseteq \dlabel{d}{f}} \\
% \latcstr{f_0 \sqsubseteq c_0} \\
% \latcstr{f_0 \sqsubseteq e_0} \\
% \ptypb{\D_2, x{:}A\dlabel{c}{e}}{Q}{d_0}{f_0}{z}{C}{d}{f}}
% {\ptyp{\D_1, \D_2}{\spwntest{x}{X}{\D_1}{\gamma}{Q}}{d_0}{f_0}{z}{C}{d}{f}}
% %  \spwn{x}{c}{e}{X}{\D_1}{c_0}{e_0}{Q}}{d_0}{f_0}{z}{C}{d}{f}}
\end{mathpar}

Rule $\msc{Spawn}$ relies on an order-preserving substitution $\substmap{\Psi}{\gamma}{\Psi'}$,
guaranteeing that the security terms provided by the spawner comply with the order expected among those terms by the spawnee.
Rule $\msc{Spawn}$ moreover establishes the above invariants for the newly spawned process,
$\latcstr{\dlabel{\substmapap{\gamma}{\psi}}{\substmapap{\gamma}{\omega}} \sqsubseteq \dlabel{d}{f}}$, and
allows the newly spawned process to start at least at the spawner's running integrity $f_0$
for running confidentiality and integrity, facilitating regrading to $f_0$ in case of a tail call.
Note that $\latcstr{f_0 \sqsubseteq \substmapap{\gamma}{\omega_0}}$ prevents the spawnee from employing more patterns checks
than the spawner because the spawnee would otherwise be affected by the spawners negligence.

\subsubsection{Termination}

\begin{mathpar}
\inferrule*[right=$\one_R$]
{\strut}
{\ptyp{\cdot}{\cls{x^{\dlabel{c}{e}}}}{c_0}{e_0}{x}{\one}{c}{e}}

\inferrule*[right=$\one_L$]
{\latcstr{\dlabel{d_1}{f_1} = \dlabel{c}{e} \sqcup \dlabel{d_0}{f_0}} \\
\ptypb{\D}{Q}{d_1}{f_1}{z}{C}{d}{f}}
{\ptyp{\D, x{:}\one\dlabel{c}{e}}{\wait{x^{\dlabel{c}{e}}}{Q}}{d_0}{f_0}{z}{C}{d}{f}}
\end{mathpar}

\subsubsection{Silent unfolding}

We adopt an \emph{equi-recursive}~\cite{CraryPLDI1999} interpretation of recursive session type definitions, avoiding explicit (un)fold messages.
The below rules allow silent unfolding of a type definition.

\begin{mathpar}
\inferrule*[right=$\msc{TVar}_R$]
{\typdef{Y}{A} \in \Sig \\
\ptypb{\D}{P}{c_0}{e_0}{x}{A}{c}{e}}
{\ptypb{\D}{P}{c_0}{e_0}{x}{Y}{c}{e}}

\inferrule*[right=$\msc{TVar}_L$]
{\typdef{Y}{A} \in \Sig \\
\ptyp{\D, x{:}A\dlabel{c}{e}}{Q}{d_0}{f_0}{z}{C}{d}{f}}
{\ptyp{\D, x{:}Y\dlabel{c}{e}}{Q}{d_0}{f_0}{z}{C}{d}{f}}
\end{mathpar}

\subsubsection{Signature checking}

The signature $\Sig$ collects process definitions (rule $\Sig_3$) and possibly recursive session type definitions (rule $\Sig_2$).
The judgment $\typwfmd{A}$ denotes a well-formed session type definition, which, if recursive,
must be \emph{equi-recursive}~\cite{CraryPLDI1999} and \emph{contractive}.
Equi-recursiveness avoids explicit (un)fold messages and relates types up to their unfolding.
Contractiveness demands an exchange before recurring.
Signature checking happens relative to a globally fixed security lattice $\Psi_0$ of concrete security levels
and relative to the entire signature $\Sig$.

\begin{mathpar}
\inferrule*[right=$\Sig_1$]
{\strut}
{\sigwfmd{(\cdot)}}

\inferrule*[right=$\Sig_2$]
{\typwfmd{A} \\
\sigwfmd{\Sig'}}
{\sigwfmd{\typdef{Y}{A}, \Sig'}}

\inferrule*[right=$\Sig_3$]
{\forall i \in \{1\dots n\}. \latcstr{\dlabel{\psi_i}{\omega_i} \sqsubseteq \dlabel{\psi}{\omega}}, \latcstr{\omega_i \sqsubseteq \psi_i} \\
\latcstr{\dlabel{\psi_0}{\omega_0} \sqsubseteq \dlabel{\psi}{\omega}} \\
\latcstr{\omega_0 \sqsubseteq \psi_0} \\
\latcstr{\omega \sqsubseteq \psi} \\
\ptypb{y_1 {:} B_1\dlabel{\psi_1}{\omega_1}, \dots, y_n {:} B_n\dlabel{\psi_n}{\omega_n}}{P}{\psi_0}{\omega_0}{x}{A}{\psi}{\omega} \\
\sigwfmd{\Sig'}}
{\sigwfmd{\procdefb{y_1 {:} B_1\dlabel{\psi_1}{\omega_1}, \dots, y_n {:} B_n\dlabel{\psi_n}{\omega_n}}{X}{P}{\psi_0}{\omega_0}{x}{A}{\psi}{\omega}, \Sig'}}
\end{mathpar}

\subsection{Synchronization patterns}\label{apx:sec:sync-patterns}

The judgment is of the form $\barbeq{P}{d}{f}{Q}$, and the rules are defined inductively below.
The judgment states that process terms $P$ and $Q$ are \emph{synchronized} in terms of their communication pattern,
meaning that if $P$ outputs along channel $x$, so must $Q$, and that if $P$ inputs along channel $x$, so must $Q$, and vice versa,
and is pairwise called for all branches of a case statement.
The check is \emph{conditioned} on
the confidentiality $d$ and running integrity $f$ of the recipient after branching.
The idea is to rule out leakage of the newly learned secret of level $d$ by the presence of
different communication patterns in the recipient's continuation
as long as the recipient may actually regrade to the level of the secret.
The confidentiality of the learned secret $d$ remains constant under inductive invocations,
the running integrity of the recipient $f$ is updated to the recipient's running integrity after a receive.

The inductive definition below distinguishes the following cases:

\begin{itemize}

\item If $\latcstr{d \sqsubseteq f}$, then $\barbeq{P}{d}{f}{Q}$ are trivially related
because the learned secret is not greater than or equal to the level to which the recipient regrades.
This case is a base case of the inductive definition.

\item If $\latcstrnot{d \sqsubseteq f}$, the rules consider whether the next action is a close, send (except close), or receive:

\begin{itemize}

\item In case of a close, the rules insist that they are synchronized.
This is the other base case of the inductive definition.

\item In case of a receive, the receives are checked to be synchronized and the pattern check is invoked inductively on the continuation,
with updated running integrity ($f \sqcup e$).

\item In case of a send, the sends must only by synchronized if
the carrier channel's minimal integrity $e$ is not greater than or equal to the level $d$ of the secret ($d \not\sqsubseteq e$).
In either case, the pattern check is invoked inductively on the continuation, with unchanged running integrity.

\end{itemize}

\end{itemize}

For convenience, we define a process term and a message term with an output prefix as follows:
\[\outact{x}{c}{e}{P} \defined
\labsnd{x^{\dlabel{c}{e}}}{k}{P} \alt
\chnsnd{y}{x^{\dlabel{c}{e}}}{P} \]

\[\outact{x}{c}{e}{\_} \defined
x^{\dlabel{c}{e}}.k \alt
\mathbf{send}\,y\,x^{\dlabel{c}{e}} \]

\subsubsection{Unsynchronized rules}

\begin{mathpar}
\inferrule[$\msc{Brb-Unsync}_1$]
{\latcstrnot{d \sqsubseteq f} \\
\latcstr{d \sqsubseteq e} \\
\barbeq{P}{d}{f}{Q}}
{\barbeq{\outact{x}{c}{e}{P}}{d}{f}{Q}}

\inferrule[$\msc{Brb-Unsync}_2$]
{\latcstrnot{d \sqsubseteq f} \\
\latcstr{d \sqsubseteq e} \\
\barbeq{P}{d}{f}{Q}}
{\barbeq{P}{d}{f}{\outact{x}{c}{e}{Q}}}

\inferrule*[right=$\msc{Brb-Unsync}_3$]
{\latcstr{d \sqsubseteq f}}
{\barbeq{P}{d}{f}{Q}}

\inferrule[$\msc{Brb-Unsync-Spawn}_1$]
{\latcstrnot{d \sqsubseteq f} \\
\latcstr{d \sqsubseteq e_0} \\
\forall y{:}B\dlabel{c'}{e'} \in \D . \latcstr{d \sqsubseteq e'}\\
\barbeq{P}{d}{f}{Q}}
{\barbeq{\spwn{x}{c}{e}{X}{\D}{c_0}{e_0}{P}}{d}{f}{Q}}

\inferrule[$\msc{Brb-Unsync-Spawn}_2$]
{\latcstrnot{d \sqsubseteq f} \\
\latcstr{d \sqsubseteq e_0} \\
\forall y{:}B\dlabel{c'}{e'} \in \D . \latcstr{d \sqsubseteq e'}\\
\barbeq{P}{d}{f}{Q}}
{\barbeq{P}{d}{f}{\spwn{x}{c}{e}{X}{\D}{c_0}{e_0}{Q}}}
\end{mathpar}

%We need to ask the programmer to provide the substitution for spawn. This is true even in our non-downgrading setting? We can say the order in the signature is important.

\subsubsection{Internal and external choice}

\begin{mathpar}
% \inferrule*[right=$\msc{Brb-SndLab}$]
\inferrule[$\msc{Brb-SndLab}$]
{\latcstrnot{d \sqsubseteq f} \\
\latcstrnot{d \sqsubseteq e} \\
\barbeq{P}{d}{f}{Q}}
{\barbeq{\labsnd{x^{\dlabel{c}{e}}}{k}{P}}{d}{f}{\labsnd{x^{\dlabel{c}{e}}}{\ell}{Q}}}

% \inferrule*[right=$\msc{Brb-RcvLab}$]
\inferrule[$\msc{Brb-RcvLab}$]
{\latcstrnot{d \sqsubseteq f} \\
\forall j \in I, k \in L. \barbeq{P_j}{d}{f \sqcup e}{Q_k}}
{\barbeq{\labrcvb{x^{\dlabel{c}{e}}}{P}}{d}{f}{\labrcv{x^{\dlabel{c}{e}}}{Q}}}
\end{mathpar}

\subsubsection{Channel output and input}

\begin{mathpar}
\inferrule*[right=$\msc{Brb-SndChn}$]
{\latcstrnot{d \sqsubseteq f} \\
\latcstrnot{d \sqsubseteq e} \\
\barbeq{P}{d}{f}{Q}}
{\barbeq{\chnsnd{y}{x^{\dlabel{c}{e}}}{P}}{d}{f}{\chnsnd{y}{x^{\dlabel{c}{e}}}{Q}}}

\inferrule*[right=$\msc{Brb-RcvChn}$]
{\latcstrnot{d \sqsubseteq f} \\
\barbeq{\subst{y}{y_1}{P}}{d}{f \sqcup e}{\subst{y}{y_2}{Q}}}
{\barbeq{\chnrcv{y_1}{x^{\dlabel{c}{e}}}{P}}{d}{f}{\chnrcv{y_2}{x^{\dlabel{c}{e}}}{Q}}}
\end{mathpar}

\subsubsection{spawn}

\begin{mathpar}
\inferrule*[right=$\msc{Brb-Spawn}$]
{\latcstrnot{d \sqsubseteq f}\\
\color{mytomato}{(\latcstrnot{d \sqsubseteq e_0} \; \m{or} \; \exists y{:}B\dlabel{c'}{e'} \in \D . \latcstrnot{d \sqsubseteq e'})} \\
\barbeq{\subst{x}{x_1}{P}}{d}{f}{\subst{x}{x_2}{Q}}}
{\barbeq{\spwn{x_1}{c}{e}{X}{\D}{c_0}{e_0}{P}}{d}{f}{\spwn{x_2}{c}{e}{X}{\D}{c_0}{e_0}{Q}}}
\end{mathpar}

\subsubsection{identity}

\begin{mathpar}
\inferrule*[right=$\msc{Brb-Fwd}$]
{\latcstrnot{d \sqsubseteq f}}
{\barbeq{\fwd{x^{\dlabel{c_1}{e_1}}}{y^{\dlabel{c_2}{e_2}}}}{d}{f}{\fwd{x^{\dlabel{c_1}{e_1}}}{y^{\dlabel{c_2}{e_2}}}}}
\end{mathpar}

\subsubsection{Forwarder process}

\begin{mathpar}
  \inferrule*[right=$\msc{Brb-D-Fwd}$]
  {\latcstrnot{d \sqsubseteq f}}
  {\barbeq{ F_Y[\gamma][x^{\dlabel{c_1}{e_1}},y^{\dlabel{c_1}{e_1}}]}{d}{f}{ F_Y[\gamma][x^{\dlabel{c_1}{e_1}},y^{\dlabel{c_1}{e_1}}]}}
  \end{mathpar}

\begin{mathpar}

  \inferrule*[right=$\msc{D-Fwd}$]
  { \psi = \psi; z_1:C[\psi] \vdash_\Sigma F_Y = \fwder{C,x^\psi_1 \leftarrow z^\psi_1} @\psi  :: x_1{:}C[\psi]\in \Sig\\ 
  Y=A \in \Sig\\
  \substmap{\Psi}{\gamma_{\m{sec}}}{\psi= \psi} \\
  z{:}C[c] , y{:}C[c]\Vdash (\gamma_{\m{sec}},\gamma_{\m{var}}):: z_1{:}C[\psi] , x_1{:}C[\psi]}
  {\Psi; z{:}C[c] \vdash_{\Sigma} F_Y[(\gamma_{\m{sec}},\gamma_{\m{var}})] @ c :: y{:}C[c]}
  \end{mathpar}
  
\subsubsection{Termination}

\begin{mathpar}
\inferrule[$\msc{Brb-Close}$]
{\latcstrnot{d \sqsubseteq f}}
{\barbeq{\cls{x^{\dlabel{c}{e}}}}{d}{f}{\cls{x^{\dlabel{c}{e}}}}}

\inferrule[$\msc{Brb-Wait}$]
{\latcstrnot{d \sqsubseteq f} \\
\barbeq{P}{d}{f \sqcup e}{Q}}
{\barbeq{\wait{x^{\dlabel{c}{e}}}{P}}{d}{f}{\wait{x^{\dlabel{c}{e}}}{Q}}}
\end{mathpar}

\subsubsection{Runtime - messages}

\begin{mathpar}
  \inferrule[$\msc{Brb-M-Unsync}_1$]
  {\latcstrnot{d \sqsubseteq f} \\
  \latcstr{d \sqsubseteq e} \\
  \barbeq{\_}{d}{f}{M}}
  {\barbeq{\outact{x}{c}{e}{\_}}{d}{f}{M}}
  
  \inferrule[$\msc{Brb-M-Unsync}_2$]
  {\latcstrnot{d \sqsubseteq f} \\
  \latcstr{d \sqsubseteq e} \\
  \barbeq{M}{d}{f}{\_}}
  {\barbeq{M}{d}{f}{\outact{x}{c}{e}{\_}}}

  \inferrule[$\msc{Brb-M-done}_2$]
{\latcstrnot{d \sqsubseteq f} \\
\latcstr{d \sqsubseteq e}}
{\barbeq{\_}{d}{f}{\_}}

  \inferrule*[right=$\msc{Brb-M-Unsync}_3$]
  {\latcstr{d \sqsubseteq f}}
  {\barbeq{M}{d}{f}{M'}}
  
  \inferrule[$\msc{Brb-M-SndLab}$]
{\latcstrnot{d \sqsubseteq f} \\
\latcstrnot{d \sqsubseteq e} }
{\barbeq{x^{\dlabel{c}{e}}.k}{d}{f}{x^{\dlabel{c}{e}}.\ell}}

\inferrule*[right=$\msc{Brb-M-SndChn}$]
{\latcstrnot{d \sqsubseteq f} \\
\latcstrnot{d \sqsubseteq e} }
{\barbeq{\mathbf{send}\,y\,{x^{\dlabel{c}{e}}}}{d}{f}{\mathbf{send}\,y\,{x^{\dlabel{c}{e}}}}}

\inferrule[$\msc{Brb-M-Close}$]
{\latcstrnot{d \sqsubseteq f}}
{\barbeq{\cls{x^{\dlabel{c}{e}}}}{d}{f}{\cls{x^{\dlabel{c}{e}}}}}

\end{mathpar}

\subsection{Asynchronous dynamics}\label{apx:sec:async_dynamics}

We define an asynchronous dynamics for $\lang$ that we adopt from prior work~\cite{BalzerICFP2017,DasCSF2021,DerakhshanLICS2021} and extend to account for polymorphic process definitions.
The dynamics is given in \Cref{apx:fig:dynamics} and uses multiset rewriting rules \cite{CervesatoARTICLE2009}.
Multiset rewriting rules express the dynamics as state transitions between configurations and are \emph{local} in that they only mention the parts of a configuration they rewrite.

\begin{figure*}
  \begin{center}
  \begin{small}
  \input{figs/dynamics-apx}
  \caption{Asynchronous dynamics of $\lang$.}
  \label{apx:fig:dynamics}
  \end{small}
  \end{center}
  \end{figure*}

\paragraph{Message Typing}
\begin{mathpar}
  \inferrule*[right=$\otimes R$]
  { }
  {\Psi_0 ; z_\beta{:}A\langle c,e \rangle,  y_{\alpha+1}{:} B\langle c,e \rangle\vdash \mathbf{send}\,z^{\langle c,e \rangle}_\beta\,y^{\langle c,e \rangle}_\alpha::  y_{\alpha}{:}A \otimes B \langle c,e \rangle}
  
  \inferrule*[right=$\multimap L$]
  { }
  {\Psi_0 ;z^{\langle c,e \rangle}_\beta{:}A{\langle c,e \rangle}, x^{\langle c,e \rangle}_\alpha{:}A \multimap B{\langle c,e \rangle} \vdash \mathbf{send}\, z^{\langle c,e \rangle}_\beta\,x^{\langle c,e \rangle}_\alpha::  x_{\alpha+1}{:}B{\langle c,e \rangle}}
  
  \inferrule*[right=$\oplus R$]
  { k\in L }
  {\Psi_0 ;y_{\alpha+1}{:} A_k{\langle c,e \rangle} \vdash y^{\langle c,e \rangle}_\alpha.k ::  y_\alpha{:}\oplus\{\ell{:}A_{\ell}\}_{\ell \in L}{\langle c,e \rangle}}
  
  \inferrule*[right=$\& L$]
  { k \in L}
  { \Psi_0 ;x_\alpha{:}\&\{\ell:A_{\ell}\}_{\ell \in I}{\langle c,e \rangle} \vdash (x^{\langle c,e \rangle}_\alpha.k; P):: x_{\alpha+1}{:}A_k{\langle c,e \rangle}}
  
  \inferrule*[right=$1 R$]
  { }
  { \Psi_0 ;\cdot \vdash \mb{close}\, y^{\langle c,e \rangle}_\alpha:: y_{\alpha}{:}1{\langle c,e \rangle}}
  \end{mathpar}

\subsection{Configuration typing}
Similar to~\cite{DerakhshanLICS2021}, we allow configurations of processes to be \emph{open}.
An open configuration can be \emph{plugged} with other process configurations to form a closed program.
Open configurations are typed with the judgment {\small$\Psi_0;\Delta\Vdash \mc{C}:: \Delta'$}, indicating that the configuration {\small$\mc{C}$} provides sessions along the channels in {\small$\Delta'$}, using sessions provided along channels in {\small$\Delta$} and given the concrete security lattice {\small$\Psi_0$}.
{\small$\Delta$} and {\small$\Delta'$} are both linear contexts, consisting of actual run-time channels of the form {\small$y_i{:}B_i\maxsec{\dlabel{d_i'}{e'_i}}$}.
For brevity, we do not indicate a channel's generation.
The typing rules are shown below.

\begin{mathpar}
\inferrule*[right=$\mathbf{emp}$]
{\strut}
{\Psi_0; x_\alpha{:}A[\langle d,e \rangle] \Vdash \cdot :: (x_\alpha{:}A[\langle d,e \rangle])}
\and
%
%\inferrule*[right=$\mathbf{emp}_2$]
%{\strut}
%{\Psi_0; \cdot \Vdash \cdot :: (\cdot)}
%
%\and
%
\inferrule*[right=$\mathbf{proc}$]
{\Psi_0 \Vdash d_1 \sqsubseteq d \\
\Psi_0 \Vdash e_1 \sqsubseteq e \\
\forall y_\beta{:}B[\langle d',e' \rangle] \in \Delta'_0, \Delta \, (\Psi_0 \Vdash d' \sqsubseteq d) \\
\forall y_\beta{:}B[\langle d',e' \rangle] \in \Delta'_0, \Delta \, (\Psi_0 \Vdash e' \sqsubseteq e) \\
\Psi_0 \Vdash e_1 \sqsubseteq d_1 \\
\Psi_0 \Vdash e \sqsubseteq d \\
\forall y_\beta{:}B[\langle d',e' \rangle] \in \Delta'_0, \Delta \, (\Psi_0 \Vdash e' \sqsubseteq d') \\
\Psi_0; \Delta_0 \Vdash \mathcal{C}:: \Delta \\
\Psi_0; \Delta'_0, \Delta \vdash P@\langle d_1,e_1 \rangle:: (x_\alpha{:}A[\langle d,e \rangle])}
{\Psi_0; \Delta_0, \Delta'_0 \Vdash \mathcal{C}, \mathbf{proc}(x_\alpha[\langle d,e \rangle], P@\langle d_1,e_1 \rangle):: (x_\alpha{:}A[\langle d,e \rangle])}
\and
\inferrule*[right=$\mathbf{msg}$]
{\forall y_\beta{:}B[\langle d',e' \rangle] \in \Delta'_0, \Delta \, (\Psi_0 \Vdash d' \sqsubseteq d) \\
\Psi_0 \Vdash e \sqsubseteq d\\
\forall y_\beta{:}B[\langle d',e' \rangle] \in \Delta'_0, \Delta \, (\Psi_0 \Vdash e' \sqsubseteq d') \\
\Psi_0; \Delta_0 \Vdash \mathcal{C}:: \Delta \\
\Psi_0; \Delta'_0, \Delta \vdash M@\langle d,e \rangle:: (x_\alpha{:}A[\langle d,e \rangle])}
{\Psi_0; \Delta_0, \Delta'_0 \Vdash \mathcal{C}, \mathbf{msg}(M):: (x_alpha{:}A[\langle d,e \rangle])}
\and
\inferrule*[right=$\mathbf{comp}$]
{\Psi_0; \Delta_0 \Vdash \mathcal{C}:: \Delta \\
\Psi_0; \Delta'_0 \Vdash \mc{C}_1:: x_\alpha{:}A[\langle d, e\rangle]}
{\Psi_0; \Delta_0, \Delta'_0 \Vdash \mathcal{C}, \mc{C}_1 :: \Delta, x_\alpha{:} A[\langle d, e\rangle]}
\end{mathpar}

Since $\Psi_0$ and $\Sigma$ are fixed, we may drop $\Psi_0$ and $\Sigma$ in a configuration typing judgment and a process typing judgment, respectively, for brevity.

\subsection{Progress and preservation}\label{apx:sec:progress-preservation}
\begin{theorem}[Progress]
  If $\Psi_0; \Delta \Vdash \mc{C}::\Delta'$, then 
  either $\mc{C} \mapsto\mc{C}'$,  or $\mc{C}$ is poised. 
  \end{theorem}
  \begin{proof}
  The proof is by induction on the configuration typing of $\mc{C}$. 
  \end{proof}

  \begin{theorem}[Preservation]
  If $\Psi_0; \Delta \Vdash \mc{C} :: \Delta'$ and $\mc{C}\mapsto \mc{C}'$, then  $\Psi_0; \Delta \Vdash \mc{C}' :: \Delta'$.  %Moreover, if every process in $\mathcal{C}$ preserves the tree invariant, so does every process in $\mathcal{C'}$. (Not sure if the second part is every going to be used.)
  %Moreover $\mc{C}'<\mc{C}$ by the multiset ordering.
  \end{theorem}
  \begin{proof}
  The proof is by considering different cases of $\mc{C}\mapsto \mc{C}'$, and then by inversion on the typing derivations.
  \end{proof}
\section{Proof of the main lemma}
\subsection{Relevant nodes}

We have three different kinds of nodes (\ie processes and messages) in an open configuration {\small$\mnode$}: \emph{relevant} nodes, \emph{semi-relevant} nodes, and \emph{irrelevant} nodes.
Relevant and semi-relevant nodes are processes or messages in {\small$\mnode$} that can directly or indirectly influence what messages are made ready for transmission across the interface {\small${\color{mygreen}\Delta} \Vdash {\color{mygreen}K}$}.
Relevant and semi-relevant nodes thus affect the \emph{observable outcome}.
Irrelevant nodes, on the other hand, do not have any observable outcome.
What distinguishes a relevant node from a semi-relevant node is that a relevant node can directly cause an observable outcome,
whereas a semi-relevant node can only indirectly cause an observable outcome in combination with regrading.
\begin{definition} [Running integrity of nodes]\label{def:running integrity}
The running integrity of a node is dependent on the kind of node:
\begin{enumerate}
\item process: running integrity of process;
\item negative message: min integrity of the message's carrier channel;
\item positive message: join of the min integrity of the message's carrier channel with the running integrity of its parent (if it has any).
\end{enumerate}
\end{definition}

\begin{definition}[Relevant channels and (semi-)relevant nodes]\label{def:relevant-node}
Consider configuration $\Delta \Vdash \mc{D}:: K$ and observer level $\xi$.
\begin{enumerate}
\item A channel $x^{\dlabel{c}{e}} \in \dom{\D \cup \{K\}}$ is relevant if $e \sqsubseteq \xi$.

\item An offering channel $x^{\dlabel{c}{e}}$ of a node in $\mc{D}$ is relevant, if
(a) $e \sqsubseteq \xi$ and if
(b) there exists a node with running integrity $e_0 \sqsubseteq \xi$ and there exists a relevant channel $y$
such that either both $x$ and $y$ are the node's children or only one is a child of the node and the other is the node's offering channel.

\item A relevant node in $\mc{D}$ has (a) a running integrity $e_0 \sqsubseteq \xi$ and
(b) a running confidentiality $c_0 \sqsubseteq \xi$ and
(c) at least one relevant channel as a child or as its offering channel.

\item A semi-relevant node in $\mc{D}$ has (a) a running integrity $e_0 \sqsubseteq \xi$ and
(b) a running confidentiality $c_0 \not\sqsubseteq \xi$ and
(c) at least one relevant channel as a child or as its offering channel.
\end{enumerate}
    
$\cproj{\mc{D}}{\xi}$ are the relevant and semi-relevant processes and messages in $\mc{D}$.
An irrelevant node has either a running integrity $e_0 \not\sqsubseteq \xi$ or does not have a relevant channel as a child or its offering channel.
An irrelevant channel is not relevant.

\end{definition}

In the proof of the fundamental theorem \Cref{thm:ap_main} we must assert that $\mnode_1$ and $\mnode_2$ will generate related outgoing observable messages
when stepped internally.
To make this assertion, we must know that the relevant and semi-relevant nodes in $\mnode_1$ and $\mnode_2$ are in sync,
denoted by $\cproj{$$\mnode_1$$}{\xi}\equiv_{\xi}\cproj{$$\mnode_2$$}{\xi}$,
indicating that $\mnode_1$ and $\mnode_2$ run related program terms.
\Cref{def:sim} defines this notion precisely, and we discuss it below.
Before doing that, however, we must introduce the notion of a dead-end node, on which \Cref{def:sim} relies.

Because our development uses an asynchronous semantics and because semi-relevant nodes are allowed to carry out unsynchronized actions,
not all the relevant and semi-relevant nodes in two program runs $\mnode_1$ and $\mnode_2$ are always in sync.
The notion of a \emph{dead-end} node allows us to identify those relevant and semi-relevant nodes in $\mnode_1$ and $\mnode_2$
that are permitted to be out of sync because they are guaranteed to no longer cause any observable outcomes.

A relevant node becomes a dead-end node when it is waiting to receive a message from an irrelevant node.
Since the sends may not happen simultaneously in two program runs, or not at all,
the receiving relevant nodes may not receive the message simultaneously either.
But once the message is received, the receiving relevant nodes become irrelevant nodes themselves.
This scenario also holds true for semi-relevant nodes.
Semi-relevant nodes, however, additionally may be out of sync in two runs because either of them is carrying out unsynchronized sends
before being ready to receive from the irrelevant sender.

It may be helpful to briefly reflect upon why a relevant or semi-relevant node becomes irrelevant after receiving a message from an irrelevant node.
According to \Cref{def:relevant-node}, an irrelevant node's running integrity is either $e_0 \not\sqsubseteq \xi$ or
does not have any relevant channel.
In the former case, it must hold by the presuppositions for its offering channel that its min integrity is $e \not\sqsubseteq \xi$,
increasing the running integrity of the recipient to at least $e$ as well.
In the latter case, the offering channel of the irrelevant node must be irrelevant too (by virtue of the offering node being irrelevant).
Since it is connected to a relevant or semi-relevant receiver, its min integrity must be $e \not\sqsubseteq \xi$,
thus increasing the running integrity of the recipient to at least $e$ as well.

Below we define the notion of a dead-end node, which relies on process terms with an input prefix:
\[\inact{x}{c}{e}{P} \defined
\labrcv{x^{\dlabel{c}{e}}}{P} \alt
\chnrcv{y}{x^{\dlabel{c}{e}}}{P} \alt
\wait{x^{\dlabel{c}{e}}}{P} \]

\noindent The notion of a dead-end node also relies on the notion of a \emph{path} between nodes,
which means that one node is reachable from the other by traversing edges in the tree.
An edge exists between two nodes if one is an ancestor of the other one.
We call a path \emph{relevant}, if every node and edge along the path is relevant or semi-relevant.
We say that a node is in the interface {\small${\color{mygreen}\Delta} \Vdash {\color{mygreen}K}$},
if there exists a channel $x$ such that either $x \in \Delta$ and $x$ is a resource of the node or
$x \in K$ and $x$ is the node's offering channel.

\begin{definition}[Dead-end nodes]\label{apx:def:dead-end-nodes}
Consider configuration $\Delta \Vdash \mc{D}:: K$ and observer level $\xi$.
\begin{enumerate}
\item A relevant or semi-relevant process $\procb{x}{c}{e}{Q}{d_1}{e_1}$ in $\mc{D}$ is a dead-end node,
if, for some $w$, $d_2$, and $e_2$ such that $e_2 \not\sqsubseteq \xi$ and where $e_1 \sqsubseteq \xi$ by assumption, either
\begin{itemize}
\item $\barbeq{\inact{w}{d_2}{e_2}{P}}{d_s}{e_1}{Q_{\mid \xi}}$, for some $d_s$ and $P$ such that $d_s \not\sqsubseteq \xi$
where $Q_{\mid \xi}$ is equal to $Q$ modulo renaming of channels with min integrity higher than or incomparable to the observer, or
\item $Q = \inact{w}{d_2}{e_2}{Q'}$, for some $Q'$.
\end{itemize}
\item A positive relevant or semi-relevant message $\mathbf{msg}(P)$ in $\mc{D}$ is a dead-end node,
if its parent is a dead-end node.
\item A relevant or semi-relevant node in $\mc{D}$ is a dead-end node, if any relevant path from the node to a node
in the interface {\small${\color{mygreen}\Delta} \Vdash {\color{mygreen}K}$}
includes at least one dead-end node (\ie the path that make the process relevant or semi-relevant includes a dead-end node.) 
\end{enumerate}
\end{definition}

Now with the definition of a dead-end node in place, we can finally define when to program runs $\mnode_1$ and $\mnode_2$ are in sync.

\begin{definition}[Synchronized program runs]\label{def:sim}
We write $\cproj{\mc{D}_1}{\xi}\equiv_{\xi}\cproj{\mc{D}_2}{\xi}$ iff 
\begin{enumerate}
    \item for any relevant node  \[\mathbf{proc}(x[\langle c,e\rangle], P@\langle d_1,e_1 \rangle)\,\, \textit{in}\,\, \mc{D}_1,\] there is a relevant node \[\mathbf{proc}(y[\langle c,e\rangle], Q@\langle d_1,e_1 \rangle)\,\, \textit{in}\,\, \mc{D}_2\] such that $\mathbf{proc}(x[\langle c,e\rangle],P@\langle d_1,e_1 \rangle)= \mathbf{proc}(y[\langle c,e\rangle],Q@\langle d_1,e_1 \rangle)$ upto renaming of high-integrity channels. In particular, we have $Q=P_{\mid \xi}$ and if $e\sqsubseteq \xi$, we have $x=y$:
    \[\Psi; \Delta \vdash P@ \langle d_1,e_1 \rangle :: x[\langle c,e\rangle]: B, \m{and}\, 
    \Psi; \Delta_{\mid \xi} \vdash P_{\mid \xi}@ \langle d_1,e_1 \rangle :: y[\langle c,e\rangle]: B,
    \]
    where $_{\mid \xi}$ stands for a renaming of channels with integrity $e'\not \sqsubseteq \xi$.

    \item for any relevant node  \[\mathbf{msg}(M)\,\, \textit{in}\,\, \mc{D}_1,\] there is a relevant node \[\mathbf{msg}(M)\,\, \textit{in}\,\, \mc{D}_2.\]
    
    \item for any relevant node  \[\mathbf{proc}(x[\langle c,e\rangle], P@\langle d_1,e_1 \rangle)\,\, \textit{in}\,\, \mc{D}_2,\] that is not a dead-end process in $\mc{D}_2$, there is a relevant node \[\mathbf{proc}(y[\langle c,e\rangle], Q@\langle d_1,e_1 \rangle)\,\, \textit{in}\,\, \mc{D}_1\] such that $\mathbf{proc}(x[\langle c,e\rangle],P@\langle d_1,e_1 \rangle)= \mathbf{proc}(y[\langle c,e\rangle],Q@\langle d_1,e_1 \rangle)$ upto renaming of high-integrity channels. In particular, we have $Q=P_{\mid \xi}$ and if $e\sqsubseteq \xi$, we have $x=y$:
    \[\Psi; \Delta \vdash P@ \langle d_1,e_1 \rangle :: x[\langle c,e\rangle]: B, \m{and}\, 
    \Psi; \Delta_{\mid \xi} \vdash P_{\mid \xi}@ \langle d_1,e_1 \rangle :: y[\langle c,e\rangle]: B,
    \]
    where $_{\mid \xi}$ stands for a renaming of channels with integrity $e'\not \sqsubseteq \xi$.

    \item  for any relevant node  \[\mathbf{msg}(M)\,\, \textit{in}\,\, \mc{D}_2,\] which is not a dead-end message in $\mc{D}_2$, there is a relevant node \[\mathbf{msg}(M)\,\, \textit{in}\,\, \mc{D}_1.\]
    
    \item for any semi-relevant node  \[\mathbf{proc}(x[\langle c,e\rangle], P@\langle d_1,e_1 \rangle)\,\, \textit{in}\,\, \mc{D}_1,\] there is a semi-relevant node \[\mathbf{proc}(y[\langle c,e\rangle], Q_{\mid \xi}@\langle d_1,e_1 \rangle)\,\, \textit{in}\,\, \mc{D}_2\] such that $\Psi \Vdash P \sim_{\langle d_s,e_1 \rangle}Q$ for some $d_s \not \sqsubseteq \xi$, and the two nodes agree on the channels with less than or equal integrity to the observer, e.g., if $e \sqsubseteq \xi$, then $x=y$, and
    \[\Psi; \Delta \vdash P@ \langle d_1,e_1 \rangle :: x[\langle c,e\rangle]: B, \m{and}\, 
    \Psi; \Delta'_{\mid \xi} \vdash Q_{\mid \xi}@ \langle d_1,e_1 \rangle :: y[\langle c,e\rangle]: B,
    \]
    where $\Delta' \Downarrow d_s = \Delta \Downarrow d_s$; meaning that any channel in $\Delta$ with integrity $e$ with $d_s \not \sqsubseteq e$, also occurs in $\Delta'$ and vice-versa.
    % stands for a renaming of channels with integrity $e'\not \sqsubseteq \xi$.
     
    \item for any semi-relevant node  \[\mathbf{msg}(M)\,\, \textit{in}\,\, \mc{D}_1,\] which is sending along $\langle d_1,e_1 \rangle$, there is a semi-relevant node \[\mathbf{msg}(M')\,\, \textit{in}\,\, \mc{D}_2,\]
    such that $\Psi \Vdash M \sim_{\langle d_s,e_1 \rangle}M'$ for some $d_s \not \sqsubseteq \xi$, and the two nodes agree on the channels connected to the messages.
    
     \item for any semi-relevant node  \[\mathbf{proc}(x[\langle c,e\rangle], P@\langle d_1,e_1 \rangle)\,\, \textit{in}\,\, \mc{D}_2,\] which is not a dead-end process in $\mc{D}_2$, there is a semi-relevant node \[\mathbf{proc}(y[\langle c,e\rangle], Q_{\mid \xi}@\langle d_1,e_1 \rangle)\,\, \textit{in}\,\, \mc{D}_1\] such that $\Psi \Vdash P \sim_{\langle d_s,e_1 \rangle}Q$ for some $d_s \not \sqsubseteq \xi$, and the two nodes agree on the channels with less than or equal integrity to the observer, e.g., if $e \sqsubseteq \xi$, then $x=y$, and
    \[\Psi; \Delta \vdash P@ \langle d_1,e_1 \rangle :: x[\langle c,e\rangle]: B, \m{and}\, 
    \Psi; \Delta'_{\mid \xi} \vdash Q_{\mid \xi}@ \langle d_1,e_1 \rangle :: y[\langle c,e\rangle]: B,
    \]
    where $\Delta' \Downarrow d_s = \Delta \Downarrow d_s$; meaning that any channel in $\Delta$ with integrity $e$ with $d_s \not \sqsubseteq e$, also occurs in $\Delta'$ and vice-versa.

    \item for any semi-relevant node  \[\mathbf{msg}(M)\,\, \textit{in}\,\, \mc{D}_2,\] which is sending along $x\langle d_1,e_1 \rangle$ and is not a dead-end message in $\mc{D}_2$, there is a semi-relevant node \[\mathbf{msg}(M')\,\, \textit{in}\,\, \mc{D}_1,\]
    such that $\Psi \Vdash M \sim_{\langle d_s,e_1 \rangle}M'$ for some $d_s \not \sqsubseteq \xi$, and the two nodes agree on the channels connected to the messages.
\end{enumerate}
\end{definition}

We say a channel is connected to a node if either the node offers the channel or uses it as a resource.

\subsection{Main lemma}

Having defined when two program runs are in sync, $\cproj{\mc{D}_1}{\xi}\equiv_{\xi}\cproj{\mc{D}_2}{\xi}$,
we state and prove a lemma that asserts that the program runs remain synced along internal transitions.
This lemma is used in the proof of the fundamental theorem \Cref{thm:ap_main} and states,
for two synced program runs $\mc{D}_1$ and $\mc{D}_2$, that if $\mc{D}_1$ takes an internal step to $\mc{D}_1'$,
there exists a way to step $\mc{D}_2$ to $\mc{D}_2'$ such that $\mc{D}_1'$ and $\mc{D}_2'$ remain synced.
The existential indicates that we assume a non-deterministic scheduler.

\begin{lemma}[Catch-up]\label{lem:indinvariant1}
    Consider $\Psi ; \Delta \Vdash \mc{D}_i :: K$ such that $\cproj{\mc{D}_1}{\xi}\equiv_\xi \cproj{\mc{D}_2}{\xi}$. 
    If $\mc{D}_1 \mapsto\mc{D}'_1$, then for some $\mc{D}'_2$, we have $\mc{D}_2 \mapsto^{*}\mc{D}'_2$ such  that $\cproj{\mc{D}'_1}{\xi}\equiv_\xi\cproj{\mc{D}'_2}{\xi}$.
\end{lemma}

\begin{proof}
    The proof is by cases on the possible $\mapsto$ steps.  We consider different cases based on the semantic rule with which $\mc{D}_1$ steps. In each case we prove that either the step does not change relevancy of any process in $\mc{D}_1$ or we can step $\mc{D}_2$ such that the same change of relevancy occurs in it too.

\begin{description}

  \item  {\color{ForestGreen} \bf Case 1. $\mc{D}_1=\mc{D}'_1\mb{proc}(y_\alpha[\langle c,e\rangle],  y^{\langle c,e\rangle}_\alpha.k ;P @\langle d_1, e_1\rangle) \mc{D}''_1$ and
\[ \mc{D}'_1\mb{proc}(y_\alpha[\langle c,e\rangle],  y^{\langle c,e\rangle}_\alpha.k ;P @\langle d_1, e_1\rangle) \mc{D}''_1\mapsto\]\[ \mc{D}'_1\mb{proc}(y_{\alpha+1}[\langle c,e\rangle], [y^{\langle c,e\rangle}_{\alpha+1}/y^{\langle c,e\rangle}_{\alpha}] P@\langle d_1, e_1 \rangle) \mb{msg}(y^{\langle c,e\rangle}_\alpha.k) \mc{D}''_1\]}
We consider subcases based on relevancy of process offering along $y_\alpha[\langle c,e\rangle]$: 

\begin{description}
\item {\color{Aquamarine}\bf Subcase 1. $\mb{proc}(y_\alpha[\langle c,e\rangle],  y^{\langle c,e\rangle}_\alpha.k ;P @\langle d_1,e_1\rangle)$ is  irrelevant.}
By inversion on the typing rules $e_1 \sqsubseteq e$. By definition either $e_1 \not\sqsubseteq \xi $ or none of the channels connected to the process including its offering channel $y_\alpha^{\langle c,e \rangle}$ are relevant. In both cases \[\mb{proc}(y_{\alpha+1}[\langle c,e \rangle], [y^{\langle c,e \rangle}_{\alpha+1}/y^{\langle c,e \rangle}_{\alpha}] P@\langle d_1,e_1 \rangle),\] and  \[\mb{msg}(y^{\langle c,e \rangle}_\alpha.k)\] are irrelevant in the post step.
The reason is: from $e_1 \sqsubseteq e$ and $e_1 \not\sqsubseteq \xi$, we get $e \not\sqsubseteq \xi $. Channel $y^{\langle c,e \rangle}_\alpha$ is not relevant in the pre-step, and both channels $y^{\langle c,e \rangle}_\alpha$ and $y^{\langle c,e \rangle}_{\alpha+1}$ are not relevant in pre-step and post-step configurations. 

Every irrelevant resource of $\mb{proc}(y^{\langle c,e \rangle}_\alpha,  y^{\langle c,e \rangle}_\alpha.k ;P @d_1)$ will remain irrelevant in the post-step too.

In this subcase, our goal is to show 
\[\begin{array}{l}
    \cproj{\mc{D}'_1\mb{proc}(y_{\alpha+1}[\langle c,e \rangle], [y^{\langle c,e \rangle}_{\alpha+1}/y^{\langle c,e \rangle}_{\alpha}] P @\langle d_1,e_1 \rangle) \mb{msg}(y^{\langle c,e \rangle}_\alpha.k) \mc{D}''_1}{\xi} \meq \\
    \cproj{\mc{D}'_1 \mc{D}''_1}{\xi}\meq \cproj{\mc{D}_1}{\xi}\meq \cproj{\mc{D}_2}{\xi}
\end{array}\]

To prove this we need two observations:
\begin{itemize}
    \item  $\mb{proc}(y_{\alpha+1}[\langle c,e \rangle], [y^{\langle c,e \rangle}_{\alpha+1}/y^{\langle c,e \rangle}_{\alpha}] P@\langle d_1,e_1 \rangle)$ and  $\mb{msg}(y^{\langle c,e \rangle}_\alpha.k)$ are irrelevant and they will be dismissed by the projection. (As explained above.)
    \item Replacing $\mb{proc}(y_\alpha[\langle c,e \rangle],  y^{\langle c,e \rangle}_\alpha.k ;P @\langle d_1,e_1 \rangle)$ with these two nodes, does not affect relevancy of the rest of processes in $\mc{D}'_1\mc{D}''_1$:

    The relevancy of processes in $\mc{D}''_1$ remains intact since channel $y^{\langle c,e \rangle}_\alpha$ is irrelevant in both pre-step and post-step.
    %  and remain $y^{\langle c,e \rangle}_{\alpha+1}$ are not relevant in the post step either.
    %

    The relevancy of processes in $\mc{D}'_1$ remains intact as we replace their irrelevant root with another irrelevant process. In particular, this step does not change relevancy of any chain positive messages that process $P$ might have as (grand) children.
    
    % However, we need to be careful about the changes in the  running integrity of a process and their effect on its (grand)children. The runningintegrity of the process offering along $y^{\langle c,e \rangle}_{\alpha+1}$ may be higher or incomparable to $e_1$ based on the code of $P$ (if it starts with a $\mb{recv}$ or $\mb{case}$). This is of significance only if $e_1 \sqsubseteq \xi$, and in the pre-step the process has a chain of positive messages along a channel with integrity level lower than or equal to the observer level (e.g. $x:\_[\langle d',e' \rangle]$ where $e'\sqsubseteq \xi$) as (grand)children. But by the assumption of the subcase, $x:\_[\langle d',e' \rangle]$ cannot be a relevant channel in the pre-step. Thus the chain of messages is not relevant in neither the pre-step and nor the post-step.
\end{itemize}

\item {\color{Aquamarine}\bf Subcase 2. $\mb{proc}(y_\alpha[\langle c,e\rangle],  y^{\langle c,e\rangle}_\alpha.k ;P @\langle d_1,e_1\rangle)$ is  semi-relevant, i.e., $e_1\sqsubseteq \xi$ and $d_1 \not \sqsubseteq \xi$.}

By assumption ($\cproj{\mc{D}_1}{\xi}\meq\cproj{\mc{D}_2}{\xi}$):
\[\mc{D}_2=\mc{D}'_2\mb{proc}(z_\delta[\langle c,e\rangle], Q_{\mid \xi} @\langle d_1, e_1 \rangle) \mc{D}''_2,\] such that $\Psi \Vdash y^{\langle c,e\rangle}_\alpha.k; P \sim_{\langle d_s,e_1 \rangle} Q$ for some $d_s \not\sqsubseteq \xi$, which means that $d_s \not\sqsubseteq e_1$.
%Moreover, $z_\delta=y_\alpha$ if $e\sqsubseteq \xi$.
%

\begin{description}
    \item[\bf Subcase 2a $d_s \sqsubseteq e$.]  we get that $e \not \sqsubseteq \xi$, meaning that the channel $y^{\langle c,e\rangle}_\alpha$ is not relevant in the pre-step and  channels $y^{\langle c,e\rangle}_\alpha$ and $y^{\langle c,e\rangle}_{\alpha+1}$ are not relevant in the post-step. Moreover,  $\mb{msg}(y^{\langle c,e\rangle}_\alpha.k)$ is  irrelevant in the post step. 

    Here, without loss of generality, we assume that when in the pattern checking two \msc{Unsync} rules are available, we apply the one that works on the first process.

    In this subcase, our goal is to show 
    \[\begin{array}{l}
        \cproj{\mc{D}'_1\mb{proc}(y_{\alpha+1}[\langle c,e \rangle], [y^{\langle c,e \rangle}_{\alpha+1}/y^{\langle c,e \rangle}_{\alpha}] P @\langle d_1,e_1 \rangle) \mb{msg}(y^{\langle c,e \rangle}_\alpha.k) \mc{D}''_1}{\xi} \meq \\
        \cproj{\mc{D}'_1\mb{proc}(y_{\alpha+1}[\langle c,e \rangle], [y^{\langle c,e \rangle}_{\alpha+1}/y^{\langle c,e \rangle}_{\alpha}] P @\langle d_1,e_1 \rangle) \mc{D}''_1}{\xi}\meq\\
        \cproj{\mc{D}'_2\mb{proc}(z_\delta[\langle c,e\rangle], Q_{\mid \xi} @\langle d_1, e_1 \rangle) \mc{D}''_2}{\xi}
    \end{array}\]

    By $e_1 \sqsubseteq \xi$ and $d_s \not \sqsubseteq \xi$ we know $d_s \not \sqsubseteq e_1$, and by the assumption of this subcase we know $d_s \sqsubseteq e$. Thus
    by inversion, we know that the last rule in the derivation is \msc{Brb-unsync$_1$}, and we get $\Psi \Vdash P \sim_{\langle d_s, e_1 \rangle} Q$. Processes that offer along $u_\alpha$ and $z_\delta$ in the post-step remain semi-relevant as their running integrity does not change, and their relevant channels in the resources remain intact.

    The relevancy of processes in $\mc{D}''_1$ remains intact since $y^{\langle c,e \rangle}_\alpha$ is neither relevant in the pre-step, nor the post-step.

    The relevancy of processes in $\mc{D}'_1$ remains intact as we replace their semi-relevant root with another semi-relevant process.

    \item[\bf Subcase 2b $d_s \not \sqsubseteq e$].
    % By inversion on the typing rules, we know $e_1 \sqsubseteq e$, and thus $d_s \not \sqsubseteq e_1$. 
    %Moreover, i
    If $e \sqsubseteq \xi$, we have $y_\alpha=z_\delta$.

    %Moreover, by the invariant of semi-relevant nodes and $d_s \not \sqsubseteq e$, we know that $y_\alpha=z_\delta$.
    We apply induction on the size of $Q$ to show that for some $Q'$ and some $\mc{D}_3$ we have
    \[\mb{proc}(z_\delta[\langle c,e\rangle],  Q_{\mid \xi} @\langle d_1, e_1 \rangle) \mapsto^* \mc{D}_3 \mb{proc}(z_\delta[\langle c,e\rangle],  z^{\langle c,e\rangle}_\delta.j ;Q'_{\mid \xi}@\langle d_1, e_1 \rangle)\]
    where every node in $\mc{D}_3$ is irrelevant, the process offering along $z_\delta[\langle c,e\rangle]$ in the post-step is semi-relavant, we have $\Psi \Vdash P \sim_{\langle d_s, e_1 \rangle} Q'$.
    
    We proceed by inversion on the derivation for $\Psi \Vdash y_\alpha^{\langle c,e\rangle}.k;P \sim_{\langle d_s, e_1 \rangle} Q$. Since $d_s \not \sqsubseteq e_1$, {\bf we know that the last rule in the derivation is one of the following three, i.e. either \msc{Brb-SndLab}, or \msc{Brb-Unsync-Spawn$_2$}, or \msc{Brb-Unsync$_2$}.}
    \begin{itemize}
    \item {\bf The first case} (\msc{Brb-SndLab}). We get exactly what we want, i.e. process $Q$ is of the form,
    $y^{\langle c,e\rangle}_\alpha.j ;Q',$
    and
    \[\mb{proc}(z_\delta[\langle c,e\rangle],  Q_{\mid \xi} @\langle d_1, e_1 \rangle) \mapsto^0 \mb{proc}(z_\delta[\langle c,e\rangle],  z^{\langle c,e\rangle}_\delta.j ;Q'_{\mid \xi}@\langle d_1, e_1 \rangle),\] where we rename $z_\delta$ for $y_\alpha$, if $e \not \sqsubseteq \xi$ 

    \item \label{case:unsyncspawn} In {\bf the second case (\msc{Brb-Unsync-Spawn$_2$})}, we know that
    \[\begin{array}{l}
        \mb{proc}(z_\delta[\langle c,e\rangle],  Q_{\mid \xi} @\langle d_1, e_1 \rangle) = \\
        \mb{proc}(z_\delta[\langle c,e\rangle], (u^{\langle c_u,e_u \rangle} \leftarrow X \leftarrow \Delta)@\langle c_0, e_0 \rangle ;Q^1_{\mid \xi} @\langle d_1, e_1 \rangle)
    \end{array}\]
     where $d_s \sqsubseteq e_0$ and $\forall w{:}B_w[\langle c_w, e_w \rangle] \in \Delta$ we have $d_s \sqsubseteq e_w$. Moreover, $\Psi \Vdash z_\delta^{\langle c,e\rangle}.k;P \sim_{\langle d_s, e_1 \rangle} Q^1$. By the dynamic rule for spawn, we have

     \[
    \begin{array}{l}
        \mb{proc}(z_\delta[\langle c,e\rangle], (u^{\langle c_u,e_u \rangle} \leftarrow X \leftarrow \Delta)@\langle c_0, e_0 \rangle ;Q^1_{\mid \xi} @\langle d_1, e_1 \rangle) \mapsto\\[5pt]
        \mb{proc}(u_0[\langle c_u,e_u\rangle], [x_0, \Delta/ x, \m{dom}(\Delta')] \hat{\gamma}(Q^2) @\langle c_0, e_0 \rangle)\, \mb{proc}(z_\delta[\langle c,e\rangle], Q^1_{\mid \xi} @\langle d_1, e_1 \rangle) 
    \end{array}    
    \]
    where $(\Psi'; \Delta' \vdash X=Q_2 @\langle \psi_1, \psi_2 \rangle:: u{[\langle \psi_3,\psi_4 \rangle]}:A')$ is defined in the signature $\Sigma$.

    Observe that (i) $e_0\not \sqsubseteq \xi$, since $d_s \not \sqsubseteq \xi$ and $d_s \sqsubseteq e_0 \sqsubseteq e$ (the condition $e_0 \sqsubseteq e$ is by well-typedness), and thus the spawned process is irrelevant. (ii) The process offering along $z_\delta[\langle c,e\rangle]$ remains semi-relevant: its running integrity does not change and the channels passed to the caller are all not relevant since $d_s \sqsubseteq e_w$ and $d_s \not \sqsubseteq \xi$. (iii) The resources of the process offering along $z_\delta[\langle c,e\rangle]$ that have integrity $e_r$ with $d_s \not \sqsubseteq e_r$ remain intact. (iv) We have $\Psi \Vdash y_\alpha^{\langle c,e\rangle}.k;P \sim_{\langle d_s, e_1 \rangle} Q^1$.

    By the above observations, we can apply the induction hypothesis on $Q^1$ with a smaller size to get
    \[
    \begin{array}{l}
        \mb{proc}(y_\alpha[\langle c,e\rangle], (u^{\langle c_u,e_u \rangle} \leftarrow X \leftarrow \Delta)@\langle c_0, e_0 \rangle ;Q^1_{\mid \xi} @\langle d_1, e_1 \rangle) \mapsto\\[5pt]
        \mb{proc}(u_0[\langle c_u,e_u\rangle], [x_0, \Delta/ x, \m{dom}(\Delta')] \hat{\gamma}(Q^2) @\langle c_0, e_0 \rangle)\, \mb{proc}(y_\alpha[\langle c,e\rangle], Q^1_{\mid \xi} @\langle d_1, e_1 \rangle)  \mapsto^*\\[4pt]
        \mc{D}'_3\,  \mb{proc}(u_0[\langle c_u,e_u\rangle], [x_0, \Delta/ x, \m{dom}(\Delta')] \hat{\gamma}(Q^2) @\langle c_0, e_0 \rangle)\, \mb{proc}(y_\alpha[\langle c,e\rangle],  y^{\langle c,e\rangle}_\alpha.j ;Q'_{\mid \xi}@\langle d_1, e_1 \rangle)
    \end{array}    
    \]

    By $d_s \not \sqsubseteq \xi$ and $d_s \sqsubseteq e_0$, we have $e_0\not \sqsubseteq \xi$ and thus the spawned process is irrelevant. By the induction hypothesis, every other process in $\mc{D}'_3$ is irrelevant. 
    
    % Moreover, the process offering along $y_\alpha[\langle c,e\rangle]$ remains semi-relevant, and that is why we can apply the induction hypothesis, in particular, we only pass irrelevant channels to the spawned process, so its relevancy status does not change. Moreover, the two processes continue to have the same observable channel names.

    \item In {\bf the third case (\msc{Brb-Unsync$_2$})},  we know that
    \[\begin{array}{l}
        \mb{proc}(z_\delta[\langle c,e\rangle],  Q_{\mid \xi} @\langle d_1, e_1 \rangle) = \\
        \mb{proc}(z_\delta^{\langle c,e\rangle}, \outact{u_\beta}{c_u}{e_u}{Q^1_{\mid \xi}}@\langle d_1, e_1 \rangle)
    \end{array}\]
     where $d_s \sqsubseteq e_u$.  Moreover, $\Psi \Vdash z_\delta^{\langle c,e\rangle}.k;P \sim_{\langle d_s, e_1 \rangle} Q^1$.  By assumption ($d_s \not \sqsubseteq e$) we know that $u_\beta$ is a resource and not the offering channel. By the dynamic rules for sending messages, we have

     \[
    \begin{array}{l}
        \mb{proc}(z_\delta[\langle c,e\rangle],  \outact{u_\beta}{c_u}{e_u}{Q^1_{\mid \xi}} @\langle d_1, e_1 \rangle) \mapsto\\[5pt]
        \mb{msg}( \outact{u_\beta}{c_u}{e_u}{\_})\,\,  \mb{proc}(z_\delta[\langle c,e\rangle], [u_{\beta+1}/u_{\beta}]Q^1_{\mid \xi} @\langle d_1, e_1 \rangle)
    \end{array}    
    \]

    Observe that (i)$u_\beta$ is not a relevant channel (by $d_s \sqsubseteq e_u$ and $d_s \not \sqsubseteq \xi$), and thus the message is irrelevant. (ii) The process offering along $z_\delta[\langle c,e\rangle]$ remains semi-relevant: its running integrity does not change and (in the case of higher-order sends) the channel passed to the message is not relevant since they are of the same secrecy $e_u$. (iii) The resources of the process offering along $z_\delta[\langle c,e\rangle]$ that have integrity $e_r$ with $d_s \not \sqsubseteq e_r$ remain intact. (iv) We have $\Psi \Vdash y_\alpha^{\langle c,e\rangle}.k;P \sim_{\langle d_s, e_1 \rangle} [u_{\beta+1}/u_{\beta}]Q^1$.
    %

    % \inact{u_\beta}{c_u}{e_u}{Q^1_{\mid \xi}}
    By the above observations, we can apply the induction hypothesis on $[u_{\beta+1}/u_{\beta}]Q^1$ with a smaller size to get
    \[
    \begin{array}{l}
        \mb{proc}(z_\delta[\langle c,e\rangle],  \outact{u_\beta}{c_u}{e_u}{Q^1_{\mid \xi}} @\langle d_1, e_1 \rangle) \mapsto\\[5pt]
        \mb{msg}( \outact{u_\beta}{c_u}{e_u}{\_})\,\,  \mb{proc}(z_\delta[\langle c,e\rangle], [u_{\beta+1}/u_{\beta}]Q^1_{\mid \xi} @\langle d_1, e_1 \rangle) \mapsto^*\\[4pt]
        \mc{D}'_3\,   \mb{msg}(\outact{u_\beta}{c_u}{e_u}{ \_})\, \mb{proc}(y_\alpha[\langle c,e\rangle],  y^{\langle c,e\rangle}_\alpha.j ;Q'_{\mid \xi}@\langle d_1, e_1 \rangle)
    \end{array}    
    \]
     
     By the induction hypothesis, every other process in $\mc{D}'_3$ is irrelevant.
     \end{itemize}

     Now that we established that  \[\mb{proc}(z_\delta[\langle c,e\rangle],  Q_{\mid \xi} @\langle d_1, e_1 \rangle) \mapsto^* \mc{D}_3 \mb{proc}(z_\delta[\langle c,e\rangle],  z^{\langle c,e\rangle}_\delta.j ;Q'_{\mid \xi}@\langle d_1, e_1 \rangle)\]
     where every node in $\mc{D}_3$ is irrelevant, the process offering along $z_\delta[\langle c,e\rangle]$ in the post-step is semi-relavant, and we have $\Psi \Vdash P \sim_{\langle d_s, e_1 \rangle} Q'$, we have:

    % \[\mc{D}_2=\mc{D}'_2\mb{proc}(z_\delta[\langle c,e\rangle],  z^{\langle c,e\rangle}_\delta.k ;Q @\langle d_1, e_1 \rangle) \mc{D}''_2,\]

     \[\mc{D}_2 \mapsto^* \mc{D}'_2\mc{D}_3\mb{proc}(z_\delta[\langle c,e\rangle],  Q'_{\mid \xi}@\langle d_1, e_1 \rangle) \mb{msg}(z^{\langle c,e\rangle}_\delta.j)\mc{D}''_2\]

     If $e \sqsubseteq \xi$, then $z_\delta=y_\alpha$, otherwise both messages in the first and second run are irrelevant. Note that by the typing rules, we know that $d_1 \sqsubseteq c$, and by $d_1 \not \sqsubseteq \xi$, we get $c \not \sqsubseteq \xi$, and thus the produced messages cannot be relevant (they are either semi-relevant or irrelevant). It is straightforward to observe that the two post-step configurations remain related; explained below:

    In this step, the process offering along $z_{\delta+1}$ may become dead-end in the post-step of the second run and as a result transforms chain of positive messages attached to it to dead-end. This means that we need to satisfy less conditions for the process and its connected channels, i.e., the conditions (7) and (8) do not need to be satisfied anymore. 
     
    We consider two subcacses based on the level $e$:
     \begin{itemize}
         \item If $e\not\sqsubseteq \xi$, then $\mb{msg}(y^{\langle c,e\rangle}_\alpha.k)$ and $\mb{msg}(z^{\langle c,e\rangle}_\delta.j)$ are  irrelevant in both runs, and will be dismissed by the projections. Moreover, in this case neither $y^{\langle c,e\rangle}_\alpha$, nor $z^{\langle c,e\rangle}_\delta$ are relevant in the pre-step and post-step configurations. Thus the relevancy of processes in $\mc{D}''_1$ and $\mc{D}''_2$ will remain intact.
Similarly, the relevancy of nodes in $\mc{D}'_1$ and $\mc{D}'_2$ remain intact as we replace their semi-relevant root with another semi-relevant process with the same set of resources. In particular, this step does not change relevancy of any chain positive messages that process $P$ might have as (grand) children.

         \item If $e \sqsubseteq \xi$, then channels $z^{\langle c,e\rangle}_\delta=y^{\langle c,e\rangle}_\alpha$ are relevant in the pre-step in both runs. In the post-step, channels $z^{\langle c,e\rangle}_{\delta+1}=y^{\langle c,e\rangle}_{\alpha+1}$ are relevant in both runs.

         Relevancy of messages $\mb{msg}(y^{\langle c,e\rangle}_\alpha.k)$ in the post-steps are determined by the running integrities ($e'_1$ and $e'_2$) of their parents ($N_1$ and $N_2$) in $\mc{D}''_1$ and $\mc{D}''_2$. 
     
         First observe that if  $e'_1 \sqsubseteq \xi$, then the message is semi-relevant in the post-step of the first run and its parent is either relevant or semi-relevant. Thus the correponding parent ($N_2$) occurs in $\mc{D}''_2$ with $e'_2 \sqsubseteq \xi$, which is also either relevant or semi-relevant in $\mc{D}''_2$. And as a result, the message is semi-relevant in the post-step of the second run.
         Thus messages $\mb{msg}(y^{\langle c,e\rangle}_\alpha.k)$ and $\mb{msg}(y^{\langle c,e\rangle}_\alpha.j)$are semi-relevant in both runs and channels $z^{\langle c,e\rangle}_\delta=y^{\langle c,e\rangle}_\alpha$ are relevant in the post-step too, satisfying condition (6) in Definition~\ref{def:sim}. 
         
         Otherwise, if $e'_1 \not \sqsubseteq \xi$, 
         the parent ($N_1$) is irrelevant and $\mb{msg}(y^{\langle c,e\rangle}_\alpha.k)$ is irrelevant in the post step of the first run, and will be dismissed by the projection.  We satisfy condition (6) trivially.

         Now consider the case where the message is not dead-end in the second configuration, meaning that its (grand-)parent process in $\mc{D}''_2$ is not dead-end, then both messages are semi-relevant in the post-step and we satisfy the condition (8) in Definition~\ref{def:sim} too.

         Finally, consider the case where the message is dead-end in the second configurations, meaning that its (grand-)parent process in $\mc{D}''_2$ is dead-end, then there is no need to satisfy condition (8) in Definition~\ref{def:sim}.

         In all the above cases the relevancy of processes in $\mc{D}''_1$ and $\mc{D}''_2$ remain intact and thus they continue to be related.

        Similarly, the relevancy of processes in $\mc{D}'_1$ and $\mc{D}'_2$ remain intact and thus they continue to be related.

     \end{itemize}

\end{description}

\item {\color{Aquamarine}\bf Subcase 3. $\mb{proc}(y_\alpha[\langle c,e\rangle],  y^{\langle c,e\rangle}_\alpha.k ;P @\langle d_1, e_1 \rangle)$ is relevant, i.e., $e_1\sqsubseteq d_1 \sqsubseteq \xi$.}

By assumption ($\cproj{\mc{D}_1}{\xi}\meq\cproj{\mc{D}_2}{\xi}$):
\[\mc{D}_2=\mc{D}'_2\mb{proc}(z_\delta[\langle c,e\rangle],  z^{\langle c,e\rangle}_\delta.k ;P_{\mid \xi} @\langle d_1, e_1 \rangle) \mc{D}''_2,\] such that $P_{\mid \xi}$ is equal to $P$ modulo renaming of some channels with higher or incomparable integrity to the observer and if $e\sqsubseteq \xi$, then $z_\delta=y_\alpha$. We have 
\[\mc{D}_2 \mapsto \mc{D}'_2\mb{proc}(z_{\delta+1}[c], [z^c_{\delta+1}/z^c_\delta]P_{\mid \xi} @d_1) \mb{msg}(z^c_\delta.k)\mc{D}''_2\]

It is enough to show 
{\small\[\cproj{\mc{D}'_1\mb{proc}(y_{\alpha+1}[\langle c,e\rangle], [y^{\langle c,e\rangle}_{\alpha+1}/y^{\langle c,e\rangle}_{\alpha}] P @\langle d_1, e_1 \rangle) \mb{msg}(y^{\langle c,e\rangle}_\alpha.k) \mc{D}''_1}{\xi} \meq \]\[\cproj{\mc{D}'_2\mb{proc}(z_{\delta+1}[\langle c,e\rangle], [z^{\langle c,e\rangle}_{\delta+1}/z^{\langle c,e\rangle}_\delta]P_{\mid \xi} @\langle d_1, e_1 \rangle) \mb{msg}(z^{\langle c,e\rangle}_\delta.k)\mc{D}''_2}{\xi}.\]}

In this step, the process offering along $z_{\delta+1}$ may become dead-end in the post-step of the second run and as a result transforms chain of positive messages attached to it to dead-end. This means that we need to satisfy less conditions for the process and its connected channels, i.e., the conditions (3) and (4) do not need to be satisfied anymore. 

We consider two subcacses based on the level $e$:
\begin{itemize}
    \item If $e\not\sqsubseteq \xi$, then $\mb{msg}(y^{\langle c,e\rangle}_\alpha.k)$ and $\mb{msg}(z^{\langle c,e\rangle}_\delta.k)$ are  irrelevant in both runs, and will be dismissed by the projections. Moreover, in this case neither $y^{\langle c,e\rangle}_\alpha$, nor $z^{\langle c,e\rangle}_\delta$ are relevant in the pre-step and post-step configurations. Thus the relevancy of processes in $\mc{D}''_1$ and $\mc{D}''_2$ will remain intact.

    It remains to show that projections of $\mc{D}'_1$ and $\mc{D}'_2$ are equal in the post-step too. 
    The relevancy of nodes in $\mc{D}'_1$ and $\mc{D}'_2$ remain intact as we replace their relevant root with another relevant process with the same set of resources. In particular, this step does not change relevancy of any chain positive messages that process $P$ might have as (grand) children.
    %The resources with integrities higher than or incomparable to the observer offered by the forests $\mc{D}'_1$ and $\mc{D}'_2$ in the pre-step will remain higher than or incomparable to the observer and thus irrelevant in the post-step too. 

    \item If $e \sqsubseteq \xi$, then channels $z^{\langle c,e\rangle}_\delta=y^{\langle c,e\rangle}_\alpha$ are relevant in the pre-step in both runs. In the post-step, channels $z^{\langle c,e\rangle}_{\delta+1}=y^{\langle c,e\rangle}_{\alpha+1}$ are relevant in both runs.
    % (by the tree invariant, the runningintegrity of the processes offering along channels $z^{\langle c,e\rangle}_{\delta+1}=y^{\langle c,e\rangle}_{\alpha+1}$ is less than or equal to $e \sqsubseteq \xi$).

    Relevancy of messages $\mb{msg}(y^{\langle c,e\rangle}_\alpha.k)$ in the post-steps are determined by the  running integrities ($e'_1$ and $e'_2$) of their parents ($N_1$ and $N_2$) in $\mc{D}''_1$ and $\mc{D}''_2$. 

    First observe that if $e'_1 \sqsubseteq \xi$, then message is relevant in the post-step of the first run and its parent is either relevant or semi-relevant. Thus the correponding parent ($N_2$) occurs in $\mc{D}''_2$ with $e'_2 \sqsubseteq \xi$, which is also either relevant or semi-relevant in $\mc{D}''_2$. And as a result, the message is relevant in post-step of the second run: we satisfy condition (2) in Definition~\ref{def:sim}.

    Otherwise, if $e\not \sqsubseteq e'_1$ and $e'_1\sqsubseteq \xi$, then the parent ($N_1$) is either relevant or semi-relevant in the first run and by assumption there is a corresponding relevant or semi-relevant parent ($N_2$) in the second run. Thus messages $\mb{msg}(y^{\langle c,e\rangle}_\alpha.k)$ are relevant in both runs and channels $z^{\langle c,e\rangle}_\delta=y^{\langle c,e\rangle}_\alpha$ are relevant in the post-step too, satisfying Condition (2) in Definition~\ref{def:sim}.

    %The same holds when $e'_2\sqsubseteq \xi$.
    If $e\not \sqsubseteq e'_1$ and $e'_1\not\sqsubseteq \xi$,
    then the parent ($N_1$) is irrelevant and $\mb{msg}(y^{\langle c,e\rangle}_\alpha.k)$ is irrelevant in the post step of the first run, and will be dismissed by the projection. 

    %Otherwise, in both runs the running integrity of the parent is higher than or incomparable to the observer (the parents are both irrelevant). Thus nodes $\mb{msg}(y^{\langle c,e\rangle}_\alpha.k; y^{\langle c,e\rangle}_{\alpha} \leftarrow y^{\langle c,e\rangle}_{\alpha+1})$ are irrelevant in the post step of both runs, and will be dismissed by the projections. 
    
    In all the above cases the relevancy of processes in $\mc{D}''_1$ and $\mc{D}''_2$ remain intact and thus they continue to be related.

   Similarly, the relevancy of processes in $\mc{D}'_1$ and $\mc{D}'_2$ remain intact and thus they continue to be related.

\end{itemize}

\item {\color{ForestGreen} \bf Case 2. $\mc{D}_1=$\\
$\mc{D}'_1\mb{msg}(y^{\langle c, e\rangle}_\alpha.k) \mb{proc}(x_\beta[\langle d, g\rangle], \mb{case} \,y^{\langle c, e\rangle}_\alpha (\ell \Rightarrow P_\ell)_{\ell \in L} @\langle d_1, g_1 \rangle)  \mc{D}''_1$ and
\[ \mc{D}'_1\mb{msg}(y^{\langle c, e\rangle}_\alpha.k)\mb{proc}(x_\beta[\langle d, g\rangle], \mb{case} \,y^{\langle c, e\rangle}_\alpha (\ell \Rightarrow P_\ell)_{\ell \in L} @\langle d_1, g_1 \rangle)  \mc{D}''_1 \mapsto\]
\[\mc{D}'_1\mb{proc}(w_{\beta}[\langle d,g\rangle], [y_{\alpha+1}^{\langle c, e\rangle}/y^{\langle c, e\rangle}_{\alpha}] P_k@(\langle d_1 \sqcup c, g_1 \sqcup e \rangle)) \mc{D}''_1\]}
We consider sub-cases based on relevancy of process offering along $x^{\langle d,g\rangle}_\beta$. 

\begin{description}
\item {\color{Aquamarine}\bf Subcase 1. $\mb{proc}(x_\beta[\langle d, g\rangle], \mb{case} \,y^{\langle c, e\rangle}_\alpha (\ell \Rightarrow P_\ell)_{\ell \in L} @\langle d_1, g_1 \rangle)$ is irrelevant.} 

By definition either $g_1 \not \sqsubseteq \xi  $ or none of the channels connected to $P$ including the channel $y_\alpha^{\langle c, e\rangle}$ are relevant.  In both cases the message $\mb{msg}(y^{\langle c, e\rangle}_\alpha.k)$ in the pre-step and the process \[\mb{proc}(x_{\beta}[\langle d, g\rangle], [y_{\alpha+1}^{\langle c, e\rangle}/y^{\langle c, e\rangle}_{\alpha}] P_k@ \langle d_1 \sqcup c, g_1 \sqcup e \rangle)\] in the post-step are irrelevant too.
Moreover, channel $x_\beta[\langle d,g \rangle]$ is irrelevant in the pre-step (by the tree invariant we know that $g_1 \sqsubseteq g$), and will remain irrelevant in the post-step.

In this subcase, our goal is to show   
\[\cproj{\mc{D}'_1\mb{proc}(x_{\beta}[\langle d, g\rangle], [y_{\alpha+1}^{\langle c, e\rangle}/y^{\langle c, e\rangle}_{\alpha}] P_k@ \langle d_1\sqcup c, g_1 \sqcup e \rangle)\mc{D}''_1}{\xi} \meq \cproj{\mc{D}'_1 \mc{D}''_1}{\xi}\meq \cproj{\mc{D}_1}{\xi}\meq \cproj{\mc{D}_2}{\xi}.\]

To prove this we need four observations:
\begin{itemize}
    \item  $\mb{proc}(x_{\beta}[\langle d, g\rangle], [y_{\alpha+1}^{\langle c, e\rangle}/y^{\langle c, e\rangle}_{\alpha}] P_k@ \langle d_1\sqcup c, g_1 \sqcup e \rangle)$ is irrelevant and will be dismissed by the projection (as explained above).
    \item $\mb{msg}(y^{\langle c, e\rangle}_\alpha.k)$ is irrelevant in the pre-step (in particular, it is not a corresponding node for a relevant or semi-relevant node in the second run as specified in conditions 4 and 8 in Definition~\ref{def:sim}). The reason is: if the message is either semi-relevant or relevant then its running integrity $g_1 \sqcup e \sqsubseteq \xi$ and its connected channels are relevant. This contradicts with  the process offering along $x_\beta$ being irrelevant in the pre-step.

    \item Replacing $\mb{msg}(y^{\langle c, e\rangle}_\alpha.k)$ and  $\mb{proc}(x_\beta[\langle d, g\rangle], \mb{case} \,y^{\langle c, e\rangle}_\alpha (\ell \Rightarrow P_\ell)_{\ell \in L} @\langle d_1, g_1 \rangle)$ with \[\mb{proc}(x_{\beta}[\langle d, g\rangle], [y_{\alpha+1}^{\langle c, e\rangle}/y^{\langle c, e\rangle}_{\alpha}] P_k@ \langle d_1\sqcup c, g_1 \sqcup e \rangle),\] does not affect the relevancy of the rest of processes in $\mc{D}''_1$ since channel $x^{\langle c,e \rangle}_\beta$ is irrelevant in both pre-step and post-step.
    %  and remain $y^{\langle c,e \rangle}_{\alpha+1}$ are not relevant in the post step either.
    %

    \item Replacing $\mb{msg}(y^{\langle c, e\rangle}_\alpha.k)$ and  $\mb{proc}(x_\beta[\langle d, g\rangle], \mb{case} \,y^{\langle c, e\rangle}_\alpha (\ell \Rightarrow P_\ell)_{\ell \in L} @\langle d_1, g_1 \rangle)$ with \[\mb{proc}(x_{\beta}[\langle d, g\rangle], [y_{\alpha+1}^{\langle c, e\rangle}/y^{\langle c, e\rangle}_{\alpha}] P_k@ \langle d_1\sqcup c, g_1 \sqcup e \rangle),\] does not affect the relevancy of the rest of processes in $\mc{D}'_1$ either:
    
    we replace their irrelevant root, either \[\mb{msg}(y^{\langle c, e\rangle}_\alpha.k)\,\,\, \m{or}\,\,\,\mb{proc}(x_\beta[\langle d, g\rangle], \mb{case} \,y^{\langle c, e\rangle}_\alpha (\ell \Rightarrow P_\ell)_{\ell \in L} @\langle d_1, g_1 \rangle)\]  with another irrelevant node, $\mb{proc}(x_{\beta}[\langle d, g\rangle], [y_{\alpha+1}^{\langle c, e\rangle}/y^{\langle c, e\rangle}_{\alpha}] P_k@ \langle d_1\sqcup c, g_1 \sqcup e \rangle)$. In particular, this step does not change relevancy of any chain of positive messages that process $P$ might have as (grand) children:
    \begin{itemize}
        \item If $g_1 \not \sqsubseteq \xi$, then any chain of messages connected to the process is  irrelevant in the pre-step, nor the post-step.
        \item If $g_1 \sqsubseteq \xi$ but the process does not have any relevant channels connected to it in the pre-step, we consider the min integrity $e_m$ of the chain of messages connected to the process. If $e_m \not \sqsubseteq \xi$, the chain of messages is irrelevant in the pre-step, and in the post-step.
        Otherwise, if $e_m \sqsubseteq \xi$, we know that the channels connecting the messages in the chain and the chain to the process oferring along $x_\beta$ are irrelevant in the pre-step (otherwise we could form a contradiction with the process offering along $x_\beta$ being irrelevant in the pre-step). The channels remain irrelevant in the post-step too, as we only (may) increase the running integrity of the messages in the chain.
    \end{itemize}
    In particular, none of the messages in these chains of messages are the corresponding node for a relevant or semi-relevant node in the second run (for satisfying conditions $4$ and $8$ in Definition~\ref{def:sim}).
\end{itemize}

\item {\color{Aquamarine}\bf Subcase 2. $\mb{proc}(x_\beta[\langle d, g\rangle], \mb{case} \,y^{\langle c,e \rangle}_\alpha (\ell \Rightarrow P_\ell)_{\ell \in L} @\langle d_1, g_1 \rangle)$ is semi-relevant, i.e. $g_1 \sqsubseteq \xi$ and $d_1 \not \sqsubseteq \xi$.}
By assumption ($\cproj{\mc{D}_1}{\xi}\meq\cproj{\mc{D}_2}{\xi}$):
\[\mc{D}_2=\mc{D}'_2\mb{proc}(z_\delta[\langle d,g\rangle], Q_{\mid \xi} @\langle d_1, g_1 \rangle) \mc{D}''_2,\] such that $\Psi \Vdash \mb{case} \,y^{\langle c, e\rangle}_\alpha (\ell \Rightarrow P_\ell)_{\ell \in L} @\langle d_1, g_1 \rangle \sim_{\langle d_s,g_1 \rangle} Q$ for some $d_s \not\sqsubseteq \xi$, which means that $d_s \not\sqsubseteq g_1$. 
We proceed by considering three subcases based on the relevancy of $\mb{msg}(y^{\langle c, e\rangle}_\alpha.k)$: it is either irrelevant, semi-relevant, or relevant.

\begin{description}
    \item[Subcase 2a.] $\mb{msg}(y^{\langle c, e\rangle}_\alpha.k)$ is irrelevant, i.e., $e \not \sqsubseteq \xi$.

    In this case, the process offering along $x_\beta$ in the post-step of the first configuration is irrelevant. Moreover, by assumption, we know that the process offering along $z_\delta$ in the second configuration is dead-end. Our goal is to show:
    \(\cproj{\mc{D}'_1\mb{proc}(x_{\beta}[\langle d, g\rangle], [y_{\alpha+1}^{\langle c, e\rangle}/y^{\langle c, e\rangle}_{\alpha}] P_k@ \langle d_1\sqcup c, g_1 \sqcup e \rangle)\mc{D}''_1}{\xi} \meq \cproj{\mc{D}'_1 \mc{D}''_1}{\xi}\meq \cproj{\mc{D}_1}{\xi}\meq \cproj{\mc{D}_2}{\xi}.\)

    To prove this we need the following observations:
    \begin{itemize}
        \item  $\mb{proc}(x_{\beta}[\langle d, g\rangle], [y_{\alpha+1}^{\langle c, e\rangle}/y^{\langle c, e\rangle}_{\alpha}] P_k@ \langle d_1\sqcup c, g_1 \sqcup e \rangle)$ is irrelevant and will be dismissed by the projection.
        \item $\mb{msg}(y^{\langle c, e\rangle}_\alpha.k)$ is irrelevant in the pre-step (in particular, it is not a corresponding node for a relevant or semi-relevant node in the second run as specified in conditions 4 and 8 in Definition~\ref{def:sim}).

        \item By the tree invariant, $e \sqcup g_1 \sqsubseteq g$ and channel $x^{\langle d,g \rangle}_\beta$ is irrelevant in both pre-step and post-step. 
        As a result, replacing $\mb{proc}(x_\beta[\langle d, g\rangle], \mb{case} \,y^{\langle c, e\rangle}_\alpha (\ell \Rightarrow P_\ell)_{\ell \in L} @\langle d_1, g_1 \rangle)$ and  $\mb{msg}(y^{\langle c, e\rangle}_\alpha.k)$ in the pre-step with 
        
        $\mb{proc}(x_{\beta}[\langle d, g\rangle], [y_{\alpha+1}^{\langle c, e\rangle}/y^{\langle c, e\rangle}_{\alpha}] P_k@ \langle d_1\sqcup c, g_1 \sqcup e \rangle)$ in the post-step, does not affect the relevancy of the rest of processes in $\mc{D}''_1$. 
       \item  However, replacing  $\mb{proc}(x_\beta[\langle d, g\rangle], \mb{case} \,y^{\langle c, e\rangle}_\alpha (\ell \Rightarrow P_\ell)_{\ell \in L} @\langle d_1, g_1 \rangle)$ and
       
       $\mb{msg}(y^{\langle c, e\rangle}_\alpha.k)$  with $\mb{proc}(x_{\beta}[\langle d, g\rangle], [y_{\alpha+1}^{\langle c, e\rangle}/y^{\langle c, e\rangle}_{\alpha}] P_k@ \langle d_1\sqcup c, g_1 \sqcup e \rangle)$, may affect the relevancy of the rest of processes in $\mc{D}'_1$. We explain below how we preserve relatedness of two cnofigurations despite this change:
    
       For those subtrees in $\mc{D}'_1$ that are connected to the process that offers along $x_\beta$ in the pre-step via an irrelevant channel, the relevancy does not change in the pre-and post-step. This, in particular, holds for the subtree offering along $v$.
        
      If a subtree that is connected to the process via a relevant channel, then by the tree invariant, we know that all nodes in the sub-tree are either relevant or semi-relevant. The interesting case is when this channel becomes irrelevant and consequently makes all the nodes in the sub-tree irrelevant too. Next, we consider this case for an arbitrary channel $u[c_u,e_u]$ and sub-tree $\mc{D}_u$ that offers this channel, where $e_u \sqsubseteq \xi$. If the nodes in $\mc{D}_u$ become irrelevant, then they are dismissed by the projection in the post-step, so there is no need to establish conditions 1,2,5, and 6 in Definition~\ref{def:sim} for them. We show that these nodes are not corresponding nodes to any process/message  in $\mc{D}_2$ in the pre-step as required by  conditions 3,4,7, and 8 in Definition~\ref{def:sim}: Assume they are, we get that there is a sub-tree $\mc{D}'_u$ offering along $u[c_u,e_u]$ in $\mc{D}_2$. By relatedness of the configurations and since $u$ is a relevant channel, we know that $u$ is a resource of the process offering along $z_\delta$ in $\mc{D}_2$. Moreover, since $\mc{D}_u$ becomes irrelevant in the post-step, we know that its only connection to the interface in the pre-step is via the process offering along $x_\beta$ in $\mc{D}_1$. Thus, $\mc{D}'_u$ has the same property: its only connection to the interface is via the process offering along $z_\delta$, meaning that all nodes in $\mc{D}'_u$ are dead-ends and do no have a counterpart in $\mc{D}'_u$ in the pre-step, which is a contradiction.
      This completes the proof of this subcase. 
    \end{itemize}

    \item[Subcase 2b.]  $\mb{msg}(y^{\langle c, e\rangle}_\alpha.k)$ is semi-relevant, i.e., $e \sqsubseteq \xi$, and $c \not \sqsubseteq \xi$.
    By assumption, we know that $\mb{msg}(y^{\langle c, e\rangle}_\alpha.j)$ is in  $\mc{D}'_2$, i.e., $\mc{D}'_2=\mc{D}^1_2 \, \mb{msg}(y^{\langle c, e\rangle}_\alpha.j)$.

    We use assumption $\Psi \Vdash  \mb{case} \,y^{\langle c, e\rangle}_\alpha (\ell \Rightarrow P_\ell)_{\ell \in L} @\langle d_1, g_1 \rangle \sim_{\langle d_s,g_1 \rangle} Q$ to show that for some $Q'$ and some $\mc{D}_3$ we have
    \[\begin{array}{l}
        \mc{D}'_2 \mb{proc}(z_\delta[\langle d,g\rangle],  Q_{\mid \xi} @\langle d_1, g_1 \rangle) \mc{D}''_2 \mapsto^*\\
        \mc{D}'_2\,\mc{D}_3 \mb{proc}(z_\delta[\langle d,g\rangle], \mathbf\mb{case} \,y^{\langle c, e\rangle}_\alpha (\ell \Rightarrow Q'_{\ell\mid \xi})_{\ell \in L}@\langle d_1, g_1 \rangle) \mc{D}''_2
    \end{array}\]
    where every node in $\mc{D}_3$ is irrelevant, the process offering along $z_\delta[\langle c,e\rangle]$ in the post-step is semi-relavant, and for every $\ell\in L$ and $j \in L$, we have $\Psi \Vdash P_\ell \sim_{\langle d_s, g_1 \sqcup g \rangle} Q'_j$.
   Similar to Case 1. Subcase 2b, we apply induction on the size of $Q$ and inversion on the derivation for $\Psi \Vdash  \mb{case} \,y^{\langle c, e\rangle}_\alpha (\ell \Rightarrow P_\ell)_{\ell \in L} @\langle d_1, g_1 \rangle \sim_{\langle d_s,g_1 \rangle} Q$ to show the above.

    Since $d_s \not \sqsubseteq g_1$, {\bf we know that the last rule in the derivation is one of the following three, i.e. either \msc{Brb-RcvLab}, or \msc{Brb-Unsync-Spawn$_2$}, or \msc{Brb-Unsync$_2$}.} The proof in the last two cases are identical to the corresponding cases in Case 1. Subcase 2b. We only provide the proof for the first case:
    \begin{itemize}
        \item In {\bf the first case} (\msc{Brb-RcvLab}), we get exactly what we want, i.e. process $Q$ is of the form,
        $\mathbf\mb{case} \,y^{\langle c, e\rangle}_\alpha (\ell \Rightarrow Q'_{\ell\mid \xi})_{\ell \in L},$
        and
        
        \(\mb{proc}(z_\delta[\langle d,g\rangle],  Q_{\mid \xi} @\langle d_1, g_1 \rangle) \mapsto^0\)
        
        \(\mb{proc}(z_\delta[\langle d,g\rangle],  \mathbf\mb{case} \,y^{\langle c, e\rangle}_\alpha (\ell \Rightarrow Q'_{\ell\mid \xi})_{\ell \in L}@\langle d_1, g_1 \rangle).\)
        Moreover, for every $j,\ell \in L$, we have $\Psi \Vdash P_\ell \sim_{\langle d_s, g_1 \sqcup g \rangle} Q'_j$.
        \item {\bf The second case (\msc{Brb-Unsync-Spawn$_2$})}. see Case 1. Subcase 2b.
        \item {\bf The third case (\msc{Brb-Unsync$_2$})}. 
        see Case 1. Subcase 2b.
    \end{itemize}

   Thus we can step the second configuration as
    \[\begin{array}{l}
        \mc{D}_2 \mapsto^*\\ \mc{D}^1_2\mc{D}_3\,\mb{msg}(y^{\langle c, e\rangle}_\alpha.j)\,\mb{proc}(z_\delta[\langle c,e\rangle],  \mathbf\mb{case} \,y^{\langle c, e\rangle}_\alpha (\ell \Rightarrow Q'_{\ell\mid \xi})_{\ell \in L}@\langle d_1, g_1 \rangle)\mc{D}''_2\mapsto\\
        \mc{D}^1_2\mc{D}_3 \,\mb{proc}(z_\delta[\langle d,g\rangle],   [y_{\alpha+1}^{\langle c, e\rangle}/y^{\langle c, e\rangle}_\alpha ]Q'_{j\mid \xi}@\langle d_1 \sqcup d, g_1 \sqcup g \rangle)\mc{D}''_2
    \end{array}\]
    We know that $\Psi \Vdash P_k \sim_{\langle d_s, g_1 \sqcup g \rangle} Q'_j$. By assumption, the process offering along $x_\beta$ in the post-step of the first run  is semi-relevant and so is the process offering along $z_\delta$ in the post-step of the second run. The process  offering along $z_\delta$ in the post-step of the second run may become dead-end and makes some other ndoes dead-end, but this only means we need to establish less conditions. The relevancy of other processes and messages in the post-step of the first and second run remains intact, and they continue to be reltaed by Definition~\ref{def:sim}.
    
  \item[Subcase 2c.]  $\mb{msg}(y^{\langle c, e\rangle}_\alpha.k)$ is relevant, i.e., $e \sqsubseteq \xi$, and $c \sqsubseteq \xi$.
    This case is similar to the previous case, except that we know the corresponding message in $\mc{D}_2$ sends the same label $k$ along channel $y$, i.e., $\mc{D}'_2=\mc{D}^1_2 \, \mb{msg}(y^{\langle c, e\rangle}_\alpha.k)$. The rest of the proof is similar.
\end{description}

\item {\color{Aquamarine}\bf Subcase 3. $\mb{proc}(x_\beta[\langle d, g\rangle], \mb{case} \,y^{\langle c,e \rangle}_\alpha (\ell \Rightarrow P_\ell)_{\ell \in L} @\langle d_1, g_1 \rangle)$ is relevant, i.e., $g_1 \sqsubseteq \xi$ and $d_1 \sqsubseteq \xi$}.

By assumption ($\cproj{\mc{D}_1}{\xi}\meq\cproj{\mc{D}_2}{\xi}$):
\[\mc{D}_2=\mc{D}'_2\mb{proc}(z_\delta[\langle d,g\rangle], \mb{case} \,y^{\langle c, e\rangle}_\alpha (\ell \Rightarrow P_{\ell\mid \xi})_{\ell \in L} @\langle d_1, g_1 \rangle) \mc{D}''_2.\]

We consider three subcases based on the relevancy of $\mb{msg}(y^{\langle c, e\rangle}_\alpha.k)$: it is either irrelevant, semi-relevant, or relevant.

\begin{description}
    \item[Subcase 3a.] $\mb{msg}(y^{\langle c, e\rangle}_\alpha.k)$ is irrelevant, i.e., $e \not \sqsubseteq \xi$.

    In this case, the process offering along $x_\beta$ in the post-step of the first configuration is irrelevant. Moreover, by assumption, we know that the process offering along $z_\delta$ in the second configuration is dead-end. Our goal is to show:

    \(\cproj{\mc{D}'_1\mb{proc}(x_{\beta}[\langle d, g\rangle], [y_{\alpha+1}^{\langle c, e\rangle}/y^{\langle c, e\rangle}_{\alpha}] P_k@ \langle d_1\sqcup c, g_1 \sqcup e \rangle)\mc{D}''_1}{\xi} \meq \cproj{\mc{D}'_1 \mc{D}''_1}{\xi}\meq \cproj{\mc{D}_1}{\xi}\meq \cproj{\mc{D}_2}{\xi}.\)

    To prove this we need the following observations:
    \begin{itemize}
        \item  $\mb{proc}(x_{\beta}[\langle d, g\rangle], [y_{\alpha+1}^{\langle c, e\rangle}/y^{\langle c, e\rangle}_{\alpha}] P_k@ \langle d_1\sqcup c, g_1 \sqcup e \rangle)$ is irrelevant and will be dismissed by the projection.
        \item $\mb{msg}(y^{\langle c, e\rangle}_\alpha.k)$ is irrelevant in the pre-step (in particular, it is not a corresponding node for a relevant or semi-relevant node in the second run as specified in conditions 4 and 8 in Definition~\ref{def:sim}).

        \item By the tree invariant,  $g_1 \sqcup e \sqsubseteq g$ and channel $x^{\langle d,g \rangle}_\beta$ is irrelevant in both pre-step and post-step. 
        As a result, replacing $\mb{proc}(x_\beta[\langle d, g\rangle], \mb{case} \,y^{\langle c, e\rangle}_\alpha (\ell \Rightarrow P_\ell)_{\ell \in L} @\langle d_1, g_1 \rangle)$ and  $\mb{msg}(y^{\langle c, e\rangle}_\alpha.k)$ in the pre-step with 
        
        $\mb{proc}(x_{\beta}[\langle d, g\rangle], [y_{\alpha+1}^{\langle c, e\rangle}/y^{\langle c, e\rangle}_{\alpha}] P_k@ \langle d_1\sqcup c, g_1 \sqcup e \rangle)$ in the post-step, does not affect the relevancy of the rest of processes in $\mc{D}''_1$. 
       \item  However, replacing  $\mb{proc}(x_\beta[\langle d, g\rangle], \mb{case} \,y^{\langle c, e\rangle}_\alpha (\ell \Rightarrow P_\ell)_{\ell \in L} @\langle d_1, g_1 \rangle)$ and 
       
       $\mb{msg}(y^{\langle c, e\rangle}_\alpha.k)$  with $\mb{proc}(x_{\beta}[\langle d, g\rangle], [y_{\alpha+1}^{\langle c, e\rangle}/y^{\langle c, e\rangle}_{\alpha}] P_k@ \langle d_1\sqcup c, g_1 \sqcup e \rangle)$, may affect the relevancy of the rest of processes in $\mc{D}'_1$. We explain below how we preserve relatedness of two cnofigurations despite this change:
    
       For those subtrees in $\mc{D}'_1$ that are connected to (as resources of) the process that offers $x_\beta$ in the pre-step via an irrelevant channel, the relevancy does not change in the pre-and post-step. This, in particular, holds for the subtree offering along $v$.
        
      If a subtree that is connected to the process via a relevant channel, then by the tree invariant, we know that all nodes in the sub-tree are either relevant or semi-relevant. The interesting case is when this channel becomes irrelevant and consequently makes all the nodes in the sub-tree irrelevant too. Next, we consider this case for an arbitrary channel $u[c_u,e_u]$ and sub-tree $\mc{D}_u$ that offers this channel, where $e_u \sqsubseteq \xi$. If the nodes in $\mc{D}_u$ become irrelevant, then they are dismissed by the projection in the post-step, so there is no need to establish conditions 1,2,5, and 6 in Definition~\ref{def:sim} for them. We show that these nodes are not corresponding nodes to any process/message  in $\mc{D}_2$ in the pre-step as required by  conditions 3,4,7, and 8 in Definition~\ref{def:sim}: Assume they are, we get that there is a sub-tree $\mc{D}'_u$ offering along $u[c_u,e_u]$ in $\mc{D}_2$. By relatedness of the configurations and since $u$ is a relevant channel, we know that $u$ is a resource of the process offering along $z_\delta$ in $\mc{D}_2$. Moreover, since $\mc{D}_u$ becomes irrelevant in the post-step, we know that its only connection to the interface in the pre-step is via the process offering along $x_\beta$ in $\mc{D}_1$. Thus, $\mc{D}'_u$ has the same property: its only connection to the interface is via the process offering along $z_\delta$, meaning that all nodes in $\mc{D}'_u$ are dead-ends and do no have a counterpart in $\mc{D}'_u$ in the pre-step, which is a contradiction.
      This completes the proof of this subcase. 
    \end{itemize}

    \item[Subcase 3b.]  $\mb{msg}(y^{\langle c, e\rangle}_\alpha.k)$ is semi-relevant, i.e., $e \sqsubseteq \xi$, and $c \not \sqsubseteq \xi$.
    By assumption, we know that $\mb{msg}(y^{\langle c, e\rangle}_\alpha.j)$ is in  $\mc{D}'_2$, i.e., $\mc{D}'_2=\mc{D}^1_2 \, \mb{msg}(y^{\langle c, e\rangle}_\alpha.j)$.

   Thus, we can step the second configuration as
    \[\begin{array}{l}
        \mc{D}_2 = \mc{D}^1_2\,\mb{msg}(y^{\langle c, e\rangle}_\alpha.j)\,\mb{proc}(z_\delta[\langle c,e\rangle],  \mathbf\mb{case} \,y^{\langle c, e\rangle}_\alpha (\ell \Rightarrow P_{\ell\mid \xi})_{\ell \in L}@\langle d_1, g_1 \rangle)\mc{D}''_2\mapsto\\
        \mc{D}^1_2\mc{D}_3 \,\mb{proc}(z_\delta[\langle d,g\rangle],   [y_{\alpha+1}^{\langle c, e\rangle}/y^{\langle c, e\rangle}_\alpha ]P_{j\mid \xi}@\langle d_1 \sqcup d, g_1 \sqcup g \rangle)\mc{D}''_2
    \end{array}\]
    By inversion on the typing rules, we know that $\Psi \Vdash P_k \sim_{\langle c, g_1 \sqcup g \rangle} P_j$ and $c \not\sqsubseteq \xi$.

    By assumption, the process offering along $x_\beta$ in the post-step of the first run  is semi-relevant and so is the process offering along $z_\delta$ in the post-step of the second run. The process  offering along $z_\delta$ in the post-step of the second run may become dead-end and makes some other ndoes dead-end, but this only means we need to establish less conditions. The relevancy of other processes and messages in the post-step of the first and second run remains intact, and they continue to be reltaed by Definition~\ref{def:sim}. 
    
    \item[Subcase 3c.]  $\mb{msg}(y^{\langle c, e\rangle}_\alpha.k)$ is relevant, i.e., $e \sqsubseteq \xi$, and $c \sqsubseteq \xi$.
    This case is similar to the previous case, except that we know the corresponding message in $\mc{D}_2$ sends the same label $k$ along channel $y$, i.e., $\mc{D}'_2=\mc{D}^1_2 \, \mb{msg}(y^{\langle c, e\rangle}_\alpha.k)$. The rest of the proof is similar.
\end{description}
\end{description}
\end{description}

\item {\color{ForestGreen} \bf Case 3.  $\mc{D}_1=\mc{D}'_1\mb{proc}(y_\alpha[\langle c,e \rangle], (x^{\langle d,g\rangle} \leftarrow X[\gamma] \leftarrow \Lambda) @\langle d_2, g_2 \rangle; P' @\langle d_1, g_1 \rangle)\mc{D}''_1$ and $\mathbf{def}(\Psi'; \Lambda' \vdash X = P @{\langle \psi_1, \psi_2 \rangle} ::x: B'[\langle \psi, \psi' \rangle])$ and
\[\begin{array}{l}
    \mc{D}'_1\mb{proc}(y_\alpha[\langle c, e \rangle], (x^{\langle d, g\rangle} \leftarrow X[\gamma] \leftarrow \Lambda)@\langle d_2, g_2 \rangle; P' @\langle d_1, g_1 \rangle)  \mc{D}''_1 \mapsto\\
    \mc{D}'_1 \mb{proc}(a_0[\langle d, g\rangle], \subst{a^{\langle d, g\rangle}_0, \Lambda}{x^{\langle d, g\rangle}, \Lambda'} \hat{\gamma}\,P @\langle d_2, g_2\rangle)\mb{proc}(y_\alpha[\langle c,e \rangle], [a^{\langle d, g\rangle}_0/x^{\langle d, g\rangle}] P' @\langle d_1, g_1 \rangle) \mc{D}''_1
\end{array}\]
}
We consider sub-cases based on relevancy of process offering along $y^{\langle c,e \rangle}_\alpha$.
\begin{description}
    \item {\color{Aquamarine}\bf Subcase 1. $\mb{proc}(y_\alpha[\langle c,e \rangle], (x^{\langle d,g\rangle} \leftarrow X[\gamma] \leftarrow \Lambda) @\langle d_2, g_2 \rangle; P' @\langle d_1, g_1 \rangle)$ is  irrelevant.}

    By inversion on the typing rules $g_1 \sqsubseteq g_2 \sqsubseteq g \sqsubseteq e$, and 
    $g_1 \sqsubseteq d_2 \sqsubseteq d \sqsubseteq c$.
    By definition either $g_1 \not\sqsubseteq \xi $ or none of the channels connected to the process including its offering channel $y_\alpha^{\langle c,e \rangle}$ are relevant. In both cases
    \[\begin{array}{l}
        \mb{proc}(a_0[\langle d, g\rangle], \subst{a^{\langle d, g\rangle}_0, \Lambda}{x^{\langle d, g\rangle}, \Lambda'} \hat{\gamma}\,P @\langle d_2, g_2\rangle), \,\, \m{and}\\
        \mb{proc}(y_\alpha[\langle c,e \rangle], [a^{\langle d, g\rangle}_0/x^{\langle d, g\rangle}] P' @\langle d_1, g_1 \rangle)
    \end{array}\]
    are irrelevant in the post step.
The reason is: from $g_1 \sqsubseteq g_2, g, e$ and $g_1 \not\sqsubseteq \xi$, we get $g_2, g, e \not\sqsubseteq \xi$. Channel $y^{\langle c,e \rangle}_\alpha$ is not relevant in the pre-step, and both channels $y^{\langle c,e \rangle}_\alpha$ and $y^{\langle c,e \rangle}_{\alpha+1}$ are not relevant in pre-step and post-step configurations. 

Every irrelevant resource of $\mb{proc}(y_\alpha[\langle c,e \rangle], (x^{\langle d,g\rangle} \leftarrow X[\gamma] \leftarrow \Lambda) @\langle d_2, g_2 \rangle; P' @\langle d_1, g_1 \rangle)$ will remain irrelevant in the post-step too.

In this subcase, our goal is to show 
{\small
\[\begin{array}{l}
    \cproj{\mc{D}'_1\mb{proc}(a_0[\langle d, g\rangle], \subst{a^{\langle d, g\rangle}_0, \Lambda}{x^{\langle d, g\rangle}, \Lambda'} \hat{\gamma}\,P @\langle d_2, g_2\rangle), 
    \mb{proc}(y_\alpha[\langle c,e \rangle], [a^{\langle d, g\rangle}_0/x^{\langle d, g\rangle}] P' @\langle d_1, g_1 \rangle)\mc{D}''_1}{\xi} \meq \\
    \cproj{\mc{D}'_1 \mc{D}''_1}{\xi}\meq \cproj{\mc{D}_1}{\xi}\meq \cproj{\mc{D}_2}{\xi}
\end{array}\]
}

To prove this we need two observations:
\begin{itemize}
    \item  Both processes in the post-step \[\mb{proc}(a_0[\langle d, g\rangle], \subst{a^{\langle d, g\rangle}_0, \Lambda}{x^{\langle d, g\rangle}, \Lambda'} \hat{\gamma}\,P @\langle d_2, g_2\rangle)\, \m{and}\]
    \[\mb{proc}(y_\alpha[\langle c,e \rangle], [a^{\langle d, g\rangle}_0/x^{\langle d, g\rangle}] P' @\langle d_1, g_1 \rangle)\] are irrelevant and they will be dismissed by the projection. (As explained above.)
    \item Replacing $\mb{proc}(y_\alpha[\langle c,e \rangle], (x^{\langle d,g\rangle} \leftarrow X[\gamma] \leftarrow \Lambda) @\langle d_2, g_2 \rangle; P' @\langle d_1, g_1 \rangle)$ with these two nodes, does not affect relevancy of the rest of processes in $\mc{D}'_1\mc{D}''_1$:

    The relevancy of processes in $\mc{D}''_1$ remains intact since channel $y^{\langle c,e \rangle}_\alpha$ is irrelevant in both pre-step and post-step.

    The relevancy of processes in $\mc{D}'_1$ remains intact as we replace their irrelevant root with another irrelevant process. In particular, this step does not change relevancy of any chain positive messages that process $P$ might have as (grand) children.
\end{itemize}

    \item {\color{Aquamarine}\bf Subcase 2. $\mb{proc}(y_\alpha[\langle c,e \rangle], (x^{\langle d,g\rangle} \leftarrow X[\gamma] \leftarrow \Lambda) @\langle d_2, g_2 \rangle; P' @\langle d_1, g_1 \rangle)$ is  semi-relevant.}

    By assumption ($\cproj{\mc{D}_1}{\xi}\meq\cproj{\mc{D}_2}{\xi}$):
\[\mc{D}_2=\mc{D}'_2\mb{proc}(z_\delta[\langle c,e\rangle], Q_{\mid \xi} @\langle d_1, g_1 \rangle) \mc{D}''_2,\] such that $\Psi \Vdash (x^{\langle d,g\rangle} \leftarrow X[\gamma] \leftarrow \Lambda) @\langle d_2, g_2 \rangle; P'  \sim_{\langle d_s,g_1 \rangle} Q$ for some $d_s \not\sqsubseteq \xi$, which means that $d_s \not\sqsubseteq g_1$.

\begin{description}
    \item[\bf Subcase 2a $d_s \sqsubseteq d_2$ and $\forall u:B_u\langle c_u,e_u\rangle \in \Lambda. d_s \sqsubseteq e_u$ .]
      we get that $d_2 \not \sqsubseteq \xi$, and $\forall u:B_u\langle c_u,e_u\rangle \in \Lambda. e_u \not \sqsubseteq \xi$,  meaning that the channels in $\Lambda$ and $y^{\langle c,e\rangle}_\alpha$ are irrelevant in the pre-step and post-step. 
      Moreover, the newly spawned channel $a_0^{\langle d,g\rangle}$ is irrelevant in the post-step, and the process spawned in the post-step $\mb{proc}(a_0[\langle d, g\rangle], \subst{a^{\langle d, g\rangle}_0, \Lambda}{x^{\langle d, g\rangle}, \Lambda'} \hat{\gamma}\,P @\langle d_2, g_2\rangle)$ is  irrelevant. 

    In this subcase, our goal is to show 
   {\small \[\begin{array}{l}
        \cproj{\mc{D}'_1\mb{proc}(a_0[\langle d, g\rangle], \subst{a^{\langle d, g\rangle}_0, \Lambda}{x^{\langle d, g\rangle}, \Lambda'} \hat{\gamma}\,P @\langle d_2, g_2\rangle), 
        \mb{proc}(y_\alpha[\langle c,e \rangle], [a^{\langle d, g\rangle}_0/x^{\langle d, g\rangle}] P' @\langle d_1, g_1 \rangle) \mc{D}''_1}{\xi} \meq \\
        \cproj{\mc{D}'_2\mb{proc}(z_\delta[\langle c,e\rangle], Q_{\mid \xi} @\langle d_1, e_1 \rangle) \mc{D}''_2}{\xi}
    \end{array}\]}

    By inversion, we know that the last rule in the derivation is \msc{Brb-Unsync-Spawn$_1$}, and we get $\Psi \Vdash P' \sim_{\langle d_s, g_1 \rangle} Q$. 
    In paricular, they have the same set of relevant resources (but not necessarily the same set of irrelevant resources).
    The relevancy of modes in $\mc{D}''_1$ does not change since by the tree-invariant the channel $y_\alpha^{\langle  c,e \rangle}$ is irrelevant in the pre-step and remains irrelevant in the post-step. 

    The relevancy of nodes in $\mc{D}'_1$ does not change either since in this spawn we only pass irrelevant nodes as resources to the spawnee. 
\item[\bf Subcase 2b $d_s \not \sqsubseteq d_2$ or  $\exists u:B_u\langle c_u,e_u\rangle \in \Lambda.\, d_s \not \sqsubseteq e_u$ .]

We use the assumption $\Psi \Vdash  (x^{\langle d,g\rangle} \leftarrow X[\gamma] \leftarrow \Lambda) @\langle d_2, g_2 \rangle; P'  \sim_{\langle d_s,g_1 \rangle} Q$ to show that for some $Q'$ and some $\mc{D}_3$ we have
    \[\begin{array}{l}
    \mc{D}'_2 \mb{proc}(z_\delta[\langle d,g\rangle],  Q_{\mid \xi} @\langle d_1, g_1 \rangle) \mc{D}''_2 \mapsto^*\\
    \mc{D}'_2\,\mc{D}_3 \mb{proc}(z_\delta[\langle c,e\rangle], (w^{\langle d,g\rangle} \leftarrow X[\gamma] \leftarrow \Lambda) @\langle d_2, g_2 \rangle; Q'_{\mid \xi}) \mc{D}''_2
    \end{array}\]
    where every node in $\mc{D}_3$ is irrelevant, the process offering along $z_\delta[\langle c,e\rangle]$ in the post-step is semi-relavant, and $\Psi \Vdash [a_0/x]P' \sim_{\langle d_s, g_1 \sqcup g \rangle} [a_0/w]Q'$.
   Similar to Case 1. Subcase 2b, we apply induction on the size of $Q$ and inversion on the derivation for $\Psi \Vdash  (x^{\langle d,g\rangle} \leftarrow X[\gamma] \leftarrow \Lambda) @\langle d_2, g_2 \rangle; P'  \sim_{\langle d_s,g_1 \rangle} Q$ to show the above.

    Since $d_s \not \sqsubseteq g_1$, {\bf we know that the last rule in the derivation is one of the following three, i.e. either \msc{Brb-Spawn}, or \msc{Brb-Unsync-Spawn$_2$}, or \msc{Brb-Unsync$_2$}.} The proof in the last two cases are identical to the corresponding cases in Case 1. Subcase 2b. We only provide the proof for the first case:
    \begin{itemize}
        \item In {\bf the first case} (\msc{Brb-Spawn}), we get exactly what we want, i.e. process $Q$ is of the form,
        $(w^{\langle d,g\rangle} \leftarrow X[\gamma] \leftarrow \Lambda) @\langle d_2, g_2 \rangle; Q'_{\mid \xi},$
        and
        \[\mb{proc}(z_\delta[\langle d,g\rangle],  Q_{\mid \xi} @\langle d_1, g_1 \rangle) \mapsto^0 \mb{proc}(z_\delta[\langle d,g\rangle],  (w^{\langle d,g\rangle} \leftarrow X[\gamma] \leftarrow \Lambda) @\langle d_2, g_2 \rangle; Q'_{\mid \xi}@\langle d_1, g_1 \rangle).\]
        Moreover,  we have $\Psi \Vdash [a_0/x]P' \sim_{\langle d_s, g_1 \sqcup g \rangle} [a_0/w]Q'$.
        \item {\bf The second case (\msc{Brb-Unsync-Spawn$_2$})}. see Case 1. Subcase 2b.
        \item {\bf The third case (\msc{Brb-Unsync$_2$})}. 
        see Case 1. Subcase 2b.
    \end{itemize}

   Thus we can step the second configuration as
    \[\begin{array}{l}
        \mc{D}_2 \mapsto^*\\ \mc{D}^1_2\mc{D}_3\,\mb{proc}(z_\delta[\langle d,g\rangle],  (w^{\langle d,g\rangle} \leftarrow X[\gamma] \leftarrow \Lambda) @\langle d_2, g_2 \rangle; Q'_{\mid \xi}@\langle d_1, g_1 \rangle)\mc{D}''_2\mapsto\\
        \mc{D}^1_2\mc{D}_3 \,\mb{proc}(a_0[\langle d, g\rangle], \subst{a^{\langle d, g\rangle}_0, \Lambda}{w^{\langle d, g\rangle}, \Lambda'} \hat{\gamma}\,P @\langle d_2, g_2\rangle), 
        \mb{proc}(y_\alpha[\langle c,e \rangle], [a^{\langle d, g\rangle}_0/w^{\langle d, g\rangle}] Q' @\langle d_1, g_1 \rangle)\,\mc{D}''_2
    \end{array}\]
    (we assume that the fresh channel being spawned will be $a$ in the second run too. Moreover, the (unique) substitution
    $\hat{\gamma}$ is the same when stepping both runs.)

    We know that $\Psi \Vdash [a_0/x]P' \sim_{\langle d_s, g_1 \rangle} [a_0/w]Q'$. 

    There are three sub-cases for the spawned processes in the two runs: they can be both relevant, or they are both semi-relevant, or they are both irrelevant. In the first case the spawned processes are related by the conditions of Definition~\ref{def:sim} trivially. In the second case, we use Lemma~\ref{lem:sim-refl} to show that the spawned processes are related by the conditions of Definition~\ref{def:sim}. In the third case, the spawned processes are both discarded by the projections.

   We may introduce dead-end nodes in the post-step of the second configuration, but this only means we need to establish less conditions.
    
    It is straightforward to see that the rest of the nodes continue to be reltaed by Definition~\ref{def:sim}.

\end{description}

    \item {\color{Aquamarine}\bf Subcase 3. $\mb{proc}(y_\alpha[\langle c,e \rangle], (x^{\langle d,g\rangle} \leftarrow X[\gamma] \leftarrow \Lambda) @\langle d_2, g_2 \rangle; P' @\langle d_1, g_1 \rangle)$ is  relevant.}
     %%edit
    By assumption  we have
    \[\begin{array}{l}
    \mc{D}_2= \mc{D}'_2 
    \mb{proc}(z_\delta[\langle c,e\rangle], (w^{\langle d,g\rangle} \leftarrow X[\gamma] \leftarrow \Lambda) @\langle d_2, g_2 \rangle; P'_{\mid \xi}) \mc{D}''_2
    \end{array}\]

   Thus we can step the second configuration as
   {\small \[\begin{array}{l}
        \mc{D}_2 \mapsto\\
        \mc{D}'_2\mc{D}_3 \,\mb{proc}(a_0[\langle d, g\rangle], \subst{a^{\langle d, g\rangle}_0, \Lambda}{w^{\langle d, g\rangle}, \Lambda'} \hat{\gamma}\,P @\langle d_2, g_2\rangle), 
        \mb{proc}(y_\alpha[\langle c,e \rangle], [a^{\langle d, g\rangle}_0/w^{\langle d, g\rangle}] P' @\langle d_1, g_1 \rangle)\,\mc{D}''_2
    \end{array}\]}
    (we assume that the fresh channel being spawned will be $a$ in the second run too. Moreover, the (unique) substitution
    $\hat{\gamma}$ is the same when stepping both runs.)

    There are three sub-cases for the spawned processes in the two runs: they can be both relevant, or they are both semi-relevant, or they are both irrelevant. In the first case the spawned processes are related by the conditions of Definition~\ref{def:sim} trivially. In the second case, we use Lemma~\ref{lem:sim-refl} to show that the spawned processes are related by the conditions of Definition~\ref{def:sim}. In the third case, the spawned processes are both discarded by the projections.

   We may introduce dead-end nodes in the post-step of the second configuration, but this only means we need to establish less conditions.
    
    It is straightforward to see that the rest of the nodes continue to be reltaed by Definition~\ref{def:sim}.
    %done
\end{description}

\end{description}

\end{proof}
\subsection{Auxiliary lemmas}
% Psi; Delta |-Sig P@<c0, e0> :: x : A<c, e>
\begin{lemma}\label{lem:sim-refl}
    For every $\ptyp{\Delta}{P}{c_0}{e_0}{x}{A}{c}{e}$ and confidentiality level $d_s \sqsubseteq c_0$, we have $\Psi \Vdash P \sim_{\langle d_s,e_0 \rangle} P$. 
\end{lemma}
\begin{proof}
    The proof is by induction on the typing derivation for $P$. We consider the last step in the derivation. The only interesting case is when $P$ is a case, i.e., $P=\mathbf{case}\, w^{\langle d_2, e_2 \rangle}(\ell \Rightarrow P_\ell)_{\ell \in L}$. To establish the induction hypothesis, we need to know that for all $i,j \in L$ we have $\Psi \Vdash P_i \sim_{\langle d_s,e_0 \sqcup e_2 \rangle} P_j$. By inversion on the typing derivation, we know that  $\Psi \Vdash P_i \sim_{\langle d_2 \sqcup c_0,e_0 \sqcup e_2\rangle} P_j$, where $d_s \sqsubseteq d_2 \sqcup c_0$. By Lemma~\ref{lem:sim-down}, we have $\Psi \Vdash P_i \sim_{\langle d_s,e_0 \sqcup e_2 \rangle} P_j$.
\end{proof}
\begin{lemma}\label{lem:sim-down}
    If $\Psi \Vdash P \sim_{\langle d_1,e_0 \rangle} Q$ and $\Psi \Vdash d_2 \sqsubseteq d_1$, then  $\Psi \Vdash P \sim_{\langle d_2,e_0 \rangle} Q$.
\end{lemma}
\begin{proof}
    The proof is by induction on the derivation for $\Psi \Vdash P \sim_{\langle d_1,e_0 \rangle} Q$. 
    We look at the condition for $\Psi \Vdash P \sim_{\langle d_2,e_0 \rangle} Q$, if the last step is a synchronized step, we can apply the induction hypothesis and apply the same rule on the result. If not, an assumption of the form $\Psi \Vdash d_1 \sqsubseteq e'$ holds in the given derivation. By $\Psi \Vdash d_2 \sqsubseteq d_1$, we get $\Psi \Vdash d_2 \sqsubseteq e'$; thus the assumption holds for $d_2$ too and we can apply the induction hypothesis and apply the same rule to build a derivation for $\Psi \Vdash P \sim_{\langle d_2,e_0 \rangle} Q$.
\end{proof}
\pagebreak
\section{Logical relation}\label{apx:sec:relation}

\begin{definition}
    The configuration $\Delta \Vdash \mc{C}:: \Delta'$ is ready to send along $\Lambda \subseteq \Delta, \Delta'$ iff for all $y_\alpha{:}C \in \Lambda$, there is a message $\mathbf{msg}(M)$ in $\mc{C}$ sending along $y_\alpha$, i.e. $M$ is of the form $y_\alpha.k$, $\mathbf{send}\,x_\beta\, y_\alpha$, or $\mathbf{close}\, y_\alpha$.

    \noindent
    Similarly, the configuration $\Delta \Vdash \mc{C}:: \Delta'$ is ready to receive along $\Lambda \subseteq \Delta, \Delta'$ iff for all $y_\alpha{:}C \in \Lambda$, there is a process $\mathbf{proc}(z_\delta, P)$ in $\mc{C}$ waiting to receive along $y_\alpha$, i.e. $P$ is of the form $\mathbf{case}\,y_\alpha(\ell \Rightarrow Q_\ell)_{\ell \in I}$, $w \leftarrow \mathbf{recv} y_\alpha;Q$, or $\mathbf{wait}\,y_\alpha; Q$.

    \end{definition}
\begin{definition}
    The set $\mathbf{Out}(\Delta \Vdash K)$, is defined as all channels with the sending semantics in $\Delta, K$, i.e., $y_\alpha \in \mathbf{Out}(\Delta \Vdash K)$ iff for some positive type $T$, we have $y_\alpha{:}T \in \Delta$ or for some negative type $T$, we have $y_\alpha{:}T \in K$.

\noindent
    Similarly, the set $\mathbf{In}(\Delta \Vdash K)$, is defined as all channels with the receiving semantics in $\Delta, K$, i.e., $y_\alpha \in \mathbf{Out}(\Delta \Vdash K)$ iff for some negative type $T$, we have $y_\alpha{:}T \in \Delta$ or for some positive type $T$, we have $y_\alpha{:}T \in K$.

\end{definition} 

\begin{definition}
    $\mathit{dom}(\Delta)$ is a set defined inductively as:
    \[\begin{array}{lll}
        \mathit{dom}(\Delta, y_{\alpha}:A) &=& \mathit{dom}(\Delta) \cup \{y\}\\
        \mathit{dom}(\cdot) &= & \emptyset
    \end{array}  
    \]
    %$\mathbf{chnl}(\mc{D}_1,\mc{D}_2)$ is the (infinite) set of channels that are fresh in both $\mc{D}_1$ and $\mc{D}_2$, where freshness is defined as the channel (or any other generation of it) does not occur in the free or bound names of $\mc{D}_1$ and $\mc{D}_2$ (and in the set that is used in the spawn when those channels are the offering channel of a spawner).  

\end{definition}
    \begin{definition}[Open configuration transitions]\label{apx:def:config-transitions}
    We provide some notations used in the logical relation and proofs.
    
    \begin{itemize}
    \item The notation $\mapsto^*$ refers to taking none or many steps with $\mapsto$. The notation $\mapsto^j$ refers to taking $j$ steps with $\mapsto$. 
    
    \item We write $\mc{D}\mapsto^{{*}_{\Upsilon}}\mc{D}'$  stating that $\mc{D} \mapsto^*\mc{D}'$ and $\mc{D}'$ is ready to send along $\Upsilon$.
        
    \item We write $\mc{D}\mapsto^{{*}_{\Upsilon; \Theta}}\mc{D}'$  stating that $\mc{D} \mapsto^*\mc{D}'$ and $\mc{D}'$ is ready to send along $\Upsilon$ and ready to receive along $\Theta$.
   
\end{itemize}
 
    \end{definition}

    Next, we define the projection function ${\Downarrow^{\m{ig}}} \xi$ to build the integrity interface, e.g., $\Delta {\Downarrow^{\m{ig}}} \xi$ projects out the channels in $\Delta$ that have lower or incomparable integrity than $\xi$. 

\begin{definition}[Typing context projections -integrity]\label{apx:def:typing-proj}
  Downward projection on linear contexts and $K$ is defined as follows:
  \[\begin{array}{lclc}
  \Delta, x_\alpha{:}T\langle c, e \rangle \Downarrow^{\m{ig}} \xi& :=& \Delta \Downarrow^{\m{ig}}\xi,  x_\alpha{:}T\langle c, e \rangle  & \m{if}\; e\sqsubseteq \xi  \\
       \Delta, x_\alpha{:}T\langle c, e \rangle  \Downarrow^{\m{ig}}\xi& :=& \Delta \Downarrow^{\m{ig}} \xi & \m{if}\;e \not \sqsubseteq \xi  \\
      \cdot \Downarrow^{\m{ig}} \xi& :=& \cdot &  \\
       x_\alpha{:}T\langle c, e \rangle  \Downarrow^{\m{ig}} \xi& :=&   x_\alpha{:}T\langle c, e \rangle  & \m{if}\; e\sqsubseteq \xi  \\
      x_\alpha{:}T\langle c, e \rangle  \Downarrow^{\m{ig}} \xi& :=& \_{:}1\langle \top, \top \rangle & \m{if}\; e\not\sqsubseteq \xi \\
  \end{array}\]

  \end{definition}\label{apx:def:highprovider}
  
  \begin{definition}[Low integrity provider and Low integrity client]
     \[\begin{array}{lll}
       \cdot \,\,\in \mb{L\text{-}IProvider}^\xi(\cdot) \\[5pt]
       \mc{B} \in \mb{L\text{-}IProvider}^{\xi}(\Delta, x_\alpha{:}A\langle c, e \rangle )\,  & \m{iff} & e \not \sqsubseteq \xi  \, \m{and}\,
       \mc{B} = \mc{B}' \mc{T} \,\, \m{and}\,\, \mc{B}' \in \mb{L\text{-}IProvider}^\xi(\Delta)\,\,\m{and}\\
       && \,\, \mc{T} \in \m{Tree}(\cdot \Vdash x_\alpha{:}A\langle c, e \rangle), \mb{or}\\
       && e\sqsubseteq \xi  \, \m{and}\,
      \mc{B} \in \mb{L\text{-}IProvider}^\xi(\Delta)\\[10pt]
   
       \mc{T} \in \mb{L\text{-}IClient}^\xi( x_\alpha{:}A\langle c, e \rangle )\,  & \m{iff} & e \not \sqsubseteq \xi  \, \m{and}\, \mc{T} \in \m{Tree}(x_\alpha{:}A\langle c, e \rangle \Vdash \_:1\langle \top, \top\rangle), \mb{or}\\
       &&e \sqsubseteq \xi  \, \m{and}\, \mc{T}= \cdot
     \end{array}\]
 
   \end{definition}

Two open configurations are  in the integrity relation if they are related by the term interpretation for every possible pair of substitutions for their low integrity channels and any number of observations:

  \begin{definition}[Equivalence of trees by the logical relation upto the observer level]\label{apx:def:noninterference}
  We define the relation \[(\Delta_1 \Vdash \mc{D}_1:: x_\alpha {:}A_1\langle c_1, e_1 \rangle) \equiv^{\Psi_0}_{\xi} (\Delta_2 \Vdash \mc{D}_2:: y_\beta {:}A_2\langle c_2,e_2 \rangle)\,\,\m{as}\]
  \[
  \begin{array}{lll}
    {\mc{D}_1} \in \m{Tree}({\Delta_1} \Vdash x_\alpha {:}A_1\langle c_1, e_1 \rangle)\,\, \m{and}\,\,{\mc{D}_2} \in \m{Tree}({\Delta_2} \Vdash y_\beta {:}A_2\langle c_2, e_2 \rangle)\,\, \m{and}\\
    \Delta_1 \Downarrow^{\m{ig}} \xi= \Delta_2 \Downarrow^{\m{ig}} \xi=\Delta\,\,\m{and}\,\, x_\alpha {:}A_1\langle c_1, e_1 \rangle \Downarrow^{\m{ig}} \xi= y_\beta {:}A_2\langle c_2, e_2 \rangle \Downarrow^{\m{ig}} \xi= K\,\,\m{and}\\
    \forall \, \mc{B}_1 \in \mb{L\text{-}IProvider}^\xi(\Delta_1).\, \forall \, \mc{B}_2 \in \mb{L\text{-}IProvider}^\xi(\Delta_2).\\
    \,\forall \mc{T}_1 \in \mb{L\text{-}IClient}^\xi(x_\alpha {:}A_1\langle c_1, e_1 \rangle). \,\forall \mc{T}_2 \in \mb{L\text{-}IClient}^\xi(y_\beta {:}A_2\langle c_2, e_2 \rangle).
  \end{array}  
  \]
  
  \[\begin{array}{l}  \forall\,m.\,(\mc{B}_1\mc{D}_1\mc{T}_1, \mc{B}_2\mc{D}_2\mc{T}_2) \in \mc{E}^\xi_{\Psi_0}\llbracket{\Delta} \Vdash {K} \rrbracket^{m},\,
  \m{and} \\[1pt]
  \forall\,m.\,(\mc{B}_2\mc{D}_2\mc{T}_2, \mc{B}_1\mc{D}_1\mc{T}_1) \in \mc{E}^\xi_{\Psi_0}\llbracket {\Delta} \Vdash {K} \rrbracket^{m}.
  \end{array}
  \]
 
  \end{definition}

% \begin{lemma}[Build a Meganode]\label{lem:buildmn}
%     Consider $\mc{C}_1\mc{D}_1{\mc{F}_1}$, and ${\mc{C}_2}{\mc{D}_2}{\mc{F}_2}$ such that  $\Psi; \cdot \Vdash \mc{C}_1 :: \Delta_1$ and 
%     $\Psi; \cdot \Vdash \mc{C}_2 :: \Delta_2$, and $\Psi; \Delta_1 \Vdash \mc{D}_1 :: x_{\alpha}{:}A[c]$, and
%     $\Psi; \Delta_2 \Vdash \mc{D}_2 :: y_{\beta}{:}B[d]$, and $\Psi; x_{\alpha}{:}A[c]\Vdash \mc{F}_1 :: \cdot$, and
%     $\Psi; y_{\beta}{:}B[d] \Vdash \mc{F}_2 :: \cdot$.
    
%     If $\Delta_1 {\Downarrow}^{\m{ig}} \xi=\Delta_1 {\Downarrow}^{\m{ig}} \xi=\Delta$ and $x_{\alpha}{:}A\langle c, e \rangle {\Downarrow}^{\m{ig}}{\xi}=y_{\beta}{:}B\langle d, g \rangle {\Downarrow}^{\m{ig}}{\xi}=K$, then we can rewrite  $\mc{C}_i\mc{D}_i\mc{F}_i$ as
%     $\mc{C}'_i\mc{D}'_i\mc{F}'_i$    such that \[{(\lr{\mc{C}'_1}{\mc{D}'_1}{\mc{F}'_1}; \lr{\mc{C}'_2}{\mc{D}'_2}{\mc{F}'_2}) \in \m{Tree}_\Psi(\Delta \Vdash K)}.\] Moreover, if $\cproj{\mc{D}_1}{\xi}=\cproj{\mc{D}_2}{\xi}$, then $\cproj{\mc{D}'_1}{\xi}=\cproj{\mc{D}'_2}{\xi}$.
%     \end{lemma}
% \begin{proof}
%  See \cite{DerakhshanLICS2021, DerakhshanBalzerARXIV2022} for the proof. 
% \end{proof}
\section{Fundamental theorem}
\begin{theorem}[Fundamental Theorem]\label{thm:ap_main}
    For all security levels $\xi$, and \[{\Psi; \Delta_1 \Vdash \mathcal{D}_1:: u_\alpha {:}T_1\langle c_1, e_1 \rangle},\quad \m{and} \quad {\Psi; \Delta_2 \Vdash \mathcal{D}_2:: v_\beta {:}T_2\langle c_2, e_2 \rangle}\]  with $\cproj{\mathcal{D}_1}{\xi} \meq \cproj{\mathcal{D}_2}{\xi}$, $\Delta_1 \Downarrow^{\m{ig}} \xi = \Delta_2 \Downarrow^{\m{ig}} \xi$, and $u_\alpha {:}T_1\langle c_1, e_1 \rangle\Downarrow^{\m{ig}} \xi = v_\beta {:}T_2\langle c_2,e_2 \rangle \Downarrow^{\m{ig}} \xi$  we have \[(\Delta_1 \Vdash \mathcal{D}_1:: u_\alpha {:}T_1\langle c_1, e_1 \rangle)  \equiv^\Psi_{\xi} (\Delta_2 \Vdash \mathcal{D}_2:: v_\beta {:}T_2\langle c_1, e_1 \rangle).\]  
    \end{theorem}

    \begin{proof}

        Our goal is to show:
        \begin{quote}
        For all well-typed {\small$\mathcal{D}_1$} and {\small$\mathcal{D}_2$}, i.e.,{\small${\Psi; \Delta_1 \Vdash \mathcal{D}_1:: u_\alpha {:}T_1\langle c_1, e_1 \rangle}$} and {\small${\Psi; \Delta_2 \Vdash \mathcal{D}_2:: v_\beta {:}T_2\langle c_2, e_2 \rangle}$},
         with {$\cproj{\mathcal{D}_1}{\xi} \meq \cproj{\mathcal{D}_2}{\xi}$}, and {$\Delta_1 \Downarrow^{\m{ig}} \xi = \Delta_2 \Downarrow^{\m{ig}} \xi= \Delta$}, and $u_\alpha {:}T_1\langle c_1, e_1 \rangle \Downarrow^{\m{ig}} \xi = v_\beta {:}T_2\langle c_2, e_2 \rangle \Downarrow^{\m{ig}} \xi= K$ we have 
         {\small\(\forall \mc{B}_1 \in \mb{L\text{-}IProvider}^\xi(\Delta_1).\, \forall \, \mc{B}_2 \in \mb{L\text{-}IProvider}^\xi(\Delta_2). \,\forall \mc{T}_1 \in \mb{L\text{-}IClient}^\xi(x_\alpha {:}A_1\langle c_1, e_1\rangle). \,\forall \mc{T}_2 \in \mb{L\text{-}IClient}^\xi(y_\beta {:}A_2\langle c_1, e_1\rangle).\)}
        { \[ \forall\,m.\,(\mc{B}_1{\mathcal{D}_1}\mc{T}_1, \mc{B}_2{\mathcal{D}}_2\mc{T}_2) \in \mc{E}^\xi_{\Psi_0}\llbracket {\Delta} \Vdash {K} \rrbracket^{m},\,
        \m{and}\, \,
        \forall\,m.\,(\mc{B}_2{\mathcal{D}}_2\mc{T}_2, \mc{B}_1{\mathcal{D}}_1\mc{T}_1) \in \mc{E}^\xi_{\Psi_0}\llbracket {\Delta} \Vdash {K} \rrbracket^{m}.\]}
        \end{quote}
        
        Without loss of generality, we only prove one part of this symmetric relation, i.e.:
        \begin{quote}
            For all well-typed {\small$\mathcal{D}_1$} and {\small$\mathcal{D}_2$}, i.e.,{\small${\Psi; \Delta_1 \Vdash \mathcal{D}_1:: u_\alpha {:}T_1\langle c_1, e_1 \rangle}$} and {\small${\Psi; \Delta_2 \Vdash \mathcal{D}_2:: v_\beta {:}T_2\langle c_2, e_2 \rangle}$},
             with {$\cproj{\mathcal{D}_1}{\xi} \meq \cproj{\mathcal{D}_2}{\xi}$}, and {$\Delta_1 \Downarrow^{\m{ig}} \xi = \Delta_2 \Downarrow^{\m{ig}} \xi= \Delta$}, and $u_\alpha {:}T_1\langle c_1, e_1 \rangle \Downarrow^{\m{ig}} \xi = v_\beta {:}T_2\langle c_2, e_2 \rangle \Downarrow^{\m{ig}} \xi= K$ we have 
             {\small$ \forall \mc{B}_1 \in \mb{L\text{-}IProvider}^\xi(\Delta_1).\, \forall \, \mc{B}_2 \in \mb{L\text{-}IProvider}^\xi(\Delta_2). \,\forall \mc{T}_1 \in \mb{L\text{-}IClient}^\xi(x_\alpha {:}A_1\langle c_1, e_1\rangle). \,\forall \mc{T}_2 \in \mb{L\text{-}IClient}^\xi(y_\beta {:}A_2\langle c_1, e_1\rangle).$}
            \[ \forall\,m.\,(\mc{B}_1{\mathcal{D}_1}\mc{T}_1, \mc{B}_2{\mathcal{D}}_2\mc{T}_2) \in \mc{E}^\xi_{\Psi_0}\llbracket {\Delta} \Vdash {K} \rrbracket^{m},\,\]
          \end{quote}

        The goal is equivalent to:
        \begin{quote}
          $(\star_1)\,\,$For all $m$, for all $\mathcal{D}_1$ and $\mathcal{D}_2$ that are IFC-typed, i.e.,{\small ${\Psi; \Delta_1 \Vdash \mathcal{D}_1:: u_\alpha {:}T_1\langle c_1, e_1\rangle}$} and {\small${\Psi; \Delta_2 \Vdash \mathcal{D}_2:: v_\beta {:}T_2\langle c_2, e_2\rangle}$},  with $\cproj{\mathcal{D}_1}{\xi} \meq \cproj{\mathcal{D}_2}{\xi}$, and $\Delta_1 \Downarrow \xi = \Delta_2 \Downarrow \xi= \Delta$, and $u_\alpha {:}T_1\langle c_1, e_1\rangle\Downarrow \xi = v_\beta {:}T_2\langle c_2, e_2\rangle \Downarrow \xi= K$ we have 
          $ \forall \mc{B}_1 \in \mb{L\text{-}IProvider}^\xi(\Delta_1).\, \forall \, \mc{B}_2 \in \mb{L\text{-}IProvider}^\xi(\Delta_2). \,\forall \mc{T}_1 \in \mb{L\text{-}IClient}^\xi(x_\alpha {:}A_1[{c_1}]). \,\forall \mc{T}_2 \in \mb{L\text{-}IClient}^\xi(y_\beta {:}A_2[{c_2}]).$
         { \[ (\mc{B}_1{\mathcal{D}}_1\mc{T}_1, \mc{B}_2{\mathcal{D}}_2\mc{T}_2) \in \mc{E}^\xi_{\Psi_0}\llbracket {\Delta} \Vdash {K} \rrbracket^{m}\]}
          \end{quote}

         First, we show that it is enough to prove the following goal.
        
        \begin{quote}
          $(\star_2)\,\,$For all $m$, for all $\mathcal{D}'_1$ and $\mathcal{D}'_2$ that are well-typed, i.e.,${\Psi; \Delta \Vdash \mathcal{D}'_1:: K}$ and ${\Psi; \Delta \Vdash \mathcal{D}'_2:: K}$,  with $\cproj{\mathcal{D}'_1}{\xi} \meq \cproj{\mathcal{D}'_2}{\xi}$, and $\Delta \Downarrow \xi = \Delta$, and $K\Downarrow \xi = K$ we have 
         { \[ ({\mathcal{D}'_1}, {\mathcal{D}'_2}) \in \mc{E}^\xi_{\Psi_0}\llbracket {\Delta} \Vdash {K} \rrbracket^{m}\]}
          \end{quote}
          We need to show that from $\star_1$, we get $\star_2$:
          \begin{quote}
          Apply a for all introduction on $\star_2$ to get an arbitrary number $m$, configurations $\mathcal{D}_1$ and $\mathcal{D}_2$, and high providers $\mc{B}_1$ and $\mc{B}_2$ and high clients $\mc{T}_1$ and $\mc{T}_2$ satisfying the given assumptions.
        
        We define $\mathcal{D}'_1= \mathcal{B}_1 \mathcal{D}_1 \mathcal{T}_1$ and $\mathcal{D}'_2= \mathcal{B}_2 \mathcal{D}_2 \mathcal{T}_2$. Note that these two configurations are both well-typed, and also we have 
        $\Psi; \Delta \Vdash \mathcal{D}'_1:: K$ and $\Psi; \Delta \Vdash \mathcal{D}'_2:: K$. We also know that $\mathcal{D}'_1 \Downarrow \xi = \mathcal{D}'_2 \Downarrow \xi$. Observe that all relevant and semi-relevant nodes occur in $\mathcal{D}_1$ and $\mathcal{D}_2$, for which by assumption we know $\mathcal{D}_1 \Downarrow \xi = \mathcal{D}_2 \Downarrow \xi$. Now we can apply $\star_2$ to get what we want.
        \end{quote}
        It remains to prove $(\star_2)$:
        
        \begin{quote}
          $(\star_2)\,\,$For all $m$, for all $\mathcal{D}_1$ and $\mathcal{D}_2$ that are IFC-typed, i.e.,${\Psi; \Delta \Vdash \mathcal{D}_1:: K}$ and ${\Psi; \Delta \Vdash \mathcal{D}_2:: K}$,  with $\cproj{\mathcal{D}_1}{\xi} \meq \cproj{\mathcal{D}_2}{\xi}$, and $\Delta \Downarrow \xi = \Delta$, and $K\Downarrow \xi = K$ we have 
         { \[ ({\mathcal{D}_1}; {\mathcal{D}_2}) \in \mc{E}^\xi_{\Psi_0}\llbracket {\Delta} \Vdash {K} \rrbracket^{m}\]}
        \end{quote}

        The proof is by induction on the index $m$.
        
        {\bf Base case. $m=0$.} We consider arbitrary configurations ($\forall I$ on the goal). By the configuration typing, we know that 
        \[ ({\mathcal{D}_1} ; {\mathcal{D}_2} ) \in\m{Tree}({\Delta} \Vdash {K}),\]
        which is enough to complete the proof.

        {\bf Inductive case. $m=m'+1$.}
        
        % We consider arbitrary configurations $\mc{B}_i$ and $\mc{T}_i$ that satisfy the conditions on the goal ($\forall I$ on the goal). By the configuration typing, we know 
        \[ (\mathcal{B}_1{\mathcal{D}_1} \mathcal{T}_1;\mathcal{B}_2 {\mathcal{D}_2} \mathcal{T}_2) \in\m{Tree}({\Delta} \Vdash {K}).\]
        
        % \fdtodo{give and proof the missing lemma}
        
        Consider an arbitrary $\mathcal{D}'_1$ such that $ {\mathcal{D}_1} \mapsto^{{*}_{\Upsilon; \Theta}} {\mathcal{D}'_1}$. 
        By Preservation we know that $\mathcal{D}'_1$ is IFC-typed for the same interface ${\Psi; \Delta \Vdash \mathcal{D}'_1:: K}$ 
        By \Cref{lem:indinvariant1}, for some $\mathcal{D}'_2$, we get ${\mathcal{D}_2} \mapsto^{{*}_{\Upsilon}} {\mathcal{D}'_2} $, such that ${\Psi; \Delta \Vdash \mathcal{D}'_2:: K}$ with $\mathcal{D}'_2\Downarrow \xi = \mathcal{D}'_2 \Downarrow \xi$.
        % \[
        % \begin{array}{lll}
        %   (\mathcal{B}_1{\mathcal{D}_1} \mathcal{T}_1;\mathcal{B}_2 {\mathcal{D}_2} \mathcal{T}_2) \in\m{Tree}({\Delta} \Vdash {K}) & \quad 
        %   {\mathcal{D}_1} \in\m{Tree}({\Delta_1} \Vdash {u_\alpha{:}T_1\langle c_1, e_1\rangle})
        %   & \quad
        %   {\mathcal{D}_2} \in\m{Tree}({\Delta_2} \Vdash {v_\beta{:}T_2\langle c_2, e_2\rangle})
        % \end{array}  
        % \]
        
        We need to show that \noindent 
        \[
        \begin{array}{l}
        \;(\dagger_1) \; \forall u'_\delta \in \mathbf{Out}({\Delta} \Vdash {K}).\, \mathbf{if}\, u'_\delta \in \Upsilon.\, \mathbf{then}\, ({\mathcal{D}'_1} ; {\mathcal{D}'_2} ) \in \mathcal{V}\llbracket {\Delta} \Vdash {K} \rrbracket_{\cdot;u'_\delta}^{m}\; \mathbf{and}\\
        \;(\dagger_2) \;\forall u'_\delta \in \mathbf{In}({\Delta} \Vdash {K}).\,\mathbf{if}\, u'_\delta \in \Theta.\, \m{then}\,  ({\mathcal{D}'_1} ; {\mathcal{D}'_2} ) \in \mathcal{V}\llbracket {\Delta} \Vdash {K}\rrbracket_{u'_\delta;\cdot}^{m} 
        \end{array}
        \]
        
        We prove both $\dagger_1$, and $\dagger_2$. We consider an arbitrary $u'_\delta$ and provide a case analysis based on its type. If there is no such $u'_\delta$ in the specified set, then the statements trivially holds. There might be a channel in both of these sets, if and only if $\mathcal{D}_1= \cdot$.

        \begin{description}
        \item[\bf Proof of $\dagger_1$.] Consider an arbitrary channel $u'_\delta{:}T \in \mathbf{Out}(\Delta \Vdash K)$ and assume $u'_\delta{:}T \in \Upsilon$. We consider different cases based on the session type $T$.
    
        \begin{description}
        \item [\bf Case 1.] {\color{RedViolet} $u'{:}T\langle c', e'\rangle=K=u_\alpha{:}T_1\langle c_1, e_1\rangle=u_\alpha{:}1\langle c_1, e_1\rangle$}, and $\mc{D}'_1= \mc{D}''_1\msg{\mb{close}\,u^{c_1}_\alpha}$. By the typing rules we know that $\mc{D}''_1=\cdot$, and $\mc{C}'_1= \cdot$, and $\Delta= \cdot$.
        Note that by the definition of relevancy, $e_1 \sqsubseteq \xi$ and the tree invariant, all nodes in $\mc{D}'_1$ and $\mc{D}'_2$ are either relevant or semi-relevant.
        %
        %\fdtodo{maybe add more detail for the semi-relevant case? make it a general note, for two related messages on the interface, they have to be synchronized pattern of sends.}
        As a result, $\mathcal{D}'_{2}=\mathbf{msg}(\mb{close}\,u^{\langle c_1, e_1\rangle}_\alpha),$ and by the typing rules $\mc{C}'_2=\cdot$.

        By row $r_1$ in the definition, we get 
        {\small\[\dagger_1 (\mc{D}'_1; \mc{D}'_2) \in \mc{V}^\xi_{\Psi_0}\llbracket \cdot \Vdash u_\alpha{:}1[\langle c_1, e_1\rangle]\rrbracket_{\cdot;u'_\delta}^{m},\]} 
        
        \item[\bf Case 2.]
        {\color{RedViolet}
        \[\begin{array}{l}
            u'{:}T[\langle c', e'\rangle]= K=u_\alpha{:}T_1[\langle c_1, e_1\rangle]=u_\alpha{:}\oplus\{\ell{:}A_{\ell}\}_{\ell \in I}[\langle c_1, e_1\rangle],\, \m{and}\\
            \mathcal{D}'_{1}=\mc{D}''_1\mathbf{msg}(u_\alpha^{\langle c_1, e_1\rangle}.k)
        \end{array}\]}

        By the assumption of the theorem, every node in $\mc{D}'_1$ and $\mc{D}'_2$ are either relevant or semi-relevant. In particular, $\mathbf{msg}(u_\alpha^{\langle c_1, e_1\rangle}.k)$ is either relevant or semi-relavant. 
        \begin{description}
            \item[Subcase 1.]  $\mathbf{msg}(u_\alpha^{\langle c_1, e_1\rangle}.k)$ is relevant.
            
            We have $\mathcal{D}'_{2}=\mc{D}''_2\mathbf{msg}(u_\alpha^{\langle c_1, e_1\rangle}.k).$ 

            Removing $\mb{msg}(u^{\langle c_1, e_1\rangle}_\alpha.k; u^{\langle c_1, e_1\rangle}_\alpha)$ from $\mc{D}'_1$ and $\mc{D}'_2$ does not change relevancy of the remaining configuration when $e_1 \sqsubseteq \xi$: \[\cproj{\mc{D}''_1}{\xi}=\cproj{\mc{D}''_2}{\xi}.\]
            
            We can apply the induction hypothesis on the index $m' < m$ to get 
            \[\begin{array}{ll}
            {{}\star}\,\,&  (\mc{D}''_1 ; \mathcal{D}''_2)\in \mathcal{E}^\xi_\Psi\llbracket \Delta \Vdash u_{\alpha+1}{:}A_k\langle c_1, e_1\rangle\rrbracket^{m'}.
            \end{array}\]
            
            By row $r_2$ in the definition of $\mc{V}^\xi_{\Psi_0}$:
               \[\begin{array}{ll}
            \dagger_1 &  (\mc{D}''_1 \msg{u^{\langle c_1, e_1\rangle}_\alpha.k};\mathcal{D}''_2 \msg{u^{\langle c_1, e_1\rangle}_\alpha.k})\in \mc{V}^\xi_{\Psi_0}\llbracket \Delta \Vdash u_\alpha{:}\oplus\{\ell:A_{\ell}\}_{\ell \in I}\langle c_1, e_1\rangle\rrbracket_{\cdot;u'_\delta}^{m'+1}.
            \end{array}
            \]
        \item[Subcase 2.]  $\mathbf{msg}(u_\alpha^{\langle c_1, e_1\rangle}.k)$ is semi-relevant, i.e. $e_1 \sqsubseteq \xi$, but $c_1 \not \sqsubseteq \xi$. We get $c_1 \not \sqsubseteq \xi$, and by the typing rules, we know $\forall i,j \in I. \, A_i=A_j=A'$, i.e. $u_\alpha{:} \oplus\{\ell{:}A'\}_{\ell \in I}[\langle c_1, e_1\rangle]$

        By the definition, we have $\mathcal{D}'_{2}=\mc{D}''_2\mathbf{msg}(u_\alpha^{\langle c_1, e_1\rangle}.j),$ for some $j \in I$.
        Removing $\mb{msg}(u^{\langle c_1, e_1\rangle}_\alpha.k)$ and $\mb{msg}(u^{\langle c_1, e_1\rangle}_\alpha.j)$ from $\mc{D}'_1$ and $\mc{D}'_2$, respectively, does not change relevancy of the remaining configuration when $e_1 \sqsubseteq \xi$: \[\cproj{\mc{D}''_1}{\xi}=\cproj{\mc{D}''_2}{\xi}.\]
        
            We can apply the induction hypothesis on the index $m' < m$ to get 
            \[\begin{array}{ll}
            {{}\star}\,\,&  (\mc{D}''_1; \mathcal{D}''_2) \in \mathcal{E}^\xi_\Psi\llbracket \Delta \Vdash u_{\alpha+1}{:}A'\langle c_1, e_1\rangle\rrbracket^{m'}.
            \end{array}\]
            
            By row $r_2$ in the definition of $\mc{V}^\xi_{\Psi_0}$:
               \[\begin{array}{ll}
            \dagger_1 &  (\mc{D}''_1 \msg{u^{\langle c_1, e_1\rangle}_\alpha.k} ;\mathcal{D}''_2 \msg{u^{\langle c_1, e_1\rangle}_\alpha.k})\in \mc{V}^\xi_{\Psi_0}\llbracket \Delta \Vdash u_\alpha{:}\oplus\{\ell:A'\}_{\ell \in I}\langle c_1, e_1\rangle\rrbracket_{\cdot;u'_\delta}^{m'+1}.
            \end{array}
            \]
        \end{description}

        \item [\bf Case 3.] 
        {\color{RedViolet}
        \[\begin{array}{l}
            u'{:}T[\langle c', e'\rangle]=K=u_\alpha{:}T_1[
            \langle c_1, e_1\rangle]=u_\alpha{:} (A \otimes B)[\langle c_1, e_1\rangle],\, \m{and}\\
            \mathcal{D}'_1=\mathcal{D}''_1\mc{T}_1\mathbf{msg}(\mb{send}x_\beta^{c_1}\,u_\alpha^{
            \langle c_1, e_1\rangle}),\, \m{where}\\
           \Delta=\Lambda_1, \Lambda_2$ and $\Psi; \Lambda_2\Vdash \mc{T}_1::(x_\beta{:}A[\langle c_1, e_1\rangle])
        \end{array}\]}

            By assumption of the theorem, we know every node in $\mc{D}'_1$ and $\mc{D}'_2$ is either relevant or semi-relevant. In particular, $\mathbf{msg}(\mb{send}x_\beta^{c_1}\,u_\alpha^{\langle c_1, e_1\rangle})$ is either relevant or semi-relevant. 
            %
            %\fdtodo{add more about semi-relevant}
            In both cases, 
            \[{\mathcal{D}'_{2}=\mc{D}''_2\mc{T}_2\mathbf{msg}(\mb{send}x_\beta^{\langle c_1, e_1\rangle}\, u_\alpha^{\langle c_1, e_1\rangle})}\] such that $\Psi; \Lambda_2\Vdash \mc{T}_2::(x_\beta{:}A\langle c_1, e_1\rangle)$.Moreover, channels $x_\beta^{\langle c_1, e_1\rangle}$ and $u_{\alpha+1}^{\langle c_1, e_1\rangle}$ are relevant and we get $\cproj{\mc{D}''_1}{\xi}=\cproj{\mc{D}''_2}{\xi}$ and 
            $\cproj{\mc{T}_1}{\xi}=\cproj{\mc{T}_2}{\xi}$. 
            %And by  we know that $\mc{T}_1=\mc{T}_2=\mc{T}$ and thus $\cproj{\mc{T}_1}{\xi}=\cproj{\mc{T}_2}{\xi}$.

        In the case where $x_\beta{:}A\langle c_1, e_1\rangle \in \Delta$, we have $\mc{T}_i=\cdot$ and $\Lambda_2=x_\beta{:}A[\langle c_1, e_1\rangle]$.

         We can apply the induction hypothesis on the index $m'<m$ to get 
        
        \[\begin{array}{ll}
        {\star'}\, &  (\mc{T}_1;\mathcal{T}_2) \in \mathcal{E}^\xi_\Psi\llbracket \Lambda_2 \Vdash x_{\beta}{:}A[\langle c_1, e_1\rangle]\rrbracket^{m'}.
        \end{array}\]
        and
        \[\begin{array}{ll}
        {\star''}\, &  (\mc{D}''_1; \mathcal{D}''_2) \in \mathcal{E}^\xi_\Psi\llbracket \Lambda_1 \Vdash u_{\alpha+1}{:}B[\langle c_1, e_1\rangle]\rrbracket^{m'}.
        \end{array}\]
        
        By row $r_4$ in the definition of $\mc{V}^\xi_{\Psi_0}$:
        {\small  
        \[\begin{array}{ll}
            \dagger_1 \, &  (\mc{D}''_1 \mc{T}_1\msg{\mb{send}x_\beta^{\langle c_1, e_1\rangle}\,u_\alpha^{\langle c_1, e_1\rangle}};\mathcal{D}''_2 \mc{T}_2\msg{\mb{send}x_\beta^{\langle c_1, e_1\rangle}\,u_\alpha^{\langle c_1, e_1\rangle}}) \in \mc{V}^\xi_{\Psi_0}\llbracket \Delta \Vdash u_\alpha{:}A \otimes B[\langle c_1, e_1\rangle]\rrbracket_{\cdot;u'_\delta}^{m'+1}.
        \end{array}\]}

        \item[\bf Case 4.]
        {\color{RedViolet} 
        \[\begin{array}{l}
            \Delta=\Delta', x_\gamma{:}\&\{\ell:A_\ell\}_{\ell \in L}\langle c,e \rangle,\, \m{and}\\ u'{:}T'\langle c',e' \rangle=x_\gamma{:}\&\{\ell:A_\ell\}_{\ell \in L}\langle c,e \rangle,\, \m{and}\\ 
            \mc{D}'_1= \msg{x^{\langle c,e \rangle}_\gamma.k}\mc{D}''_1
        \end{array}
        \]}

        By the assumption of the theorem, we know that $\msg{x^{\langle c,e \rangle}_\gamma.k}$ is either relevant or semi-relevant. 
        \begin{description}
        \item[\bf Subcase 1.] $\msg{x^{\langle c,e \rangle}_\gamma.k}$ is relevant.
            
             By definition, $\mathcal{D}'_{2}=\msg{x^{\langle c,e \rangle}_\gamma.k}\mc{D}''_2,$  and $\cproj{\mc{D}''_1}{\xi}=\cproj{\mc{D}'_2}{\xi}$: $e\sqsubseteq \xi$ and channel $u^{\langle c,e \rangle}_{\alpha+1}$ is relevant in $\mc{D}_i$ and no relevancy changes in the configurations after removing the negative message.
         
            We can apply the induction hypothesis on the smaller index $m'<m$ to get 
            \[\begin{array}{ll}
             {\star}\,& (\mc{D}''_1; \mathcal{D}''_2) \in \mathcal{E}^\xi_\Psi\llbracket \Delta', u_{\alpha+1}{:}A_k\langle c,e \rangle\Vdash K\rrbracket^{m'}.
            \end{array}\]

            By line ($l_3$) in the definition of $\mc{V}^\xi_{\Psi_0}$:
            \[\begin{array}{ll}
            \dagger_1\, &  (\msg{x^{\langle c,e \rangle}_\gamma.k} \mc{D}''_1 ; \msg{x^{\langle c,e \rangle}_\gamma.k} \mathcal{D}''_2)\in \mc{V}^\xi_{\Psi_0}\llbracket \Delta', x_{\gamma}{:}\&\{\ell{:}A_\ell\}_{\ell\in L}[\langle c,e \rangle]\Vdash K \rrbracket_{\cdot;u'_\delta}^{m'+1}.
            \end{array}
            \]
         
% \fdtodo{add renaming stuff from popl}

        \item[\bf Subcase 2.] $\msg{x^{\langle c,e \rangle}_\gamma.k}$ is semi-relevant, i.e. $e \sqsubseteq \xi$ and  $c \not \sqsubseteq \xi$. We get $c \not \sqsubseteq e$. By inversion on the typing rules, we get $\forall i,j \in L. A_i=A_j=A'$, i.e., $x_\gamma{:}\&\{\ell:A'\}_{\ell \in L}\langle c,e \rangle$.
        %

       % \fdtodo{more detail on semi-relevant}
        By definition, $\mathcal{D}'_{2}=\msg{x^{\langle c,e \rangle}_\gamma.j}\mc{D}''_2,$  and $\cproj{\mc{D}''_1}{\xi}=\cproj{\mc{D}'_2}{\xi}$: $e\sqsubseteq \xi$ and channel $u^{\langle c,e \rangle}_{\alpha+1}$ is relevant in $\mc{D}_i$ and no relevancy changes in the configurations after removing the negative message.
         
        We can apply the induction hypothesis on the smaller index $m'<m$ to get 
        \[\begin{array}{ll}
         {\star}\,& (\mc{D}''_1; \mathcal{D}''_2) \in \mathcal{E}^\xi_\Psi\llbracket \Delta', u_{\alpha+1}{:}A'\langle c,e \rangle\Vdash K\rrbracket^{m'}.
        \end{array}\]

        By row $l_3$ in the definition of $\mc{V}^\xi_{\Psi_0}$:
        \[\begin{array}{ll}
        \dagger_1\, &  (\msg{x^{\langle c,e \rangle}_\gamma.k} \mc{D}''_1 ; \msg{x^{\langle c,e \rangle}_\gamma.j}\mathcal{D}''_2)\in \mc{V}^\xi_{\Psi_0}\llbracket \Delta', x_{\gamma}{:}\&\{\ell{:}A'\}_{\ell\in L}[\langle c,e \rangle]\Vdash K \rrbracket_{\cdot;u'_\delta}^{m'+1}.
        \end{array}
        \]

        \item[\bf Case 5.]
         {\color{RedViolet}
         \[
        \begin{array}{l}
            \Delta=\Delta', \Delta'', x_\gamma{:}(A\multimap B)[\langle c,e \rangle],\, \m{and}\\
            u'{:}T'\langle c', e' \rangle=x_\gamma{:}(A \multimap B)[\langle c,e \rangle],\, \m{and}\\
             \mc{D}'_1= \mc{T}_1\msg{\mb{send} y^{\langle c,e \rangle}_\delta\,x^{\langle c,e \rangle}_\gamma}\mc{D}''_1,\, \m{such\, that}\\ 
             \ctype{\Psi}{\Delta''}{\mc{T}_1}{y_\delta{:}A[\langle c,e \rangle]}\, \m{and}\\ 
             \ctype{\Psi}{\Delta',u_{\alpha+1}{:}B[\langle c,e \rangle]}{\mc{D}''_1}{K}.
        \end{array}
         \]}

          The message $\msg{\mb{send} y^{\langle c,e \rangle}_\delta\,x^{\langle c,e \rangle}_\gamma}$ is either relevant or semi-relevant in $\mc{D}'_1$. 

          %\fdtodo{explain semi-relavant.}

          In both cases, $\mathcal{D}'_{2}=\mc{T}_2\msg{\mb{send} y^{\langle c,e \rangle}_\delta\,x^{\langle c,e \rangle}_\gamma},$  such that $\ctype{\Psi}{\Delta''}{\mc{T}_2}{y_\delta{:}A[c]}$ and $\ctype{\Psi}{\Delta',u_{\alpha+1}{:}B[c]}{\mc{D}''_2}{K}$.

         By typing configuration, we have $\mc{C}'_i=\mc{A}_i \mc{A}'_i$ such that $\cdot \Vdash \mc{A}_i :: \Delta'$ and $\cdot \Vdash \mc{A}'_i :: \Delta'', x_\gamma {:} (A \multimap B)[\langle c, e \rangle]$.

         By assumption we know that  $\cproj{\mc{D}'_1}{\xi}=\cproj{\mc{D}'_2}{\xi}$. By definition of relevancy and the tree invariant, we know that every node in the tree $\mc{T}_1$ is either relevant or semi-relevant in $\mc{D}_1$ and thus $\cproj{\mc{T}_1}{\xi} = \cproj{\mc{T}_2}{\xi}$. Removing these from both configurations does not change relevancy of the rest of the configuration since $w^{\langle c,e \rangle}_{\eta}$ will remain relevant, and the relevancy of the message's parent does not change: $\cproj{\mc{D}''_1}{\xi}=\cproj{\mc{D}''_2}{\xi}$.
         
        We can apply the induction hypothesis on the smaller index $m'<m$ to get 
         \[\begin{array}{ll}
        { \star\,}&  (\mc{T}_1;\mc{T}_2)\in \mathcal{E}^\xi_\Psi\llbracket \Delta'' \Vdash y_\delta{:}A\langle c, e \rangle\rrbracket^{m'}\, \m{and}
        \end{array}\]
    
        \[\begin{array}{ll}
        { \star'\,} & (\mc{D}''_1;\mathcal{D}''_2)\in \mathcal{E}^\xi_\Psi\llbracket \Delta', u_{\alpha+1}{:}B\langle c, e \rangle\Vdash K\rrbracket^{m'}.
        \end{array}\]

        By row $l_5$ in the definition of $\mc{V}^\xi_{\Psi_0}$:
   {\small \[\begin{array}{ll}
         \dagger_1 &  (\mc{T}_1\msg{\mb{send} y^{\langle c, e \rangle}_\delta\,x^{\langle c, e \rangle}_\gamma} \mc{D}''_1 ; \mc{T}_2\msg{\mb{send} y^{\langle c, e \rangle}_\delta\,x^{\langle c, e \rangle}_\gamma} \mathcal{D}''_2) \in \mc{V}^\xi_{\Psi_0}\llbracket \Delta'_1, x_{\gamma}{:}A \multimap B[\langle c, e \rangle]\Vdash K \rrbracket_{\cdot; u'}^{m'+1}.
        \end{array}
        \]}
        
        \end{description}

        \end{description}
   
        \item[\bf Proof of $\dagger_2$.] Consider an arbitrary channel $u':T[\langle c', e' \rangle ] \in \Theta_1 \cap \Theta_2$. We consider cases based on the type of this channel, and whether there is a message or a forwarding process ready to be executed along it in the first configuration; (it cannot be both message and forward).
      
        \begin{description}
        \item[\bf Case 1.] 
        {\color{RedViolet} \[
            \begin{array}{l}
                u'{:}T\langle c', e' \rangle= K=u_\alpha{:}T_1\langle c_1, e_1 \rangle =u_\alpha{:}\&\{\ell{:}A_{\ell}\}_{\ell \in I}\langle c_1, e_1 \rangle
            \end{array}
        \]}
        
        We consider two subcases:
        \begin{description}
         \item[\bf Subcase 1.] $c_1, e_1\sqsubseteq \xi$. 
         
         We use row $r_3$ in the definition of $\mc{V}^\xi_{\Psi_0}$ to prove the result.
         
        %  \fdtodo{add missing references}
        From $\cproj{\mc{D}'_1}{\xi}=\cproj{\mc{D}'_2}{\xi}$ and  $e_1 \sqsubseteq \xi$, we get \[\cproj{\mc{D}'_1\msg{u^{\langle c_1, e_1 \rangle}_\alpha.k}}{\xi}=\cproj{\mc{D}'_2\msg{u^{\langle c_1, e_1 \rangle}_\alpha.k}}{\xi}.\]
        We can apply the induction hypothesis on the index $m' < m$.
        By induction hypothesis
          \[\begin{array}{ll}
            {\star} & (\mc{D}'_1 \msg{u^{\langle c_1, e_1 \rangle}_\alpha.k};\mathcal{D}'_2 \msg{u^{\langle c_1, e_1 \rangle}_\alpha.k}) \in \mathcal{E}^\xi_\Psi\llbracket \Delta \Vdash u_{\alpha+1}{:}A_k\langle c_1, e_1 \rangle\rrbracket^{m'}.
         \end{array}\]
         By line ($r_3$) in the definition of $\mc{V}^\xi_{\Psi_0}$:
         { \small
         \[\begin{array}{ll}
         \dagger_2&  (\mc{D}'_1 ;\mathcal{D}'_2) \in \mc{V}^\xi_{\Psi_0}\llbracket \Delta \Vdash u_\alpha{:}\&\{\ell:A_{\ell}\}_{\ell \in I}\langle c_1, e_1 \rangle\rrbracket_{u'_\delta; \cdot}^{m'+1}.
         \end{array}\]}

         \item[\bf Subcase 2.] $e_1\sqsubseteq \xi$ and $c_1 \not \sqsubseteq \xi$.

         By inversion on the typing rules, we know that $\forall i,j \in I.\, A_i=A_j=A'$. 
       
         We use $(r_3)$ in the definition of $\mc{V}^\xi_{\Psi_0}$ to prove the result.

        %\fdtodo{more on semi-relevant} 
        From $\cproj{\mc{D}'_1}{\xi}=\cproj{\mc{D}'_2}{\xi}$ and  $e_1 \sqsubseteq \xi$, we get \[\cproj{\mc{D}'_1\msg{u^{\langle c_1, e_1 \rangle}_\alpha.k}}{\xi}=\cproj{\mc{D}'_2\msg{u^{\langle c_1, e_1 \rangle}_\alpha.j}}{\xi}.\]
        We can apply the induction hypothesis on the index $m' < m$.
        By induction hypothesis
          \[\begin{array}{ll}
            {\star} & (\mc{D}'_1 \msg{u^{\langle c_1, e_1 \rangle}_\alpha.k};\mathcal{D}'_2 \msg{u^{\langle c_1, e_1 \rangle}_\alpha.j}) \in \mathcal{E}^\xi_\Psi\llbracket \Delta \Vdash u_{\alpha+1}{:}A'\langle c_1, e_1 \rangle\rrbracket^{m'}.
         \end{array}\]

         By row $r_3$ in the definition of $\mc{V}^\xi_{\Psi_0}$:
         { \small
         \[\begin{array}{ll}
         \dagger_2&  (\mc{D}'_1 ;\mathcal{D}'_2) \in \mc{V}^\xi_{\Psi_0}\llbracket \Delta \Vdash u_\alpha{:}\&\{\ell:A'\}_{\ell \in I}\langle c_1, e_1 \rangle\rrbracket_{u'_\delta; \cdot}^{m'+1}.
         \end{array}\]}

        \end{description}
        \item[\bf Case 2.]{\color{RedViolet}
        \[\begin{array}{l}
            u'_\delta{:}T\langle c', e' \rangle=K=u_\alpha{:}T_1\langle c_1, e_1 \rangle=u_\alpha{:}A \multimap B\langle c_1, e_1 \rangle
             \end{array}\]}
       
        We use row $r_5$ in the definition of $\mc{V}^\xi_{\Psi_0}$ to prove the result. 

        By assumption and $e_1 \sqsubseteq \xi$, we know that every node in $\mc{D}'_1$ and $\mc{D}'_2$ is either relevant or semi-relevant. Adding a negative message along $e_1 \sqsubseteq \xi$ as the root does not change their relevancy: 
        \[\begin{array}{l}
            \cproj{\mc{D}'_1\msg{\mb{send} x^{\langle c_1, e_1 \rangle}_{\beta}\,u^{\langle c_1, e_1 \rangle}_{\alpha}}}{\xi}=  \cproj{\mc{D}'_2\msg{\mb{send} x^{\langle c_1, e_1 \rangle}_{\beta}\,u^{\langle c_1, e_1 \rangle}_{\alpha}}}{\xi} 
        \end{array}
        \]

        We apply the induction hypothesis on $m' < m$:
       {\small \[\begin{array}{ll}
            {\star}&
                 (\mc{D}'_1 \msg{\mb{send} x^{\langle c_1, e_1 \rangle}_{\beta}\,u^{\langle c_1, e_1 \rangle}_{\alpha}} ; \mathcal{D}'_2 \msg{\mb{send} x^{\langle c_1, e_1 \rangle}_{\beta}\,u^{\langle c_1, e_1 \rangle}_{\alpha}}) \in \mathcal{E}^\xi_\Psi\llbracket \Delta, x_\beta{:}A\langle c_1, e_1 \rangle\Vdash u_{\alpha+1}{:}B\langle c_1, e_1 \rangle\rrbracket^{m'}.
        \end{array}\]}
        
        By row $r_5$ in the definition of $\mc{V}^\xi_{\Psi_0}$:
       {\small \[\begin{array}{ll}
        \dagger_2 &  (\mc{D}'_1; \mathcal{D}'_2)\in \mc{V}^\xi_{\Psi_0}\llbracket \Delta \Vdash u_\alpha{:}A\multimap B\langle c_1, e_1 \rangle\rrbracket_{u'; \cdot}^{m'+1}.
        \end{array}\]}

        \item[\bf Case 3.]
        {\color{RedViolet}
        \[\begin{array}{l}
        \Delta=\Delta', x_\gamma{:}1\langle c, e \rangle.
        \end{array}\]}

        We use row $l_1$ in the definition of $\mc{V}^\xi_{\Psi_0}$ to prove the result. 
        Assume $\mathcal{C}'_{2}= \mc{C}''_2\mathbf{msg}(\mb{close}\,x^{\langle c, e \rangle}_\gamma)$.

        We first briefly explain why the invariant of induction holds after we bring the closing message inside $\mc{D}'_i$ and remove the channel $x_\gamma^{\langle c, e \rangle}$ from $\Delta$, i.e.
        \[\cproj{\msg{\mb{close}\,x^{\langle c, e \rangle}_\gamma}\mc{D}'_1}{\xi}= \cproj{\msg{\mb{close}\,x^{\langle c, e \rangle}_\gamma}\mc{D}'_2}{\xi}. \]
        
        If the parent of $\msg{\mb{close}\,x^{\langle c, e \rangle}_\gamma}$ in $\mc{D}'_1$ is relevant or semi-relevant in  $\mc{D}'_1$ before bringing the message inside then by the assumption of the theorem, there is a related node as the parent of $\msg{\mb{close}\,x^{\langle c, e \rangle}_\gamma}$ in $\mc{D}'_2$.
        \begin{itemize}
            \item If after adding the message to $\mc{D}'_1$ the parent still remains relevant or semi-relevant, it means that it has at least another relevant channel other than $x^{\langle c, e \rangle}_\gamma$; that channel also exists in $\mc{D}'_2$ and will be relevant after adding the message to $\mc{D}'_2$.
            \item If after adding the message, the parent in $\mc{D}'_1$ becomes irrelevant, it means that it does not have a relevant path to any other channel in $\Delta'$ and $K$.
            In this case, if the the parent in $\mc{D}'_2$ does not become irrelevant, it means that it has a relevant path to another channel in $\Delta'$ and $K$. But by the assumption of theorem, this only is possible if this path goes over a dead-end node in $\mc{D}'_2$, meaning that the parent is dead-end.
            The same argument holds for any other node that becomes irrelevant because of adding the closing message to $\mc{D}'_i$. 
        \end{itemize}

        If the parent of $\msg{\mb{close}\,x^{\langle c, e \rangle}_\gamma}$ in $\mc{D}'_1$ is irrelevant in  $\mc{D}'_1$ before bringing the message inside then by the assumption of the theorem, the the parent of $\msg{\mb{close}\,x^{\langle c, e \rangle}_\gamma}$ in $\mc{D}'_2$ has to be either irrelevant or a dead-end node before bringing in the message; the parent remains irrelevant or dead-end after bringing the message inside too.
        
        Now that the invariant holds, we can apply the induction hypothesis on the smaller index $m'<m$: 
        \[{ \star}\; (\msg{\mb{close}\,x^{\langle c, e \rangle}_\gamma}\mc{D}'_1; \msg{\mb{close}\,x^{\langle c, e \rangle}_\gamma}\mc{D}'_2) \in \mathcal{E}^\xi_\Psi\llbracket \Delta' \Vdash K\rrbracket^{m'}.\]
        
        By row $l_1$ in the definition of $\mc{V}^\xi_{\Psi_0}$,  
        { \[\dagger_2 \,(\mc{D}'_1; \mc{D}'_2) \in \mc{V}^\xi_{\Psi_0}\llbracket \Delta', x_\gamma{:}1\langle c, e \rangle \Vdash K\rrbracket_{u'; \cdot}^{m'+1}\]}

        \item [\bf Case 4.] {\color{RedViolet} 
        \[\begin{array}{l}
            \Delta=\Delta', x_\gamma{:}\oplus\{\ell:A_\ell\}_{\ell \in L}\langle c, e \rangle\, \m{and}\,
            u'{:}T'\langle c', e' \rangle=x_\gamma{:}\oplus\{\ell:A_\ell\}_{\ell \in L}[\langle c, e \rangle]
        \end{array}\]}

         We consider two subcases:
         \begin{description}
            \item[\bf Subcase 1.] $c, e \sqsubseteq \xi$.
            We use line ($l_2$) in definition of $\mc{V}^\xi_{\Psi_0}$ to prove the result.
    
            We have 
         \[\cproj{\msg{x^{\langle c, e \rangle}_\gamma.k}\mc{D}'_1}{\xi}=\cproj{\msg{x^{\langle c, e \rangle}_\gamma.k}\mc{D}'_2}{\xi}.\] 
         The channel $u^{\langle c, e \rangle}_{\alpha+1}$
         is relevant and the positive messages $\msg{x^{\langle c, e \rangle}_\gamma.k}$ in both configurations are relevant only if their parents are either relevant or semi-relevant.  If the parent is either relevant or semi-relevant in the first configuration, the parent is relevant or semi-relevant in the second configuration too, and the messages are both relevant. 
         If the parent is irrelevant in the first configuration, the parent has to be either irrelevant or a dead-end node in the second configuration; the message in the first configuration is irrelevant and in the second configuration is either irrelevant or dead-end.
         Adding the message does not change relevancy of any other process.

        We can apply the induction hypothesis on the smaller index $m'<m$ to get
         \[\begin{array}{ll}
        {\star\;} &  (\msg{x^{\langle c, e \rangle}_\gamma.k}\mc{D}'_1;\msg{x^{\langle c, e \rangle}_\gamma.k}\mathcal{D}'_2 )  \in \mathcal{E}^\xi_\Psi\llbracket \Delta', u_{\alpha+1}:A_k\langle c, e \rangle\Vdash K\rrbracket^{m'}.
        \end{array}\]

        By row $l_2$ in the definition of $\mc{V}^\xi_{\Psi_0}$:
        {
        \[\begin{array}{ll}
        \dagger_2&  (\mc{D}'_1; \mathcal{D}'_2)\in \mc{V}^\xi_{\Psi_0}\llbracket \Delta', x_\gamma{:}\oplus\{\ell:A_{\ell}\}_{\ell \in L}\langle c, e \rangle\Vdash K\rrbracket_{u'_\delta;\cdot}^{m'+1}.
        \end{array}\]}
        \item[\bf Subcase 2.] $e \sqsubseteq \xi$  and $c \not \sqsubseteq \xi$.
        We know $c \not \sqsubseteq \xi$, and thus $\forall i,j \in L.\, A_i=A_j=A'$. 
      
        We use  row $l_2$ in definition of $\mc{V}^\xi_{\Psi_0}$ to prove the result.

        We have 
     \[\cproj{\msg{x^{\langle c, e \rangle}_\gamma.k}\mc{D}'_1}{\xi}=\cproj{\msg{x^{\langle c, e \rangle}_\gamma.j}\mc{D}'_2}{\xi}.\] 
     The channel $u^{\langle c, e \rangle}_{\alpha+1}$
     is relevant and the positive messages $\msg{x^{\langle c, e \rangle}_\gamma.k}$ in both configurations are semi-relevant only if their parents are either relevant or semi-relevant.
    If the parent is either relevant or semi-relevant in the first configuration, the parent is relevant or semi-relevant in the second configuration too, and the messages are both semi-relevant. 
     If the parent is irrelevant in the first configuration, the parent has to be either irrelevant or a dead-end node in the second configuration; the message in the first configuration is irrelevant and in the second configuration is either irrelevant or dead-end.
     Adding the message does not change relevancy of any other process.

    We can apply the induction hypothesis on the smaller index $m'<m$ to get
     \[\begin{array}{ll}
    {\star\;} &  (\msg{x^{\langle c, e \rangle}_\gamma.k}\mc{D}'_1; \msg{x^{\langle c, e \rangle}_\gamma.j}) \in \mathcal{E}^\xi_\Psi\llbracket \Delta', u_{\alpha+1}:A'\langle c, e \rangle\Vdash K\rrbracket^{m'}.
    \end{array}\]
    
    By row $l_2$ in the definition of $\mc{V}^\xi_{\Psi_0}$:
    {
    \[\begin{array}{ll}
    \dagger_2&  (\mc{D}'_1; \mathcal{D}'_2)\in \mc{V}^\xi_{\Psi_0}\llbracket \Delta', x_\gamma{:}\oplus\{\ell:A'\}_{\ell \in L}\langle c, e \rangle\Vdash K\rrbracket_{u';\cdot}^{m'+1}.
    \end{array}\]}
         \end{description}
        
    \item[\bf Case 5.]
       {\color{RedViolet}
       \[\begin{array}{ll}
        \Delta=\Delta', x_\gamma{:}(A \otimes B)\langle c,e \rangle\,\m{and}\, u'_\delta{:}T'\langle c', e' \rangle= x_\gamma{:}(A \otimes B)\langle c,e \rangle
       \end{array}\]
       }

        We use line $l_4$ in the definition of $\mc{V}^\xi_{\Psi_0}$ to prove this case.

        We have 
        \[\begin{array}{l}
            \msg{\mb{send}\,y_\delta^{\langle c,e \rangle}\,x_\gamma^{\langle c, e \rangle}}\cproj{\mc{D}'_1}{\xi}\meq\\
            \cproj{\msg{\mb{send}\,y_\delta^{\langle c,e \rangle}\,x_\gamma^{\langle c,e \rangle}}\mc{D}'_2}{\xi}.
        \end{array}\] 
        
        Channels $u^{\langle c, e \rangle}_{\alpha+1}$ and $y^{\langle c, e \rangle}_{\alpha}$
     are relevant and in both configurations.

      The positive messages $\msg{\mb{send}y_\delta^{\langle c,e \rangle}\,x_\gamma^{\langle c,e \rangle}}$ are (semi-)relevant only if their parents are either relevant or semi-relevant.
    If the parent is either relevant or semi-relevant in the first configuration, the parent is relevant or semi-relevant in the second configuration too, and the messages are both (semi-)relevant. 
     If the parent is irrelevant in the first configuration, the parent has to be either irrelevant or a dead-end node in the second configuration; the message in the first configuration is irrelevant and in the second configuration is either irrelevant or dead-end.
     Adding the message does not change relevancy of any other process.

    We can apply the induction hypothesis on  the smaller index $m'<m$ to get 
    {\small \[\begin{array}{ll}
            {\star}& (\msg{\mb{send}y_\delta^{\langle c,e \rangle}\,x_\gamma^{\langle c,e \rangle}}\mc{D}'_1; \msg{\mb{send}y_\delta^{\langle c,e \rangle}\,x_\gamma^{\langle c,e \rangle}}\mathcal{D}'_2) \in  \mathcal{E}^\xi_\Psi\llbracket \Delta',  y_\delta{:}A\langle c,e \rangle, u_{\alpha+1}{:}B\langle c,e \rangle\Vdash K\rrbracket^{m'}.
    \end{array}\]}
        
    By row $l_4$ in the definition of $\mc{V}^\xi_{\Psi_0}$:
    {\small
    \[\begin{array}{ll}
    \dagger_2& (\mc{D}'_1; \mathcal{D}'_2) \in \mc{V}^\xi_{\Psi_0}\llbracket \Delta', x_\gamma{:}A \otimes B\langle c,e \rangle\Vdash K\rrbracket_{u'_\delta; \cdot}^{m'+1}.
    \end{array}\]}
  
\end{description}
        \end{description}
         
        \end{proof}
\section{Useful properties of the logical relation}
%%%Lemma 1%%% :done
\begin{lemma}[Building minimal sending configuration]\label{apx:lem:least}

    Consider {\small$\Delta \Vdash \mc{D}_2:: K$}, a set of channels {\small$\Upsilon_1 \subseteq \Delta, K$} and {\small$\mc{D}_2 \mapsto^{*_{\Upsilon_2}} \mc{D}'_2$} for some {\small$\Upsilon_2 \supseteq \Upsilon_1$}. There exists a set {\small$\Upsilon$} and a configuration {\small$\mc{D}''_2$} such that {\small$\Upsilon_1 \subseteq \Upsilon \subseteq \Upsilon_2$} and {\small$\mc{D}_2\mapsto^{*_{\Upsilon}} \mc{D}''_2$} and {\small$\forall \mc{D}^1_2, \Upsilon_3 \supseteq \Upsilon_1.\, \m{if} \,\mc{D}_2 \mapsto^{*_{\Upsilon_3}} \mc{D}^1_2\,\m{then}\,\mc{D}''_2 \mapsto^{*_{\Upsilon_3}} \mc{D}^1_2$}. We call {\small$\mc{D}''_2$} and {\small$\Upsilon$} the minimal sending configuration and the minimal sending set with respect to {\small$\Upsilon_1$} and {\small$\mc{D}_2$}, respectively. 
\end{lemma}
\begin{proof}
    A property of session-typed processes. Refer to~\cite{BalzerARXIV2023} for the proof and the algorithm.
\end{proof} 

\begin{lemma}[Backward closure on the second run]\label{lem:backsecondrun} The second run enjoys backward closure:
\begin{enumerate}
 \item If $(\mc{D}_1;\mc{D}_2) \in \mc{E}_{\Psi_0}^\xi\llbracket \Delta\Vdash K \rrbracket^k$ 
%  and $\mc{B}_i=\lre{\mc{C}_i}{ \mc{D}_i}{\mc{F}_i}$ with $(\lr{\mc{C}_1}{ \mc{D}_1}{\mc{F}_1};\lr{\mc{C}_2}{\mc{D}_2}{\mc{F}_2}) \in \m{Tree}( \Delta\Vdash K )$
 and for $\mc{D}''_2 \in \m{Tree}(\Delta \Vdash K)$, we have $\mc{D}''_2\mapsto^* \mc{D}_2$ then
 $(\mc{D}_1;\mc{D}''_2) \in \mc{E}_{\Psi_0}^\xi\llbracket \Delta\Vdash K \rrbracket^k$. 
% $(\lre{\mc{C}_1}{ \mc{D}_1}{\mc{F}_1};\lre{\mc{C}_2}{\mc{D}''_2}{\mc{F}_2}) \in \mc{E}_{\Psi_0}^\xi\llbracket \Delta\Vdash K \rrbracket^k.$

\item If $(\mc{D}_1;\mc{D}_2) \in \mc{V}_{\Psi_0}^\xi\llbracket \Delta\Vdash K \rrbracket_{\cdot;y_\alpha}^{k+1}$ and for $\mc{D}''_2 \in \m{Tree}(\Delta \Vdash K)$, we have $\mc{D}''_2\mapsto^* \mc{D}_2$ with $\mc{D}''_2$ sending along channel $y_\alpha$, 
%Because of tensor and lolli, we need backward closure on the substitutions of the second run too.
then
$(\mc{D}_1;\mc{D}''_2) \in \mc{V}_{\Psi_0}^\xi\llbracket \Delta\Vdash K \rrbracket_{\cdot;y_\alpha}^{k+1}.$ 
\item If $(\mc{D}_1;\mc{D}_2) \in \mc{V}_{\Psi_0}^\xi\llbracket \Delta\Vdash K \rrbracket_{y_\alpha;\cdot}^{k+1}$ and for $\mc{D}''_2 \in \m{Tree}(\Delta \Vdash K)$, we have $\mc{D}''_2\mapsto^* \mc{D}_2$, then
$(\mc{D}_1;\mc{D}''_2)\in \mc{V}_{\Psi_0}^\xi\llbracket \Delta\Vdash K \rrbracket_{y_\alpha; \cdot}^{k+1}.$
\end{enumerate}
\end{lemma}
\begin{proof}
We prove the first statement separately and then use it to prove the second and third statements.

%by a simultaneous lexicographic induction on $k$ and $\mc{V} < \mc{E}$.
\begin{enumerate}
 \item If $k=0$, the proof is trivial. Consider $k=m+1$.
 By assumption, we have $(\mc{D}_1;\mc{D}_2)  \in \mc{E}_{\Psi_0}^\xi\llbracket \Delta\Vdash K \rrbracket^{m+1}$.
 By the term interpretation~\ref{fig:rel_term} of the logical relation we know

\[\begin{array}{l}
(\star)\,\,\forall\, \Upsilon_1, \Theta_1, \mc{D}'_1. \,\m{if}\,  \mc{D}_1 \mapsto^{*_{ \Upsilon_1; \Theta_1}} \mc{D}'_1,\, \m{then}\, \exists \Upsilon_2, \mc{D}'_2 \,\m{such\, that} \, \,\mc{D}_2\mapsto^{{*}_{ \Upsilon_2}} \mc{D}'_2\, \m{and}\\ 
\; \; \forall\, x_\alpha \in \mb{Out}(\Delta \Vdash K).\, \mb{if}\, x_\alpha \in \Upsilon_1.\, \mb{then}\,\, (\mc{D}'_1; \mc{D}'_2)\in \mc{V}^\xi_{\Psi_0}\llbracket  \Delta \Vdash K \rrbracket_{\cdot;x_\alpha}^{m+1}\,\m{and}\\ 
\; \; \forall\, x_\alpha \in \mb{In}(\Delta \Vdash K).\, \mb{if}\, x_\alpha \in \Theta_1.\, \mb{then}\,\, (\mc{D}'_1; \mc{D}'_2)\in \mc{V}^\xi_{\Psi_0}\llbracket  \Delta \Vdash K \rrbracket_{x_\alpha;\cdot}^{m+1}.
\end{array}\]
Consider $\mc{D}''_2$, for which by the assumption we have $\mc{D}''_2\mapsto^* \mc{D}_2$. Our goal is to prove $( \mc{D}_1;\mc{D}''_2)  \in \mc{E}_{\Psi_0}^\xi\llbracket \Delta\Vdash K \rrbracket^{m+1}$.
We need to show that 
\[\begin{array}{l}
 (\dagger)\,\,\forall\, \Upsilon_1, \Theta_1, \mc{D}'_1.\, \m{if}\,  \mc{D}_1 \mapsto^{*_{ \Upsilon_1; \Theta_1}} \mc{D}'_1,\, \m{then}\, \exists \Upsilon_2, \mc{D}'_2 \,\m{such\, that} \, \,\mc{D}''_2\mapsto^{{*}_{ \Upsilon_2}} \mc{D}'_2\, \m{and}\, \Upsilon_1 \subseteq \Upsilon_2\, \m{and}\\ 
  \; \; \forall\, x_\alpha \in \mb{Out}(\Delta \Vdash K).\, \mb{if}\, x_\alpha \in \Upsilon_1.\, \mb{then}\,\, (\mc{D}'_1; \mc{D}'_2)\in \mc{V}^\xi_{\Psi_0}\llbracket  \Delta \Vdash K \rrbracket_{\cdot;x_\alpha}^{m+1}\,\m{and}\\
   \; \; \forall\, x_\alpha \in \mb{In}(\Delta \Vdash K).\,\mb{if}\, x_\alpha \in \Theta_1.\, \mb{then}\,\,  (\mc{D}'_1; \mc{D}'_2)\in \mc{V}^\xi_{\Psi_0}\llbracket  \Delta \Vdash K \rrbracket_{x_\alpha;\cdot}^{m+1}.
 \end{array}\]

With a $\forall$-Introduction, an $\m{if}$-Introduction on the goal followed by a $\forall$-Elimination and an $\m{if}$-Elimination on the assumption, we get the assumption 

\[\begin{array}{l}
 (\star')\,\,\exists \Upsilon_2, \mc{D}'_2 \,\m{such\, that} \, \,\mc{D}_2\mapsto^{{*}_{ \Upsilon_2}} \mc{D}'_2\, \m{and}\\ 
  \; \; \forall\, x_\alpha \in \mb{Out}(\Delta \Vdash K).\, \mb{if}\, x_\alpha \in \Upsilon_1.\, \mb{then}\,\, (\mc{D}'_1; \mc{D}'_2)\in \mc{V}^\xi_{\Psi_0}\llbracket  \Delta \Vdash K \rrbracket_{\cdot;x_\alpha}^{m+1}\,\m{and}\\
   \; \; \forall\, x_\alpha \in \mb{In}(\Delta \Vdash K).\,\mb{if}\, x_\alpha \in \Theta_1.\, \mb{then}\,\,  (\mc{D}'_1; \mc{D}'_2)\in \mc{V}^\xi_{\Psi_0}\llbracket  \Delta \Vdash K \rrbracket_{x_\alpha;\cdot}^{m+1}.
 \end{array}\]

and the goal 
\[\begin{array}{l}
 (\dagger')\,\,\exists \Upsilon_2,\,\mc{D}'_2 \,\m{such\, that} \, \,\mc{D}''_2\mapsto^{{*}_{ \Upsilon_2}} \mc{D}'_2\, \m{and}\\ 
  \; \; \forall\, x_\alpha \in \mb{Out}(\Delta \Vdash K).\, \mb{if}\, x_\alpha \in \Upsilon_1.\, \mb{then}\,\, (\mc{D}'_1; \mc{D}'_2)\in \mc{V}^\xi_{\Psi_0}\llbracket  \Delta \Vdash K \rrbracket_{\cdot;x_\alpha}^{m+1}\,\m{and}\\
   \; \; \forall\, x_\alpha \in \mb{In}(\Delta \Vdash K).\,\mb{if}\, x_\alpha \in \Theta_1.\, \mb{then}\,\,  (\mc{D}'_1; \mc{D}'_2)\in \mc{V}^\xi_{\Psi_0}\llbracket  \Delta \Vdash K \rrbracket_{x_\alpha;\cdot}^{m+1}.
 \end{array}\]
We apply an $\exists$-Elimination on the assumption to get $\Upsilon_2$ and $\mc{D}'_2$ that satisfies the conditions and use the same $\Upsilon_2$ and $\mc{D}'_2$ to instantiate the existential quantifier in the goal. Since $\mc{D}''_2\mapsto^* \mc{D}_2$, we get $\mc{D}''_2\mapsto^{{*}_{\Upsilon_2}} \mc{D}'_2$, and the proof is complete.
\item We proceed the proof by cases on the row of the logical relation that makes the assumption $(\mc{D}_1;\mc{D}_2) \in \mc{V}\llbracket \Delta\Vdash K \rrbracket_{\cdot;y_\alpha}^{k+1}$ correct.
Here we only consider two interesting cases, the proof of other cases is similar.

\begin{description}
    \item[\bf Row $r_2$.]
    By the conditions of this row, we know that $K=y_\alpha{:}\oplus \{\ell{:}A_\ell\}_{\ell \in I}\langle c,e \rangle$. 

    We have 
    \[ \begin{array}{ll}
       & \mathcal{D}_1=\mathcal{D}'_1\mathbf{msg}(y_\alpha^{\langle c,e \rangle}.k)\, \m{and}\,  \mathcal{D}_2=\mathcal{D}'_2 \mathbf{msg}(y_\alpha^{\langle c,e \rangle}.j)\, \m{and}\\
        \dagger_1&  (\mc{D}'_1; \mc{D}'_2)  \in  \mathcal{E}^\xi_{ \Psi}\llbracket \Delta \Vdash y_{\alpha+1}{:}A_k\langle c,e \rangle\rrbracket^{k}
    \end{array}\]
    where $j=k$ if $c \sqsubseteq \xi$.

    These are also the statements that we need to prove when replacing $\mc{D}_2$ with $\mc{D}''_2$.
    By the assumption that $\mc{D}''_2$ also sends along $y$,  uniqueness of channels, and $\mc{D}''_2 \mapsto^* \mc{D}_2$ we get 
    \[\mathcal{D}''_2=\mathcal{D}^1_2 \mathbf{msg}(y_\alpha^{\langle c,e \rangle}.j)\]
    Moreover, $\mc{D}^1_2 \mapsto^* \mc{D}'_2$.
    Note that if $c \not \sqsubseteq \xi$, then $c \not \sqsubseteq e$, then $A_k=A_j$.
    
    Now we can apply the the first item of this lemma on $\dagger_1$ to get
    \[\begin{array}{ll}
        \dagger''_2 &  (\mc{D}'_1; \mc{D}^1_2)  \in  \mathcal{E}^\xi_{ \Psi}\llbracket \Delta \Vdash y_{\alpha+1}{:}A_k\langle c,e \rangle\rrbracket^{k}
    \end{array} 
    \]
    and the proof of this subcase is complete.

\item[\bf Row $r_4$.]
By the conditions of this row, we know that $\Delta=\Delta',\Delta''$, and $K=y_\alpha{:}A\otimes B\langle c,e \rangle$. 

We have 
\[ \begin{array}{ll}
   & \mathcal{D}_1=\mathcal{D}'_1\mathcal{T}_1\mathbf{msg}( \mb{send}x_\beta^{\langle c,e \rangle}\,y_\alpha^{\langle c,e \rangle})\, \m{for}\, \mc{T}_1 \in \m{Tree}_\Psi(\Delta'' \Vdash x_\beta{:}A\langle c,e \rangle)\, \m{and}\\
   & \mathcal{D}_2=\mathcal{D}'_2\mc{T}_2 \mathbf{msg}( \mb{send},x_\beta^{\langle c,e \rangle}\,y_\alpha^{\langle c,e \rangle} )\,\m{for}\, \mc{T}_2 \in \m{Tree}_\Psi(\Delta'' \Vdash x_\beta{:}A\langle c,e \rangle)\, \m{and}\\
    \dagger_1 & (\mc{T}_1; \mc{T}_2)  \in  \mathcal{E}^\xi_{ \Psi}\llbracket \Delta'' \Vdash x_{\beta}{:}A\langle c,e \rangle\rrbracket^{k}\, \m{and}\\
\dagger_2&  (\mc{D}'_1; \mc{D}'_2)  \in  \mathcal{E}^\xi_{ \Psi}\llbracket \Delta' \Vdash y_{\alpha+1}{:}B\langle c,e \rangle\rrbracket^{k}
\end{array}\]

These are also the statements that we need to prove when replacing $\mc{D}_2$ with $\mc{D}''_2$.
By the assumption that $\mc{D}''_2$ also sends along $y$,  uniqueness of channels, and $\mc{D}''_2 \mapsto^* \mc{D}_2$ we get 
\[\mathcal{D}''_2=\mathcal{D}^1_2\mc{T}^1_2 \mathbf{msg}( \mb{send},x_\beta^{\langle c,e \rangle}\,y_\alpha^{\langle c,e \rangle})\,\m{for}\, \mc{T}^1_2 \in \m{Tree}_\Psi(\Delta'' \Vdash x_\beta{:}A\langle c,e \rangle)\]
Moreover, since $\mc{T}^1_2$ and $\mathcal{D}^1_2$ are disjoint sub-trees and cannot communicate with each other internally, we have
$\mc{T}^1_2 \mapsto^* \mc{T}_2$ and $\mc{D}^1_2 \mapsto^* \mc{D}'_2$.

Now we can apply the the first item of this lemma on $\dagger_1$ to get
\[ \begin{array}{ll}
    \dagger'_1 & (\mc{T}_1; \mc{T}^1_2) \in  \mathcal{E}^\xi_{ \Psi}\llbracket \Delta'' \Vdash x_{\beta}{:}A\langle c,e \rangle\rrbracket^{k}.
\end{array}\]

 Similarly we apply the first item of this lemma on $\dagger_2$ to get
\[\begin{array}{ll}
    \dagger''_2 &  (\mc{D}'_1; \mc{D}^1_2)  \in  \mathcal{E}^\xi_{ \Psi}\llbracket \Delta' \Vdash u_{\delta}{:}B\langle c,e \rangle\rrbracket^{k}
\end{array} 
\]
and the proof of this subcase is complete.

\end{description}

\item The proof is by considering the row of the logical relation that ensures the assumption $(\mc{D}_1;\mc{D}_2) \in \mc{V}\llbracket \Delta\Vdash K \rrbracket_{y_\alpha;\cdot}^{k}$. Here we only consider two interesting cases, the proof of other cases is similar.
\begin{description}
\item[\bf Row $r_3$.] By the conditions of this row, we know that $K=y_\alpha{:}\&\{\ell{:}A_\ell\}_{\ell \in I}\langle c,e \rangle$.
We have 
 \[\begin{array}{l}
(\mc{D}_1; \mc{D}_2) \in \m{Tree}_\Psi(\Delta \Vdash y_\alpha{:}\&\{\ell{:}A_\ell\}_{\ell \in I}\langle c,e \rangle)
 \end{array}
 \]

Moreover, we have 
 \[ \begin{array}{ll}
    & \forall k,j \in I\,\m{where}\, k=j\, \m{if}\, c \sqsubseteq \xi\, \m{then}\\
    \dagger & (\mc{D}_1\mathbf{msg}(y_\alpha^{\langle c,e \rangle}.k) ; \mc{D}_2\mathbf{msg}(y_\alpha^{\langle c,e \rangle}.j))  \in  \mathcal{E}^\xi_{ \Psi}\llbracket \Delta\Vdash y_{\alpha+1}{:}A_k\langle c,e \rangle\rrbracket^{k}
 \end{array}\]

 These are also the statements that we need to prove when replacing $\mc{D}_2$ with $\mc{D}''_2$. The first tree statements are straightforward by type preservation.

By assumption, \[\mc{D}''_2\mathbf{msg}(y_\alpha^{\langle c,e \rangle}.j)\mapsto^*\mc{D}_2\mathbf{msg}(y_\alpha^{\langle c,e \rangle}.j).\]
Note that if $c \not \sqsubseteq \xi$, then $c \not \sqsubseteq e$, and $A_k=A_j$.

 We can apply the first item of this lemma on $\dagger$ to get 
 \[ \begin{array}{ll}
 \dagger' & (\mc{D}_1\mathbf{msg}(y_\alpha^{\langle c,e \rangle}.k); \mc{D}''_2\mathbf{msg}(y_\alpha^{\langle c,e \rangle}.j))   \in  \mathcal{E}^\xi_{ \Psi}\llbracket \Delta\Vdash y_{\alpha+1}{:}A_k\langle c,e \rangle\rrbracket^{k}
 \end{array}\]
 which completes the proof of this case.
\item[\bf Row $r_5$.] By the conditions of this row, we know that $K=y_\alpha{:}A\multimap B\langle c,e \rangle$.
We have 
 \[\begin{array}{l}
(\mc{D}_1; \mc{D}_2) \in \m{Tree}_\Psi(\Delta \Vdash y_\alpha{:}A\multimap B\langle c,e \rangle)
 \end{array}
 \]

Moreover, we have 
 \[ \begin{array}{ll}
    & \forall x_\beta \not \in \mathit{dom}(\Delta, y_{\alpha+1}:A_k\langle c, e \rangle).\\
    \dagger & (\mc{D}_1\mathbf{msg}( \mb{send}x_\beta^{\langle c,e \rangle}\,y_\alpha^{\langle c,e \rangle}); \mc{D}_2\mathbf{msg}(\mb{send}x_\beta^{\langle c,e \rangle}\,y_\alpha^{\langle c,e \rangle}))  \in  \mathcal{E}^\xi_{ \Psi}\llbracket \Delta, x_\beta{:}A\langle c,e \rangle\Vdash y_{\alpha+1}{:}B\langle c,e \rangle\rrbracket^{k}
 \end{array}\]

 These are also the statements that we need to prove when replacing $\mc{D}_2$ with $\mc{D}''_2$. The first tree statements are straightforward by type preservation.

By assumption, \[\mc{D}''_2\mathbf{msg}(\mb{send}x_\beta^{\langle c,e \rangle}\,y_\alpha^{\langle c,e \rangle})\mapsto^*\mc{D}_2\mathbf{msg}(\mb{send}x_\beta^{\langle c,e \rangle}\,y_\alpha^{\langle c,e \rangle}).\]
 We can apply the first item of this lemma on $\dagger$ to get 
 \[ \begin{array}{ll}
 \dagger' & (\mc{D}_1\mathbf{msg}( \mb{send}x_\beta^{\langle c,e \rangle}\,y_\alpha^{\langle c,e \rangle}) ; \mc{D}''_2\mathbf{msg}(\mb{send}x_\beta^{\langle c,e \rangle}\,y_\alpha^{\langle c,e \rangle}))  \in  \mathcal{E}^\xi_{ \Psi}\llbracket \Delta, x_\beta{:}A\langle c,e \rangle\Vdash y_{\alpha+1}{:}B\langle c,e \rangle\rrbracket^{k}
 \end{array}\]
 which completes the proof of this case.

\item[\bf Rows $l_1$, $l_2$, and $l_4$.] Similar to the previous cases.
\end{description}
\end{enumerate}
\end{proof}

%%%Lemma 4%%%: done
\begin{lemma}[Forward closure on the first run]\label{apx:lem:fwdfirst}
    Consider $(\mc{D}_1;\mc{D}_2) \in \mc{E}_{\Psi_0}^\xi\llbracket \Delta\Vdash K \rrbracket^k$ and   $\mc{D}_1\mapsto^{*} \mc{D}''_1$. We have
    $( \mc{D}''_1;\mc{D}_2) \in \mc{E}_{\Psi_0}^\xi\llbracket \Delta\Vdash K \rrbracket^k.$
    \end{lemma}
    \begin{proof}
    If $k=0$ the proof is trivial. Consider $k=m+1$.
    By assumption, we have $(\mc{D}_1;\mc{D}_2)  \in \mc{E}_{\Psi_0}^\xi\llbracket \Delta\Vdash K \rrbracket^{m+1}$.
    
    By the term interpretation (Figure ~\ref{fig:rel_term}) of the logical relation we get  
     
    \[\begin{array}{l}
        \star\,\, \forall\, \Upsilon_1, \Theta_1, \mc{D}'_1. \,\m{if}\,  \mc{D}_1 \mapsto^{*_{ \Upsilon_1; \Theta_1}} \mc{D}'_1,\, \m{then}\, \exists \Upsilon_2, \mc{D}'_2 \,\m{such\, that} \, \,\mc{D}_2\mapsto^{{*}_{ \Upsilon_2}} \mc{D}'_2\, \m{and}\, \Upsilon_1 \subseteq \Upsilon_2\, \m{and}\\ 
         \; \; \forall\, x_\alpha \in \mb{Out}(\Delta \Vdash K).\, \mb{if}\, x_\alpha \in \Upsilon_1.\, \mb{then}\,\, (\mc{D}'_1; \mc{D}'_2)\in \mc{V}^\xi_{\Psi_0}\llbracket  \Delta \Vdash K \rrbracket_{\cdot;x_\alpha}^{m+1}\,\m{and}\\ 
         \; \; \forall\, x_\alpha \in \mb{In}(\Delta \Vdash K). \,\mb{if}\, x_\alpha \in \Theta_1. \,\mb{then}\, (\mc{D}'_1; \mc{D}'_2)\in \mc{V}^\xi_{\Psi_0}\llbracket  \Delta \Vdash K \rrbracket_{x_\alpha;\cdot}^{m+1}.
        \end{array}\]

    Consider $\mc{D}''_1$, for which by the assumption we have $\mc{D}_1\mapsto^* \mc{D}''_1$. Our goal is to prove $( \mc{D}''_1;\mc{D}_2)  \in \mc{E}_{\Psi_0}^\xi\llbracket \Delta\Vdash K \rrbracket^{m+1}$.
    We need to show that 
    \[\begin{array}{l}
        \dagger,\,\forall\, \Upsilon_1, \Theta_1, \mc{D}'_1. \,\m{if}\,  \mc{D}''_1 \mapsto^{*_{ \Upsilon_1; \Theta_1}} \mc{D}'_1,\, \m{then}\, \exists \mc{D}'_2 \,\m{such\, that} \, \,\mc{D}_2\mapsto^{{*}_{ \Upsilon_2}} \mc{D}'_2\, \m{and}\, \Upsilon_1 \subseteq \Upsilon_2\, \m{and}\\ 
         \; \; \forall\, x_\alpha \in \mb{Out}(\Delta \Vdash K).\, \mb{if}\, x_\alpha \in \Upsilon_1.\, \mb{then}\,\, (\mc{D}'_1; \mc{D}'_2)\in \mc{V}^\xi_{\Psi_0}\llbracket  \Delta \Vdash K \rrbracket_{\cdot;x_\alpha}^{m+1}\,\m{and}\\ 
         \; \; \forall\, x_\alpha \in \mb{In}(\Delta \Vdash K).\, \,\mb{if}\, x_\alpha \in \Theta_1. \,\mb{then}\,  (\mc{D}'_1; \mc{D}'_2)\in \mc{V}^\xi_{\Psi_0}\llbracket  \Delta \Vdash K \rrbracket_{x_\alpha;\cdot}^{m+1}.
        \end{array}\]

    With a $\forall$-Introduction and an $\m{if}$-Introduction on the goal, we assume $ \mc{D}''_1 \mapsto^{*_{ \Upsilon_1; \Theta_1}} \mc{D}'_1$. By assumption of $\mc{D}_1 \mapsto^* \mc{D}''_1$ we get $\mc{D}_1 \mapsto^{*_{ \Upsilon_1; \Theta_1}} \mc{D}'_1$. We use this to apply $\forall$-Elimination and $\m{if}$- Elimination on the assumption, and get 
    
    \[\begin{array}{l}
        \star'\,\,\exists \Upsilon_2, \mc{D}'_2 \,\m{such\, that} \, \,\mc{D}_2\mapsto^{{*}_{ \Upsilon_2}} \mc{D}'_2\, \m{and} \Upsilon_1 \subseteq \Upsilon_2\, \m{and}\\ 
         \; \; \forall\, x_\alpha \in \mb{Out}(\Delta \Vdash K).\, \mb{if}\, x_\alpha \in \Upsilon_1.\, \mb{then}\,\, (\mc{D}'_1; \mc{D}'_2)\in \mc{V}^\xi_{\Psi_0}\llbracket  \Delta \Vdash K \rrbracket_{\cdot;x_\alpha}^{m+1}\,\m{and}\\ 
         \; \; \forall\, x_\alpha \in \mb{In}(\Delta \Vdash K).\,\mb{if}\, x_\alpha \in \Theta_1. \,\mb{then}\,  (\mc{D}'_1; \mc{D}'_2)\in \mc{V}^\xi_{\Psi_0}\llbracket  \Delta \Vdash K \rrbracket_{x_\alpha;\cdot}^{m+1}.
        \end{array}\]
    
    Which exactly matches our goal and the proof is complete.  
    \end{proof}

\begin{lemma}[Forward closure on the second run with some specific conditions]\label{apx:lem:fwdsecond}
    Consider $(\mc{D}_1;\mc{D}_2) \in \mc{E}_{\Psi_0}^\xi\llbracket \Delta\Vdash K \rrbracket^k$ and  $\mc{D}_2\mapsto^{*_{\Upsilon}} \mc{D}''_2$ such that if $\mc{D}_1$ sends along the set $\Upsilon_1$, we have $\Upsilon_1 \subseteq \Upsilon$ and also $\mc{D}''_2$ is the minimal configuration built by \Cref{apx:lem:least} given the set $\Upsilon_1$ and configuration $\mc{D}_2$. We have
    $(\mc{D}_1;\mc{D}''_2) \in \mc{E}_{\Psi_0}^\xi\llbracket \Delta\Vdash K \rrbracket^k.$
    \end{lemma}
    \begin{proof}
    If $k=0$ the proof is trivial. Consider $k=m+1$.
    By assumption, we have $(\mc{D}_1;\mc{D}_2)  \in \mc{E}_{\Psi_0}^\xi\llbracket \Delta\Vdash K \rrbracket^{m+1}$.
        By the term interpretation (Figure ~\ref{fig:rel_term}) of the logical relation we get  
    
    \[\begin{array}{l}
    \star\,\,\forall\, \Upsilon_1, \Theta_1, \mc{D}'_1. \, \m{if}\,  \mc{D}_1 \mapsto^{*_{ \Upsilon_1; \Theta_1}} \mc{D}'_1,\, \m{then}\, \exists \mc{D}'_2 \,\m{such\, that} \, \,\mc{D}_2\mapsto^{{*}_{ \Upsilon_2}} \mc{D}'_2\, \m{and}\, \Upsilon_1 \subseteq \Upsilon_2\, \m{and}\\ 
        \; \; \forall\, x_\alpha \in \mb{Out}(\Delta \Vdash K).\, \mb{if}\, x_\alpha \in \Upsilon_1.\, \mb{then}\,\, (\mc{D}'_1; \mc{D}'_2)\in \mc{V}^\xi_{\Psi_0}\llbracket  \Delta \Vdash K \rrbracket_{\cdot;x_\alpha}^{m+1}\,\m{and}\\ 
        \; \; \forall\, x_\alpha \in \mb{In}(\Delta \Vdash K).\,\mb{if}\, x_\alpha \in \Theta_1.\, (\mc{D}'_1; \mc{D}'_2)\in \mc{V}^\xi_{\Psi_0}\llbracket  \Delta \Vdash K \rrbracket_{x_\alpha;\cdot}^{m+1}.
    \end{array}\]
    
    Consider $\mc{D}''_2$, for which by the assumption we have $\mc{D}_2\mapsto^{*_\Upsilon} \mc{D}''_2$. Our goal is to prove $(\mc{D}_1;\mc{D}''_2)  \in \mc{E}_{\Psi_0}^\xi\llbracket \Delta\Vdash K \rrbracket^{m+1}$.
    We need to show that 
    \[\begin{array}{l}
    \dagger\,\forall\, \Upsilon_1, \Theta_1, \mc{D}'_1. \,\,\m{if}\,  \mc{D}_1 \mapsto^{*_{ \Upsilon_1; \Theta_1}} \mc{D}'_1,\, \m{then}\, \exists \mc{D}'_2 \,\m{such\, that} \, \,\mc{D}''_2\mapsto^{{*}_{ \Upsilon_2}} \mc{D}'_2\, \m{and} \Upsilon_1 \subseteq \Upsilon_2\, \m{and}\\ 
    \; \; \forall\, x_\alpha \in \mb{Out}(\Delta \Vdash K).\, \mb{if}\, x_\alpha \in \Upsilon_1.\, \mb{then}\,\, (\mc{D}'_1; \mc{D}'_2)\in \mc{V}^\xi_{\Psi_0}\llbracket  \Delta \Vdash K \rrbracket_{\cdot;x_\alpha}^{m+1}\,\m{and}\\ 
    \; \; \forall\, x_\alpha \in \mb{In}(\Delta \Vdash K).\,\mb{if}\, x_\alpha \in \Theta_1.\, (\mc{D}'_1; \mc{D}'_2)\in \mc{V}^\xi_{\Psi_0}\llbracket  \Delta \Vdash K \rrbracket_{x_\alpha;\cdot}^{m+1}.
    \end{array}\]

    With an $\m{if}$-Introduction on the goal followed by an $\m{if}$- Elimination on the assumption, we get the assumption 
    \[\begin{array}{l}
        \star'\,\,\exists \Upsilon_2, \mc{D}'_2 \,\m{such\, that} \, \,\mc{D}_2\mapsto^{{*}_{ \Upsilon_1; \Theta_1}} \mc{D}'_2\, \m{and}\Upsilon_1 \subseteq \Upsilon_2\, \m{and}\\ 
            \; \; \forall\, x_\alpha \in \mb{Out}(\Delta \Vdash K).\, \mb{if}\, x_\alpha \in \Upsilon_1.\, \mb{then}\,\, (\mc{D}'_1; \mc{D}'_2)\in \mc{V}^\xi_{\Psi_0}\llbracket  \Delta \Vdash K \rrbracket_{\cdot;x_\alpha}^{m+1}\,\m{and}\\ 
            \; \; \forall\, x_\alpha \in \mb{In}(\Delta \Vdash K).\,\mb{if}\, x_\alpha \in \Theta_1.\, (\mc{D}'_1; \mc{D}'_2)\in \mc{V}^\xi_{\Psi_0}\llbracket  \Delta \Vdash K \rrbracket_{x_\alpha;\cdot}^{m+1}.
        \end{array}\]
    
    and the goal 
    \[\begin{array}{l}
        \dagger'\,\,\exists \Upsilon_2, \mc{D}'_2 \,\m{such\, that} \, \,\mc{D}''_2\mapsto^{{*}_{ \Upsilon_2}} \mc{D}'_2\, \m{and} \Upsilon_1 \subseteq \Upsilon_2\, \m{and}\\ 
        \; \; \forall\, x_\alpha \in \mb{Out}(\Delta \Vdash K).\, \mb{if}\, x_\alpha \in \Upsilon_1.\, \mb{then}\,\, (\mc{D}'_1; \mc{D}'_2)\in \mc{V}^\xi_{\Psi_0}\llbracket  \Delta \Vdash K \rrbracket_{\cdot;x_\alpha}^{m+1}\,\m{and}\\ 
        \; \; \forall\, x_\alpha \in \mb{In}(\Delta \Vdash K).\,\mb{if}\, x_\alpha \in \Theta_1.\,x (\mc{D}'_1; \mc{D}'_2)\in \mc{V}^\xi_{\Psi_0}\llbracket  \Delta \Vdash K \rrbracket_{x_\alpha;\cdot}^{m+1}.
        \end{array}\]
    We apply an $\exists$-Elimination on the assumption to get $\mc{D}'_2$ that satisfies the conditions, i.e., $\mc{D}_2 \mapsto^{*_{\Upsilon_2}} \mathcal{D}'_2$ and $\Upsilon_1 \subseteq \Upsilon_2$. We use the same $\mc{D}'_2$ to instantiate the existential quantifier in the goal, we need to show that $\mc{D}''_2 \mapsto^{*_{\Upsilon_1}} \mathcal{D}'_2$. Since $\mc{D}''_2$ is the minimal configuration built for $\Upsilon_1$ and $\mc{D}_2$, and $\Upsilon_1 \subseteq \Upsilon_2$, by Lemma~\ref{apx:lem:least} we get $\mc{D}''_2\mapsto^{*_{\Upsilon_2}} \mc{D}'_2$,  and the proof is complete.  
    \end{proof}

\begin{lemma} [Moving existential over universal quantifier]\label{apx:lem:moving-existential}
    if we have 
     
    \[\begin{array}{l}
        (\dagger)\,\,\forall m.\, \forall\, \Upsilon_1, \Theta_1, \mathcal{D}'_1.\, \, \m{if}\,  \mc{D}_1 \mapsto^{*_{ \Upsilon_1; \Theta_1}} \mc{D}'_1,\, \m{then}\, \exists \Upsilon_2, \mc{D}'_2 \,\m{such\, that} \, \,\mc{D}_2\mapsto^{{*}_{ \Upsilon_2}} \mc{D}'_2\, \m{and}\, \Upsilon_1 \subseteq \Upsilon_2 \, \m{and}\\ 
         \; \; \forall\, x_\alpha \in \mb{Out}(\Delta \Vdash K).\, \mb{if}\, x_\alpha \in \Upsilon_1.\, \mb{then}\,\, (\mc{D}'_1; \mc{D}'_2)\in \mc{V}_{\Psi_0}^\xi\llbracket  \Delta \Vdash K \rrbracket_{\cdot;x_\alpha}^{m+1}\,\m{and}\\ 
         \; \; \forall\, x_\alpha \in \mb{In}(\Delta \Vdash K).\mb{if}\, x_\alpha \in \Theta_1.\,\mb{then}\, (\mc{D}'_1; \mc{D}'_2)\in \mc{V}_{\Psi_0}^\xi\llbracket  \Delta \Vdash K \rrbracket_{x_\alpha;\cdot}^{m+1}.
        \end{array}\]
    
    then 
    \[\begin{array}{l}
        \forall\, \Upsilon_1, \Theta_1, \mathcal{D}'_1. \m{if}\,  \mc{D}_1 \mapsto^{*_{ \Upsilon_1; \Theta_1}} \mc{D}'_1,\, \m{then}\, \exists \Upsilon_2,\mc{D}'_2 \,\m{such\, that} \, \,\mc{D}_2\mapsto^{{*}_{ \Upsilon_2}} \mc{D}'_2\, \m{and}\, \Upsilon_1 \subseteq \Upsilon_2 \, \m{and}\\ 
        \; \; \forall\, x_\alpha \in \mb{Out}(\Delta \Vdash K).\, \mb{if}\, x_\alpha \in \Upsilon_1.\, \mb{then}\,\, \forall \,k.\, (\mc{D}'_1; \mc{D}'_2)\in \mc{V}_{\Psi_0}^\xi\llbracket  \Delta \Vdash K \rrbracket_{\cdot;x_\alpha}^{k+1}\,\m{and}\\ 
        \; \; \forall\, x_\alpha \in \mb{In}(\Delta \Vdash K). \mb{if}\, x_\alpha \in \Theta_1.\,\mb{then}\,\,\forall k.\, (\mc{D}'_1; \mc{D}'_2)\in \mc{V}_{\Psi_0}^\xi\llbracket  \Delta \Vdash K \rrbracket_{x_\alpha;\cdot}^{k+1}.
    \end{array}\]
    \end{lemma}
    
    \begin{proof}
    First put $m=1$ to apply $\forall\, \mathbf{E}.$ on the assumption (instantiating $\forall\, m$ only). We get as an assumption
    \[\begin{array}{l}
     (\dagger')\,\, \forall\, \Upsilon_1, \Theta_1, \mathcal{D}'_1.\, \, \m{if}\,  \mc{D}_1 \mapsto^{*_{ \Upsilon_1; \Theta_1}} \mc{D}'_1,\, \m{then}\, \exists \Upsilon_2, \mc{D}'_2 \,\m{such\, that} \, \,\mc{D}_2\mapsto^{{*}_{ \Upsilon_2}} \mc{D}'_2\, \m{and}\, \Upsilon_1 \subseteq \Upsilon_2 \, \m{and}\\ 
         \; \; \forall\, x_\alpha \in \mb{Out}(\Delta \Vdash K).\, \mb{if}\, x_\alpha \in \Upsilon_1.\, \mb{then}\,\, (\mc{D}'_1; \mc{D}'_2)\in \mc{V}_{\Psi_0}^\xi\llbracket  \Delta \Vdash K \rrbracket_{\cdot;x_\alpha}^{0+1}\,\m{and}\\ 
         \; \; \forall\, x_\alpha \in \mb{In}(\Delta \Vdash K).\mb{if}\, x_\alpha \in \Theta_1.\,\mb{then}\, (\mc{D}'_1; \mc{D}'_2)\in \mc{V}_{\Psi_0}^\xi\llbracket  \Delta \Vdash K \rrbracket_{x_\alpha;\cdot}^{0+1}.
        \end{array}\]

    Next, apply a $\forall \mathbf{I}.$ and $\mathbf{if}\, \mathbf{I}.$ on the goal followed by a corresponding $\forall \mathbf{E}.$ and $\mathbf{if}\, \mathbf{E}.$ on the assumption $(\dagger')$.  Now apply an $\exists \mathbf{E}.$ on the assumption $(\dagger')$ to get $\mc{D}'_2$ such that $\mc{D}_2 \mapsto^{*_{\Upsilon_2}} \mc{D}'_2$ and $\Upsilon_1 \subseteq \Upsilon_2$. Given $\mc{D}'_2$, by \Cref{apx:lem:least}, we can build the minimal $\mc{D}''_2$ such that $\mc{D}_2 \mapsto^{*_\Upsilon} \mc{D}''_2$ and $\Upsilon_1 \subseteq \Upsilon$. Moreover, we know that for every $\mc{D}$ such that $\mc{D}_2 \mapsto^{*_{\Upsilon_3}} \mc{D}$ and $\Upsilon_1 \subseteq \Upsilon_3$, we get $\mc{D}''_2 \mapsto^{*_{\Upsilon_3}} \mc{D}$.
    
    We use this minimal $\mc{D}''_2$ to instantiate the existential ($\exists\, \mathbf{I}.$) in the goal, and use $\forall \, \mathbf{I}.$ on the goal. In particular, we instantiate $k$ with an arbitrary natural number. 
    The goals are:
    {\[\begin{array}{l}
     (\mc{D}'_1; \mc{D}''_2)\in \mc{V}_{\Psi_0}^\xi\llbracket  \Delta \Vdash K \rrbracket_{\cdot;x_\alpha}^{k+1}\,\,\m{and}\,\, ({\mc{D}'_1}; {\mc{D}''_2})\in \mc{V}_{\Psi_0}^\xi\llbracket  \Delta \Vdash K \rrbracket_{x_\alpha;\cdot}^{k+1}.
    \end{array}\]}
    
    Next, we instantiate the $\forall$ quantifier in the original assumption ($\dagger$) once again, this time with $m=k$ followed by a $\forall \mathbf{E}$, and $\mathbf{if} \mathbf{E}$ instantiating the quantifiers with similar $\mc{D}_1'$, $\Upsilon_1$, and $\Theta_1$ as the first time. We get as an assumption:
    
    \[\begin{array}{l}
        (\dagger'')\,\,\exists \Upsilon_2, \mc{D}'_2 \,\m{such\, that} \, \,\mc{D}_2\mapsto^{{*}_{ \Upsilon_2}} \mc{D}'_2\, \m{and}\, \Upsilon_1 \subseteq \Upsilon_2 \, \m{and}\\ 
         \; \; \forall\, x_\alpha \in \mb{Out}(\Delta \Vdash K).\, \mb{if}\, x_\alpha \in \Upsilon_1.\, \mb{then}\,\, (\mc{D}'_1; \mc{D}'_2)\in \mc{V}_{\Psi_0}^\xi\llbracket  \Delta \Vdash K \rrbracket_{\cdot;x_\alpha}^{k+1}\,\m{and}\\ 
         \; \; \forall\, x_\alpha \in \mb{In}(\Delta \Vdash K).\mb{if}\, x_\alpha \in \Theta_1.\,\mb{then}\, (\mc{D}'_1; \mc{D}'_2)\in \mc{V}_{\Psi_0}^\xi\llbracket  \Delta \Vdash K \rrbracket_{x_\alpha;\cdot}^{k+1}.
        \end{array}\]

    Next, apply $\exists\, \mathbf{I}.$ to get a $\Upsilon'$ and $\mc{D}'$ that satisfies the conditions, i.e., $\mathcal{D} \mapsto^{*_{\Upsilon'}} \mathcal{D}'$ and $\Upsilon_1 \subseteq \Upsilon'$. Instantiate $x_\alpha$ as those chosen for the goal. We have as assumptions: 
    {\[\begin{array}{l}
     (\dagger''')\,\,({\mc{D}'_1}; {\mc{D}'_2})\in \mc{V}_{\Psi_0}^\xi\llbracket  \Delta \Vdash K \rrbracket_{\cdot;x_\alpha}^{k+1}\,\,\m{and}\,\, ({\mc{D}'_1}; {\mc{D}'_2})\in \mc{V}_{\Psi_0}^\xi\llbracket  \Delta \Vdash K \rrbracket_{x_\alpha;\cdot}^{k+1}.
    \end{array}\]}
    Since $\mc{D}''_2$ is the minimal configuration built for $\Upsilon_1$ and $\mc{D}_2$, we know that $\mc{D}''_2 \mapsto^{*_{\Upsilon'}} \mc{D}'$. We can apply the backward closure results of \Cref{lem:backsecondrun} to get the goal from the assumptions $ (\dagger''')$, and this completes the proof.
    
    \end{proof}

%%%Lemma 7%%%:done
\begin{lemma}[Compositionlaity]
Consider $\Delta , u_\alpha\langle c,e \rangle{:}T \Vdash \mc{D}_i:: K$ and $\Delta'  \Vdash \mc{T}_i::  u_\alpha\langle c,e \rangle{:}T$.
\[\begin{array}{l}
    (\dagger_1)\,\, \forall m.\, (\mc{D}_1; \mc{D}_2) \in \mc{E}_{\Psi_0}^\xi\llbracket \Delta, u_\alpha^{\langle c,e \rangle}{:}T \Vdash K \rrbracket^m\, \m{and}
    (\dagger_2)\,\, \forall m.\, (\mc{T}_1; \mc{T}_2) \in \mc{E}_{\Psi_0}^\xi\llbracket \Delta'\Vdash  u_\alpha^{\langle c,e \rangle}{:}T \rrbracket^m\, \m{iff}\\ 
    \dagger\,\forall k.\, (\mc{T}_1 \mc{D}_1; \mc{T}_2\mc{D}_2) \in \mc{E}_{\Psi_0}^\xi\llbracket \Delta', \Delta\Vdash K \rrbracket^k.
\end{array}\]

\end{lemma}
\begin{proof}
    The left to right direction is a corollary of Lemma~\ref{lem:puttogether} in which we compose multiple configurations instead of just two.
    For the right to left direction, we put $\mc{T}_i=\cdot$, and $\Delta'=u_\alpha{:}T$, and the rest of the proof is straightforward. 
\end{proof}

%%%Lemma 8%%%: done

\begin{lemma}[General Compositionality]\label{lem:puttogether}
    For $i\in \{1,2\}$, and index set $I$ consider  well-typed tree-shaped configurations $\Delta^n \Vdash \mc{B}^n_i:: K^n$
    such that $n \in I$ and their compositions form  well-typed tree-shaped configurations $\Delta \Vdash \mc{D}_i::K$, i.e.,  $\mc{D}_i= \{\mc{B}^n_i\}_{n\in I}$. If for all $n \in I$, we have 
 $\dagger_n\,\,\forall m.\, (\mc{B}^n_1;\mc{B}^n_2) \in \mc{E}_{\Psi_0}^\xi\llbracket \Delta^n \Vdash K^n \rrbracket^m$ then $\dagger\,\forall k.\, (\mc{D}_1;\mc{D}_2) \in \mc{E}_{\Psi_0}^\xi\llbracket \Delta\Vdash K \rrbracket^k$.
\end{lemma}
\begin{proof}

    Our goal is to prove the following:
    \begin{quote} For all index set $I$ and all session-typed configurations $\Delta^n \Vdash \mc{B}^n_i:: K^n$ with $n\in I$ such that
    $\dagger_1\,\, \forall m.\, (\mc{B}^n_1;\mc{B}^n_2) \in \mc{E}_{\Psi_0}^\xi\llbracket \Delta^n \Vdash K^n \rrbracket^m$ we have $\dagger \,\,\forall k.\, (\{\mc{B}^n_1\}_{n \in I};\{\mc{B}^n_2\}_{n \in I}) \in \mc{E}_{\Psi_0}^\xi\llbracket \Delta\Vdash K \rrbracket^k$.
    \end{quote}
    
    This is equivalent to the following statement which we prove:
    \begin{quote} For all natural numbers, $k$, and for all index set $I$, and for all session-typed configurations $\Delta^n \Vdash \mc{B}^n_i:: K^n$ such that
    $\dagger_n\,\,\forall m.\, (\mc{B}^n_1;\mc{B}^n_2) \in \mc{E}_{\Psi_0}^\xi\llbracket \Delta^n \Vdash K^n \rrbracket^m$ we have $\dagger'\,\, (\{\mc{B}^n_1\}_{n \in I};\{\mc{B}^n_2\}_{n \in I}) \in \mc{E}_{\Psi_0}^\xi\llbracket \Delta \Vdash K \rrbracket^k$.
    \end{quote}
    
    We proceed the proof by an induction on k.
    
    {\bf Base case ($k=0$). }The proof is straightforward, since by the definition of the logical relation for session-typed configurations $\Delta^n \Vdash \mc{B}^n_i:: K^n$ with $\Delta \Vdash \{\mc{B}^n_i\}_{n \in I}::K$, we have $(\{\mc{B}^n_1\}_{n \in I};\{\mc{B}^n_2\}_{n \in I}) \in \mc{E}_{\Psi_0}^\xi\llbracket \Delta \Vdash K \rrbracket^0$.
    
    {\bf Inductive case ($k=k'+1$). }
    Our goal is to prove the following:
    \begin{quote} 
         For all index set $I$, and for all session-typed configurations $\Delta^n \Vdash \mc{B}^n_i:: K^n$ where $n\in I$ such that
        $\dagger_n\,\,\forall m.\, (\mc{B}^n_1;\mc{B}^n_2) \in \mc{E}_{\Psi_0}^\xi\llbracket \Delta^n \Vdash K^n \rrbracket^m$ we have $\dagger'\,\, (\{\mc{B}^n_1\}_{n \in I};\{\mc{B}^n_2\}_{n \in I}) \in \mc{E}_{\Psi_0}^\xi\llbracket \Delta \Vdash K \rrbracket^{k'+1}$
    \end{quote}
    
    where $\dagger'$ is defined in the term interpretation (Figure ~\ref{fig:rel_term}) of the logical relation as 
    \[\begin{array}{l}
        \forall \,{\Upsilon_1}, \Theta_1, \mc{D}'_1.\,  \forall \, j \in \mathbb{N}.\, \mathbf{if}\, {\{\mc{B}^n_1\}_{n \in I}}\mapsto^{{j}_{ \Upsilon_1; \Theta_1}} {\mc{D}'_1}\,\m{then}\, \exists \Upsilon_2, \mc{D}'_2 \,\m{such\, that} \, \,\{\mc{B}^n_2\}_{n \in I}\mapsto^{{*}_{ \Upsilon_2}} \mc{D}'_2\, \m{and} \, \Upsilon_1 \subseteq \Upsilon_2\, \m{and}\\ 
         \; \; \forall\, x_\gamma \in \mb{Out}( \Delta  \Vdash K).\, \mb{if}\, x_\gamma \in \Upsilon_1.\, \mb{then}\,\, (\mc{D}'_1; \mc{D}'_2)\in \mc{V}_{\Psi_0}^\xi\llbracket   \Delta \Vdash K \rrbracket_{\cdot;x_\gamma}^{k'+1}\,\m{and}\\ 
         \; \; \forall\, x_\gamma \in \mb{In}( \Delta  \Vdash K).\,\mb{if}\, x_\gamma \in \Theta_1.\, \mb{then}\,\, (\mc{D}'_1; \mc{D}'_2)\in \mc{V}_{\Psi_0}^\xi\llbracket   \Delta \Vdash K \rrbracket_{x_\gamma;\cdot}^{k'+1}.
     \end{array}\]
    
    Again, we rewrite the above goal as an  equivalent statement as follows:
    
    \begin{quote} For all natural numbers $j$, for all $ I$, and for all session-typed configurations $\Delta^n \Vdash \mc{B}^n_i:: K^n$ with $n \in I$ such that
        $\dagger_n\,\,\forall m.\, (\mc{B}^n_1;\mc{B}^n_2) \in \mc{E}_{\Psi_0}^\xi\llbracket \Delta^n \Vdash K^n \rrbracket^m$ we have 
    \[\begin{array}{l}
        \forall \,\Upsilon_1, \Theta_1, \mc{D}'_1.\, \, \mathbf{if}\, {\{\mc{B}^n_1\}_{n \in I}}\mapsto^{{j}_{ \Upsilon_1; \Theta_1}} {\mc{D}'_1}\,\m{then}\, \exists \Upsilon_2, \mc{D}'_2 \,\m{such\, that} \, \,\{\mc{B}^n_2\}_{n \in I}\mapsto^{{*}_{ \Upsilon_2}} \mc{D}'_2\, \m{and}\, \Upsilon_1 \subseteq \Upsilon_2\, \m{and}\\ 
         \; \; \forall\, x_\gamma \in \mb{Out}(\Delta  \Vdash K).\, \mb{if}\, x_\gamma \in \Upsilon_1.\, \mb{then}\,\, (\mc{D}'_1; \mc{D}'_2)\in \mc{V}_{\Psi_0}^\xi\llbracket   \Delta \Vdash K \rrbracket_{\cdot;x_\gamma}^{k'+1}\,\m{and}\\ 
         \; \; \forall\, x_\gamma \in \mb{In}( \Delta  \Vdash K).\mb{if}\, x_\gamma \in \Theta_1.\, \mb{then}\, (\mc{D}'_1; \mc{D}'_2)\in \mc{V}_{\Psi_0}^\xi\llbracket  \Delta \Vdash K \rrbracket_{x_\gamma;\cdot}^{k'+1}.
     \end{array}\]
    \end{quote}
    
    We proceed the proof by a nested induction on $j$. 

\begin{description}
\item{\bf Base case ($j=0$)}. Consider an arbitrary index set $I$ and an arbitrary session-typed configurations $\Delta^n\Vdash \mc{B}^n_i:: K^n$ that satisfy the conditions $\dagger_n$. We need to show that

\[\begin{array}{l}
    \forall \Upsilon_1, \Theta_1, \,\mc{D}'_1.\, \mathbf{if}\, {\{\mc{B}^n_1\}_{n \in I}}\mapsto^{{0}_{ \Upsilon_1; \Theta_1}} {\mc{D}'_1}\,\m{then}\, \exists \Upsilon_2, \mc{D}'_2 \,\m{such\, that} \, \,\{\mc{B}^n_2\}_{n \in I}\mapsto^{{*}_{ \Upsilon_2}} \mc{D}'_2\, \m{and} \, \Upsilon_1 \subseteq \Upsilon_2\, \m{and}\\ 
     \; \; \forall\, x_\gamma \in \mb{Out}( \Delta  \Vdash K).\, \mb{if}\, x_\gamma \in \Upsilon_1.\, \mb{then}\,\, (\mc{D}'_1; \mc{D}'_2)\in \mc{V}_{\Psi_0}^\xi\llbracket \Delta \Vdash K \rrbracket_{\cdot;x_\gamma}^{k'+1}\,\m{and}\\ 
     \; \; \forall\, x_\gamma \in \mb{In}( \Delta \Vdash K).\, \mb{if}\, x_\gamma \in \Theta_1. \,(\mc{D}'_1; \mc{D}'_2)\in \mc{V}_{\Psi_0}^\xi\llbracket   \Delta \Vdash K \rrbracket_{x_\gamma;\cdot}^{k'+1}.
 \end{array}\]

Consider an arbitrary $\Upsilon_1$, $\Theta_1$, ans $\mc{D}'_1$, and apply $\mathbf{If}\, \mathbf{I}.$ on the goal. By the assumption $\{\mc{B}^n_1\}_{n \in I}\mapsto^{{0}_{ \Upsilon_1; \Theta_1}} {\mc{D}'_1}$  we know that $\mc{D}'_1=\{\mc{B}^n_1\}_{n \in I}$, and $\{\mc{B}^n_1\}_{n \in I}$ sends along $\Upsilon_1$ and receives along $\Theta_1$. Our goal is to show the following:

\[\begin{array}{l}
   \star\,\, \exists \Upsilon_2, \mc{D}'_2 \,\m{such\, that} \, \,\{\mc{B}^n_2\}_{n \in I}\mapsto^{{*}_{ \Upsilon_2}} \mc{D}'_2\, \m{and}  \, \Upsilon_1 \subseteq \Upsilon_2\, \m{and}\\ 
     \; \; \forall\, x_\gamma \in \mb{Out}(\Delta \Vdash K).\, \mb{if}\, x_\gamma \in \Upsilon_1.\, \mb{then}\,\, (\{\mc{B}^n_1\}_{n \in I}; \mc{D}'_2)\in \mc{V}_{\Psi_0}^\xi\llbracket   \Delta \Vdash K \rrbracket_{\cdot;x_\gamma}^{k'+1}\,\m{and}\\ 
     \; \; \forall\, x_\gamma \in \mb{In}( \Delta \Vdash K). \mb{if}\, x_\gamma \in \Theta_1.\, \mb{then} \,(\{\mc{B}^n_1\}_{n \in I}; \mc{D}'_2)\in \mc{V}_{\Psi_0}^\xi\llbracket  \Delta \Vdash K \rrbracket_{x_\gamma;\cdot}^{k'+1}.
 \end{array}\]

By the definition of the logical relation, and from assumptions $\dagger_n$ for $n \in I$ we get

\[\begin{array}{l}
    \dagger'_n\,\forall m. \forall \Upsilon_n, \Theta_n, \mc{B}^{n'}_1.\,\m{if}\, \mc{B}^n_1 \mapsto^{*_{\Upsilon_n; \Theta_n}} \mc{B}^{n'}_1\, \m{then}\, \exists \Upsilon_{n_2}, \mc{B}^{n'}_2 \,\m{such\, that} \, \,\mc{B}^n_2\mapsto^{{*}_{ \Upsilon_{n_2}}} \mc{B}^{n'}_2\, \m{and} 
    \, \Upsilon_n \subseteq \Upsilon_{n_2}\, \m{and}\\ 
     \; \; \forall\, x_\gamma \in \mb{Out}(\Delta^n \Vdash K^n).\, \mb{if}\, x_\gamma \in \Upsilon_n.\, \mb{then}\,\, (\mc{B}^{n'}_1; \mc{B}^{n'}_2)\in \mc{V}_{\Psi_0}^\xi\llbracket  \Delta^n \Vdash K^n\rrbracket_{\cdot;x_\gamma}^{m+1}\,\m{and}\\ 
     \; \; \forall\, x_\gamma \in \mb{In}(\Delta^n \Vdash K^n).\,\mb{if}\, x_\gamma \in \Theta_n.\, \mb{then} \,(\mc{B}^{n'}_1; \mc{B}^{n'}_2)\in \mc{V}_{\Psi_0}^\xi\llbracket \Delta^n \Vdash K^n\rrbracket_{x_\gamma;\cdot}^{m+1}.
 \end{array}\]

By \Cref{apx:lem:moving-existential}, we get 
\[\begin{array}{l}
    \dagger''_n\,\,\forall \Upsilon_n, \Theta_n, \mc{B}^{n'}_1.\,\m{if}\, \mc{B}^n_1 \mapsto^{*_{\Upsilon_n; \Theta_n}} \mc{B}^{n'}_1\, \m{then}\, \exists \Upsilon_{n_2} \mc{B}^{n'}_2 \,\m{such\, that} \, \,\mc{B}^n_2\mapsto^{{*}_{ \Upsilon_{n_2}}} \mc{B}^{n'}_2\, \m{and} 
    \, \Upsilon_n \subseteq \Upsilon_{n_2}\, \m{and}\\ 
     \; \; \forall\, x_\gamma \in \mb{Out}(\Delta^n \Vdash K^n).\, \mb{if}\, x_\gamma \in \Upsilon_n.\, \mb{then}\,\, \forall\, m.\,(\mc{B}^{n'}_1; \mc{B}^{n'}_2)\in \mc{V}_{\Psi_0}^\xi\llbracket  \Delta^n \Vdash K^n\rrbracket_{\cdot;x_\gamma}^{m+1}\,\m{and}\\ 
     \; \; \forall\, x_\gamma \in \mb{In}(\Delta^n \Vdash K^n).  \, \mb{if}\, x_\gamma \in \Theta_n.\, \mb{then}\, \forall\, m. \,(\mc{B}^{n'}_1; \mc{B}^{n'}_2)\in \mc{V}_{\Psi_0}^\xi\llbracket \Delta^n \Vdash K^n\rrbracket_{x_\gamma;\cdot}^{m+1}.
 \end{array}\]

We instantiate the forall quantifier in $\dagger''_n$ by $\mathcal{B}^n_1$ and the sets $\Upsilon_n$ and $\Theta_n$ along which $\mathcal{B}^n_1$ sends and receives. Note that by definition we have $\Upsilon_1 \subseteq \bigcup\{\Upsilon_n\}_{n \in I}$ and $\Theta_1 \subseteq \bigcup\{\Theta_n\}_{n \in I}$. We get:

\[\begin{array}{l}
    \exists \Upsilon_{n_2}, \mc{B}^{n'}_2 \,\m{such\, that} \, \,\mc{B}^n_2\mapsto^{{*}_{ \Upsilon_{n_2}}} \mc{B}^{n'}_2\, \m{and}\\ 
    \; \; \forall\, x_\gamma \in \mb{Out}(\Delta^n \Vdash K^n).\, \mb{if}\, x_\gamma \in \Upsilon_n.\, \mb{then}\,\, \forall\, m.\,(\mc{B}^{n}_1; \mc{B}^{n'}_2)\in \mc{V}_{\Psi_0}^\xi\llbracket  \Delta^n \Vdash K^n\rrbracket_{\cdot;x_\gamma}^{m+1}\,\m{and}\\ 
    \; \; \forall\, x_\gamma \in \mb{In}(\Delta^n \Vdash K^n). \mb{if}\, x_\gamma \in \Theta_n.\, \mb{then}\, \forall\, m. \,(\mc{B}^{n}_1; \mc{B}^{n'}_2)\in \mc{V}_{\Psi_0}^\xi\llbracket \Delta^n \Vdash K^n\rrbracket_{x_\gamma;\cdot}^{m+1}.
 \end{array}\]

 By existential elimination, for all $n \in I$, we get a $\mc{B}^{n'}_2$ and $\Upsilon_{n_2}$ such that  $\mc{B}_2\mapsto^{*_{\Upsilon_{n_2}}} \mc{B}^{n'}_2$, and  $ \Upsilon_n \subseteq \Upsilon_{n_2}$ and

 \[\begin{array}{l}
    \dagger'''_n\, 
     \; \; \forall\, x_\gamma \in \mb{Out}(\Delta^n \Vdash K^n).\, \mb{if}\, x_\gamma \in \Upsilon_n.\, \mb{then}\,\, \forall\, m.\,(\mc{B}^{n}_1; \mc{B}^{n'}_2)\in \mc{V}_{\Psi_0}^\xi\llbracket  \Delta^n \Vdash K^n\rrbracket_{\cdot;x_\gamma}^{m+1}\,\m{and}\\ 
     \; \; \forall\, x_\gamma \in \mb{In}(\Delta^n \Vdash K^n). \mb{if}\, x_\gamma \in \Theta_n.\, \mb{then}\, \forall\, m. \,(\mc{B}^{n}_1; \mc{B}^{n'}_2)\in \mc{V}_{\Psi_0}^\xi\llbracket \Delta^n \Vdash K^n\rrbracket_{x_\gamma;\cdot}^{m+1}.
 \end{array}\]

 Note that by definition $\mb{Out}(\Delta \Vdash K)\subseteq \bigcup_{n \in I}\mb{Out}(\Delta^n \Vdash K^n)$ and  $\mb{In}(\Delta \Vdash K)\subseteq \bigcup_{n \in I}\mb{In}(\Delta^n \Vdash K^n)$.

We apply \Cref{apx:lem:least} to get the minimal sending configurations $\mc{B}^{n''}_2$ and set $\Upsilon'_n$ for each given $\Upsilon_n$ and $\mc{B}^n_2$. We get $\mc{B}^{n''}_2$ such that  $\mc{B}^n_2\mapsto^{*_{\Upsilon'_n}} \mc{B}^{n''}_2$. Since these configurations are minimal, for all $n \in I$, we have $\mc{B}^{n''}_2\mapsto^{*_{\Upsilon_{n_2}}} \mc{B}^{n'}_2$.

Since $\Upsilon_n \subseteq \Upsilon'_n \subseteq \Upsilon_{n_2}$,  we can apply the backward closure (\Cref{lem:backsecondrun}) on $\dagger'''_n$ to get

\[\begin{array}{l}
    \dagger^4_n\, 
    \; \; \forall\, x_\gamma \in \mb{Out}(\Delta^n \Vdash K^n).\, \mb{if}\, x_\gamma \in \Upsilon_n.\, \mb{then}\,\, \forall\, m.\,(\mc{B}^{n}_1; \mc{B}^{n''}_2)\in \mc{V}_{\Psi_0}^\xi\llbracket  \Delta^n \Vdash K^n\rrbracket_{\cdot;x_\gamma}^{m+1}\,\m{and}\\ 
     \; \; \forall\, x_\gamma \in \mb{In}(\Delta^n \Vdash K^n). \,  \mb{if}\, x_\gamma \in \Theta_n.\forall\, m. \,(\mc{B}^{n}_1; \mc{B}^{n''}_2)\in \mc{V}_{\Psi_0}^\xi\llbracket \Delta^n \Vdash K^n\rrbracket_{x_\gamma;\cdot}^{m+1}.
 \end{array}\]

 Moreover, we apply forward closure on the second run \Cref{apx:lem:fwdsecond} on assumptions $\dagger_n$ to get for all $n \in I$:
 \[\circ_n\,\,\forall m.\, (\mc{B}^n_1;\mc{B}^{n''}_2) \in \mc{E}_{\Psi_0}^\xi\llbracket \Delta^n \Vdash K^n\rrbracket^m\]
 We can apply the forward closure on the second run since the conditions of \Cref{apx:lem:fwdsecond} are satisfied, i.e. $\Upsilon_n\subseteq \Upsilon'_n$ and $\mc{B}^{n''}_2$is the minimal sending configuration with respect to $\mc{B}^n_2$ and $\Upsilon_n$.

We build $\mc{D}'_2$ to be $\{\mc{B}^{n''}_2\}_{n \in I}$. We have $\{\mc{B}^{n}_2\}_{n \in I} \mapsto^{*_{\Upsilon_2}}\{\mc{B}^{n''}_2\}_{n \in I}$ with $\Upsilon_1 \subseteq \Upsilon_2$, and $\{\mc{B}^{n''}_2\}_{n \in I} \mapsto^{*_{\Upsilon_3}} \{\mc{B}^{n'}_2\}_{n \in I}$ with $\Upsilon_2 \subseteq \Upsilon_3$.
We, then instantiate the existential quantifier in the goal ($\star$) with $\Upsilon_2$ and $\mc{D}'_2$ that we built for which we know $\{\mc{B}^{n}_2\}_{n \in I} \mapsto^{*_{\Upsilon_2}}\mc{D}'_2$.  We need to show  
\[\begin{array}{l}
    \star'
      \; \; \forall\, x_\gamma \in \mb{Out}( \Delta \Vdash K).\, \mb{if}\, x_\gamma \in \Upsilon_1.\, \mb{then}\,\, (\{\mc{B}^{n}_1\}_{n \in I} ; \{\mc{B}^{n''}_2\}_{n \in I} )\in \mc{V}_{\Psi_0}^\xi\llbracket   \Delta \Vdash K \rrbracket_{\cdot;x_\gamma}^{k'+1}\,\m{and}\\ 
      \; \; \forall\, x_\gamma \in \mb{In}( \Delta \Vdash K). \, \mb{if}\, x_\gamma \in \Theta_1.\, \mb{then}\, \,(\{\mc{B}^{n}_1\}_{n \in I} ; \{\mc{B}^{n''}_2\}_{n \in I} )\in \mc{V}_{\Psi_0}^\xi\llbracket   \Delta \Vdash K \rrbracket_{x_\gamma;\cdot}^{k'+1}.
  \end{array}\]

  \begin{description}
\item{\bf Part 1.} Consider arbitrary  $x_\gamma \in \mb{Out}( \Delta \Vdash K)$ and assume $x_\gamma \in \Upsilon_1$. By the structure of the configurations, for some $n \in I$, $x_\gamma \in \Delta^n, K^n$ and thus $x_\gamma \in \mb{Out}(\Delta^n \Vdash K^n)$ and $x_\gamma \in \Upsilon_n$. The goal is to prove
\[ \star_1\,(\{\mc{B}^{n}_1\}_{n \in I} ; \{\mc{B}^{n''}_2\}_{n \in I} )\in \mc{V}_{\Psi_0}^\xi\llbracket   \Delta \Vdash K \rrbracket_{\cdot;x_\gamma}^{k'+1}\]

\item{\bf Part 2.}
Consider arbitrary  $x_\gamma \in \mb{In}( \Delta \Vdash K)$ and assume $x_\gamma \in \Theta_1$. By the structure of the configurations, for some $n \in I$, $x_\gamma \in \Delta^n, K^n$ and thus $x_\gamma \in \mb{In}(\Delta^n\Vdash K^n)$ and $x_\gamma \in \Theta_n$. The goal is to prove
\[\star_2\,(\{\mc{B}^{n}_1\}_{n \in I} ; \{\mc{B}^{n''}_2\}_{n \in I} )\in \mc{V}_{\Psi_0}^\xi\llbracket   \Delta \Vdash K \rrbracket_{x_\gamma;\cdot}^{k'+1}\]
\end{description}

In both parts, we continue the proof by considering the type of $x_\gamma$. The type of $x_\gamma$ determines whether we need to prove {\bf Part 1.} or {\bf Part 2}. We provide the detailed proof for two interesting cases, the proof of the rest of cases is similar. Also see the proof of similar lemma in~.
\begin{description}
    \item[\bf Subcase 1.]
    $x_\gamma{:}\oplus \{\ell{:}A_\ell\}_{\ell \in L}\langle c, e\rangle \in K$. This case corresponds to {\bf Part 1.} of the goal in which we have $x_\gamma \in \mb{Out}(\Delta \Vdash x _\gamma{:}\oplus \{\ell{:}A_\ell\}_{\ell \in L} \langle c, e\rangle)$ and $x_\gamma \in \Upsilon_1$. By the structure of the configuration, there exists a tree $\mc{B}_1^\kappa$ that provides the root channel $K= K^\kappa= x_\gamma{:}\oplus \{\ell{:}A_\ell\}_{\ell \in L}\langle c, e\rangle $. We use assumption $\dagger^4_n$ for that specific channel ($\dagger^4_\kappa$), we have
\[ \dagger^5_\kappa\,\, \forall\, m.\,(\mc{B}^{\kappa}_1; \mc{B}^{\kappa''}_2)\in \mc{V}_{\Psi_0}^\xi\llbracket  \Delta^\kappa \Vdash K^\kappa\rrbracket_{\cdot;x_\gamma}^{m+1}\]

First, instantiate the forall quantifier with $m=0$.
By {\bf Row $r_2$.} of the logical relation, we have $\Delta^\kappa=\Delta^\kappa_1, \Delta^\kappa_2$ and for some $k_1,k_2\in L$ , we have:
$\mc{B}^\kappa_1 = \mc{B}^{\kappa'}_1\msg{x_\gamma^{\langle c, e \rangle }. k_1} $ and $\mc{B}^{\kappa''}_2= \mc{B}^{\kappa'''}_2\msg{\msg{x_\gamma^{\langle c, e \rangle }. k_2}}$, such that if $c \sqsubseteq \xi,$ we have $k_1=k_2$. Moreover, by well-typedness we know that $A_{k_1}= A{k_2}=B$.

We want to prove \[\circ' \forall \, m.\, (\mc{B}^{\kappa'}_1, \mc{B}^{\kappa'''}_2) \in \mc{E}_{\Psi_0}^\xi\llbracket \Delta^\kappa \Vdash x_{\gamma+1}{:}B \langle c,e \rangle  \rrbracket^m.\]

Consider an arbitrary $m$ given by $\forall \mathbf{I}$ on the goal $\circ'$. Once again, instantiate the quantifier in $\dagger^5_\kappa$, this time with the arbitrary $m$. Again, we get $\Delta^\kappa=\Delta^\kappa_1, \Delta^\kappa_2$ and  for some $k_1, k_2\in L$ , we have:
$\mc{B}^\kappa_1 = \mc{B}^{\kappa'}_1\msg{x_\gamma^{\langle c, e \rangle }. k_1} $ and $\mc{B}^{\kappa''}_2= \mc{B}^{\kappa'''}_2\msg{\msg{x_\gamma^{\langle c, e \rangle }. k_2}}$, such that if $c \sqsubseteq \xi,$ we have $k_1=k_2$.
 Moreover, 
\[(\mc{B}^{\kappa'}_1, \mc{B}^{\kappa'''}_2) \in \mc{E}_{\Psi_0}^\xi\llbracket \Delta^\kappa \Vdash x_{\gamma+1}{:}B\langle c,e \rangle  \rrbracket^m.\]

Since the namings in the configuration is unique, we get the above for the same $j,k$ as we got in the case of $m=0$ and the proof of $\circ'$ is complete.

From this we can prove 
\[\begin{array}{ll}
    \mc{D}_1=\{\mc{B}^n_1\}_{n \in I}= \{\mc{B}^n_1\}_{n \in I \& n \neq \kappa}\,\, \mc{B}^{\kappa'}_1\msg{x^{\langle c, e \rangle}_\gamma. k_1}\,\m{and}\\
    \mc{D}'_2=\{\mc{B}^{n''}_2\}_{n \in I }= \{\mc{B}^{n''}_2\}_{n \in I \& n \neq \kappa}\,\,\mc{B}^{\kappa'''}_2\msg{x^{\langle c, e \rangle}_\gamma. k_2}
\end{array}\]
where $k_1=k_2$ if $c \sqsubseteq \xi$.
First, observe that by the structure of the configuration, there is no tree $\mc{B}_1^n$ or $\mc{B}_2^{n''}$ using $K^\kappa=x_\gamma{:}\oplus \{\ell{:}A_\ell\}_{\ell \in L}\langle c, e\rangle$ as its resource, i.e., $x_\gamma$ is the root.
We can break down the resources $\Delta^\kappa$  as $\Delta^\kappa= \Delta^{\kappa'}, \Delta^{\kappa''}$, such that  $\Delta^{\kappa'}$ is in the interface of $\mc{D}_1$ and $\mc{D}'_2$ and $\Delta^{\kappa''}_1$ are the resources provided by other trees. 
% We can partition $I\backslash \{\kappa\}$ into two disjoint sets $I_1$ and $I_2$ such that the configurations $\{\mc{B}^{n}_1\}_{n \in I_1}$ and $\{\mc{B}^{n''}_2\}_{n \in I_1}$ provide the resources in  $\Delta^{\kappa''}$ and configurations $\{\mc{B}^{n}_1\}_{n \in I_2}$ and $\{\mc{B}^{n''}_2\}_{n \in I_2}$ provide the resources in $\Delta^{\kappa''}_2$.
 In other words, we have $\Delta= \Delta_1, \Delta^{\kappa'}_1$ and $K=x_\gamma{:}\oplus \{\ell{:}A_\ell\}_{\ell \in L}\langle c, e\rangle$ and
\[
\begin{array}{llll}
 (i) &  \Delta_1 \Vdash \{\mc{B}^{n}_1\}_{n \in I_1}:: \Delta^{\kappa''}_1 &\;\Delta^{\kappa'}_1, \Delta^{\kappa''}_1 \Vdash \mc{B}^{\kappa'}_1::x_{\gamma+1}{:}B\langle c,e \rangle & \; \Delta_1, \Delta^{\kappa'}_1  \Vdash \{\mc{B}^{n}_1\}_{n \in I_1} \mc{B}^{\kappa'}_1:: x_{\gamma+1}{:}B\langle c,e \rangle\\
  &  \Delta_1 \Vdash \{\mc{B}^{n''}_2\}_{n \in I_1}:: \Delta^{\kappa''}_1 &\; \Delta^{\kappa'}_1, \Delta^{\kappa''}_1 \Vdash \mc{B}^{\kappa'''}_2::x_{\gamma+1}{:}B\langle c,e \rangle & \;\Delta_1, \Delta^{\kappa'}_1  \Vdash \{\mc{B}^{n''}_2\}_{n \in I_1}\mc{B}^{\kappa'''}_2:: x_{\gamma+1}{:}B\langle c,e \rangle
\end{array}    
\]
We also have 
\[\begin{array}{lll}
  \circ''' &  \mc{D}_1=\{\mc{B}^n_1\}_{n \in I}= \{\mc{B}^n_1\}_{n \in I \& n \neq \kappa}\,\, \mc{B}^{\kappa'}_1\msg{x^{\langle c, e \rangle}_\gamma. k_1}= \{\mc{B}^n_1\}_{n \in I_1}\, \{\mc{B}^n_1\}_{n \in I_2}\, \mc{B}^{\kappa'}_1\msg{x^{\langle c, e \rangle}_\gamma. k_1} \,\m{and}\\
  &  \mc{D}'_2=\{\mc{B}^{n''}_2\}_{n \in I }= \{\mc{B}^{n''}_2\}_{n \in I \& n \neq \kappa}\,\,\mc{B}^{\kappa'''}_2\msg{x^{\langle c, e \rangle}_\gamma. k_2}= \{\mc{B}^{n''}_2\}_{n \in I_1}\{\mc{B}^{n''}_2\}_{n \in I_2}\,\,\mc{B}^{\kappa'''}_2\msg{x^{\langle c, e \rangle}_\gamma. k_2}
\end{array}\]

Recall that earlier we proved 
\[\circ' \, \forall \, m.\, (\mc{B}^{\kappa'}_1, \mc{B}^{\kappa'''}_2) \in \mc{E}_{\Psi_0}^\xi\llbracket \Delta^\kappa \Vdash x_{\gamma+1}{:}B\langle c,e \rangle  \rrbracket^m,\]
and we also have for all $n \in I_1 \cup I_2,$ 
\[\circ_2\,\,\forall m.\, (\mc{B}^n_1;\mc{B}^{n''}_2) \in \mc{E}_{\Psi_0}^\xi\llbracket \Delta^{n}\Vdash K^n\rrbracket^m\]

We can apply the induction hypothesis on the smaller index $k'$, (i), $\circ''$, and $\circ_2$ for those trees indexed in $I_1$ to get 
    \[\circ_3\,(\{\mc{B}^n_1\}_{n \in I_1}\, \mc{B}^{\kappa'}_1;\{\mc{B}^{n''}_2\}_{n \in I_1}\, \mc{B}^{\kappa'''}_2) \in \mc{E}_{\Psi_0}^\xi\llbracket \Delta,\Delta^{\kappa'}\Vdash  x_{\gamma+1}{:}B\langle c,e \rangle\rrbracket^{k'}.\]

By {\bf Row $r_2$.} of the logical relation, this is enough to establish the goal
\[ \star_1\,(\mc{D}_1; \mc{D}'_2)\in \mc{V}_{\Psi_0}^\xi\llbracket  \Delta \Vdash K \rrbracket_{\cdot;x_\gamma}^{k'+1}\] 

\item[\bf Subcase 2.] $\oplus \{\ell{:}A_\ell\}_{\ell \in L}\langle c, e\rangle\in \Delta$, i.e., $\Delta=\Delta', x_\gamma{:}\oplus \{\ell{:}A_\ell\}_{\ell \in L}\langle c, e\rangle$. This case corresponds to {\bf Part 2.} of the goal in which we have $x_\gamma^{\langle c, e\rangle}\in \mb{In}(\Delta\Vdash K)$ and $x_\gamma \in \Theta_1$, and the goal is to prove
$(\{\mc{B}^{n}_1\}_{n \in I} ; \{\mc{B}^{n''}_2\}_{n \in I} )\in \mc{V}\llbracket   \Delta \Vdash K \rrbracket_{x_\gamma;\cdot}^{k'+1}$.

By the structure of the configuration, for some index $\kappa \in I$, we have $x_\gamma{:} \oplus \{\ell{:}A_\ell\}_{\ell \in L}\langle c, e\rangle \in \Delta^\kappa$, i.e. $\Delta^\kappa=\Delta^{\kappa'}, x_\gamma{:} \oplus \{\ell{:}A_\ell\}_{\ell \in L}\langle c, e\rangle$.

By assumption $\dagger^4_\kappa$, we have
\[  \forall\, m.\,(\mc{B}^\kappa_1; \mc{B}^{\kappa''}_2)\in \mc{V}\llbracket  \Delta^{\kappa'}, x_\gamma{:}\oplus \{\ell{:}A_\ell\}_{\ell \in L}\langle c, e\rangle \Vdash K^\kappa \rrbracket_{x_\gamma; \cdot}^{m+1}\]

By {\bf Row $l_3$.} of the logical relation, we get for all $k_1, k_2 \in L$ such that $k_1=k_2$ if $c \sqsubseteq \xi$:

{\small\[ \dagger^5_\kappa\,\,\forall m.\,(\msg{x_\gamma^{\langle c,e \rangle}. k_1}\mc{B}^\kappa_1; \msg{x_\gamma^{\langle c,e \rangle}. k_2}\mc{B}^{\kappa''}_2)\in \mc{E}\llbracket  \Delta^{\kappa'}, y_\beta{:}A,  x_{\gamma+1}{:} B \langle c, e \rangle \Vdash K^\kappa \rrbracket^{m}\]}
Note that by well-typedness of the configuration, we know that $A_{k_1}= A_{k_2}=B$.

Moreover, we have
\[\Delta, x_\gamma{:} B\langle c,e \rangle \Vdash  \msg{x_\gamma^{\langle c,e \rangle}. k_1}\{\mc{B}^{n}_1\}_{n \in I}::K \qquad \Delta,x_\gamma{:} B \langle c,e \rangle\Vdash  \msg{x_\gamma^{\langle c,e \rangle}. k_2}\{\mc{B}^{n''}_2\}_{n \in I}::K \]

Recall that for all $n \in I\backslash \{\kappa\}$ we also have \[\circ_2\,\,\forall m.\, (\mc{B}^{n}_1;\mc{B}^{n''}_2) \in \mc{E}\llbracket \Delta^n\Vdash  K^n\rrbracket^m\]

We can apply the induction hypothesis on the smaller index $k'$, $\circ_2$, and $\dagger^5_\kappa$ to get for all $k_1, k_2 \in L$
, such that $k_1=k_2$ if $c \sqsubseteq \xi$
{\small\[(\msg{x_\gamma^{\langle c,e \rangle}. k_1}\{\mc{B}^{n}_1\}_{n \in I};\msg{x_\gamma^{\langle c,e \rangle}. k_2}\{\mc{B}^{n''}_1\}_{n \in I}) \in \mc{E}\llbracket \Delta,  x_{\gamma+1}{:}B \langle c,e \rangle\Vdash  K\rrbracket^{k'}.\]}
By {\bf Row 9.} of the logical relation, this is enough to establish the goal.

\end{description}

\item{\bf Inductive case ($j=j'+1$).}
Consider an arbitrary index set $I$ and an arbitrary session-typed configurations $\Delta^n\Vdash \mc{B}^n_i:: K^n$ that satisfy the conditions $\dagger_n$. 
We need to show that

\[\begin{array}{l}
    \forall \Upsilon_1, \Theta_1,\,\mc{D}'_1.\, \mathbf{if}\, {\{\mc{B}^n_1\}_{n \in I}}\mapsto^{{j'+1}_{ \Upsilon_1;\Theta_1}} {\mc{D}'_1}\,\m{then}\, \exists \Upsilon_2, \mc{D}'_2 \,\m{such\, that} \, \,\{\mc{B}^n_2\}_{n \in I}\mapsto^{{*}_{ \Upsilon_2}} \mc{D}'_2\, \m{and}\,\Upsilon_1 \subseteq \Upsilon_2 \, \m{and}\\ 
     \; \; \forall\, x_\gamma \in \mb{Out}( \Delta  \Vdash K).\, \mb{if}\, x_\gamma \in \Upsilon_1.\, \mb{then}\,\, (\mc{D}'_1; \mc{D}'_2)\in \mc{V}_{\Psi_0}^\xi\llbracket \Delta \Vdash K \rrbracket_{\cdot;x_\gamma}^{k'+1}\,\m{and}\\ 
     \; \; \forall\, x_\gamma \in \mb{In}( \Delta \Vdash K). \mb{if}\, x_\gamma \in \Theta_1.\, \mb{then} \,(\mc{D}'_1; \mc{D}'_2)\in \mc{V}_{\Psi_0}^\xi\llbracket   \Delta \Vdash K \rrbracket_{x_\gamma;\cdot}^{k'+1}.
 \end{array}\]

We apply a $\forall \mathbf{I}.$ and $\mathbf{if}\,\mathbf{I}.$ on the goal: consider an arbitrary $\Upsilon_1$, $\Theta_1$, and $\mc{D}'_1$ and the first step of ${\{\mc{B}^n_1\}_{n \in I}}\mapsto^{{j'+1}_{ \Upsilon_1}; \Theta_1} {\mc{D}'_1}$. There are two cases to consider:

\item{\bf Case 1. the first step is an internal step but not a communication between the sub-trees.}

Without loss of generality, let's assume that the communication is internal to $\mc{B}^\kappa_1$ for some $\kappa \in I$ and all other trees $\mc{B}^n_1$ for $n \neq \kappa$ remain intact, i.e. we have $\{\mc{B}^n_1\}_{n \in I}=\{\mc{B}^n_1\}_{n \in I_1}\, \mc{B}^{\kappa}_1 \{\mc{B}^n_1\}_{n \in I_2}$ and \[ \{\mc{B}^n_1\}_{n \in I_1}\, \mc{B}^{\kappa}_1 \{\mc{B}^n_1\}_{n \in I_2}\;\mapsto \{\mc{B}^n_1\}_{n \in I_1}, \mc{B}^{\kappa'}_1 \{\mc{B}^n_1\}_{n \in I_2} \;\mapsto^{{j'}_{\Upsilon_1; \Theta_1}} {\mc{D}'_1}.\]

By forward closure (\Cref{apx:lem:fwdfirst}) on the assumption $\dagger_\kappa$ (i.e.,$\dagger_n$for $n = \kappa$), we get
\[\dagger'_\kappa\,\,\forall m.\, (\mc{B}^{\kappa'}_1;\mc{B}_2) \in \mc{E}_{\Psi_0}^\xi\llbracket \Delta^\kappa\Vdash K^\kappa \rrbracket^m.\]

Recall that by $\dagger_n$, we also have for $n \neq \kappa$
\[\forall m.\, (\mc{B}^n_1;\mc{B}^n_2) \in \mc{E}_{\Psi_0}^\xi\llbracket \Delta^n \Vdash K^n\rrbracket^m.\]
 We can apply the induction hypothesis on the number of steps $j'$ to get:
 \[\begin{array}{l}
    \forall \Upsilon_1, \Theta_1, \,\mc{D}'_1.\,\forall\, {\Upsilon_1}. \, \mathbf{if}\,  \{\mc{B}^n_1\}_{n \in I_1}\, \mc{B}^{\kappa'}_1 \{\mc{B}^n_1\}_{n \in I_2} \mapsto^{{j'}_{ \Upsilon_1; \Theta_1}} {\mc{D}'_1}\,\m{then}\, \exists \Upsilon_2\,\mc{D}'_2. \, \,\{\mc{B}^n_2\}_{n \in I}\mapsto^{{*}_{ \Upsilon_2}} \mc{D}'_2\, \m{and}\, \Upsilon_1 \subseteq \Upsilon_2\,\m{and}\\ 
     \; \; \forall\, x_\gamma \in \mb{Out}(\Delta  \Vdash K).\, \mb{if}\, x_\gamma \in \Upsilon_1.\, \mb{then}\,\, (\mc{D}'_1; \mc{D}'_2)\in \mc{V}_{\Psi_0}^\xi\llbracket   \Delta \Vdash K \rrbracket_{\cdot;x_\gamma}^{k'+1}\,\m{and}\\ 
     \; \; \forall\, x_\gamma \in \mb{In}( \Delta  \Vdash K).\,\mb{if}\, x_\gamma \in \Theta_1.\, \mb{then}\, (\mc{D}'_1; \mc{D}'_2)\in \mc{V}_{\Psi_0}^\xi\llbracket  \Delta \Vdash K \rrbracket_{x_\gamma;\cdot}^{k'+1}.
 \end{array}\]

 Since $\{\mc{B}^n_1\}_{n \in I_1}\, \mc{B}^{\kappa}_1 \{\mc{B}^n_1\}_{n \in I_2}\;\mapsto \{\mc{B}^n_1\}_{n \in I_1}, \mc{B}^{\kappa'}_1 \{\mc{B}^n_1\}_{n \in I_2} \;\mapsto^{{j'}_{\Upsilon_1; \Theta_1}} {\mc{D}'_1},$
 %note that $ \{\mc{B}^n_1\}_{n \in I_1}, \mc{B}^{\kappa}_1 \{\mc{B}^n_1\}_{n \in I_2}\;\mapsto \{\mc{B}^n_1\}_{n \in I_1}, \mc{B}^{\kappa'}_1 \{\mc{B}^n_1\}_{n \in I_2}$, we have for all $\,\mc{D}'_1$ and  ${\Upsilon_1}$ if $ \{\mc{B}^n_1\}_{n \in I_1}, \mc{B}^{\kappa'}_1 \{\mc{B}^n_1\}_{n \in I_2}\mapsto^{{j'}_{ \Upsilon_1}} {\mc{D}'_1}$  then $ \{\mc{B}^n_1\}_{n \in I_1}, \mc{B}^{\kappa}_1 \{\mc{B}^n_1\}_{n \in I_2}\mapsto^{{j'+1}_{ \Upsilon_1}} {\mc{D}'_1}$.
we can prove the goal:
\[\begin{array}{l}
    \exists \Upsilon_2, \mc{D}'_2 \,\m{such\, that} \, \,\{\mc{B}^n_2\}_{n \in I}\mapsto^{{*}_{ \Upsilon_2}} \mc{D}'_2\, \m{and}\, \Upsilon_1 \subseteq \Upsilon_2\, \m{and}\\ 
     \; \; \forall\, x_\gamma \in \mb{Out}(\Delta  \Vdash K).\, \mb{if}\, x_\gamma \in \Upsilon_1.\, \mb{then}\,\, (\mc{D}'_1; \mc{D}'_2)\in \mc{V}_{\Psi_0}^\xi\llbracket   \Delta \Vdash K \rrbracket_{\cdot;x_\gamma}^{k'+1}\,\m{and}\\ 
     \; \; \forall\, x_\gamma \in \mb{In}( \Delta  \Vdash K).\,\mb{if}\, x_\gamma \in \Theta_1.\, (\mc{D}'_1; \mc{D}'_2)\in \mc{V}_{\Psi_0}^\xi\llbracket  \Delta \Vdash K \rrbracket_{x_\gamma;\cdot}^{k'+1}.
 \end{array}\]
which completes the proof of this case.

\item{\bf Case 2. the first step is a communication between two sub-configurations.} Without loss of generality, we assume that the communication is between trees indexed by $\kappa, \lambda \in I$, i.e., $\mc{B}^\kappa_1$ offers a resource $u_\alpha{:}T$ and $\mc{B}^\lambda_1$ uses the resource, and there is a message available along $u_\alpha$ that is received in the first step.
The proof proceeds by case analysis on type of $u_\alpha^c{:}T$. We provide the detailed proof for one case. The proof of the rest of the cases is similar.

\begin{description}
    \item{\bf  Subcase 1. $T= \oplus\{\ell{:}A_\ell\}_{\ell\in L}$}, i.e., we have
    \[\begin{array}{llll}
       (i)& \Delta^\kappa \Vdash \mc{B}^{\kappa}_1:: K^\kappa & \quad \m{where}\,\, K^\kappa= u_\alpha{:}\oplus\{\ell{:}A_\ell\}_{\ell\in L}\langle c,e \rangle \\ &\Delta^\kappa \Vdash \mc{B}^{\kappa}_2 ::K^\kappa &  \\[4pt]
       (ii)& \Delta^\lambda \Vdash \mc{B}^{\lambda}_1:: K^\lambda & \quad \m{where}\,\, \Delta^\lambda= \Delta^{\lambda'}, u_\alpha{:}\oplus\{\ell{:}A_\ell\}_{\ell\in L}\langle c,e \rangle  \\ &\Delta^\lambda \Vdash \mc{B}^{\lambda}_2 ::K^\lambda& \\\
    \end{array}
       \]
    
    By the assumptions of this case, we know that there is a message sent along $u_\alpha^{\langle c,e \rangle}{:}A \otimes B$ in $\mc{B}^\kappa_1$ that is received by a process in $\mc{B}^\lambda_1$, i.e., $\mc{B}^{\kappa}_1= \mc{B}^{\kappa'}_1\mb{msg}( u_\alpha^{\langle c,e \rangle}. k_1)$ and
    
    \[\begin{array}{l}
        \{\mc{B}^n_1\}_{n \in I}=\{\mc{B}^n_1\}_{n \in I_1}, \mc{B}^{\kappa}_1 \mc{B}^{\lambda}_1 \{\mc{B}^n_1\}_{n \in I_2}= \{\mc{B}^n_1\}_{n \in I_1}, \mc{B}^{\kappa'}_1\mb{msg}( u_\alpha^{\langle c,e \rangle}. k_1) \mc{B}^{\lambda}_1 \{\mc{B}^n_1\}_{n \in I_2}\,\, \m{and}\\
        \{\mc{B}^n_1\}_{n \in I_1}, \mc{B}^{\kappa'}_1 \mb{msg}( u_\alpha^{\langle c,e \rangle}. k_1) \mc{B}^{\lambda}_1 \{\mc{B}^n_1\}_{n \in I_2}\;\mapsto \{\mc{B}^n_1\}_{n \in I_1}, \mc{B}^{\kappa'}_1 \mc{B}^{\lambda'}_1 \{\mc{B}^n_1\}_{n \in I_2} \;\mapsto^{{j'}_{\Upsilon_1; \Theta_1}} {\mc{D}'_1}.
       
    \end{array} \]
    
    By $\dagger_{\kappa}$, \Cref{apx:lem:moving-existential}, and $\mc{B}^{\kappa'}_1\mb{msg}( u_\alpha^{\langle c,e \rangle}. k_1)\mapsto^{0_{\Upsilon'_\kappa; \Theta'_\kappa}}\mc{B}^{\kappa'}_1 \mb{msg}( u_\alpha^{\langle c,e \rangle}. k_1)$, we get
    
    {\small\[\begin{array}{ll}
        \,\exists \Upsilon_{\kappa_2}\mc{B}^{\kappa'}_2 \,\m{such\, that} \, \,\mc{B}^\kappa_2\mapsto^{{*}_{ \Upsilon_{\kappa_2}}} \mc{B}^{\kappa'}_2\, \m{and}\, \Upsilon'_\kappa \subseteq \Upsilon_{\kappa_2}\\ 
        \; \; \forall\, x_\gamma \in \mb{Out}(\Delta^\kappa \Vdash u_{\alpha}{:}\oplus\{\ell{:}A_\ell\}_{\ell\in L}\langle c,e \rangle ).\, \mb{if}\, x_\gamma \in \Upsilon'_\kappa.\,\\
        \quad \mb{then}\,\,\forall\, m.\, (\mc{B}^{\kappa'}_1 \mb{msg}( u_\alpha^{\langle c,e \rangle}. k_1); \mc{B}^{\kappa'}_2)\in \mc{V}\llbracket  \Delta^ \kappa \Vdash u_{\alpha}{:}\oplus\{\ell{:}A_\ell\}_{\ell\in L}\langle c,e \rangle\rrbracket_{\cdot;x_\gamma}^{m+1}\,\m{and}\\ 
        \; \; \forall\, x_\gamma \in \mb{In}(\Delta^\kappa \Vdash u_{\alpha}{:} \oplus\{\ell{:}A_\ell\}_{\ell\in L}\langle c,e \rangle).\, \mb{if}\, x_\gamma \in \Theta'_\kappa.\,\\
        \quad  \mb{then}\, \forall\,m. \,(\mc{B}^{\kappa'}_1 \mb{msg}( u_\alpha^{\langle c,e \rangle}. k_1); \mc{B}^{\kappa'}_2)\in \mc{V}\llbracket \Delta^\kappa \Vdash u_{\alpha}{:}\oplus\{\ell{:}A_\ell\}_{\ell\in L} \langle c,e \rangle\rrbracket_{x_\gamma;\cdot}^{m+1}
    \end{array}\]}

    In particular, we know that $u_\alpha{:}\oplus\{\ell{:}A_\ell\}_{\ell\in L} \langle c,e \rangle \in \mb{Out}(\Delta^\kappa \Vdash  u_{\alpha}{:} \oplus\{\ell{:}A_\ell\}_{\ell\in L} \langle c,e \rangle)$ and $u_\alpha \in \Upsilon'_\kappa$, which gives us:
    \[\begin{array}{ll}
    \forall\, m.\, (\mc{B}^{\kappa'}_1 \mb{msg}( u_\alpha^{\langle c,e \rangle}. k_1); \mc{B}^{\kappa'}_2)\in \mc{V}\llbracket  \Delta^\kappa \Vdash u_{\alpha}{:}\oplus\{\ell{:}A_\ell\}_{\ell\in L} \langle c,e \rangle\rrbracket_{\cdot;u_\alpha}^{m+1}
    \end{array}\]
    
    By {\bf Row $r_2$.} of the logical relation, we have for $k_2 \in L$ such that $k_1=k_2$ if $c \sqsubseteq \xi$
    \[\begin{array}{ll}
        \mc{B}^{\kappa'}_2=\mc{B}^{\kappa''}_2\mb{msg}( u_\alpha^{\langle c,e \rangle}. k_2),\,\m{and}\\
      \circ'\,  \forall\, m.\, (\mc{B}^{\kappa'}_1 ; \mc{B}^{\kappa''}_2)\in \mc{E}_{\Psi_0}^\xi\llbracket  \Delta^{\kappa}_1 \Vdash u_{\alpha+1}{:}B\rrbracket^{m},\\
    \end{array}\]
    where $B=A_{k_1}=A_{k_2}$, which by well-typedness we know it's true.

    Next, we consider $\mc{B}^\lambda_1$.
    By $\dagger_\lambda$, \Cref{apx:lem:moving-existential}, and $\mc{B}^\lambda_1 \mapsto^{0_{\Upsilon'_\lambda}}\mc{B}^\lambda_1$, we get
    \[\begin{array}{ll}
        \,\exists \Upsilon_{\lambda_2} \mc{B}^{\lambda'}_2 \,\m{such\, that} \, \,\mc{B}^{\lambda}_2\mapsto^{{*}_{ \Upsilon_{\lambda_2}}} \mc{B}^{\lambda'}_2\, \m{and}\\ 
        \; \; \forall\, x_\gamma \in \mb{Out}(\Delta^{\lambda'}, u_{\alpha}{:}\oplus\{\ell{:}A_\ell\}_{\ell\in L}\langle c,e \rangle \Vdash K^\lambda).\, \mb{if}\, x_\gamma \in \Upsilon'_\lambda.\,\\
        \quad\, \mb{then}\,\,\forall\, m.\, (\mc{B}^\lambda_1; \mc{B}^{\lambda'}_2)\in \mc{V}\llbracket  \Delta^{\lambda'}, u_{\alpha}{:}\oplus\{\ell{:}A_\ell\}_{\ell\in L}\langle c,e \rangle \Vdash K^\lambda \rrbracket_{\cdot;x_\gamma}^{m+1}\,\m{and}\\ 
        \; \; \forall\, x_\gamma \in \mb{In}(\Delta^{\lambda'}, u_{\alpha}{:} \oplus\{\ell{:}A_\ell\}_{\ell\in L} \langle c,e \rangle \Vdash K^\lambda). \,\mb{if}\, x_\gamma \in \Theta'_\lambda.\, \\
        \, \quad \mb{then}\,\,\forall\,m. \,(\mc{B}^{\lambda}_1; \mc{B}^{\lambda'}_2)\in \mc{V}\llbracket  \Delta^{\lambda}, u_{\alpha}{:} \oplus\{\ell{:}A_\ell\}_{\ell\in L} \langle c, e \rangle  \Vdash K^{\lambda} \rrbracket_{x_\gamma;\cdot}^{m+1}
    \end{array}\]
    
    In particular, we know that $u_\alpha{:}\oplus\{\ell{:}A_\ell\}_{\ell\in L} \langle c,e \rangle\in \mb{In}(\Delta^{\lambda'}, u_{\alpha}{:} \oplus\{\ell{:}A_\ell\}_{\ell\in L}\langle c,e \rangle \Vdash K^\lambda)$ and also $u_\alpha \in \Theta'_\lambda$, which gives us:
    \[\begin{array}{ll}
    \forall\, m.\, (\mc{B}^\lambda_1; \mc{B}^{\lambda'}_2)\in \mc{V}\llbracket  \Delta^{\lambda'}, u_{\alpha}{:}\oplus\{\ell{:}A_\ell\}_{\ell\in L}\langle c,e \rangle \Vdash K^{\lambda}\rrbracket_{u_\alpha; \cdot}^{m+1}
    \end{array}\]
    
    By {\bf Row $l_2$.} of the logical relation, knowing that $k_1=k_2$ if $c \sqsubseteq \xi$, we have
    
    \[\begin{array}{ll}
        \forall m.\, (\mb{msg}( u_\alpha^{\langle c,e \rangle}. k_1); \mb{msg}( u_\alpha^{\langle c,e \rangle}. k_2)\mc{B}^{\lambda'}_2)\in \mc{E}_{\Psi_0}^\xi\llbracket  \Delta^\lambda, u_{\alpha+1}{:}B \Vdash K^\lambda \rrbracket^{m}.\\
    \end{array}\]

    By forward closure (\Cref{apx:lem:fwdfirst}) and $\mb{msg}( u_\alpha^{\langle c,e \rangle}. k_1)\mc{B}^\lambda_1 \mapsto \mc{B}^{\lambda'}_1$ we have:
    \[\begin{array}{ll}
       \circ''\,\, \forall m.\, (\mc{B}^{\lambda'}_1 ;\mb{msg}( u_\alpha^{\langle c,e \rangle}. k_2)\mc{B}^{\lambda'}_2)\in \mc{E}_{\Psi_0}^\xi\llbracket  \Delta,  u_{\alpha+1}{:}B \Vdash K \rrbracket^{m}\\
    \end{array}\]

    By induction on the number of steps, $\circ'$, $\circ''$, and $\dagger_n$ for $n \neq \kappa, \lambda$ we get 

    \[\begin{array}{l}
        \forall \Upsilon_1, \Theta_1,\,\mc{D}'_1. \, \mathbf{if}\, {\{\mc{B}^n_1\}_{n \in I_1} \mc{B}^{\kappa'}_1 \mc{A}_1\mc{B}^{\lambda'}_1 \{\mc{B}^n_1\}_{n \in I_2}}\mapsto^{{j'}_{ \Upsilon_1; \Theta_1}} {\mc{D}'_1}\\\,\mb{then}\, \exists \Upsilon_2, \mc{D}'_2 \,\m{such\, that} \, \,\{\mc{B}^n_2\}_{n \in I_1} \mc{B}^{\kappa''}_2 \mb{msg}( u_\alpha^{\langle c,e \rangle}. k_2)\mc{B}^{\lambda'}_2 \{\mc{B}^n_2\}_{n \in I_2} 
        \mapsto^{{*}_{ \Upsilon_2}} \mc{D}'_2\, \m{and}\Upsilon_1 \subseteq \Upsilon_2\, \m{and}\\[2pt]
         \; \; \forall\, x_\gamma \in \mb{Out}(\Delta  \Vdash K).\, \mb{if}\, x_\gamma \in \Upsilon_1.\, \mb{then}\,\, (\mc{D}'_1; \mc{D}'_2)\in \mc{V}_{\Psi_0}^\xi\llbracket   \Delta \Vdash K \rrbracket_{\cdot;x_\gamma}^{k'+1}\,\m{and}\\[2pt] 
         \; \; \forall\, x_\gamma \in \mb{In}( \Delta  \Vdash K).\,\mb{if}\, x_\gamma \in \Theta_1.\, \mb{then}\,\, (\mc{D}'_1; \mc{D}'_2)\in \mc{V}_{\Psi_0}^\xi\llbracket  \Delta \Vdash K \rrbracket_{x_\gamma;\cdot}^{k'+1}.
     \end{array}\]
    
    By a $\forall \mathbf{E}.$ and ${\{\mc{B}^n_1\}_{n \in I_1} \mc{B}^{\kappa'}_1 \mc{B}^{\lambda'}_1 \{\mc{B}^n_1\}_{n \in I_2}}\mapsto^{{j'}_{ \Upsilon_1}} {\mc{D}'_1}$ we get:
    \[\begin{array}{l}
        \exists \Upsilon_2, \mc{D}'_2 \,\m{such\, that} \, \,\{\mc{B}^n_2\}_{n \in I_1} \mc{B}^{\kappa''}_2 \mb{msg}( u_\alpha^{\langle c,e \rangle}. k_2)\mc{B}^{\lambda'}_2 \{\mc{B}^n_2\}_{n \in I_2} 
        \mapsto^{{*}_{ \Upsilon_2}} \mc{D}'_2\, \m{and}\Upsilon_1 \subseteq \Upsilon_2\, \m{and}\\ 
         \; \; \forall\, x_\gamma \in \mb{Out}(\Delta  \Vdash K).\, \mb{if}\, x_\gamma \in \Upsilon_1.\, \mb{then}\,\, (\mc{D}'_1; \mc{D}'_2)\in \mc{V}_{\Psi_0}^\xi\llbracket   \Delta \Vdash K \rrbracket_{\cdot;x_\gamma}^{k'+1}\,\m{and}\\ 
         \; \; \forall\, x_\gamma \in \mb{In}( \Delta  \Vdash K).\,\mb{if}\, x_\gamma \in \Theta_1.\, \mb{then}\,\, (\mc{D}'_1; \mc{D}'_2)\in \mc{V}_{\Psi_0}^\xi\llbracket  \Delta \Vdash K \rrbracket_{x_\gamma;\cdot}^{k'+1}.
     \end{array}\]
    
    We also know that 
    
    $\{\mc{B}^n_2\}_{n \in I}=\{\mc{B}^n_2\}_{n \in I_1} \mc{B}^{\kappa}_2 \mc{B}^{\lambda}_2 \{\mc{B}^n_2\}_{n \in I_2}\mapsto^* \{\mc{B}^n_2\}_{n \in I_1} \mc{B}^{\kappa''}_2   \mb{msg}(\mb{send}y_\beta\, u_\alpha) \mc{B}^{\lambda'}_2 \{\mc{B}^n_2\}_{n \in I_2}.$
    We get the following for the same $\Upsilon_2$ and $\mc{D}'_2$:
    \[\begin{array}{l}
        \exists \Upsilon_2, \mc{D}'_2 \,\m{such\, that} \, \,\{\mc{B}^n_2\}_{n \in I_1} \mc{B}^{\kappa''}_2  \mb{msg}(\mb{send}y_\beta\, u_\alpha)\mc{B}^{\lambda'}_2 \{\mc{B}^n_2\}_{n \in I_2} 
        \mapsto^{{*}_{ \Upsilon_2}} \mc{D}'_2\, \m{and}\Upsilon_1 \subseteq \Upsilon_2\, \m{and}\\ 
         \; \; \forall\, x_\gamma \in \mb{Out}(\Delta  \Vdash K).\, \mb{if}\, x_\gamma \in \Upsilon_1.\, \mb{then}\,\, (\mc{D}'_1; \mc{D}'_2)\in \mc{V}_{\Psi_0}^\xi\llbracket   \Delta \Vdash K \rrbracket_{\cdot;x_\gamma}^{k'+1}\,\m{and}\\ 
         \; \; \forall\, x_\gamma \in \mb{In}( \Delta  \Vdash K).\,\mb{if}\, x_\gamma \in \Theta_1.\, \mb{then}\,\, (\mc{D}'_1; \mc{D}'_2)\in \mc{V}_{\Psi_0}^\xi\llbracket  \Delta \Vdash K \rrbracket_{x_\gamma;\cdot}^{k'+1}.
     \end{array}\]

    and the proof of this subcase is complete.
\end{description}
\end{description}

\end{proof}

\section{Adequacy as a bisimulation}

\begin{figure}
  
    \begin{flalign*}
      &\textit{Internal transition $\xrightarrow{\tau}$ defined as:}&\\
      &\mc{D}_1\xrightarrow{\tau} \mc{D}_1'\, \m{iff}\, \mc{D}_1 \mapsto \mc{D}'_1& \\[10pt]
      &\textit{Actions $\xrightarrow{\overline{y_\alpha}\,q}$,   $\xrightarrow{{\mathbf{L}\,y_\alpha\,q}}$ and $\xrightarrow{{\mathbf{R}\,y_\alpha\,q}}$ defined as below when $y\in \mathbf{fn}(\mc{D}_1)$: }&
    \end{flalign*}
       \[\begin{array}{llll}
      (1)\, \mc{D}_1{\mb{msg}(\mb{close}\,y_\alpha)}\,{\mc{D}_2} & \xrightarrow{\overline{y_\alpha}\,\mb{close}} & \mc{D}_1\mc{D}_2 \\[10pt]
      (2)\mc{D}_1\,\mb{msg}(y_\alpha.k)\mc{D}_2&   \xrightarrow{\overline{y_\alpha}\,k} & \mc{D}_1\mc{D}_2 \\[10pt]
      (3) \mc{D}_1\, \mb{msg}(\mb{send}\,x_\beta\,y_\alpha)\mc{D}_2&   \xrightarrow{\overline{y_\alpha}\,x_\beta} & \mc{D}_1\mc{D}_2 \\[10pt]
      (4)\, \mc{D}_1\mb{proc}(z_\delta, \mathbf{wait}\, y_\alpha; P){\mc{D}_2} & \xrightarrow{{\mathbf{L}\,y_\alpha\,\mb{close}}} & {\mb{msg}(\mb{close}\,y_\alpha)\,\mc{D}_1\mb{proc}(z_\delta, \mathbf{wait}\, y_\alpha;P)\,\mc{D}_2}\\[10pt]
       (5)\,\mc{D}_1\mb{proc}(z_\delta, \mathbf{case}\,y_\alpha\,(\ell \Rightarrow P_\ell)_{\ell \in I})\, \mc{D}_2  & \xrightarrow{{\mathbf{L}\,y_\alpha\,k}} & {\mb{msg}(y_\alpha.k)\mc{D}_1 \mb{proc}(z_\delta, \mathbf{case}\,y_\alpha\,(\ell \Rightarrow P_\ell)_{\ell \in I}) \mc{D}_2} \\[10pt]
       (6) \mc{D}_1 \,\mb{proc}(z_\delta, w \leftarrow \mathbf{recv}\,y_\alpha) \mc{D}_2& \xrightarrow{{\mathbf{L}\,y_\alpha\,x_\beta}} &  \mb{msg}(\mb{send}\,x_\beta\,y_\alpha)\mc{D}_1\, \mb{proc}(z_\delta,  w \leftarrow \mathbf{recv}\,y_\alpha) \,\mc{D}_2  \\[10pt]
      %  (5)\mc{D}_1\mb{msg}(y_\alpha.k)\, \mc{D}_2&   \xrightarrow{y_\alpha.k} & \mc{D}_1\mc{D}_2 \\[10pt]
       (7) \mc{D}_1 \,\mb{proc}(y_\alpha, \mathbf{case}\,y_\alpha\,(\ell \Rightarrow P_\ell)_{\ell \in I}) \mc{D}_2& \xrightarrow{{\mathbf{R}\,y_\alpha\,k}} &  \mc{D}_1\, \mb{proc}(y_\alpha, \mathbf{case}\,y_\alpha\,(\ell \Rightarrow P_\ell)_{\ell \in I}) \,\mc{D}_2  \mb{msg}(y_\alpha.k) \\[10pt]
       (8) \mc{D}_1\mb{proc}(y_\alpha, w \leftarrow \mathbf{recv}\,y_\alpha)\,\mc{D}_2 \, & \xrightarrow{{\mathbf{R}\,y_\alpha\,x_\beta}} &    \mc{D}_1\mb{proc}(y_\alpha,  w \leftarrow \mathbf{recv}\,y_\alpha) \mc{D}_2 \, \mb{msg}(\mb{send}\,x_\beta\,y_\alpha) 
      \end{array}\]
        \caption{The transition rules.  An overline indicates that an outgoing message, an otherwise the message is incoming.}
         \label{apx:fig:LTSdef}
      \end{figure}

\begin{definition}[Free names of a configuration]
    for $\Delta \Vdash \mc{D}:: \Delta'$, we define $\m{fn}(\mc{D})$ as $\mathit{dom}(\Delta,\Delta')$.
  \end{definition}
  
  \begin{definition}[Weak transition relations]$\,$\\
  
    \begin{enumerate}
      \item[(1)] $\xRightarrow{}$ is the reflexive and transitive closure of $\xrightarrow{\tau}$.
      \item[(2)] $\xRightarrow{\alpha}$ is $\xRightarrow{} \xrightarrow{\alpha}$.
    \end{enumerate}

  \end{definition}
  \begin{definition}[Asynchronous bisimilarity]\label{apx:def:asynchronousbisim}
  Asynchronous bisimilarity is the largest symmetric relation such that whenever { $\mc{D}_1 \approx_a \mc{D}_2$}, we have
  \begin{enumerate}
    \item{$(\tau-\mathtt{step})$} if  { $\mc{D}_1 \xrightarrow{\tau} \mc{D}'_1$} then {$\exists \mc{D}'_2. \, \mc{D}_2 \xRightarrow{\tau} \mc{D}'_2$} and  {$\mc{D}'_1 \approx_a \mc{D}'_2,$}
    \item{$(\m{output})$} if $\mc{D}_1 \xrightarrow{\overline{x_\alpha}\,q} \mc{D}'_1$ then $\exists \mc{D}'_2. \, \mc{D}_2 \xRightarrow{\overline{x_\alpha}\,q} \mc{D}'_2$ and {$\mc{D}'_1 \approx_a \mc{D}'_2.$}
    \item{$(\m{left\, input})$} for all $q\not \in \m{fn}(\mc{D}_1).$, if $\mc{D}_1 \xrightarrow{{\mathbf{L}\, x_\alpha\,q}} \mc{D}'_1$ then  $\exists \mc{D}'_2. \, \mc{D}_2 \xRightarrow{\tau} \mc{D}'_2$ and {$\mc{D}'_1 \approx_a \mathbf{msg}(x_\alpha.q) \mc{D}'_2,$}
  \item{$(\m{right\, input})$} for all $q\not \in \m{fn}(\mc{D}_1).$, if $\mc{D}_1 \xrightarrow{{\mathbf{R}\, x_\alpha\,q}} \mc{D}'_1$ then  $\exists \mc{D}'_2. \, \mc{D}_2 \xRightarrow{\tau} \mc{D}'_2$ and   {$\mc{D}'_1 \approx_a \mc{D}'_2\mathbf{msg}(x_\alpha.q).$}
  \end{enumerate}
  where $\mb{msg}(x_\alpha.q)$ is defined as $\mb{msg}(\mathbf{close}\, x_\alpha)$ if $q=\mathbf{close}$, $\mb{msg}(x_\alpha.k)$ if $q=k$, and $\mb{msg}(\mathbf{send}\,z_\delta\, x_\alpha)$ if $q=z_\delta$.

  \end{definition}

  \begin{definition}[Typing context projections -integrity]\label{apx:def:typing-proj-r}
    Downward projection on linear contexts and $K$ is defined as follows:
    \[\begin{array}{lclc}
    \Delta, x_\alpha{:}T\langle c, e \rangle \Downarrow^{\m{ig}} \xi& :=& \Delta \Downarrow^{\m{ig}}\xi,  x_\alpha{:}T\langle c, e \rangle  & \m{if}\; e\sqsubseteq \xi  \\
         \Delta, x_\alpha{:}T\langle c, e \rangle  \Downarrow^{\m{ig}}\xi& :=& \Delta \Downarrow^{\m{ig}} \xi & \m{if}\;e \not \sqsubseteq \xi  \\
        \cdot \Downarrow^{\m{ig}} \xi& :=& \cdot &  \\
         x_\alpha{:}T\langle c, e \rangle  \Downarrow^{\m{ig}} \xi& :=&   x_\alpha{:}T\langle c, e \rangle  & \m{if}\; e\sqsubseteq \xi  \\
        x_\alpha{:}T\langle c, e \rangle  \Downarrow^{\m{ig}} \xi& :=& \_{:}1\langle \top, \top \rangle & \m{if}\; e\not\sqsubseteq \xi \\
    \end{array}\]
  
\end{definition}
    \begin{definition}[Typing context projections -confidentiality]\label{apx:def:typing-proj-cf}
        Downward projection on linear contexts and $K$ is defined as follows:
        \[\begin{array}{lclc}
        \Delta, x_\alpha{:}T\langle c, e \rangle \Downarrow^{\m{cf}} \xi& :=& \Delta \Downarrow^{\m{cf}}\xi,  x_\alpha{:}T\langle c, e \rangle  & \m{if}\; c\sqsubseteq \xi  \\
             \Delta, x_\alpha{:}T\langle c, e \rangle  \Downarrow^{\m{cf}}\xi& :=& \Delta \Downarrow^{\m{cf}} \xi & \m{if}\;c \not \sqsubseteq \xi  \\
            \cdot \Downarrow^{\m{cf}} \xi& :=& \cdot &  \\
             x_\alpha{:}T\langle c, e \rangle  \Downarrow^{\m{cf}} \xi& :=&   x_\alpha{:}T\langle c, e \rangle  & \m{if}\; c\sqsubseteq \xi  \\
            x_\alpha{:}T\langle c, e \rangle  \Downarrow^{\m{cf}} \xi& :=& \_{:}1\langle \top, \top \rangle & \m{if}\; c\not\sqsubseteq \xi \\
        \end{array}\]
    \end{definition}   

  \begin{definition}[High confidentiality provider and High confidentiality client]
    \[\begin{array}{lll}
      \cdot \,\,\in \mb{H\text{-}CProvider}^\xi(\cdot) \\[5pt]
      \mc{B} \in \mb{H\text{-}CProvider}^{\xi}(\Delta, x_\alpha{:}A\langle c, e \rangle )  & \m{iff} & c \not \sqsubseteq \xi  \, \m{and}\,
      \mc{B} = \mc{B}' \mc{T} \,\, \m{and}\,\, \mc{B}' \in \mb{H\text{-}CProvider}^\xi(\Delta)\\
      &&\,\,\m{and}\,\, \mc{T} \in \m{Tree}(\cdot \Vdash x_\alpha{:}A\langle c, e \rangle), \mb{or}\\
      &&c \sqsubseteq \xi  \, \m{and}\,
     \mc{B} \in \mb{H\text{-}CProvider}^\xi(\Delta)\\[10pt]
  
      \mc{T} \in \mb{H\text{-}CClient}^\xi( x_\alpha{:}A\langle c, e \rangle )\,  & \m{iff} & c \not \sqsubseteq \xi  \, \m{and}\, \mc{T} \in \m{Tree}(x_\alpha{:}A\langle c, e \rangle \Vdash \_:1\langle \top, \top\rangle), \mb{or}\\
      &&c \sqsubseteq \xi  \, \m{and}\, \mc{T}= \cdot
    \end{array}\]

  \end{definition}

  \begin{definition}[Related nigh integrity provider and  client]
    \[\begin{array}{lll}
      (\cdot, \cdot) \,\,\in \mb{Hrel\text{-}IProvider}^\xi(\cdot) \\[5pt]
      (\mc{B}_1, \mc{B}_2) \in \mb{Hrel\text{-}IProvider}^{\xi}(\Delta, x_\alpha{:}A\langle c, e \rangle )  & \m{iff} & e \not \sqsubseteq \xi  \, \m{and}\,
      \mc{B}_i = \mc{B}'_i \mc{T}_i \,\, \m{and}\,\, \mc{B}' \in \mb{Hrel\text{-}IProvider}^\xi(\Delta), \mb{or}\\[3pt]
      &&e \sqsubseteq \xi  \, \m{and} \mc{B}_i = \mc{B}'_i \mc{T}_i \,\, \m{and}\,\, \mc{B}' \in \mb{Hrel\text{-}IProvider}^\xi(\Delta)\\[0pt]
      &&\,\,\quad \quad\m{and}\,\, \cdot \Vdash\mc{T}_1 \equiv_\xi^{\Psi_0} \mc{T}_2:: x_\alpha{:}A\langle c, e \rangle.\\[10pt]
      (\cdot, \cdot) \in \mb{Hrel\text{-}IClient}^\xi( \_\langle \top, \top \rangle )\\[5pt]
      (\mc{T}_1, \mc{T}_2) \in \mb{Hrel\text{-}IClient}^\xi( x_\alpha{:}A\langle c, e \rangle )\,  & \m{iff} & e \not \sqsubseteq \xi  \, \m{and}\, , \mb{or}\\
      &&e \sqsubseteq \xi  \, \m{and}\, x_\alpha{:}A\langle c, e \rangle  \Vdash \mc{T}_1 \equiv_{\xi}^{\Psi_0} \mc{T}_2 :: \_:1\langle \top, \top\rangle
    \end{array}\]

  \end{definition}

\begin{definition}\label{apx:def:simequiv}
    For $\mc{D}_1\in \m{Tree}({\Delta_1} \Vdash x_\alpha {:}A_1\langle c_1 , e_1 \rangle)$, $\mc{D}_2\in \m{Tree}({\Delta_2} \Vdash y_\beta {:}A_2\langle c_2 , e_2 \rangle)$ we define
     \[\Delta_1 \Vdash \mc{D}_1 :: x_\alpha {:}A_1\langle c_1 , e_1 \rangle \approx^\xi_a \Delta_2 \Vdash \mc{D}_2:: y_\beta {:}A_2\langle c_2 , e_2 \rangle\,\,\m{as}\]
  \[\begin{array}{ll}
    \Delta= \Delta_1 \Downarrow^{\m{cf}} \xi= \Delta_2  \Downarrow^{\m{cf}} \xi\,\,\,\m{and}\,\,\, K=y_\beta {:}A_2\langle c_2 , e_2 \rangle \Downarrow^{\m{cf}} \xi =  x_\alpha {:}A_1\langle c_1 , e_1 \rangle \Downarrow^{\m{cf}} \xi\,\, \,\m{and}\, \\
    \Delta'= \Delta_1 \Downarrow^{\m{ig}} \xi= \Delta_2 \Downarrow^{\m{ig}} \xi\,\,\,\m{and}\,\,\, K'=y_\beta {:}A_2\langle c_2 , e_2 \rangle \Downarrow^{\m{ig}} \xi =  x_\alpha {:}A_1\langle c_1 , e_1 \rangle \Downarrow^{\m{ig}} \xi\,\, \,\m{and}\, \\
    \forall \mc{B}_1 \in \mb{H\text{-}CProvider}^\xi(\Delta_1), \mc{B}_2\in \mb{H\text{-}CProvider}^\xi(\Delta_2),\\
    \forall \mc{T}_1 \in \mb{H\text{-}CCLient}^\xi(x_\alpha{:}A_1\langle c_1 , e_1 \rangle),\, \mc{T}_2 \in \mb{H\text{-}CClient}^\xi(y_\beta{:}A_2\langle c_2 , e_2 \rangle).\\
     \m{if} (\mc{B}_1, \mc{B}_2)\in \mb{Hrel\text{-}IProvider}^\xi(\Delta'\backslash \Delta)\, \m{and}\,(\mc{T}_1, \mc{T}_2)\in \mb{Hrel\text{-}IClient}^\xi(K'\backslash K)\,\\
      \m{then}\,|\mc{B}_1\mc{D}_1\mc{T}_1| \approx_a |\mc{B}_2\mc{D}_2\mc{T}_2|.
  \end{array}
  \] 
\end{definition}
  \begin{theorem}\label{apx:thm:soundness}
    For all $\mc{D}_1\in \m{Tree}({\Delta_1} \Vdash x_\alpha {:}A_1)$ and $\mc{D}_2\in \m{Tree}({\Delta_2} \Vdash y_\beta {:}A_2)$, if
   $(\Delta_1 \Vdash \mc{D}_1:: x_\alpha {:}A_1[{c_1}]) \equiv^{\Psi_0}_{\xi} (\Delta_2 \Vdash \mc{D}_2:: y_\beta {:}A_2[{c_2}])$ then $(\Delta_1 \Vdash \mc{D}_1:: x_\alpha {:}A_1[{c_1}]) \approx_a^{\xi} (\Delta_2 \Vdash \mc{D}_2:: y_\beta {:}A_2[{c_2}])$.
    \end{theorem}
\begin{proof}
    Our goal is to show that     \[(\star_1)\,\,\,\Delta=\Delta_1 \Downarrow^{\m{cf}} \xi= \Delta_2,  \Downarrow^{\m{cf}} \xi\,\,\,\m{and}\,\,\, K=y_\beta {:}A_2\langle c_2 , e_2 \rangle \Downarrow^{\m{cf}} \xi =  x_\alpha {:}A_1\langle c_1 , e_1 \rangle \Downarrow^{\m{cf}} \xi\,\, \,\m{and}\]
     \[(\star'_1)\,\,\,\Delta'= \Delta_1 \Downarrow^{\m{ig}} \xi= \Delta_2 \Downarrow^{\m{ig}} \xi\,\,\,\m{and}\,\,\, K'=y_\beta {:}A_2\langle c_2 , e_2 \rangle \Downarrow^{\m{ig}} \xi =  x_\alpha {:}A_1\langle c_1 , e_1 \rangle \Downarrow^{\m{ig}} \xi\,\,\,\, \,\m{and}\]
    \[(\star_2)\,\,\,\forall \mc{B}_1 \in \mb{H\text{-}CProvider}^\xi(\Delta_1), \mc{B}_2\in \mb{H\text{-}CProvider}^\xi(\Delta_2), \] \[\quad \, \forall\mc{T}_1 \in \mb{H\text{-}CClient}^\xi(x_\alpha{:}A_1\langle c_1 , e_1 \rangle),
    \mc{T}_2 \in \mb{H\text{-}CClient}^\xi(y_\beta{:}A_2\langle c_2 , e_2 \rangle).\]
    \[\quad\, \m{if} (\mc{B}_1, \mc{B}_2)\in \mb{Hrel\text{-}IProvider}^\xi(\Delta)\, \m{and}\,(\mc{T}_1, \mc{T}_2)\in \mb{Hrel\text{-}IClient}^\xi(K)\,\]
    \[\quad \,\m{then}\,\mc{B}_1\mc{D}_1\mc{T}_1 \approx_a \mc{B}_2\mc{D}_2\mc{T}_2.\]
We know that 
\[
  \begin{array}{lll}
    {\mc{D}_1} \in \m{Tree}({\Delta_1} \Vdash x_\alpha {:}A_1\langle c_1, e_1 \rangle)\,\, \m{and}\,\,{\mc{D}_2} \in \m{Tree}({\Delta_2} \Vdash y_\beta {:}A_2\langle c_2, e_2 \rangle)\,\, \m{and}\\
    \Delta_1 \Downarrow^{\m{ig}} \xi= \Delta_2 \Downarrow^{\m{ig}} \xi=\Delta'\,\,\m{and}\,\, x_\alpha {:}A_1\langle c_1, e_1 \rangle \Downarrow^{\m{ig}} \xi= y_\beta {:}A_2\langle c_2, e_2 \rangle \Downarrow^{\m{ig}} \xi= K'\,\,\m{and}\\
    \forall \, \mc{B}_1 \in \mb{L\text{-}IProvider}^\xi(\Delta_1).\, \forall \, \mc{B}_2 \in \mb{L\text{-}IProvider}^\xi(\Delta_2).\\
     \,\forall \mc{T}_1 \in \mb{L\text{-}IClient}^\xi(x_\alpha {:}A_1\langle c_1, e_1 \rangle). \,\forall \mc{T}_2 \in \mb{L\text{-}IClient}^\xi(y_\beta {:}A_2\langle c_2, e_2 \rangle).
  \end{array}  
  \]
  
  \[\begin{array}{l}  \forall\,m.\,(\mc{B}_1\mc{D}_1\mc{T}_1, \mc{B}_2\mc{D}_2\mc{T}_2) \in \mc{E}^\xi_{\Psi_0}\llbracket{\Delta'} \Vdash {K'} \rrbracket^{m},\,
  \m{and} \\[1pt]
  \forall\,m.\,(\mc{B}_2\mc{D}_2\mc{T}_2, \mc{B}_1\mc{D}_1\mc{T}_1) \in \mc{E}^\xi_{\Psi_0}\llbracket {\Delta'} \Vdash {K'} \rrbracket^{m}.
  \end{array}
  \]

  The first part of the goal ($\star_1$) is starightforward from the assumption, the defintions of $\Downarrow^{\m{cf}}$ and $\Downarrow^{\m{ig}}$ and the fact that the min integrity ($e$) of a channel is always bounded by its max secrecy ($c$) in well-typed processes, i.e., $e \sqsubseteq c$ and thus $c \sqsubseteq \xi \rightarrow e \sqsubseteq \xi$. We get $\Delta= \Delta_1 \Downarrow^{\m{cf}} \xi= \Delta_2,  \Downarrow^{\m{cf}} \xi\,\,\,\m{and}\,\,\, K= y_\beta {:}A_2\langle c_2 , e_2 \rangle \Downarrow^{\m{cf}} \xi =  x_\alpha {:}A_1\langle c_1 , e_1 \rangle \Downarrow^{\m{cf}} \xi$, such that $\Delta \subseteq \Delta'$ and $K \subseteq K'$. The part ($\star'_1$) is given by the assumption.

  To prove  the second part of the goal ($\star_2$) we apply a forall introduction to get arbitrary configurations
\(\mc{B}_1 \in \mb{H\text{-}CProvider}^\xi(\Delta_1)\), $\mc{B}_2\in \mb{H\text{-}CProvider}^\xi(\Delta_2)$, \(\mc{T}_1 \in \mb{H\text{-}CClient}^\xi(x_\alpha{:}A_1\langle c_1 , e_1 \rangle)\), $\mc{T}_2 \in \mb{H\text{-}CClient}^\xi(y_\beta{:}A_2\langle c_2 , e_2 \rangle)$. We assume that $(\mc{B}_1, \mc{B}_2)\in \mb{Hrel\text{-}IProvider}^\xi(\Delta'\backslash \Delta)$ and $(\mc{T}_1, \mc{T}_2)\in \mb{Hrel\text{-}IClient}^\xi(K'\backslash K)$. We observe that by the tree invariant for any channel in $\Delta_i$, the min integrity ($e$) is always bounded by the max secrecy ($c$), i.e., $e \sqsubseteq c$, and thus $c \sqsubseteq \xi \rightarrow e \sqsubseteq \xi$ and $e \not \sqsubseteq \xi \rightarrow c \not\sqsubseteq \xi$. Similarly, $e_1 \sqsubseteq c_1$ and $e_2 \sqsubseteq c_2$. Thus we can rewrite $\mc{B}_i$ and $\mc{T}_i$ as $\mc{B}_i=\mc{B}'_i\mc{B}''_i$ and $\mc{T}_i=\mc{T}'_i\mc{T}''_i$, such that  $\mc{B}'_1 \in \mb{L\text{-}IProvider}^\xi(\Delta_1), \mc{B}'_2\in \mb{L\text{-}IProvider}^\xi(\Delta_2), \mc{T}'_1 \in \mb{L\text{-}ICLient}^\xi(x_\alpha{:}A_1\langle c_1 , e_1 \rangle),\, \mc{T}'_2 \in \mb{L\text{-}IClient}^\xi(y_\beta{:}A_2\langle c_2 , e_2 \rangle)$. 

  Moreover, by the definition of $ \mb{Hrel\text{-}IProvider}^\xi(\Delta'\backslash \Delta)$ and $\mb{Hrel\text{-}IClient}^\xi(K'\backslash K)$, we get that all trees in $\mc{B}''_1$ and  $\mc{B}''_2$ are related by the logical relation, and also $\mc{T}''_1$, and $\mc{T}''_2$ are related by the logical relation (for every $m$).

  With $\mc{B}'_i$ and $\mc{T}'_i$, we instantiate the for all quantifier in the assumption and we get
  \[\begin{array}{l}  \forall\,m.\,(\mc{B}'_1\mc{D}_1\mc{T}'_1, \mc{B}'_2\mc{D}'_2\mc{T}'_2) \in \mc{E}^\xi_{\Psi_0}\llbracket{\Delta'} \Vdash {K'} \rrbracket^{m},\,
    \m{and} \\[1pt]
    \forall\,m.\,(\mc{B}'_2\mc{D}_2\mc{T}'_2, \mc{B}'_1\mc{D}_1\mc{T}'_1) \in \mc{E}^\xi_{\Psi_0}\llbracket {\Delta'} \Vdash {K'} \rrbracket^{m}.
    \end{array}
    \]

    Now, we can use relatedness of the trees in $\mc{B}''_1$ and $\mc{B}''_2$ and relatedness of $\mc{T}''_1$ and $\mc{T}''_2$ and the compositionality lemma~\ref{lem:puttogether} to get 

    \[\begin{array}{l}  \forall\,m.\,(\mc{B}''_1\mc{B}'_1\mc{D}_1\mc{T}''_1\mc{T}'_1, \mc{B}''_2\mc{B}'_2\mc{D}'_2\mc\mc{T}''_2{T}'_2) \in \mc{E}^\xi_{\Psi_0}\llbracket{\Delta} \Vdash {K} \rrbracket^{m},\,
        \m{and} \\[1pt]
        \forall\,m.\,(\mc{B}''_2\mc{B}'_2\mc{D}_2\mc{T}''_2\mc{T}'_2, \mc{B}''_2\mc{B}'_1\mc{D}_1\mc{T}''_2\mc{T}'_1) \in \mc{E}^\xi_{\Psi_0}\llbracket {\Delta} \Vdash {K} \rrbracket^{m}.
        \end{array}
        \]
Observe that all channels in $\Delta$ and $K$ have lower max secrecy $c$ than the observer $\xi$. In this special case, the term interpretation of our integrity logical relation implies the term interpretation of the logical relation for equivalence of programs that has been introduced in prior work~\cite{BalzerARXIV2023}, and proved to be sound with respect to bisimulation.
We apply~\Cref{lem:priorrel}, on our goal to get 
\[\begin{array}{l}  \forall\,m.\,(|\mc{B}''_1\mc{B}'_1\mc{D}_1\mc{T}''_1\mc{T}'_1|, |\mc{B}''_2\mc{B}'_2\mc{D}'_2\mc\mc{T}''_2{T}'_2|) \in \mc{E}^c\llbracket{|\Delta|} \Vdash {|K|} \rrbracket^{m},\,
  \m{and} \\[1pt]
  \forall\,m.\,(|\mc{B}''_2\mc{B}'_2\mc{D}_2\mc{T}''_2\mc{T}'_2|, |\mc{B}''_2\mc{B}'_1\mc{D}_1\mc{T}''_2\mc{T}'_1|) \in \mc{E}^c\llbracket {|\Delta|} \Vdash {|K|} \rrbracket^{m}.
  \end{array}
  \]
Now, we can apply~\Cref{apx:lem:soundness-completeness} to get $|\mc{B}_1\mc{D}_1\mc{T}_1| \approx_a |\mc{B}_2\mc{D}_2\mc{T}_2|$ and the proof is complete.
\end{proof}

\begin{figure}
  \centering
  \begin{small} 
   \input{appendix-arxiv/apx-logical-relation-i.tex}
  \end{small}
  \caption{Term and value interpretation of the recursive session logical relation for equivalence of programs with confidentiality interface from prior work~\cite{BalzerARXIV2023}.}
  \label{fig:rslr}
  \end{figure}

  To be self-contained, we have included the term and value interpretations of the logical relation introduced in prior work~\cite{BalzerARXIV2023} in~\Cref{fig:rslr}. This logical relation from previous work is defined considering the vanilla session typed system, meaning that there are no security annotations in the configurations $\mc{D}_1$, $\mc{D}_2$, and the interface $\Delta$ and $K$. The authors use an erasure function to erase the confidentiality annotations from the configurations and the context. We assume a similar erasure function $|\_|$, which removes all the confidentiality and integrity annotations from the configuration and the context. For the sake of simplicity, in the next two lemmas, we use the same notations for configurations (e.g. $\mc{D}$) and the context (e.g. $\Delta$) before and after the erasure.
The authors use the tree predicate, to state that both configurations are well-typed in the vanilla session type system.

  The next theorem demonstrates that the term interpretation of our integrity logical relation implies the term interpretation of the logical relation introduced in prior work~\cite{BalzerARXIV2023}. 

  \begin{lemma}\label{lem:priorrel}
  If $\forall\,m.(\mc{D}_1, \mc{D}_2) \in \mc{E}^\xi_{\Psi_0}\llbracket{\Delta} \Vdash {K} \rrbracket^{m}$ and every channel in $\Delta$ and $K$ are observable, i.e., if $x_\alpha{:}A_ \langle c_1, e_1 \rangle  \in \Delta, K$ then $c_1 \sqsubseteq \xi$, then  $\forall\,m.(|\mc{D}_1| , |\mc{D}_2|) \in \mc{E}^c\llbracket{|\Delta|} \Vdash {|K|} \rrbracket^{m}$, where $|\_|$ is the erasure function that removes all the confidentiality and integrity annotations from the configuration and the context.
  \end{lemma}
  \begin{proof}
  The proof is by induction on the index $m$. For $m=0$, the proof is straightforward as all processes in our system are well-typed in the vanilla type system.
  For $m=m'+1$, we observe that it's enough to show (a) if $(\mc{D}_1, \mc{D}_2) \in \mc{V}^\xi_{\Psi_0}\llbracket{\Delta} \Vdash {K} \rrbracket^{m'+1}_{\cdot;y_\alpha}$ then $(|\mc{D}_1|, |\mc{D}_2|) \in \mc{V}^c\llbracket{|\Delta|} \Vdash {|K|} \rrbracket^{m'+1}_{\cdot;y_\alpha}$ and (b) if $(\mc{D}_1, \mc{D}_2) \in \mc{V}^\xi_{\Psi_0}\llbracket{\Delta} \Vdash {K} \rrbracket^{m'+1}_{\cdot;y_\alpha}$ then $(|\mc{D}_1|, |\mc{D}_2|) \in \mc{V}^c\llbracket{|\Delta|} \Vdash {|K|} \rrbracket^{m'+1}_{\cdot;y_\alpha}$. The proof of the two statements proceed by case analysis on the type of $y_\alpha$. Again, the proof is straightforward by observing that the definition of all cases in our integrity logical relation implies the ones in~\Cref{fig:rslr}. In particular, for the cases responsible for communicating labels, $(r_2)$, $(r_3)$, $(l_2)$, and $(l_3)$ in~\Cref{fig:rel_value_right,fig:rel_value_left} we know that the same lables $k_1$ and $k_2$ are sent/received since all the channels in $\Delta$ and $K$ are observable. This means that they imply lines $(2)$, $(3)$, $(7)$, and $(8)$ in~\Cref{fig:rslr} resp.
  \end{proof}

  Finally, we state the soundness result from the TR of~\cite{BalzerARXIV2023} that shows the term interpretation of the confidentiality logical relation is sound with respect to the asynchronous bisimulation. 
  \begin{lemma}[Lemma 24 from the TR of~\cite{BalzerARXIV2023}]\label{apx:lem:soundness-completeness}
    Consider a pair of vanilla session-typed forests $\mc{C}_1$, $\mc{C}_2$ with $\Delta \Vdash \mc{C}_1 :: \Delta'$ and $\Delta \Vdash \mc{C}_2 :: \Delta'$ (in the vanilla session typed system) consisting of multiple session-typed trees indexed in the set $I$, i.e., $\mc{C}_i=\{\mc{C}^j_i\}_{j\in I}$  and $\Delta^j\Vdash \mc{C}^j_i :: K^j$, for $i \in \{1,2\}$, with $\Delta= \{\Delta^j\}_{j\in I}$ and $\Delta'= \{K^j\}_{j\in I}$. We have:
     \(\forall j\in I. \forall m.\, (\mc{C}^j_1, \mc{C}^j_2) \in \mc{E}\llbracket \Delta^j \Vdash K^j\rrbracket^m \,\, \m{and}\,\, (\mc{C}^j_2, \mc{C}^j_1) \in \mc{E}\llbracket \Delta^j \Vdash K^j\rrbracket^m \,\,\m{iff} \,\,\mc{C}_1 \approx_a \mc{C}_2.\)
    \end{lemma}
    \begin{proof}
    For the proof please see Lem. 24 from the technical report of~\cite{BalzerARXIV2023}.
    \end{proof}

 \pagebreak

\section{Implementation}\label{apx:implementation}

\subsection{Overview}

We implemented a \emph{type checker} for $\lang$, whose features we now briefly describe. 
Our base calculus is an intuitionistic session-typed process calculus, for which type checking is well 
understood and implementations exist \cite{GayHoleARTICLE2005,DasPfenningFSCD2020}. The novelty of our work lies 
in our algorithmic treatment of the IFC fragment of the calculus, which as far as we know is new 
and has not been attempted yet.

\subsubsection{User-specified security semilattice.} Our type checker allows the user to specify a 
concrete security join-semilattice $\Psi_0$ (see \Cref{apx:def:lattice}) and verify its validity.  It will reject security levels that 
fail the semilattice axioms. To specify $\Psi_0$, the user must prefix the source code with a
\texttt{secrecy} section, where the user must list the elements of the lattice in decreasing order of secrecy.
In other words, elements must be listed by their topological order,
if we were to view the lattice as a directed graph where arrows point from high secrecy to low secrecy. For each 
element, the programmer must then list its immediate predecessors.
For example, we indicate our banking example lattice as follows, where \code{#bank} stands for the top element:

\begin{codeblock}
secrecy
  #bank  < ()
  #alice < (#bank)
  #bob   < (#bank)
  #guest < (#alice, #bob)
end
\end{codeblock}

% \begin{flushleft}
% \begin{small}
% \begin{minipage}[t]{\textwidth}
% \begin{tabbing}
% \texttt{secrecy} \\ 
% \quad \=\texttt{\#bank  < ()} \\
% \> \texttt{\#alice < (\#bank)} \\
% \> \texttt{\#bob   < (\#bank)} \\
% \> \texttt{\#guest < (\#alice, \#bob)} \\
% \texttt{end}
% \end{tabbing}
% \end{minipage}
% \end{small}
% \end{flushleft}

\begin{lemma}
Suppose the user lists the elements of an $n$-element lattice with a sequence $(e_0, \cdots, e_{n-1})$. 
Define $L_i = (e_0, \cdots, e_i)$ for $0 \leq i < n$ ($L_i$ is the preorder obtained by taking the 
first $i$ elements specified along with arrows between them). Then the full preorder $L_{n-1}$ is a 
join-semilattice iff for all $i < n$, $L_i$ is a join-semilattice.
\end{lemma}
\begin{proof}
The \emph{if} part is trivial (take $i = n - 1$). 
It remains to be shown that if $L_{n-1}$ is a semilattice then for all $i < n$, $L_i$ is a semilattice. 
We note that we will show if $L_i$ is a semilattice then $L_{i-1}$ is a semilattice. 
We sketch the proof, following a similar argument as in \cite{KyunoARTICLE1979}: $L_{i-1}$ is obtained 
by removing one of the minimal elements from $L_i$. Therefore, all semilattice axioms are preserved for 
the remaining elements, concluding $L_{i-1}$ is by definition a semilattice. From here by an argument 
of induction (base case is just $L_{n-1}$, we show that for all $i < n$, $L_i$ is a join-semilattice).
\end{proof}

The lemma proves that any semilattice can be specified by listing its elements in topological order, 
and verifying that each increment extends the prior semilattice to form another semilattice. This can be done algorithmically
within reasonable computational complexity: at every step we keep track of the minimal elements 
in the current semilattice. For every additional element we ensure that it has a join with every 
existing minimal element. This can be done efficiently by computing the 
\emph{least common ancestor} and check uniqueness (a problem with known fast solutions \cite{SleatorTarjanJCSS1983}). 
The type checker implemented a more expensive but straightforward method of directly checking 
the definition of supremum, because the number of elements in a lattice is generally small.

\subsubsection{Lattice theories for polymorphic process definitions} Our type checker allows the programmer 
to specify any number of \emph{lattice theories}, as defined in \Cref{apx:def:lattice}. A lattice theory postulates a series of security terms $c, d \in \mc{S}$,
ranging over the elements of the concrete security lattice as well as a set of security variables $\psi \in \mathcal{V}$, along with a list of 
relations $E \in \mathcal{E}$ that the variables must satisfy.
As defined in \Cref{apx:def:lattice}, $E$ may only be of the form $c \sqsubseteq d$,
but the implementation also allows programmers to specify 
equality between formulas (\eg $c \sqcup d = e \sqcup \texttt{Alice}$), 
which is desugared into two $\sqsubseteq$ relations 
($c \sqcup d \sqsubseteq e \sqcup \texttt{Alice}$, and vice versa). Lattice theories are named.

A process is then parametrized over a lattice theory of the programmer's choice.
By postulating 
a lattice theory a process gains access to the security variables within, and therefore can specify
security terms involving these variables.

During the type checking we will encounter relations that we must establish. The type checker 
implements a search algorithm that automatically derives the desired relation from
the induction definition in \Cref{apx:fig:latcstr-def},
and if it fails to do the process is rejected.
The relation it can attempt to show includes $\sqsubseteq$ and equality of security terms (again 
treated as conjunction of separate $\sqsubseteq$ relations). The latter is useful because this allows 
the programmer to write down intuitively ``equivalent but not literally equal'' secrecy terms (\eg 
$\textsf{Alice} \sqcup \textsf{Bob}$ versus $\textsf{Bank}$) and have it type checked.
The same mechanism is also used to validate security theory substitution during spawning.

The type checker will insist that the user pick a main process and provide a substitution of its security 
theory from the concrete lattice (\eg \texttt{execute P[\#bank, \#alice, \#bob, \#guest]}), thus establishing validity of the security theories of all subsequently 
referenced process definitions. This is crucial as the user can postulate secrecy theories that are 
inconsistent with the concrete lattice. Processes that assume these theories \emph{will} leak information 
if executed. However, our type system prevents their execution because no substitution from the concrete 
lattice exists. This is akin to the fact that types such as $\forall t. t$ is uninhabited in polymorphic 
$\lambda$-calculus.

\subsubsection{Mutual recursive definition} The type checker also implements limited support for 
equi-recurisve definitions, in line with the rules laid out in the main body of the paper.
In particular, for (recursive) type definitions, although definitions can refer to each other, they must 
all be individually \emph{contractive}. Furthermore, it should be emphasized that each definition is 
a recursive data type definition, not to be confused with a \emph{type definition} 
(i.e. \texttt{type t = ...} declaration in Standard ML). 
Because all definitions are implicitly mutually recursive, the type checker will reject attempts to shadow
previous definitions.

\subsubsection{Synchronization pattern check} The type checker implements the synchronization check described 
in \Cref{apx:sec:sync-patterns}. The check will first attempt to apply the unsynchronized rules.
If any of the rules apply, the type checker commits to that choice.
If none of the rules apply, the type checker attempts to apply the synchronized rules,
rejecting the program when this fails.
This algorithm is fine because the unsynchronized rules are \emph{invertible}.

The synchronization type checker also employ a few optimizations that are justified by transitivity and 
symmetry of the $\barbeq{P}{d}{f}{Q}$ relation. Example 1: in type checking 
rule $\oplus_L$, instead of checking all pairs of branches, we will arbitrarily pick one branch and 
check this against all other branches, relying on transitivity and symmetry to conclude the rest. 
Example 2: in checking \textsc{RcvLab}, we will pick one branch from $P_j$ and $Q_k$ each and 
compare them. This is safe because the type checker will establish the relation within branches of 
$P_j$ and $Q_k$ respectively, and by an argument of transitivity and symmetry we conclude the rest.

% \paragraph{Web interface} The type checker is written in OCaml and compiled to a Javascript application, 
% available \url{https://tripack45.github.io/session-impl/index.html}. It offers a list of existing examples and produces type checking outputs. The 
% application itself is hosted statically by the GitHub Pages service.

\subsubsection{Implementation} The type checker is written in OCaml.
It offers a list of existing examples and produces type checking outputs. The 
application itself is hosted statically by the GitHub Pages service.

\subsubsection{Guide to the concrete syntax}
Here we describe the concrete syntax used in the implementation.
The syntax differs from the one in the main part of the paper,
which is more compact for the purpose of formal development.

\subsubsection{General structure} A valid input file of $\lang$ must contain the following sections, 
in this order, 

\begin{itemize}
\item A specification of the concrete lattice (\code{secrecy ... end});
\item One or more lattice theories (\code{theory[...] ... end});
\item A session type definition signature (\code{stype signature ... end});
\item A process definition signature (\code{proc signature ... end});
\item And the main process to be executed (\code{exec P[...]})
\end{itemize}

The language uses ML-style comments \texttt{(* comment *)}, and it is entirely 
whitespace (indentation) insensitive.

The language has no restriction in terms of identifier capitalization. However, it is 
conventional to capitalize the session type names, process names and labels.

We will now introduce the ``official'' syntax of the language. The author acknowledges that 
people often have their own preferences, therefore the language supports a few ``unofficial'' 
variations for a subset of the keywords.

\subsubsection{Concrete lattice} Concrete lattices are specified by entries 
within \code{secrecy ... end}. Each entry defines a concrete secrecy level
and immediate provide its nearest ancestors, listed in the shape of a tuple.
Concrete secrecy levels must start with a pound (\code{#}). The levels 
must be listed in topological order (i.e. they may only refer to levels that 
are already present).

The following defines the concrete lattice used for our running example. 
Here \code{#bank} is known to be the top element because it has no ancestor.

\begin{codeblock}
secrecy
  #bank  < ()
  #alice < (#bank)
  #bob   < (#bank)
  #guest < (#alice, #bob)
end
\end{codeblock}

\subsubsection{Secrecy theories} The programmer may postulate a number of lattice 
theories. For each theory the programmer must provide a name and list the 
security variables. Then the programmer may list one or more relations, separated by 
a semi colon.  Each equation relates two \emph{security terms}.
% \yytodo{need to unite terminology}. 

The following chart lists the syntax for expressing security terms. 

\begin{synchart}{syn:sec}
\text{Security terms} & c, d & \bnfdef & \textbf{alice} & \code{#alice}  & \text{concrete levels} \\
                      &      & \bnfalt & a              & \code{a}       & \text{variables} \\
                      &      & \bnfalt & c \sqcup d     & \code{c |_| d} & \text{join of terms} \\
\end{synchart}

For example, the following theory postulate a theory \code{Th} with variables \code{a} and \code{b}.
From the relations postulated the type checker is capable of deducing $a \sqsubseteq b$.

\begin{codeblock}
theory Th[a, b]
  a <= #alice;
  #alice |_| #bob <= b;
end
\end{codeblock}

\subsubsection{Type signature} Session type signature block contains definitions of 
(possibly mutual recursive) session types. The block itself contains zero or more session type 
definitions. 

The following block defines two mutually recursive session types $A$ and $B$. 

\begin{codeblock}
stype signature
    stype A = 1 * B
    stype B = 1 * A
end
\end{codeblock}

The following chart lists concrete syntax for session types. 

\begin{synchart}{syn:sec}
\text{Session types}  & A, B & \bnfdef & A              & \code{A}                     & \text{name of session type} \\
                      &      & \bnfalt & \one           & \code{unit}                  & \text{closing session} \\
                      &      & \bnfalt & \intchoice{A}  & \code{ + \{l -> A_l | ...\}} & \text{internal choice} \\
                      &      & \bnfalt & \extchoice{A}  & \code{ & \{l -> A_l | ...\}} & \text{external choice} \\
                      &      & \bnfalt & \chanout{A}{B} & \code{A * B}                 & \text{channel output} \\
                      &      & \bnfalt & \chanin{A}{B}  & \code{A -o B}                & \text{channel input} \\
\end{synchart}

\subsubsection{Process signature} Process signature block contains definition of
(possibly mutual recursive) process definitions. The block contains zero or more process 
definitions. To define process $X$:
$$
\ptypb{y_1 {:} B_1\dlabel{\psi_1}{\omega_1}, \dots, y_n {:} B_n\dlabel{\psi_n}{\omega_n}}{X = P}{\psi_0}{\omega_0}{x}{A}{\psi}{\omega}
$$

The programmer writes down the following signature

\begin{codeblock}
proc X[Psi] provide (x : A[psi, w]) 
  using (y1 : B1[psi1, w1], ...) at [psi0, w0] = ...
\end{codeblock}

where $X$ is the process name; \code{Psi} is the name of the lattice theory it assumes; \code{x} 
is the name of its offering channel and \code{A[psi, w]} is confidentiality and integrity of \code{x};
\code{[psi0, w0]} is the level that the process is spawned at, and \code{y1 : B1[psi1, w1], ...} are 
its arguments, and the dots stands for the code $P$ of the process. The $\Sigma$ context is 
the entire process definition context.

If a channel does not utilize regrading (that is its integrity coincides with is confidentiality), 
the programmer may specify \code{A[psi]} meaning \code{A[psi, psi]}. This is just a quality-of-life 
feature.

The following table lists syntax for process terms (spawning described in text separately):

\begin{synchart}{syn:sec}
\text{Proc term} & P, Q      & \bnfdef & \labsnd{x}{k}{P}   & \code{select x K; P}           & \text{send label} \\
                      &      & \bnfalt & \labrcv{x}{P}      & \code{case x \{l -> P_x | ...\}} & \text{recieve label} \\
                      &      & \bnfalt & \chnsnd{y}{x}{P}   & \code{send y to x; P}          & \text{send channel} \\
                      &      & \bnfalt & \chnrcv{z}{x}{P}   & \code{x = receive z; P}        & \text{receive channel} \\
                      &      & \bnfalt & \wait{x}{P}        & \code{wait x; P}               & \text{wait on channel} \\
                      &      & \bnfalt & \cls{x}            & \code{close x}                 & \text{close channel} \\
                      &      & \bnfalt & \fwd{x}{y}         & \code{forward y to x}          & \text{forward channel} \\
\end{synchart}

Within the process body it has access to the security variables postulated in its signature.

To spawn $\spwn{x}{c}{e}{X}{\D}{c_0}{e_0}{Q}$, the programmer writes 
\begin{codeblock}
instantiate x = X[c, ...](y, ...); Q
\end{codeblock}
Here the programmer instantiates the process of definition \code{X}, 
supplying the lattice theory substitution $\gamma$ given by \code{c, ...}, 
one security term for every variable in the theory assumed by 
\code{X}, in the order they are listed in the theory. The security term can use variables defined in the 
spawner's theory. The arguments $\D$ are supplied in the parentheses as \code{x, ...}. The type checker automatically 
infers the confidentiality and initial running confidentiality of the spawnee.

The following signature defines two processes that spawn each other. This assumes the availability 
of the security variable \code{a} in the theory \code{Th}. Both processes do not require arguments, 
and do not use regrading.

\begin{codeblock}
proc signature
    proc M1[Th] provide (x : unit[a]) using () at [a] = 
	instantiate y = M2[]();
	forward y to x

    proc M2[Th] provide (x : unit[a]) using () at [a] = 
	instantiate y = M1[]();
	forward y to x
end
\end{codeblock}

\subsubsection{Main process} The programmer must designate a process $P$ as the main process 
and provide a substitution from the concrete lattice into the lattice theory of $P$ 
(the security terms \code{c}, \code{d} can only refer to the concrete lattice). Doing 
so completes the input file.

\begin{codeblock}
exec P[c, d, ...]
\end{codeblock}

\section{extra examples}\label{apx:sec:example}

\subsection{Polymorphic banking example}
\begin{center}
    \begin{small}
    \begin{minipage}[t]{1\textwidth}
    \begin{tabbing}
    %%%bank
    $\Psi;y_1{:}\m{customer}\maxsec{\dlabel{\psi}{\mb{guest}}},  y_2{:}\m{customer}\maxsec{\dlabel{\psi'}{\mb{guest}}}, w_1{:}\m{account}\maxsec{\dlabel{\psi}{\psi}}, w_2{:}\m{account}\maxsec{\dlabel{\psi'}{\psi'}} \vdash \m{Bank} :: x{:}\m{1}\maxsec{\dlabel{\mb{bank}}{\mb{bank}}}$ \\
    $x \leftarrow \m{Bank}\leftarrow y_1, y_2, w_1, w_2=$\\
    $\mb{case} \, y_1 \, $\= $(\mathit{tok_{white}} \Rightarrow y_1.\mathit{succ}; \, $ \= $(u\leftarrow \m{Statement}[\gamma]\leftarrow \cdot);\, \mb{send}\, u\, w_1; (x'\leftarrow \m{Bank}[\gamma']\leftarrow y_2, y_1, w_2, w_1);\; x \leftarrow x'$\\
    $\quad \quad\, $\>  $\mid \mathit{tok_{black}} \Rightarrow y_1.\mathit{fail};(x'\leftarrow \m{Bank}[\gamma']\leftarrow y_2, y_1, w_2, w_1);\; x \leftarrow x'\,)\runsec{@\dlabel{\mb{guest}}{\mb{guest}}}$\\[4pt]

    $\Psi;\cdot \vdash \m{Customer} :: y{:}\m{customer}\maxsec{\dlabel{\psi}{\mb{guest}}}$ \\
    $ y\leftarrow \m{Customer}\leftarrow \cdot=$\\
    $y.\mathit{tok_{white}}; \mb{case} \, y \, $\= $(\mathit{succ} \Rightarrow  (y'\leftarrow \m{Customer}[\gamma]\leftarrow \cdot);\; y \leftarrow y'$\\
    $\quad \quad\, $\>  $\mid \mathit{fail} \Rightarrow (y'\leftarrow \m{Customer}[\gamma]\leftarrow \cdot);\; y \leftarrow y')\;\;\runsec{@\dlabel{l}{\mb{guest}}}$\\[4pt]
    $\Psi;\cdot \vdash \m{Statement} :: u{:}\m{transfer}\maxsec{\dlabel{\psi}{\psi}}$ \\
    $ u\leftarrow \m{Statement}\leftarrow \cdot=\;\;$
    $u.\mathit{transaction}; \mb{close} \, u\;\; \runsec{@\dlabel{\psi}{\psi}}$\\[4pt]

    $\Psi_0;  w_1{:}\m{customer}\maxsec{\dlabel{\mb{alice}}{\mb{alice}}}, w_2{:}\m{customer}\maxsec{\dlabel{\mb{bob}}{\mb{bob}}} \vdash \m{Main} :: x{:}1\maxsec{\dlabel{\mb{bank}}{\mb{bank}}}$ \\
    $ x\leftarrow \m{Main}\leftarrow w_1, w_2=$\\
    $(y_1\leftarrow \m{Customer}[\gamma_1]\leftarrow \cdot);\; (y_2\leftarrow \m{Customer}[\gamma'_1]\leftarrow \cdot);$\\$\;\;(x'\leftarrow \m{Bank}[\gamma_1]\leftarrow y_1,y_2, w_1,w_2);\; x \leftarrow x' \; \runsec{@\dlabel{\mb{guest}}{\mb{guest}}}$
\end{tabbing}
    \end{minipage}
    \end{small}
\end{center}    

\[\begin{array}{llll}
    \gamma:=\{ \psi \mapsto \psi, \psi' \mapsto \psi'\} & 
    \gamma_1:=\{ \psi \mapsto \mb{alice}, \psi' \mapsto \mb{bob}\}\\
    \gamma':= \{\psi \mapsto \psi', \psi' \mapsto \psi\}  & \gamma'_1:=\{\psi \mapsto \mb{bob}, \psi' \mapsto \mb{alice}\}
\end{array}\]

% \bibliographystyle{ACM-Reference-Format}
% \bibliography{../ref}

\end{document}